\documentclass[11pt]{report}

\usepackage{graphicx}
\graphicspath{ {images/} }

\usepackage[a4paper,height=190mm,width=142mm]{geometry}

\usepackage[utf8]{inputenc}
\usepackage{amsthm}
\usepackage{amssymb}
\usepackage{latexsym}
\usepackage{amsmath}
\usepackage{amsfonts}
\usepackage{mathrsfs}
\usepackage{amstext}

\usepackage{comment}

\usepackage{enumitem}
\setitemize{noitemsep,topsep=2pt,parsep=2pt,partopsep=2pt}
\usepackage{mathtools}

\DeclarePairedDelimiter{\floor}{\lfloor}{\rfloor}
\mathtoolsset{showonlyrefs}
\usepackage{stmaryrd}
\usepackage{tocloft}
\usepackage{cite}
\usepackage{graphicx}
\usepackage{xcolor}
\usepackage{color,hyperref}
\definecolor{darkblue}{rgb}{0.0,0.0,0.5}
\hypersetup{
  linkcolor  = darkblue,
  citecolor  = darkblue,
  urlcolor   = darkblue,
  colorlinks = true,
}

\renewcommand{\text}{\mathrm}

\usepackage{cancel}
\newcommand{\red}{\color{red}}

\usepackage{tikz}
\usetikzlibrary{calc}
\usepackage{ifthen}
\usepackage{xcolor}
\usepackage[normalem]{ulem}  

\newcommand{\hstrike}[1]{%
  \ifmmode
    \tikz[baseline=(X.base)]{
      \node[
        inner sep=0pt, outer sep=0pt, anchor=base,
        text=red
      ] (X) {\ensuremath{#1}};
      \draw[line width=0.4pt] 
        ($(X.south west)!0.5!(X.north west)$) -- 
        ($(X.south east)!0.5!(X.north east)$);
    }%
  \else
    \textcolor{red}{\sout{#1}}%
  \fi
}

\usepackage{mathtools}  

\usepackage{lipsum}

\newcommand{\Om}{\Omega}
\newcommand{\mcF}{\mathcal{F}}
\newcommand{\hil}{\mathcal{H}} 
\newcommand{\lan}{\langle}
\newcommand{\ran}{\rangle}

\newcommand{\pnew}{p'}

\newcommand{\bN}{\mathbb{N}} 
\newcommand{\bY}{\mathbb{Y}} 
\newcommand{\bE}{\mathbb{E}}
\newcommand{\bC}{\mathbb{C}}
\newcommand{\bZ}{\mathbb{Z}}
\newcommand{\bR}{\mathbb{R}}

\newcommand{\bS}{\mathbb{S}}
\newcommand{\bP}{\mathbb{P}}

\newcommand{\cA}{\mathcal{A}}
\newcommand{\cS}{\mathcal{S}}
\newcommand{\cP}{\mathcal{P}}
\newcommand{\cD}{\mathcal{D}}
\newcommand{\cR}{\mathcal{R}}
\newcommand{\cT}{\mathcal{T}}
\newcommand{\cF}{\mathcal{F}}
\newcommand{\cN}{\mathcal{N}}
\newcommand{\cK}{\mathcal{K}}
\newcommand{\cM}{\mathcal{M}}
\newcommand{\cX}{\mathcal{X}}
\newcommand{\cY}{\mathcal{Y}}
\newcommand{\cL}{\mathcal{L}}
\newcommand{\cH}{\mathcal{H}}
\newcommand{\cU}{\mathcal{U}}
\newcommand{\cZ}{\mathcal{Z}}

\newcommand{\sS}{\mathscr{S}}
\newcommand{\sD}{\mathscr{D}}

\newcommand{\ri}{\mathrm{i}}
\newcommand{\rx}{\mathrm{x}}
\newcommand{\ry}{\mathrm{y}}

\newcommand{\rz}{\mathrm{z}}
\newcommand{\rd}{\mathrm{d}}
\newcommand{\re}{\mathrm{e}}

\newcommand{\rc}{\mathrm{c}}
\newcommand{\rb}{\mathrm{b}}

\newcommand{\supp}{\mathrm{supp}}
\makeatletter
\newcommand*\botimes{{\mathpalette\botimes@{1.5}}}
\newcommand*\botimes@[2]{\mathbin{\vcenter{\hbox{\scalebox{#2}{\hspace{0.0mm}$\m@th#1\otimes$\hspace{0.0mm}}}}}}
\newcommand*\Cdot{{\mathpalette\Cdot@{.6}}}
\newcommand*\Cdot@[2]{\mathbin{\vcenter{\hbox{\scalebox{#2}{\hspace{0.5mm}$\m@th#1\bullet$\hspace{0.5mm}}}}}}
\makeatother

\usepackage[T1]{fontenc}
\usepackage[utf8]{inputenc}

\newcommand{\hf}{\hat{f}}
\newcommand{\hg}{\hat{g}}

\newcommand{\mb}{\mathbb}

\newtheorem{thm}{Theorem}[section]
\newtheorem{lem}[thm]{Lemma}
\newtheorem{prop}[thm]{Proposition}

\newtheorem{rem}[thm]{Remark}

\theoremstyle{dfn}
\newtheorem{dfn}[thm]{Definition}

\newtheorem{Remark}[thm]{Remark}
\newtheorem{Example}[thm]{Example}

\numberwithin{equation}{section}
\numberwithin{thm}{section}
\newcommand{\hK}{\hat{K}_{R,N}}
\newcommand{\ti}{\tilde}
\newcommand{\beqa}{\begin{eqnarray}}   
\newcommand{\eeqa}{\end{eqnarray}} 
\newcommand{\om}{\omega}
\newcommand{\mrm}{\mathrm}
\newcommand{\De}{\Delta}
\newcommand{\2}{\!\!\!&}
\newcommand{\one}{\mathmybb{1}}

\usepackage{chngcntr}

\usepackage{bbm}

\DeclareMathAlphabet{\mathmybb}{U}{bbold}{m}{n}
\newcommand{\1}{\mathmybb{1}}

\usepackage{tikz-cd}
\usepackage{caption}
\usepackage{subcaption}
\usepackage{titlesec}

\titleformat{\chapter}
  {\Large\bfseries} 
  {}                
  {0pt}             
  {\huge}           
\usepackage[nottoc]{tocbibind}
\usepackage{afterpage}

\newcommand\blankpage{%
    \null
    \thispagestyle{empty}%
    \addtocounter{page}{-1}%
    \newpage}

\usepackage{hyperref}
    
\begin{document}

\thispagestyle{empty}

	\begin{center}

		\large
		\textbf{Adam Mickiewicz University, Poznań}

		\large
		Faculty of Mathematics and Computer Science \\
		Poland
		
		\vspace{0.2cm}
		
		\includegraphics[width=0.25\textwidth]{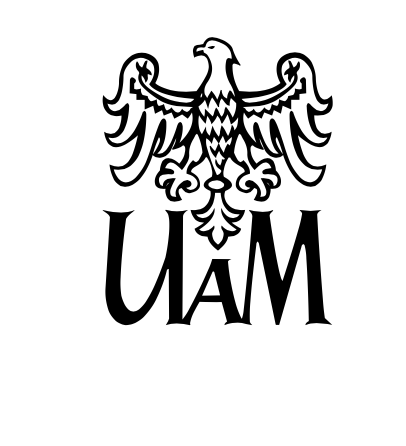}
		
		\vspace{0.3cm}
		
		\Large
		{Azam Jahandideh}
		
		\vspace{0.9cm}

		\Large
	\textbf{Stochastic quantization and Osterwalder--Schrader axioms for quantum field theory models}

		\vspace{0.4cm}
		{\normalsize A doctoral dissertation in natural sciences\\
		in the area of mathematics\\}
		\vspace{0.4cm}
		\normalsize {Advisor: prof. UAM dr hab.}~\Large {Wojciech Dybalski}\\

		\vspace{1.5cm}

	\Large
	\textbf{Kwantowanie stochastyczne i aksjomaty Osterwaldera--Schradera dla modeli kwantowej teorii pola}
		
		\vspace{0.4cm}

		\normalsize
		
		Rozprawa doktorska w naukach ścisłych i przyrodniczych
		w dyscyplinie matematyka\\
		{Promotor: prof. UAM dr hab. Wojciech Dybalski}

		\vspace{1.2cm}

		\normalsize
		Wydział Matematyki i Informatyki\\
		Uniwersytet im. Adama Mickiewicza w Poznaniu\\
	
		\vspace{0.2cm}
		October~2025
		
	\end{center}

\afterpage{\blankpage}

\chapter*{Acknowledgements}
\addcontentsline{toc}{chapter}{Acknowledgments}
\thispagestyle{empty}
I would like to thank Prof. Wojciech Dybalski and Dr. Pawe\l~Duch for our joint meetings during the completion of~\cite{DDJ}, on which this thesis is built. I would like to thank Wojciech separately for reading my thesis multiple times and for providing feedback, especially for completing the proof of Lemma~\ref{lem:prob_of_density} and for his alternative succinct proof in Appendix~\ref{sec:translation}.
\\[0.3cm]
The financial support provided by the grant ‘Sonata Bis’~$\mathrm{2019/34/E/ST1/00053}$ of the National Science Centre, Poland  between Oct\,2020\,--\,Oct\,2025 is well appreciated.

\thispagestyle{empty}

\newpage
\tableofcontents

\chapter*{Summary}
\addcontentsline{toc}{chapter}{Summary}
The purpose of this thesis is to construct the $P(\Phi)_2$ Euclidean quantum field theory, extending our discussion from~\cite{DDJ}. We first construct the finite volume measure of the $P(\Phi)_2$ model on $\sD'(\bS_R)$, where $\bS_R$ is a 2-sphere of the radius $R\in\bN_+$. To this end, we rely on the Nelson hypercontractivity estimate. Then, we utilize the parabolic stochastic quantization and the energy method with the help of the pull-back map $\jmath_R^*\,:\,\sD'(\bS_R)\to\sS'(\bR^2)$, induced by the stereographic projection $\jmath_R$, to obtain the infinite volume measure on $\sS'(\bR^2)$. To this end, we verify the global existence in time of solutions for the corresponding stochastic partial differential equations with fixed cutoffs, a topic that was only briefly addressed in~\cite{DDJ}.
\\[0.3cm]
The rest of the thesis is devoted to verifying the Osterwalder--Schrader axioms: regularity, reflection positivity and  invariance under translations, rotations, and reflections. Such a combination of axioms, which is crucial for obtaining a non-trivial local relativistic QFT on Minkowski space-time, was not obtained in earlier constructions based on the stochastic quantization, such as~\cite{AK17, AK21,GH21}.

\chapter*{Podsumowanie}
\addcontentsline{toc}{chapter}{Podsumowanie}
Celem tej pracy jest skonstruowanie euklidesowej kwantowej teorii pola  $P(\Phi)_2$, rozszerzając naszą dyskusję z~\cite{DDJ}. Najpierw konstruujemy  miarę na $\sD'(\bS_R)$ dla modelu $P(\Phi)_2$ w skończonej  objętości, gdzie $\bS_R$ jest 2-sferą o promieniu $R\in\bN_+$. W tym celu opieramy się na oszacowaniu  Nelsona dotyczącym półgrup hiperkontraktywnych. Następnie wykorzystujemy paraboliczne kwantowanie stochastyczne  i metodę energetyczną aby za pomocą pull-backu $\jmath_R^*\,:\,\sD'(\bS_R)\to\sS'(\bR^2)$, zadanego przez rzut  stereograficzny $\jmath_R$,  uzyskać miarę w nieskończonej objętości, mianowicie  na $\sS'(\bR^2)$. W tym celu sprawdzamy  istnienie globalnych w czasie rozwiązań dla odpowiednich stochastycznych równań różniczkowych cząstkowych z ustalonymi obcięciami, który to temat został jedynie krótko poruszony w~\cite{DDJ}.
\\[0.3cm]
Pozostała część pracy poświęcona jest weryfikacji aksjomatów Osterwaldera--Schradera: regularności, dodatniości względem odbicia oraz niezmienniczości względem przesunięć, obrotów i odbić. Takiej kombinacji aksjomatów, kluczowej dla uzyskania nietrywialnej lokalnej relatywistycznej kwantowej teorii pola  na czasoprzestrzeni Minkowskiego, nie uzyskano we wcześniejszych konstrukcjach opartych na kwantowaniu stochastycznym, takich jak~\cite{AK17, AK21,GH21}.

\chapter{Introduction and the main result}\label{ch:introduction}
Over the past 60 years several attempts have been made to reconcile the principles of quantum mechanics  with special relativity. The resulting theories are known as quantum field theories (QFTs). In order to make QFTs mathematically precise, using axiomatic setting, several approaches were proposed. The first satisfactory approach was the Wightman formalism on Minkowski spacetime, which assigns 
 to each subset of space-time an operator-valued distribution as a quantum observable~\cite{SW80}. Later, the Haag-Kastler setting was invented, which assigns $C^*$-algebras of observables to open bounded regions of space-time~\cite{Ha96}. Another significant endeavour is the Osterwalder--Schrader framework, which provides conditions
under which Wightman quantum field theory can be constructed from
a Euclidean QFT~\cite{OS75}. It is also the setting of this work.
\\[0.3cm]
To set the scene, we discuss in heuristic terms the Euclidean field theory using the probabilistic approach of~\cite{Le11}. Let $(M,g)$ be a $d-$dimensional Riemannian manifold with the metric tensor $g$, which might enjoy some symmetries, $\Sigma$ be another Riemannian manifold equipped with a volume measure $\sigma$, and let $\Phi(M)$ denote the set consisting of all the mappings  $\phi\,:\,M \to \Sigma$. We aim to define a measure $\mu$ on $\Phi(M)$ representing the desired 
classical field measure such that it respects the symmetries of the underlying manifold $M$. To this end, consider the functional $\cL\,:\,\Phi(M)\times M\to\bR_+\,$, which is called the Lagrangian density and formally assigns to each point $\mathrm{m}\in M$ and each $\phi\in\Phi(M)$ a real number, i.e., $[\cL(\phi)](\mathrm{m})=\mathcal{L}(\phi,\mathrm{m})$. Integrating the Lagrangian density functional over the manifold $M$, one defines the action $\cA\,:\,\Phi(M)\to \bR$ by
$$\cA(\phi):=\int_M\,[\mathcal{L}(\phi)](\mathrm{m})\,\mathrm{vol}(\rd \mathrm{m}),$$
where $\mathrm{vol}(\rd \mathrm{m})$ is the Riemannian volume measure of $M$ determined using the metric $g$. Using the action $\cA$, one defines the measure 
\begin{equation}\label{eq:functional-measure}
\mu(\rd\phi):=\frac{1}{\cZ}\,\exp\left(-\cA(\phi)\right)\nu(\rd\phi),
\end{equation}
on $\Phi(M)$, where $\nu$ is a reference measure on $\Phi(M)$ informally given by $\nu(\rd\phi):=\bigotimes_{\mathrm{m}\in M}\,\sigma(\rd\phi(\mathrm{m}))$. The partition function 
$$\cZ:=\int_{\Phi(M)}\,\mu(\rd\phi)=\int_{\Phi(M)}\,\re^{-\,\cA(\phi)}\,\nu(\rd\phi)$$
denotes the total mass of the measure $\mu$. Note that most of the time such a measure is not well-defined and calculating $\cZ$ for a given theory might not always be feasible. But if the measure~\eqref{eq:functional-measure} can be controlled, its moments give the Schwinger functions of the QFT in question.
To define the moments, we need a multiplication in $\Sigma$. As we are interested in scalar fields, we set $\Sigma=\mathbb{R}$ in the following.
\\[0.3cm]
\noindent
Parabolic stochastic quantization for a Euclidean scalar field was introduced by Parisi and Wu~\cite{parisi1981}. It is an approach to control mathematically the measure~\eqref{eq:functional-measure} using methods from the theory of stochastic partial differential equations (SPDEs). In our context of fields defined on a Riemannian manifold $M$ it  works as follows:  we start with the field $\phi(\mathrm{m})$, $\mathrm{m}\in M$, introduced above. Then, we define its counterpart, $\varPhi(t,\mathrm{m}),~(t,\mathrm{m})\in[0, T]\times M$, by introducing a fictitious time variable $t\in[0,T]$ such that $\varPhi(0, \mathrm{m}) = \phi(\mathrm{m})$. The random field $\varPhi(t,\mathrm{m})$ is coupled to an irregular Gaussian space-time white noise $\xi(t,\mathrm{m})$ whose typical realizations are distributions of negative regularity~\cite[Thm.~2.7]{CW17}. Next, we shall demand that the time evolution of $\varPhi(t,\mathrm{m})$ obeys a stochastic differential equation of Langevin type\footnote{There is a choice involved at this stage. Instead of the Langevin dynamics, one can choose the Fokker-Planck dynamics, the elliptic stochastic quantization, canonical stochastic quantization or the variational method~\cite{Gub21,DOR}. }
\begin{equation}\label{eq:Langevin}
\frac{\partial\,\varPhi(t,\mathrm{m})}{\partial t}=-\frac{1}{2}\,
 \frac{\delta \cA}{\delta \varPhi }\,(\varPhi(t,\Cdot))(\mathrm{m})
 +\xi(t,\mathrm{m}),
\end{equation}
where $\frac{\delta \cA}{\delta \varPhi}$ denotes the functional derivative of the action. Even if the action $\mathcal{A}$ describes the free field, the solution $\Phi$ inherits the distributional character of the white noise. If, subsequently, interaction is switched on, the functional derivative above becomes non-linear. Thus, evaluation  on the distribution $\varPhi$ is ill-defined, which is a manifestation of
the ultraviolet (UV) problem. It requires a modification of the equation, which is called renormalization. The simplest case, which suffices for polynomial interactions in $d=2$,  is the Wick ordering of the action.
 Higher dimensions and more complicated interactions require more advanced methods such as the theory of regularity structures~\cite{Hai14}, paracontrolled calculus~\cite{GH18}, and renormalization group techniques~\cite{Ku16, Du21}. In these particular
 frameworks it is possible to classify the equations~\eqref{eq:Langevin} into sub-critical, critical and super-critical, depending on their UV behaviour. This corresponds to the physicist's terminology of super-renormalizable, (just-)renormalizable and non-renormalizable actions, dictated by their behaviour under the Wilsonian block-averaging~\cite{Go}. In the renormalizable cases these actions move away from the massless free field action (i.e., the Gaussian fixed point) and give rise to an interacting effective  action at macroscopic scales. This mechanism is particularly robust in the super-renormalizable case, in which the coupling constants have strictly negative dimensions (in the units of length). In the non-renormalizable case, the action moves  towards the Gaussian fixed point which leads to trivialization.  
\\[0.3cm]
\noindent
If $\varPhi$ evolves according to equation~\eqref{eq:Langevin}, its law evolves according to the corresponding Fokker--Planck equation and is expected to approach the measure~\eqref{eq:functional-measure} as $t\to\infty$. Moreover, this measure is expected to be invariant under the dynamics, and hence to be the law of a stationary solution of~\eqref{eq:Langevin}. Physically, equation~\eqref{eq:Langevin} describes the stochastic dynamics of a statistical-mechanical system whose equilibrium distribution is determined by the Euclidean action~$\mathcal{A}$.
\\[0.3cm]
\noindent
Once the measure $\mu$ is constructed, one moves on to the study of its properties. In the case of the Euclidean spacetime $M=\mathbb{R}^d$, a convenient set of conditions, which guarantee the existence of an underlying Minkowskian QFT, are the Osterwalder--Schrader axioms. They concern an abstract family of distributions $S_m\in \sS'(\mathbb{R}^d)$, $m\in \mathbb{N}_0$, which are interpreted as the moments of the measure $\mu$:
\begin{equation}
S_m(f_1\otimes\ldots \otimes f_m)=\int_{\sS'(\bR^d)} \phi(f_1)\ldots \phi(f_m)\,\mu(\rd\phi), \label{Schwinger-functions-measure}
\end{equation}
$f_1,\ldots, f_m\in \sS(\bR^d)$, i.e., the Schwinger functions. The axioms encode their physically expected properties and add sufficient technical input for  
the subsequent reconstruction of a Wightman QFT on Minkowski spacetime.

\begin{dfn}\label{def:OS_axioms} (The Osterwalder-Schrader axioms)~\cite{OS75}. 
\begin{itemize}

\item[$\mathrm{OS0}$](Regularity). It holds that $S_0=1$. There is a Schwartz norm $\|\Cdot\|_{\mathrm{s}}$ on $\sS'(\mathbb{R}^d)$ and $K,\beta>0$ such that for all $m\in \bN$ 
and $f_1, \ldots, f_m\in \sS(\bR^d)$
\begin{equation}
|S_m(f_1 \otimes \ldots \otimes f_m)|\leq K^m(m!)^{\beta} \prod_{i=1}^m \|f_i\|_{\mathrm{s}}.
\end{equation}

\item[$\mathrm{OS1}$](Euclidean invariance). For each $m\in\bN$  and  all $f_1, \ldots, f_m \in\sS(\bR^{d})$, 
\begin{equation}
S_m(f_{1,(a,Q)}\otimes \ldots \otimes f_{m,(a,Q)})=S_m(f_1\otimes\ldots \otimes f_m),
\end{equation}
where $f_{j,(a, Q)}(x):=f_j( Q^{-1}(x-a))$, $Q\in SO(d)$, $a\in\bR^{d}$, and $j=1,\cdots,m$.

\item[$\mathrm{OS2}$](Reflection positivity). 
Let $\mathbb{R}^{dm}_+:=\{ (x^{(1)}, \ldots, x^{(m)})\in \mathbb{R}^{dm}\, : x_0^{(j)}>0, j=1, \ldots, m   \}$ and 
\begin{equation}
\sS(\mathbb{R}_+^{dm}; \mathbb{C}):=\{\, f\in \sS(\mathbb{R}^{dm}; \mathbb{C})\,:\, \mathrm{supp}(f)\subset  \mathbb{R}^{dm}_+ \}.
\end{equation}
For all sequences $f^{(m)}\in \sS(\mathbb{R}^{dm}_+; \mathbb{C})$, $m\in \mathbb{N}_0$, with finitely many non-zero elements
\begin{equation}
\sum_{\ell, m\in  \mathbb{N}_0} S_{\ell+m}( \overline{\Theta f^{(\ell)}} \otimes f^{(m)})\geq 0,
\end{equation}
where $(\Theta f^{(\ell)})(x_0^{(1)}, \vec{x}^{(1)}, \ldots, x^{(\ell)}_0, \vec{x}^{(\ell)}):=f^{(\ell)}( -x_0^{(1)}, \vec{x}^{(1)}, \ldots, -x^{(\ell)}_0, \vec{x}^{(\ell)}  )$.

\item[$\mathrm{OS3}$](Symmetry). For all $m \in \bN$, $f_1,\ldots, f_m\in \sS(\bR^d)$ and $\pi$ a permutation of $m$ elements 
\begin{equation}
S_m(f_{\pi(1)}\otimes\ldots \otimes f_{\pi(m)}) =  S_m(f_1\otimes\ldots \otimes f_m).
\end{equation}

\item[$\mathrm{OS4}$](Clustering).  For all sequences $f^{(m)}, g^{(m)}\in \sS(\mathbb{R}^{dm}_+; \mathbb{C})$, $m\in \mathbb{N}_0$, with finitely many non-zero elements
\begin{equation}
\lim_{\lambda \to \infty}\sum_{\ell, m \in  \mathbb{N}_0} \bigg\{ S_{\ell+m}( \overline{\Theta f^{(\ell)}} \otimes g^{(m)}_{(\lambda a,I)} )- 
S_{\ell}( \overline{\Theta f^{(\ell)}} )  S_{m}(  g^{(m)} )  \}=0,
\end{equation}
where $a=(0, \vec{a})$.
\end{itemize}
\end{dfn}
\noindent
There are several, possibly inequivalent, versions of the Osterwalder--Schrader axioms in the literature~\cite{OS75, EE79, GlimmJaffe}. We chose the above variant because the regularity condition 
OS0 is easy to check by methods of stochastic quantization (unlike the variant from the textbook of Glimm and Jaffe~\cite{GlimmJaffe}).  The reflection positivity  axiom OS2 is somewhat stronger than in the work of Osterwalder and Schrader~\cite{OS75} and in this form it was formulated by Eckmann and Epstein in~\cite{EE79} to ensure the existence of time-ordered products on the Minkowski side.   Clearly, the symmetry axiom OS3 is trivially satisfied if the distributions $S_m$ are defined by~\eqref{Schwinger-functions-measure} as moments of a measure.
\\[0.3cm]
\noindent
The axioms are satisfied, in particular, by free field theories, in which case the measure $\mu$, defining the Schwinger functions in 
\eqref{Schwinger-functions-measure} is Gaussian. Therefore, non-Gaussianity
of the measure is generally accepted as a useful criterion for non-triviality of a QFT in question. Due to the necessity of renormalization, mentioned above, non-triviality of a QFT cannot be inferred from the mere presence of interaction terms in the action $\mathcal{A}$. For example,
the $\phi^4$-interaction term in the action gives rise to a trivial theory in $d\geq 4$~\cite{Ai82, Fr82,AD21}.
\\[0.3cm]
\noindent
If the Schwinger functions, defined by \eqref{Schwinger-functions-measure}, satisfy the Osterwalder--Schrader axioms of Def.~\ref{def:OS_axioms}, then their analytic continuation from imaginary to real time $x_{0, j}\to -\ri x_{0,j}$, $j=1,\ldots,m$, is possible and gives rise to the Wightman functions of the underlying QFT on Minkowski spacetime~\cite{OS75}. 
The fact that the Osterwalder--Schrader axioms imply the existence of a corresponding Wightman QFT is a central result of constructive 
QFT and an important motivation to study Euclidean QFT.   
Let us recall, for completeness, the Wightman setting as it links our discussion with quantum theory.
\newcommand{\hphi}{\hat{\phi}}
\\[0.3cm]
\noindent
Let  $M=\bR^d$ be equipped with the Minkowski metric $\eta$ which  is invariant under  the restricted Poincar\'e group $\cP_+^\uparrow$.  Consider a quadruple $\langle H,U, \hphi, \mathrm{Dom}\rangle$, where $H$ is a separable Hilbert space consisting of physical states, $\mathrm{Dom}\subset H$ is a dense subspace, $\cP_+^\uparrow \ni  g  \mapsto U(g)\in  \mathcal{U}(H)$  is a group homomorphism, where $\mathcal{U}(H)$ is the group of unitary operators on $H$~\footnote{The Wigner theorem gives an operator $U(g)$ for each $g\in\cP_+^\uparrow $, which is unique up to phase.} and $\hphi$ is an operator-valued tempered distribution. Furthermore, a unit vector $\Omega\in H$  is called a vacuum state if $U(g)\,\Omega \subset \Omega $, $g\in \cP_+^\uparrow$, and it is called unique if $\Omega$ is the only such vector in $H$ up to multiplication by a phase.

\begin{dfn}(The Wightman axioms)\cite[IX.8]{RSII}\label{def:wightman}
A relativistic real-valued QFT on $\sS'(\bR^d)$ is a quadruple $\langle H,U,\hphi,\mathrm{Dom}\rangle$ satisfying the Wightman axioms:
\begin{itemize}

\item[(A)]
$H$ is a separable Hilbert space with a unique vacuum state $\Omega\in\mathrm{Dom}\subset H$ and $U$ is a strongly continuous unitary representation on $H$ of the restricted Poincar\'e group $\cP_+^\uparrow$ such that $U(a, \Lambda)\mathrm{Dom}\subset \mathrm{Dom}$ for all $(a, \Lambda)\in \cP_+^\uparrow$. 

\item[(B)]
The projection-valued measure $E$ on $\bR^d$ corresponding to $U(a) = \re^{ia^\mu P_\mu}$ has support in the closed forward light cone, i.e., $\overline{V}_+ = \{ (p_0,p_1)\in \bR^d\,|\, p_0\geq |p_1| \}$. 
\item[(C)]

For $f\in\sS(\bR^d)$ it holds $\mathrm{Dom}\subset \mathrm{Dom}(\hphi(f))$, $\hphi(f)\mathrm{Dom}\subset \mathrm{Dom}$ and for any fixed $\varphi\in \mathrm{Dom}$, the map $f\mapsto \hphi(f)\varphi$ is linear. Moreover,
for any  $\varphi_1,\varphi_2\in\mathrm{Dom}$ the mapping
$\sS(\bR^d)\ni f\mapsto \langle \varphi_1,\,\hphi_{\rm}(f)\,\varphi_2\rangle\in\bR$ is a tempered distribution.

\item[(D)](Locality) $[\hphi(f_1),\hphi(f_2)]=0$ if $\supp f_1$ and $\supp f_2$ spacelike separated.
(In the sense of weak commutativity on $\mathrm{Dom}$).

\item[(E)](Covariance) $U(a,\Lambda)\,\hphi(f)\,U(a,\Lambda)^*=\hphi(f_{(a,\Lambda)})$ for all $(a,\Lambda)\in\cP_+^\uparrow$ and all $f\in\sS(\bR^d)$ with $f_{(a,\Lambda)}=f(\Lambda^{-1}(x-a))$.

\item[(F)]
The dense subset $\mathrm{Dom}\subset H$ is constructed out of the vacuum via $\mathrm{Dom}=\mathrm{Span}\{\hphi(f_1),\cdots\hphi(f_m)\,\Omega\,|\, f_1,\cdots, f_m\in \sS(\bR^d),  m\in\bN_0\}$.

\end{itemize}
\end{dfn}
\noindent
The spectral condition in Def.~\ref{def:wightman},~Item~(B), determines the physical masses of the quantum fields via the point spectrum of the mass operator $M:=(P_{\mu} P^{\mu})^{1/2}$. For non-zero isolated masses Haag-Ruelle scattering theory gives 
a canonical construction of the scattering matrix~\cite{Ha96}. If it is different from the identity, we say that the QFT is interacting. This condition for non-triviality is stronger than the non-Gaussianity of the functional measure $\mu$ mentioned above. It was verified for some polynomial interactions for $d=2, 3$, cf.~\cite{GlimmJaffe}. In this thesis we will only consider non-Gaussianity.
\subsection*{Review of the literature}
There are a large number of works devoted to constructions of the $\Phi^4_d$ quantum field theory in different dimensions. Let us briefly describe the most relevant results in $d=2,3$ settings.
\\[0.1cm]
\begin{itemize}

\item[$\blacksquare$]  The finite volume measure of the $\Phi^4_2$ model was first constructed by Nelson~\cite{Nelson}. Due to the fact that finite volume $\Phi^4_2$ measure is absolutely continuous with respect to the Gaussian reference measure, the construction of this model was relatively easy to carry out.

\item[$\blacksquare$] Glimm and Jaffe constructed the infinite volume measure of the $P(\Phi)_2$ model on Minkowski space-time and verified all the Wightman axioms~\cite{GlimmJaffe}. Glimm, Jaffe, and Spencer constructed the infinite volume measure of the $P(\Phi)_2$ model and verified all of the Osterwalder--Schrader axioms including exponential decay of correlation functions for small coupling constant~\cite{GJS74}.

\item[$\blacksquare$] Glimm showed the existence of the $\Phi^4_3$ Hamiltonian in the Minkowski setting~\cite{Glim68}. Glimm and Jaffe proved the positivity of the $\Phi^4_3$ Hamiltonian~\cite{GJ73}. The first complete constructions of the model are given in~\cite{Fe74, FO76, MS76}.

\item[$\blacksquare$] Hairer studied the dynamical $\Phi^4_3$ model using the theory of regularity structures and verified the local in time well-posedness of this model~\cite{Hai14}. Mourrat and Weber investigated its long-time behaviour~\cite{MW17a}. Gubinelli and Hofmanov\'a constructed global in space and time solutions of the $\Phi^4_3$ model~\cite{GH18}. Later, Jagannath and Perkowski gave a simpler argument for the global well-posedness of the $\Phi^4_3$ model~\cite{JP23}.

\item[$\blacksquare$] Shen, Zhua, and Zhu obtained the perturbative expansion for the $k$-point correlation functions of the $\Phi^4_2$ model on the plane using the stochastic quantization. The same approach is claimed to give the perturbative expansion for the $P(\Phi)_2$ model~\cite{SZZ23}.

\item[$\blacksquare$] Albeverio and Kusuoka constructed the measure of the $\Phi^{4}_3$ model of Euclidean quantum field theory using the stochastic quantization and verified rotational invariance, non-Gaussianity and reflection positivity~\cite{AK17,AK21}. Note that based on the observation made in~\cite{BDFLL}, reflection positivity in~\cite{AK21} might not hold true.

\item[$\blacksquare$] Gubinelli and Hofmanov\'a  constructed the measure of the $\Phi^{4}_3$ model of Euclidean quantum field theory via stochastic quantization and verified non-Gaussianity, reflection positivity, translation invariance and exponential integrability. However, rotational invariance and clustering were not addressed. Their construction also includes the dynamical fractional $\Phi^4_3$ measure, which is proved to be reflection positive~\cite{GH21}.

\item[$\blacksquare$] Duch, Dybalski and Jahandideh constructed the Euclidean $P(\Phi)_2$ model using the stochastic quantization and verified the rotational and translational invariance, integrability, non-Gaussianity, and reflection positivity~\cite{DDJ}. This thesis is based on \cite{DDJ}.

\item[$\blacksquare$]
Duch, Gubinelli, Rinaldi constructed the measure of the fractional $\Phi^{4}_3$ model of Euclidean quantum field theory via stochastic quantization and verified non-Gaussianity, reflection positivity, translation invariance and exponential integrability~\cite{DGR23}.

\item[$\blacksquare$]
Bailleul, Dang, Ferdinand, and T$\hat{o}$ constructed the $\Phi^4_3$ measure on an arbitrary 3-dimensional compact Riemannian manifold without boundary using stochastic quantization~\cite{Bai23,BDFT1,BDFT2}. To this end, they used the technique introduced by Jagannath–Perkowski in~\cite{JP23} to develop a new Cole-Hopf
transform. In particular, the used counter-terms in their construction do not depend on the metric of the underlying Riemannian manifold.

\item[$\blacksquare$]
Bauerschmidt, Dagallier, Weber proved uniqueness of the invariant measure of the $\Phi^4_2\,$ SPDE up to the critical temperature using the Holley–Stroock approach, which was originally established to verify the uniqueness of invariant measures of Glauber dynamics of lattice spin systems~\cite{BDH25}.

\item[$\blacksquare$]
Duch, Hairer, Yi, and Zhao characterised the infinite-volume $\Phi^4_3$ measure at high temperature (small
coupling) as the unique
invariant measure of the associated stochastic equation, and proved that it satisfies all Osterwalder--Schrader axioms, including invariance under translations, rotations, and reflections,
as well as exponential decay of correlations~\cite{DHYZ25}.

\end{itemize}

\subsection*{$P(\Phi)_2$ model and the main result}
Let us adapt the above setting for a class of polynomial interactions known as the $P(\Phi)_2$ models, which constitute some of the simplest examples of interacting QFTs. Fix $n\in2\bN_+$, $n\geq 4$, and a real polynomial $
 P(\tau) = \sum_{m=0}^n\,a_m \tau^m,\quad\tau\in\bR$, with $a_0,\ldots,a_{n-1}\in\bR$ and $a_n=1/n\,.$
Observe that such a polynomial is bounded from below. Recall the setting introduced in the beginning of this chapter:~Let $M=\bS_R$ be the two-sphere of radius $R\in\bN_+$, equipped with the metric tensor $g_R$ induced by the Euclidean metric on $\bR^3$, and with the symmetry group $O(3)$. Moreover, $\Sigma=\bR$, replace $\Phi(M)$ by $\sD'(\bS_R)$, and $\mathrm{vol}(\rd \mathrm{m})=\rho_R(\rd\rx)$, where $\rho_R$ is invariant under the action of $SO(3)$. Consider the probability measure
\begin{equation}\label{eq:measure_mu_R}
 \mu_{R}(\rd\phi):=\frac{1}{\cZ_R}\exp\big(-\,\int_{\bS_R}\,\lambda\,:P(\phi(\rx)):\,\rho_R(\rd\rx)\big) \,\nu_{R}(\rd\phi),
\end{equation}
where $\nu_{R}:= \cN(0,G_{R})$ indicates the free field Gaussian measure with covariance 
$G_{R}:=(1-\Delta_R)^{-1}$, $(-\Delta_R)$ is the Laplace-Beltrami operator, $\lambda\in(0,\infty)$ is the coupling constant, $:\Cdot:$ refers to the Wick renormalization with respect to the measure $\nu_R$ and $\cZ_R\in (0,\infty)$ is the partition function. Observe that we used the Gaussian measure $\nu_R$ as the reference measure on $\sD'(\bS_R)$ due to the fact that the Lebesgue measure does not exist in infinite dimensional settings. Note that $(\sD'(\bS_R),\mathrm{Borel}(\sD'(\bS_R)), \nu_R)$ forms a probability space. The existence of the measure~\eqref{eq:measure_mu_R} is shown by first introducing an ultraviolet (UV) regularized measure $\mu_{R,N}$ and then controlling the 
limit $N\to \infty$ using the Nelson argument, cf. Ch.~\ref{ch:finite_volume}. Such simple treatment of the UV problem is special to polynomial interactions in $d=2$ and reflects the fact that the perturbed measure~\eqref{eq:measure_mu_R} is absolutely continuous 
with respect to $\nu_R$.
\\[0.3cm]
\noindent
On account of Eq.~\eqref{eq:Langevin}, one writes the corresponding parabolic SPDE associated to the measure~\eqref{eq:measure_mu_R}
\begin{equation}\label{eq:sq_pphi}
(\partial_t  + \frac{1}{2}Q_R)\,\varPhi_{R}(t,\rx)
=- \frac{1}{2}\lambda\,:P'(\varPhi_{R}(t,\rx)):+\,\xi(t,\rx),
\end{equation}
where $(t,\rx)\in \bR_+\times \bS_R$,  $P'(\tau) := \partial_{\tau}P(\tau)$ and $Q_R:=(1-\Delta_R)$. For $R\in\bN_+$ let $\jmath_R: \mathbb{R}^2 \to \bS_R \backslash \{(0,0, R)\}$ be the stereographic projection, and let
$$\jmath_R^*\,:\,\sD'(\bS_R)\to\sS'(\bR^2)$$ denote the corresponding pull-back of distributions. We denote by $\jmath_R^*\sharp\mu_R$ the measure on $\sS'(\bR^2)$, obtained by the push-forward of $\mu_R$ with $\jmath_R^*$. We shall make use of the SPDE~\eqref{eq:sq_pphi} and the intertwining property of the map $\jmath_R^*$ in order to prove the existence of limit points of the measures $(\jmath_R^*\sharp\mu_R)_{R\in \mathbb{R}_+}$ on $\sS'(\bR^2)$ and to verify that such limiting measures satisfy the Osterwalder-Schrader axioms. The following theorem is our main result.


\begin{thm}\label{thm:main-theorem}
The sequence of measures $(\jmath_R^*\sharp\mu_R)_{R\in\bN_+}$ on $\sS'(\bR^2)$ has a weakly convergent subsequence. Every accumulation point $\mu$ of $(\jmath_R^*\sharp\mu_R)_{R\in\bN_+}$ is invariant under the Euclidean symmetries of the plane and reflection positive. Moreover, there exists a ball $B\subset\sS(\bR^2)$ with respect to some Schwartz semi-norm centered at the origin such that for all $f\in B$, it holds that
\begin{equation}\label{eq:main_thm_bound_intr}
 \int\exp(\phi(f)^n)\,\mu(\rd\phi) 
 \leq 2\,.
\end{equation}
\end{thm}
\begin{Remark}
By the above theorem, $\mu$ satisfies all the Osterwalder--Schrader axioms of Def.~\ref{def:OS_axioms}, possibly with the exception of the cluster property OS4. It is known that the $P(\Phi)_2$ measure on the plane is unique provided $\lambda\in(0,\infty)$ is sufficiently small~\cite{GJS74}. In general, uniqueness does not hold and the model exhibits phase transitions~\cite[Ch.~16]{GJ85}. The cluster property is only expected to hold in pure phases. Our construction of the $P(\Phi)_2$ measure does not need any smallness assumption on $\lambda$. However, it does not give any information about the uniqueness of the infinite volume limit or the decay of correlation functions. In what follows we set $\lambda=1$. 
\end{Remark}
\begin{proof}
The existence of a weakly convergent subsequence of $(\jmath_R^*\sharp\mu_R)_{R\in\bN_+}$ follows from tightness and Prokhorov's theorem. The proof of tightness is presented in Sec.~\ref{sec:tightness} and uses parabolic stochastic quantization combined with a PDE energy estimate. The invariance of $\mu$ under the Euclidean symmetries is established in Sec.~\ref{sec:Euclidean} and is based on the fact that for all $R\in\bN_+$, the measure $\mu_R$ is invariant under the action of the orthogonal group $O(3)$. The proof that $\mu$ is reflection positive is given in Sec.~\ref{sec:reflection_posivity} and is based on the fact that for all $R\in\bN_+$, the measure $\mu_R$ is reflection positive. The bound~\eqref{eq:main_thm_bound_intr} is proved in Prop.~\ref{prop:integrability} with the help of the Hairer-Steele argument~\cite{HairerSteele}. 
\end{proof}

\subsection*{Plan of the thesis}
In Ch.~\ref{ch:preliminaries}, we provide the background used throughout the sequel. Sec.~\ref{sec:spaces} is an expanded version of~\cite[App.,A]{DDJ}. Lemmas~\ref{lem:trace}, and~\ref{lem:tight_comapct_embedding} correspond to \cite[Lemmas B.1, B.4]{DDJ}, respectively.
\\[0.3cm]
In Ch.~\ref{ch:finite_volume}, we introduce the regularized $P(\Phi)_2$ measure $\mu_{R,N}$ on $\bS_R$ with a certain UV cutoff $N\in\bN_+$, in the frequency space. We shall write  $\mu_{R,N}$ as a perturbation of the regularized Gaussian measure $\nu_{R,N}$ by the $P(\Phi)_2$ interaction. In Sec.~\ref{sec:existence_UV_limit}, we verify that the measure $\mu_{R,N}$  converges weakly, as $N\to\infty$, to the probability measure $\mu_R$ formally given by Eq.~\eqref{eq:measure_mu_R}. We use the Nelson hypercontractivity estimate for this purpose. We also investigate a certain auxiliary measure $\mu_{R,N}^g$, 
$g\in \sS(\bR^2)$, which coincides with $\mu_{R,N}$ when $g=0$ and will be used in~Ch.~\ref{ch:Osterwalder-Schrader axioms}. Ch.~\ref{ch:finite_volume} extends~\cite[Sec.\,2]{DDJ}.
\\[0.3cm]
In Ch.~\ref{ch:stochastic_quantization}, we study the stochastic quantization equations of the measures $\mu_{R,N}^g$ and $\nu_{R,N}$ using the Da Prato-Debussche trick. We show the local-in-time and global-in-time existence of solutions with fixed $R,N\in\bN_+$, in~Sec.~\ref{subsec:local_solution_remainder} and in Sec.~\ref{subsec:global_solution_remainder}, respectively. Local existence is established using the Banach fixed point theorem together with an appropriate Sobolev embedding, while the global existence is shown with the help of the energy technique. The fact that the measures $\mu_{R,N}^g$ and $\nu_{R,N}$ are invariant for the associated semigroups is the subject of~Sec.~\ref{sec:invariant_measures}. This chapter is largely independent of~\cite{DDJ} as these topics were only briefly addressed in Sec.\,3 of this latter publication.
\\[0.3cm]
In Ch.~\ref{ch:infinite-measure}, we apply the energy technique to prove a certain a priori bound, which is uniform in $R,N\in\bN_+$. This a priori bound will be used in order to verify that the sequence of measures $(\jmath_R^*\sharp \mu^g_{R})_{R\in\bN_+}$ on $\sS'(\bR^2)$ is weakly converging along a subsequence. This chapter expands on~\cite[Secs.\,5,\,6]{DDJ}.
\\[0.3cm]
In Ch.~\ref{ch:Osterwalder-Schrader axioms}, we investigate the Osterwalder--Schrader axioms. Sec.~\ref{sec:regularity} is devoted to the proof of regularity, while Sec.~\ref{sec:reflection_posivity} concerns reflection positivity. In Sec.~\ref{sec:Euclidean}, we use the intertwining property of the map $\jmath_R^*$ to show that an infinite volume $P(\Phi)_2$ measure is invariant under translations, rotations and reflections of the plane. The proof relies on the invariance of the finite volume measure $\mu_R$ on $\sD'(\bS_R)$ under the action of the orthogonal group $O(3)$, combined with the fact that the symmetry groups of $\bS_R$ and $\bR^2$ have the same dimension. This chapter extends~\cite[Secs.\,7,\,8,\,9]{DDJ}.
\\[0.3cm]
Appendix~\ref{sec:bessel} provides some well-known facts on Bessel potentials. In Appendix~\ref{sec:stochastic_estimates}, which extends~\cite[Apps.\,B,\,C]{DDJ}, we prove bounds for the Wick ordered fields, which are essential for proving the local-in-time and the global-in-time existence of the stochastic random fields in Ch.~\ref{ch:stochastic_quantization}, as well as to prove tightness in Ch.~\ref{ch:infinite-measure}. In Appendices~\ref{sec:mild_regularity_time_space},~\ref{sec:proof_wieght_plane_laplacian}, and~\ref{sec:translation} we provide the complete proofs of Lemma~\ref{lem:mild_regularity_time_space}, Lemma~\ref{lem:wieght_plane_laplacian}, and Remark~\ref{rem:translation}. These proofs were largely left for the reader in~\cite{DDJ}.

\chapter{Preliminaries}\label{ch:preliminaries}
In this section, we collect various technical results on a 2-dimensional sphere of radius $R\in[1,\infty)$, i.e., $\bS_R$. We gather some necessary facts on spherical harmonics in Sec.~\ref{sec:spherical_harmonics}. In Sec.~\ref{sec:spaces}, we define function spaces on $\bR^2$ and $\bS_R$ and collect some useful Sobolev embeddings. We summarize all required functional analysis techniques in Sec.~\ref{sec:functional_analysis}. Sec.~\ref{sec:probability} contains a review of all the essential techniques in probability. Sec.~\ref{sec:stochastic_analysis} is meant to provide the necessary background on the stochastic differential equations.

\section{2-dimensional sphere}\label{sec:sphere}
Let $R\in[1,\infty)$. We define the two-dimensional sphere of radius $R$, i.e., $(\bS_R,g_R)$ as a regular surface of the Euclidean space $\bR^3$ via
\begin{equation}
 \bS_R:=\{\rx=(\rx_1,\rx_2,\rx_3)\in\bR^3\,|\, \rx_1^2+\rx_2^2+\rx_3^2=R^2\},
\end{equation}
where $g_R$ denotes for the metric tensor on $\bS_R$ induced from the metric on $\bR^3$. We let $N= (0,0,R)$ be the north pole of $\bS_R$. Note that $\bS_R$ is a Riemannian manifold with the Ricci scalar curvature $2/R^2$~\cite[Eq.~(3.158)]{Ca03}~\footnote{The sectional curvature of an $d$-sphere of radius $R$ is $1/R^2$ and its Ricci scalar curvature is $d(d-1)/R^2$.}. We let $\rho_R(\rd\rx)=\sqrt{\det g_R}\,\rd\rx$ denote the Riemannian volume measure~\cite[Thm.~3.11]{Gri09} or~\cite[Def.~1.74]{Aub98} such that $\int_{\bS_R}\rho_R(\rd\rx) = 4\pi R^2$. Moreover, $\|\Cdot\|_{L_p(\bS_R)}$ denotes the $L_p$-norm on $\bS_R$ with respect to the canonical measure $\rho_R(\rd\rx)$ for all $p\in[1,\infty)$. 
\\[0.3cm]
The symmetry group of a two-dimensional sphere, including reflections, is known as the orthogonal group $O(3)$~\cite[Prop.~18.5.2]{Berger2}. The symmetry group consisting of only rotations is the special orthogonal group $SO(3)$. $O(3)$ and $SO(3)$ act transitively on $\bS_R$\footnote{Namely, for all $\rx,\ry\in\bS_R$ there exists $Q\in SO(3)$ such that $Q\cdot\rx=\ry$~\cite[Def.~1.4.3]{Berger1}.} as groups of isometries and they give rise to groups of unitary operators on $L_2(\bS_R)$. Note that for all $R>0$, the measure $\rho_R(\rd\rx)$ is invariant under the action of $SO(3)$.
\noindent
The transitive action of $SO(3)$ makes $\bS_R$ a homogeneous space\footnote{The set $X$ is called a homogeneous space under $G$ (for the action $\phi$) if $(G,X,\phi)$ is transitive~\cite[Def.~1.5.4]{Berger1}.}, i.e., for all $R>0$ the two-dimensional sphere $\bS_R$ can be identified with the quotient space $SO(3)/SO(2)$~\footnote{
To see this, fix a unit vector in $\bR^{3}$ such that its orbit under the $SO(3)$-action gives rise to the canonical embedding $\bS_1\hookrightarrow \bR^{3}$. Moreover, the subgroup $SO(2)$, embedded inside $SO(3)$, stabilizes the axis formed by this vector. Hence, by the Orbit-Stabilizer theorem we have the identification $\bS_1\equiv SO(3)/SO(2)$.}.
\\[0.3cm]
The tangent space of $\bS_R$ at $\rx\in\bS_R$, i.e., $\mathrm{T}_{\rx}\bS_R$ is the orthogonal complement of $\rx$,
$$\mathrm{T}_{\rx}\bS_R=\rx^{\perp}=\{v\in \bR^3\,|\,\langle \rx,v\rangle =0\},$$
where $\langle \rx,v\rangle=\sum_{i=1}^3\,\rx_i\cdot v_i$ denotes the inner product on $\bR^3$. Note that if $v,w$ form an orthogonal basis of $\mathrm{T}_{\rx}\,\bS_R$, then $\rx,v,w$ form an orthogonal basis of $\bR^{3}$.

\begin{Remark}\label{rem:O(3)_splitting}
For all $n\geq 1$, the determinant on $O(n)$ yields a short exact sequence
$$
\1\to SO(n)\xrightarrow{\imath} O(n)\xrightarrow{\det} \bZ_2\to \1.
$$
Since there exists $\mathcal{R}\in O(n)$ with $\det\mathcal{R}=-1$, this sequence splits. Consequently,
$$O(n)\cong SO(n)\rtimes_{\varphi} \bZ_2,$$
where $\varphi\,:\,\bZ_2\to \mathrm{Aut}(SO(n))$ is defined by $\varphi(1)(Q)=Q$ and $\varphi(-1)(Q)= \mathcal{R}Q \mathcal{R}^{-1}$ for $Q\in SO(n)$.
\end{Remark}

\begin{Remark}\label{rem:sphere_quotient}
Let $N=(0,0,R)$ be the north pole of $\bS_R$.
Observe that 
$$N^{\perp}=\mathrm{T}_{N}\bS_R=\{(x_1,x_2,0)\in\bR^3\}.$$
This implies that the isotropy group of $N$ is isomorphic to $SO(2)$. More precisely, the subgroup of $SO(3)$ fixing $N$ is a copy of $SO(2)$, which stabilizes the $x_1x_2-$plane. In Sec.~\ref{sec:Euclidean} this fact will be used to establish a one-to-one correspondence between the action of the orthogonal group $O(3)$ on $\bS_R$ and the action of the Euclidean group $T(2)\rtimes O(2)$ on $\bR^2$, see Remark~\ref{rem:euclidean_R_two}.
\end{Remark}

\begin{Remark}\label{rem:tangent_lei_algebra}
Let $p\,:\,SO(3)\to SO(3)/SO(2)$ denote the canonical projection from the special orthogonal group $SO(3)$ onto the homogeneous space $SO(3)/SO(2)$, which is smooth. The differential of $p$ at the identity $\re\in SO(3)$ is the linear map
$$\mathrm{T}_{\re}\,p\,:\, so(3)\to \mathrm{T}_{[\re]}(SO(3)/SO(2)),\quad \text{where}\quad [\re]=p(\re).$$
 Note that both $p$ and $\mathrm{T}_{\re}\,p$ are surjective. This implies that $\mathrm{T}_{[\re]}(SO(3)/SO(2))$ is isomorphic to $so(3)/so(2)$. One similarly obtains that $\mathrm{T}_{[g]}(SO(3)/SO(2))$ is also isomorphic to $so(3)/so(2)$ for arbitrary $[g]\in SO(3)/SO(2)$ thanks to the fact that $SO(3)/SO(2)$ is a homogeneous space. 
\end{Remark}

\begin{dfn}\label{def:sphere_distance}
For all $\rx,\ry \in \bS_R$, we set the geodesic distance on $\bS_R$ by
\begin{equation}\label{eq:geode_distance}
\begin{gathered}
\mathrm{dist}(\rx,\ry) := d_R(\rx,\ry)=R\,\theta(\rx,\ry)\,,
\quad
\theta(\rx,\ry)=\arccos \frac{\rx\cdot\ry}{R^2}\in [0,\pi].
\end{gathered}
\end{equation}
In fact, the geodesic distance on the unit sphere corresponds to the arc length of the great circle connecting the points~\cite[p.~82]{Le97}. In particular, one has $0\leq d_R(\rx,\ry)\leq R\,\pi$. Moreover, 
\begin{equation}\label{eq:distance_sphere_1}
\|\rx-\ry\| = \sqrt{2R^2-2R^2\,\cos(\theta(\rx,\ry))}
= 2R\,\sin\frac{d_R(\rx,\ry)}{2R}\,,
\end{equation}
where $\|\Cdot\|$ denotes the Euclidean norm on $\bR^3$. It holds that
$$\frac{2}{\pi}\,d_R(\rx,\ry)\leq \|\rx-\ry\|\leq d_R(\rx,\ry),$$
which implies that the two distances are comparable.
\end{dfn}

\subsection*{Laplace-Beltrami operator}\label{sec:laplace_beltrami}
Let $(M,g)$ be a $d$-dimensional Riemannian manifold with local coordinates $\rx=(\rx_1,\rx_2,\cdots,\rx_d)$. The Laplace–Beltrami operator is given by~\cite[Eq.~(3.40)]{Gri09}
\begin{equation}\label{eq:manifold_laplacian}
\Delta_g = \frac{1}{\sqrt{|g|}}\,\sum_{i,j=1}^d\frac{\partial}{\partial \rx_i}\,(\sqrt{|g|}\,g^{ij}\,\frac{\partial}{\partial \rx_j})\,,
\end{equation}
where $\rd s^2 = \sum_{i,j=1}^d g_{ij}\,\rd \rx^i\,\rd\rx^j$ and $|g| = |\det (g)|$. In particular, in the Euclidean space $\bR^3$ with the Euclidean metric $g_{ij}=\delta_{ij}$, one has 
\begin{equation}\label{eq:sphere_R3_laplacian}
\Delta_{\bR^3} = \sum_{i=1}^3\frac{\partial^2 }{\partial x_i^2}\equiv \frac{\partial^2}{\partial r^2} +\frac{2}{r}\frac{\partial}{\partial r}+\frac{1}{r^2}\Delta_{\bS_1},
\end{equation}
where $\Delta_{\bS_1}$ is the Laplace-Beltrami operator on the unit sphere with the usual round metric. It is often written as $-L^2$ in physics texts. In spherical coordinates $(\theta,\varphi)$, one has
\begin{equation}
\Delta_{\bS_1}=
\frac{1}{\sin\theta}\frac{\partial}{\partial\theta}(\sin\theta\frac{\partial}{\partial\theta}) +\frac{1}{\sin^2\theta}\frac{\partial^2}{\partial\varphi^2}\,.
\end{equation}
The Laplace-Beltrami operator on $\bS_R$ satisfies $\Delta_R := \frac{1}{R^2}\,\Delta_{\bS_1}$ for all $R>0$. 

\begin{Example}\label{ex:conformal_measure}
Let $(M,g)$ be a $d$-dimensional manifold. Consider the conformal deformation $\tilde{g}$ associated to the metric $g$, i.e., $\tilde g= \re^{2\nu} g$, where $\nu\in C^{\infty}(M)$. One has
\begin{equation}
\tilde g_{ij}=  \re^{2\nu}\,g_{ij},\quad \tilde g^{ij}=  \re^{-2\nu}\,g^{ij},\quad \sqrt{\det\, \tilde g}= \re^{d\nu}\,\sqrt{ {\det g}}\,.
\end{equation}
Hence,
\begin{equation}
\begin{aligned}
\Delta_{\tilde g}
&= \frac{1}{\re^{d\nu}\sqrt{\det g}}  \sum_{i,j=1}^d\frac{\partial}{\partial \rx_i}\Big(\re^{(d-2)\nu}\sqrt{\det g} \,g^{ij}\frac{\partial}{\partial \rx_j} \Big) 
\\
&=\re^{-2\nu}\Delta_{g}+(d-2)\re^{-2\nu}\,g^{ij}\,\frac{\partial\nu}{\partial\rx_i}\,\frac{\partial}{\partial \rx_j}\,. 
\end{aligned}
\end{equation}
In particular, if $(M,\tilde g)$ is a 2-dimensional Riemannian manifold with a metric given by $(\tilde g)_{jk}(x)= \re^{2\nu(x)}\,\delta_{jk}$ it holds that~\cite[App.~C, Cor.~3.5]{TaylorII}
\begin{equation}
\Delta_{\tilde g} f=
\re^{-2\nu}\sum_{j,k=1}^2\frac{\partial}{\partial \rx_j}  ( \re^{-2\nu}\,\delta^{jk}\, \re^{2\nu}\frac{\partial}{\partial \rx_k}  f)= \re^{-2\nu}\Delta f,\,\,\,\,\,\,\,\Delta f=\frac{\partial^2 f}{\partial x_1^2} + \frac{\partial^2 f}{\partial x_2^2}\,.
\end{equation}
We shall make use of this fact once we introduce the stereographic projection.
\end{Example}

\begin{Remark}\label{rem:lap_bel_identities}
For all $R\in\bN_+$, we have $\Delta_R= \vec{\nabla}_R\cdot\vec{\nabla}_R$, where $\vec{\nabla}_R$ is the gradient on the sphere. For any $f,g\in C^2(\bS_R)$, by the Green–Beltrami identity, one has~\cite[Prop.~3.3~and~Cor.~3.4]{AH12} or~\cite[Thm.~3.16]{Gri09}
\begin{equation}
\int_{\bS_R}g(\rx)\, (\Delta_R f)(\rx)\,\rho_R(\rd \rx)
=
-\int_{\bS_R}(\vec{\nabla}_R g)(\rx)\cdot (\vec{\nabla}_R f)(\rx)\,\rho_R(\rd \rx)
= \int_{\bS_R}f(\rx)(\Delta_R g)(\rx)\,\rho_R(\rd \rx).
\end{equation}
Let $\langle\Cdot,\Cdot\rangle_{L_2(\bS_R)}$ denote the inner product on $L_2(\bS_R)$. Then
\begin{equation}
\begin{gathered}
\langle f,\Delta_R\,g\rangle_{L_2(\bS_R)} = \langle \Delta_R\,f, g\rangle_{L_2(\bS_R)}, 
\quad
\langle f,\Delta_R\,f\rangle_{L_2(\bS_R)} \leq 0\,.
\end{gathered}
\end{equation}
In general, on a complete Riemannian manifold $M$, the Laplace-Beltrami operator is essentially self-adjoint~\cite[Thm.~2.4]{St83} or~\cite[Thm.~11.5]{Gri09}, i.e., it has a unique self-adjoint extension.
\end{Remark}

\subsection*{Stereographic projection}\label{sec:stereographic}
\begin{dfn}\label{def:stereoraphic}
Let $R\in [1,\infty)$ and $N= (0,0,R)$. Define the map $\jmath_R\,:\, \bR^2\to \bS_R\setminus N\subset \bR^3$ by
\begin{equation}
\jmath_R(x_1,x_2):= \rx \equiv(\rx_1,\rx_2,\rx_3)= \frac{R(4Rx_1, 4Rx_2, x_1^2 +x_2^2 -4R^2)}{x_1^2+x_2^2 +4R^2}\,.
\end{equation}
Note that $-R\leq \rx_3<R$. We call the point $(x_1,x_2) = x\in\bR^2$ the stereographic coordinate of the point $\rx \in\bS_R\setminus N$. Note that $\jmath_R$ is always locally invertible. Let $\jmath_R^{-1}\,:\, (\bS_R\setminus N)\subset \bR^3\to \bR^2$. It holds that
\begin{equation}
\jmath_R^{-1}(\rx_1,\rx_2,\rx_3):=x \equiv (x_1,x_2) = \frac{2R(\rx_1,\rx_2)}{R-\rx_3}\,.
\end{equation}
Moreover, for $x,y\in\bR^2$ such that $\jmath_R(x)=\rx$ and $\jmath_R(y)= \ry$ one has
\begin{equation}
\sum_{i=1}^3(\rx_i-\ry_i)^2 = \|\rx-\ry\|^2 = \frac{16R^4}{(4R^2+|x|^2)(4R^2+|y|^2)}\,\|x-y\|^2\,.
\end{equation}
Observe that for all $R\in\bN_+$, the mappings $\jmath_R$ and $\jmath_R^{-1}$ are rational functions, hence smooth. Thus, $\jmath_R$ is a local diffeomorphism. 
\end{dfn}
\noindent
Recall the Riemannian volume measure was given by $\rho_R(\rd\ry)=\sqrt{\det g_R}\,\rd\ry$. One verifies that the metric tensor $g_R$ on $\bS_R$ in the stereographic coordinates is of the form
\begin{equation}\label{eq:conformal_factor}
(\jmath_R^*\,g_R)_{ij}(x) := w_R(x)\,\delta_{ij}
\end{equation}
with
$$w_R(x):=\sqrt{\det (\jmath_R^* g_R)(x)} =\frac{16R^4}{(4R^2+|x|^2)^2} \in C^{\infty}(\bR^2),$$
where $|x|^2= x_1^2+x_2^2$~\cite[p.~12]{Sch08}. This shows that $\bS_R$ is locally conformally flat\footnote{A Riemannian manifold $(M,g)$ is said to be locally conformally flat if it is locally conformal to Euclidean space. Recall that a conformal structure on $M$ is an equivalence class of metrics with $\hat{g}=\re^{2\omega}\,g$ for some function $\omega\in C^{\infty}(M)$.}. Note that for all $R\in\bN_+$ and all $x\in\bR^2$, it holds that $w_R(x)\leq 1$.

\begin{Remark}\label{rem:norm_stereo}
Let $f\in C(\bS_R)$, then $\jmath_R^*f=f\circ\jmath_R\in C(\bR^2)$. Note that for $f\in C^\infty_\rc(\bR^2)$ the function $f\circ \jmath_R^{-1}\in C_\rc(\bS_R\setminus N)$ has unique continuous
extension to $\bS_R$. 
From Def.~\ref{def:measure_preserving} Items~(A) and~(B) for all $f\in C^{\infty}(\bS_R\setminus N)$ it holds that
\begin{equation}
\int_{\bR^2}\,\jmath_R^*f(x)\,w_R(x)\,\rd x=\int_{\bR^2}\,f\circ\jmath_R(x)\,w_R(x)\,\rd x=\int_{\bS_R}\,f(\rx)\,\rho_R(\rd\rx)\,.
\end{equation}
Moreover, by Def.~\ref{def:measure_preserving} Item~(C) for $\phi\in\sD'(\bS_R)$ and $f\in C_{\rc}^{\infty}(\bR^2)$ one has
\begin{equation}
\begin{aligned}
\langle\jmath^*_R\phi,f\rangle_{L_2(\bR^2)}&=
\int_{\bR^2}\,\jmath_R^*\phi(x)\,f(x)\,\rd x
\\=&\int_{\bS_R}\,\phi(\ry)\,f(\jmath_R^{-1}(\ry))\,w^{-1}_{R}(\jmath_R^{-1}(\ry))\,\rho_R(\rd \ry)
= \langle \phi, (w_R^{-1}f)\circ \jmath^{-1}_R\rangle_{L_2(\bS_R)}\,.
\end{aligned}
\end{equation}
\end{Remark}

\begin{Remark}\label{rem:comf_laplacian}
Let $\hat\Delta:= \jmath_R^*\Delta_R$ denote the Laplacian associated to the metric $\jmath^*_R g_R$ on $\bR^2$. Using Example~\ref{ex:conformal_measure} one gets $\hat\Delta = w_R^{-1}\Delta$, where $\Delta$ denotes the standard Laplacian on $\bR^2$. Alternatively, one can use~\cite[Sec.~1]{PR87} with $g_1=\jmath_R^* g_{R}$, $g=\delta_{ij}$, $f=\frac{1}{2}\, \log(w_R(x))$ and $n=2$. It holds that $\jmath_R^*\,\Delta_R= w^{-1}_R\Delta\,\jmath_R^*$.
\end{Remark}

\begin{Remark}\label{rem:sphere_mobius_plan}
On account of Remark~\ref{rem:sphere_quotient}, under stereographic projection from the north pole, the action of $O(3)$ on $\bS_R$ induces (anti-)M\"obius transformations on $\bR^2$. As $R\to\infty$, when restricted to bounded subsets of $\bR^2$, e.g., the neighborhoods of the south pole, these transformations converge to the action of the Euclidean isometry group $E(2)=T(2)\rtimes O(2)$ on $\bR^2$.
\end{Remark}

\section{Spherical harmonics}\label{sec:spherical_harmonics}
Let $\bY_l(\bR^d)$ be the space of the homogeneous harmonics of degree $l$ in $\bR^d$, and $\bY_l^d:= \bY_l(\bR^d)|_{\bS_1^{d-1}}$ be the space of spherical harmonics of degree $l$ in a unit $(d-1)$-dimensional sphere. We set $\{Y_{l,m}\,:\,1\leq m\leq N_{l,d}\}$ to be an orthonormal basis of $\bY^d_l$, i.e.,
\begin{equation}\label{eq:ortho_harmonic}
\int_{\bS_1^{d-1}} Y_{l,j}(\eta)\,\overline{Y_{l',k}(\eta)}\,\rd\eta = \delta_{ll'}\,\delta_{jk},\quad 1\leq j,k\leq N_{l,d}\,,
\end{equation}
where $\rd\eta$ is the Riemannian volume measure on $\bS_1^{d-1}$ and $N_{l,d}= \text{dim}\bY_l^d$ such that for $l$ sufficiently large $\text{dim}\bY^d_l = O(l^{d-2})$. Given $Y_l\in \bY_l^d$ and $H_l\in \bY_l(\bR^d)$, it holds $H_l(r\rx) = r^l Y_l(\rx)$ for all $\rx\in\bS_1^{d-1}$ and $r>0$~\cite[Sec.~2.1.3]{AH12}. For $d=3$, using Eq.~\eqref{eq:sphere_R3_laplacian} one gets
\begin{equation}
\Delta H_l(r\rx) = \Delta (r^lY_l)(\rx) = 0\quad \text{iff}\quad \Delta_{\bS_1} Y_l(\rx) = -l(l+1)\,Y_l(\rx)\,.
\end{equation}
Recall that $\Delta_R=\Delta_{\bS_1}/R^2$. The above expression implies that the Laplace-Beltrami operator $(-\Delta_R)$ has non-negative discrete spectrum and an isolated simple eigenvalue at zero, i.e., $0=\lambda_0<\lambda_1\leq \lambda_2\leq \cdots\to \infty$. Moreover, the non-zero functions in the space $\bY^3_l$ are the eigenfunctions of $(-\Delta_R)$ with eigenvalue $l(l+1)/R^2$ for all $l\in\bN_0$ and with the multiplicity $N_{l,3}=\mathrm{dim}\bY^3_l = (2l+1)$~\footnote{In general, the eigenvalues of the (negative) Laplace-Beltrami operator on $\bS_1^{d-1}$ are $l^2+(d-2)l$  where $l\in\bN_0$. The multiplicity of this eigenvalue is\,$
\frac{(l+d/2 -1)\prod^{d-3}_{j=1}(l+j)}{(d/2-1)\cdots (d-3)!}$~\cite[Thm.~2.9]{STV18}.}\cite[Thm.~1.4.5]{DX13}. Note that each $\bY^3_l$ is the eigenspace of $(-\Delta_R)$, which is preserved by the action of $O(3)$.

\begin{Remark}
Observe that the first non-zero eigenvalue is $\lambda_1= 2/R^2$, which can be viewed as twice the sectorial curvature $1/R^2$. This is an instance of a more general result known as the Lichnerowicz–Obata theorem~\cite{Li58} or~\cite[Thm.~5.1]{Li09}.
\end{Remark}

\subsection*{Legendre polynomials}\label{Legendre Polynomials}
The Legendre polynomial of degree $l$ in $d$-dimensions can be determined as the restriction of the Legendre harmonic on the unit sphere via~\cite[Eq.~(2.19)]{AH12}
\begin{equation}
P_{l,d}(t) = l!\,\Gamma(\frac{d-1}{2})\,\sum_{k=0}^{\lfloor l/2\rfloor}\,(-1)^k\,\frac{(1-t^2)^k\,t^{l-2k}}{4^k\,k!\, (l-2k)!\,\Gamma(k+(d-1)/2)}\,.
\end{equation}
In particular, for all $l\in\bN_0, d\geq 2, t\in[-1,1]$ it holds~\cite[Eq.~(2.116)]{AH12}
\begin{equation}\label{eq:bound_legendre}
|P_{l,d}(t)|\leq 1 = P_{l,d}(1)\,.
\end{equation}
The Legendre polynomial $P_{l,d}$ satisfies the following differential equation~\cite[Eq.~(2.82)]{AH12} or~\cite[(15.75)]{Arfken}
\begin{equation}
(1-t^2)^{\frac{3-d}{2}}\frac{\rd}{\rd t}\left[(1-t^2)^{\frac{d-1}{2}}\,\frac{\rd}{\rd t}P_{l,d}(t)\right] + l(l+d-2)P_{l,d}(t) = 0\,.
\end{equation}
Using the preceding expression with $t=\cos\theta$ and $d=3$ one has
\begin{equation}\label{eq:legendre_1}
\frac{\rd}{\rd\theta}\left(\cos\theta\,P_{l}(\cos\theta)\right) -\,\frac{\rd^2}{\rd\theta^2}\,\left(\sin\theta\,P_{l}(\cos\theta)\right) 
=
 l(l+1)P_{l}(\cos\theta)\sin\theta,
\end{equation}
where $P_l$ denotes the $l-$th Legendre polynomial in $d=3$.

\begin{Remark}\label{rem:er_legendre}
Let $j\in\bN$. We denote the derivatives of $t \mapsto P_{l}(t)$ with respect to $t\in [-1,1]$ of order $j$ by $ t \mapsto P^{(j)}_{l}(t)$. It holds that~\cite[Eqs.~(2.118--120)]{AH12}
\begin{equation}\label{eq:legendre_2}
\mathrm{max}_{t\in[-1,1]}\,|P^{(j)}_{l}(t)| = O(l^{2j})\,.
\end{equation}
Utilizing the mean value theorem for any $t,s\in[-1,1]$, and some $\tau$ between $t$ and $s$, one obtains
\begin{equation}\label{eq:legendre_3}
|P_l(t) - P_l(s)| = |P^1_l(\tau)(t-s)|\leq c\,l^2\,|t-s|,
\end{equation}
where we used Eq.~\eqref{eq:legendre_2} with $j=1$. This implies that
\begin{equation}
|P_{l}(\frac{\rx\cdot\ry}{R^2}) - P_{l}(\frac{\rz\cdot\ry}{R^2})|\leq c\,l^2\,\frac{|\rx-\rz|}{R},\qquad \forall~\rx,\ry,\rz\in\bS_R\,.
\end{equation}
\end{Remark}

\begin{Remark}\label{rem:spherical_projection}
For all $\rx,\ry\in\bS_R$, one has~\cite[Eq.~(16.57)]{Arfken},~\cite[Thm.~2.9]{AH12},~\cite[Eq.~(1.6.7)]{DX13}
\begin{equation}\label{eq:spherical_projec}
\sum_{m=-l}^{l} Y_{l,m}(\rx)\overline{Y}_{l,m}(\ry)
= \frac{(2l+1)}{4\pi R^2}\,P_l(\frac{\rx\cdot\ry}{R^2})
=(2l+1) \cP_{R,l}(\rx,\ry),
\end{equation}
where $\cP_{R,l}\,:\,L_2(\bS_R)\to L_2(\bS_R)$ is the integral operator with kernel $\cP_{R,l}(\rx,\ry)$, so that $(2l+1)\cP_{R,l}$ is the orthogonal projection onto the eigenspace $\bY_l^3$ of $-\Delta_R$. In particular, one obtains
\begin{equation}
\sum_{m=-l}^{l}\, Y_{l,m}(\rx)\overline{Y_{l,m}}(\rx)= \frac{(2l+1)}{4\pi R^2} \,P_l(1) =  \frac{(2l+1)}{ 4\pi R^2}\,.
\end{equation}
This implies that
$$\max \{|Y_{l,m}(\rx)|\,:\, \rx\in \bS_R\,\,,\,-l\leq m \leq l \} \leq \sqrt{(2l+1)/(4\pi R^2)}.$$
\end{Remark}

\begin{Remark}\label{rem:parseval_them}
Let $f\in L_2(\bS_R)$. One has~\cite[Thm.~2.38]{AH12}
\begin{equation}
f(\rx)=\sum_{l=0}^{\infty}\sum_{m=-l}^{l} f_{lm}\,Y_{l,m}(\rx) 
= \int_{\bS_R} \sum_{l=0}^{\infty}\sum_{m=-l}^{l} f(\ry) \overline{Y_{l,m}}(\ry)\, Y_{l,m}(\rx)\,\rho_R(\rd\ry)\,.
\end{equation}
Moreover, the projection of $f\in L_2(\bS_R)$ into $\bY^3_l$ can be read off as
\begin{equation}
(2l+1)\,(\cP_{R,l}f)(\rx) := f_{l}(\rx) = \frac{2l+1}{4\pi R^2}\int_{\bS_R}f(\ry)\,P_{l}(\rx\cdot \ry/R^2)\,\rho_R(\rd\ry),
\end{equation}
where we used Eq.~\eqref{eq:spherical_projec}. As a result, we have Parseval's equality on $L_2(\bS_R)$,
$$\|f\|^2_{L_2(\bS_R)} = \sum_{l=0}^{\infty}\|f_{l}\|^2_{L_2(\bS_R)},$$
which is a consequence of the fact that $L_2(\bS_R)= \bigoplus_{l=0}^{\infty}\bY^3_l$~\cite[Eq.~(2.143)]{AH12}.
\end{Remark}

\begin{Remark}\label{rem:legendre_ortho}
It holds~\cite[Eq.~2.79]{AH12}
\begin{equation}
\int_{-1}^1 P_{l_1}(t)\,P_{l_2}(t)\,\rd t = \frac{2}{2l_1+1}\,\delta_{l_1,l_2}\,.
\end{equation}
\end{Remark}

\section{Function spaces}\label{sec:spaces}
\begin{dfn}\label{dfn:weights}\cite[Def.~6.1]{Triebel3}
We say that $w\in C^\infty(\bR^d)$ is an admissible weight iff there exist $b\in[0,\infty)$ and $c\in(0,\infty)$ such that $0<w(x)\leq c\,w(y)\,(1+|x-y|)^b$ for all $x,y\in\bR^d$ and for every $a\in\mathbb{N}_0^d$ there exists $c_a\in(0,\infty)$ such that $|\partial^a w (x)|\leq c_a\, w(x)$ for all $x\in\bR^d$.
\end{dfn}

\begin{dfn}\label{def:sobolev_norms}
\begin{itemize}
\item[(A)]
Let $w\in C^\infty(\bR^d)$ be an admissible weight, $p\in[1,\infty]$ and $\alpha\in\bR$, $n\in\bN_0$. By definition $L_p(\bR^d,w)$ is the Banach space with the norm~\cite[Eq.~6.5]{Triebel3}

$$\|f\|_{L_p(\bR^d,w)}:=\|w f\|_{L_p(\bR^d)}.$$
\item[(B)]
The weighted Bessel potential space $L^\alpha_p(\bR^d,w)$ is the Banach space with the norm
$$ \|f\|_{L^\alpha_p(\bR^d,w)} := \|(1-\Delta)^{\alpha/2}f\|_{L_p(\bR^d,w)}.$$
We also set $L^\alpha_p(\bR^d)=L^\alpha_p(\bR^d,1)$. 
\item[(C)]
The weighted Sobolev space $W^n_p(\bR^d,w)$ is the Banach space with the norm
\begin{equation}
 \|f\|_{W^n_p(\bR^d,w)} = \textstyle\sum_{\substack{a\in\bN^d,|a|\leq n}} \|\partial^a f\|_{L_p(\bR^d,w)}\,.
 \end{equation}
 \end{itemize}
\end{dfn}

\begin{dfn}
Let $\|\Cdot\|_{L_p(\bS_R)}$ denote the $L_p$-norm on $\bS_R$ defined with respect to the canonical measure $\rho_R(\rd\rx)$. Note that for all $0<q\leq p\leq \infty$ and all $R\in\bN_+$, there exists $C\in (0,\infty)$ such that $\|\Cdot\|_{L_q(\bS_R)}\leq C\, \|\Cdot\|_{L_p(\bS_R)}$. For $R\in(0,\infty)$, we define the Bessel potential space $L^\alpha_p(\bS_R)$ on $\bS_R\subset\bR^3$ as the Banach space with the norm
$$ \|f\|_{L^\alpha_p(\bS_R)} := \|(1-\Delta_R)^{\alpha/2}f\|_{L_p(\bS_R)}.$$ 
\end{dfn}

\begin{dfn}\cite[Def.~2.3.1]{Triebel1}
Let $\Phi(\bR^d)$ be a collection of functions $\{\varphi_j\}_{j\in\bN_0}\subset \sS(\bR^d)$ such that for some $A,B,C>0$, it holds that
\begin{equation}
\supp\,\varphi_0 \subset \{x\in\bR^d\,:\, |x|<A\},\quad \supp\,\varphi_j \subset \{x\in\bR^d\,:\, B\,2^{j-1}\leq |x|\leq C\,2^{j+1}\}.
\end{equation}
Moreover, for all $\alpha\in\bN_0^n$, there exists $c_{\alpha}>0$ such that for all $x\in\bR^d$, one has
\begin{equation}
\sup_{x\in\bR^d}\,\sup_{j=0,1,2\cdots}\,2^{j|\alpha|}\,|\partial^{\alpha}\,\varphi_j(x)|\leq c_{\alpha},
\quad \sum_{j=0}^{\infty}\varphi_j(x) = 1\,.
\end{equation}
We call $\{\varphi_j\}_{j\in\bN_0}$ a resolution of unity.
\end{dfn}
\begin{Remark}
Let $\{\varphi_j\}_{j\in\bN_0}\in \Phi(\bR^d)$. By construction, for all $f\in \sS'(\bR^d)$ one has
$$f=\sum_{j=0}^{\infty}\,\cF^{-1}\varphi_j\cF\,f =\sum_{j=0}^{\infty}\,\Delta_j\,f,$$
where $\Delta_j\,f:=\cF^{-1}\varphi_j\cF\,f $, and the operator $\Delta_j f$ is called Littlewood-Paley dyadic block. Observe that each $\Delta_j f$ is an analytic function, i.e., it can be evaluated at a point.
\end{Remark}

\begin{dfn}~\cite[Def.~6.3]{Triebel3}
Let $s\in\bR, 0<q\leq \infty, \{\varphi_j\}_{j\in\bN_0}\in\Phi(\bR^d)$ and $f\in \sS'(\bR^d)$.
We define the weighted Triebel-Lizorkin norm for $0<p< \infty$ and the weighted Besov norm for $0<p\leq \infty$ via
\begin{equation}
\begin{aligned}
\|f\|_{F^s_{p,q}(\bR^d,w)}:=& \| (\sum_{j=0}^{\infty}\,2^{jsq}\,|\cF^{-1}\varphi_j\cF f(\Cdot)|^q)^{1/q} \|_{L_p(\bR^d,w)}<\infty\, ,
\\
\|f\|_{B^s_{p,q}(\bR^d,w)}:=& \Big(\sum_{j=0}^{\infty}\,2^{jsq}\,\|\cF^{-1}\varphi_j\cF f\|^q_{L_p(\bR^d,w)}\Big)^{1/q}<\infty\,.
\end{aligned}
\end{equation}
\end{dfn}

\begin{Remark}\label{rem:lift_wight}
Let $\alpha\in\bR$, $w\in C^\infty(\bR^d)$ be an admissible weight, and $f\in\sS'(\bR^d)$. 
Suppose $0<p\leq \infty$ ($p<\infty$ for the Triebel-Lizorkin spaces) and $0<q\leq \infty$.
\begin{itemize}
\item[(A)]
The operator $f\mapsto w f$ is an isomorphic mapping from $B^s_{p,q}(\bR^d,w)$ onto $B^s_{p,q}(\bR^d)$ and from $F^s_{p,q}(\bR^d,w)$ onto $F^s_{p,q}(\bR^d)$~\cite[Thm.~6.5]{Triebel3}. This implies that the following norms are equivalent~\footnote{If $c\|\Cdot\|_B\leq \|\Cdot\|_A\leq C\|\Cdot\|_{B}$ for $c,C\in(0,\infty)$, then $\|\Cdot\|_A\sim\|\Cdot\|_B$.}
\begin{equation}
\|w\,f\|_{B^s_{p,q}(\bR^d)}\sim\|f\|_{B^s_{p,q}(\bR^d,w)},
\quad
\|w\,f\|_{F^s_{p,q}(\bR^d)}\sim\|f\|_{F^s_{p,q}(\bR^d,w)}.
\end{equation}
\item[(B)]
Let $I_{\alpha}$ be the operator defined by
$$I_{\alpha}\,:\, f\mapsto \cF^{-1}\left( 
(1+|\zeta|^{2})^{\alpha/2}\cF f\right),$$
called the lifting operator. It holds that~\cite[Thm.~6.9]{Triebel3}
\begin{equation}
I_{\alpha} B^{s}_{p,q}(\bR^d,w) = B^{s-\alpha}_{p,q}(\bR^d,w),
\quad
I_{\alpha} F^{s}_{p,q}(\bR^d,w) = F^{s-\alpha}_{p,q}(\bR^d,w)\,.
\end{equation}
\item[(C)]
For all $1<p<\infty$ and $m\in\bN_0$ it holds that~\cite[Eq.~(6.21)]{Triebel3}
$$W^m_p(\bR^d,w)=F^m_{p,2}(\bR^d,w).$$
\end{itemize}
\end{Remark}

\begin{lem}\label{rem:sobolev}
Let $w\in C^\infty(\bR^d)$ be an admissible weight, $\alpha\in\bR$, and $n\in\bN_0$. 
\begin{itemize}
\item[(A)]
For all $p\in(1,\infty)$, the norms $\|\Cdot\|_{L^\alpha_p(\bR^d,w)}$ and $\|w\Cdot\|_{L^\alpha_p(\bR^d)}$ are equivalent.

\item[(B)]
For all $p\in(1,\infty)$, the Bessel potential space $L^\alpha_p(\bR^d,w)$ coincides with the Triebel-Lizorkin space $F^\alpha_{p,2}(\bR^d,w)$ with equivalent norms. 

\item[(C)]
For all $p\in(1,\infty)$, the Sobolev space $W^n_p(\bR^d,w)$ coincides with the Bessel potential space $L^n_p(\bR^d,w)$ with equivalent norms.

\item[(D)]
For all $1\leq p\leq\infty$, there are continuous embeddings
$$B_{p,1}^{\alpha}(\bR^d,w) \to L_p^{\alpha}(\bR^d,w) \to B_{p,\infty}^{\alpha}(\bR^d,w).$$
\end{itemize}
\end{lem}

\begin{proof}
From Remark~\ref{rem:lift_wight} Item~(A) with $q=2$, one has $\|\Cdot\|_{F^{\alpha}_{p,2}(\bR^d,w)}\sim\|w\,\Cdot\|_{F^{\alpha}_{p,2}(\bR^d)}$. Moreover, for all $p\in(1,\infty)$, it holds that $\|\Cdot\|_{L^{\alpha}_p(\bR^d)}\sim \|\Cdot\|_{F^{\alpha}_{p,2}(\bR^d)}$~\cite[Sec.~2.5.6]{Triebel1} or~\cite[Sec.~3.1.2]{YS18}. This concludes Item~(A), and Item~(B). Evoking Item~(B) with $\alpha=n\in\bN_+$, combined with Remark~\ref{rem:lift_wight} Item~(C), one infers that both $\|\Cdot\|_{L^n_p(\bR^d,w)}$ and $\|\Cdot\|_{W^n_p(\bR^d,w)}$ are equivalent to $\|\Cdot\|_{F^n_{p,2}(\bR^d,w)}$. This concludes Item~(C). To prove Item~(D), one utilizes the fact that~\cite[Sec.~2.5.7, Prop.~1]{Triebel1}
$$B^0_{p,1}(\bR^d) \to L_p(\bR^d) \to B^0_{p,\infty}(\bR^d).$$
This yields the the continuous embeddings
$$B^0_{p,1}(\bR^d,w) \to F^0_{p,2}(\bR^d,w) \to B^0_{p,\infty}(\bR^d,w),$$
which follows from Item~(B) and Remark~\ref{rem:lift_wight} Item~(A). Applying Remark~\ref{rem:lift_wight} Item~(B) with $I_{-\alpha}$, one obtains
 $$B^{\alpha}_{p,1}(\bR^d,w) \to F^{\alpha}_{p,2}(\bR^d,w) \to B^{\alpha}_{p,\infty}(\bR^d,w).$$
Item~(B) concludes Item~(D) and finishes the proof.
\end{proof}

\begin{Remark} \label{Hoelder-remark}
For $\alpha\in\bR$ and $p,q\in[1,\infty]$ we have the following generalized H{\"o}lder inequality
\begin{equation}
 |\langle f,g\rangle_{L_2(\bR^d,w^{1/2})}| \leq C\, \|f\|_{L^\alpha_p(\bR^d,w^{1/p})}\,\|g\|_{L^{-\alpha}_q(\bR^d,w^{1/q})},
 \qquad
 \frac{1}{p}+\frac{1}{q}=1,
\end{equation}
where $\langle \Cdot,\Cdot\rangle_{L_2(\bR^d,w^{1/2})}$ is the scalar product in $L_2(\bR^d,w^{1/2})$ and the constant $C\in(0,\infty)$ depends only on the weight $w$. 
\end{Remark}
\noindent
\newcommand{\al}{\alpha}

\begin{proof} 
Suppose first that $1<p<\infty$. By the H\"older inequality, we obtain
\begin{align}
\langle f,g\rangle_{L_2(\bR^d,w^{1/2})}=& \int \rd x \,\bar{f}(x)  g(x) w(x)
 = \int \rd x\, (w^{1/p}\bar{f}) (x)\, (w^{1/q}g)(x)  \\
=& \int \rd x\, (1-\De)^{\al/2}(w^{1/p}\bar{f}) (x)\,   (1-\De)^{-\al/2}(w^{1/q}g)(x) \\
\leq \,& \|w^{1/p} f\|_{L^{\al}_p(\bR^d)}   \|w^{1/q} g\|_{L^{-\al}_q(\bR^d)}\,.
\end{align}
Now, by Lemma~\ref{rem:sobolev}, Item~(B), and Remark~\ref{rem:lift_wight}, Item~(A), we write
\begin{equation}
 \|w^{1/p} f\|_{L^{\al}_p(\bR^d)}\leq c  \|w^{1/p} f\|_{F^\alpha_{p,2}(\bR^d)} \leq c' \|f\|_{F^\alpha_{p,2}(\bR^d, w^{1/p} )} \leq c'' \|f\|_{L^{\al}_p(\bR^d, w^{1/p} )}  \label{wf} 
 \end{equation}
and analogously for $\|w^{1/q} g\|_{L^{-\al}_q(\bR^d)}$. This concludes the proof for $1<p<\infty$. 
For $p=\infty$ and $q=1$ (or vice versa)  $L^{\alpha}_{\infty}(\bR^d)$ is understood as the H\"older-Zygmund class $B^{\alpha}_{\infty,\infty}(\bR^d)$. As the restriction $p<\infty$ in Remark~\ref{rem:lift_wight} does not apply
to the Besov spaces, we conclude as in Eq.~\eqref{wf}. 
\end{proof}

\begin{thm}\label{thm:embedding}
Let $w,v$ be admissible weights and 
\begin{equation}
 -\infty<\alpha_2\leq \alpha_1 <\infty,
 \qquad
 1\leq p_1\leq p_2\leq \infty.
\end{equation}
(A) The embedding $L^{\alpha_1}_{p_1}(\bR^d,w)\to L^{\alpha_2}_{p_2}(\bR^d,v)$ is continuous if
\begin{equation}
 p_2<\infty,
 \qquad
 \alpha_1-d/p_1\geq \alpha_2-d/p_2 
 \quad\textrm{and}\quad
 \sup_{x\in\bR^d} v(x)/w(x) < \infty.
\end{equation}
(B) The embedding $L^{\alpha_1}_{p_1}(\bR^d,w)\to L^{\alpha_2}_{\infty}(\bR^d,v)$ is continuous if
\begin{equation}
 \alpha_1-d/p_1> \alpha_2
 \quad\textrm{and}\quad
 \sup_{x\in\bR^d} v(x)/w(x) < \infty.
\end{equation}
(C) The embedding $L^{\alpha_1}_{p_1}(\bR^d,w)\to L^{\alpha_2}_{p_2}(\bR^d,v)$ is compact if  
\begin{equation}
p_2<\infty,
\qquad
\alpha_1-d/p_1>\alpha_2-d/p_2
\quad\textrm{and}\quad
\lim_{|x|\to\infty} v(x)/w(x) =0.
\end{equation}
\end{thm}
\begin{proof} Parts (A) and (C) follow from  \cite[Sec.\,4.2.3, Thm.]{EdmundsTriebel} and the equivalence between $L^\alpha_p(\bR^d,w)$ and $F^\alpha_{p,2}(\bR^d,w)$ mentioned in Lemma~\ref{rem:sobolev} above. Part (B) is covered by~\cite[Sec. 4.2.3, Remark]{EdmundsTriebel} and the embeddings stated 
in Lemma~\ref{rem:sobolev}.
\end{proof}

\begin{thm}\label{thm:sobolev_multiplication}
Let $w$ be an admissible weight, $\alpha\in[0,\infty)$ and $p,p_1,p_2\in[1,\infty)$ be such that $1/p=1/p_1+1/p_2$. Then, there exists $C\in(0,\infty)$ such that for all $f\in L^\alpha_{p_1}(\bR^d,w^{1/p_1})$, and $g\in L^\alpha_{p_2}(\bR^d,w^{1/p_2})$
\begin{equation}
 \|fg\|_{L^\alpha_p(\bR^d,w^{1/p})}
 \leq C\,\|f\|_{L^\alpha_{p_1}(\bR^d,w^{1/p_1})} 
 \,
 \|g\|_{L^\alpha_{p_2}(\bR^d,w^{1/p_2})}.
\end{equation}
\end{thm}
\begin{proof}
The statement follows from the equivalence of the norms $\|\Cdot\|_{L^\alpha_p(\bR^d,w)}$ and $\|w\Cdot\|_{L^\alpha_p(\bR^d)}$, the fractional Leibniz rule~\cite[Ch.~2]{Muscalu} and Theorem A.5 (A).  Alternatively, one can use \cite[Lemma 5]{BM01}.
\end{proof}

\begin{thm}\label{thm:sobolev_interpolation}
Let $w$ be an admissible weight, $p_1,p_2\in[1,\infty)$, $\alpha_1,\alpha_2\in\bR$, $\theta\in(0,1)$ and 
\begin{equation}
 \alpha=\theta\,\alpha_1+(1-\theta)\,\alpha_2,
 \qquad\qquad
 \frac{1}{p} = \frac{\theta}{p_1} + \frac{1-\theta}{p_2}.
\end{equation}
There exists $C\in(0,\infty)$ such that for all $f\in L^{\alpha_1}_{p_1}(\bR^d,w^{1/p_1})\cap L^{\alpha_2}_{p_2}(\bR^d,w^{1/p_2})$ it holds
\begin{equation}
 \|f\|_{L^\alpha_p(\bR^d,w^{1/p})} \leq C\,\|f\|^\theta_{L^{\alpha_1}_{p_1}(\bR^d,w^{1/p_1})}\, \|f\|^{1-\theta}_{L^{\alpha_2}_{p_2}(\bR^d,w^{1/p_2})}\,.
\end{equation}
\end{thm}
\begin{proof} 
The statement is a consequence of the H{\"o}lder inequality, cf. \cite[Sec 3]{BM01}. Applying first the H\"older inequality for series with $p'=\frac{1}{\theta}$ and $q'=\frac{1}{1-\theta}$ and then the H\"older inequality for integrals with $\frac{1}{p} = \frac{\theta}{p_1} + \frac{1-\theta}{p_2}$ one has
\begin{multline}
 \|f\|_{L^\alpha_p(\bR^d,w^{1/p})}= \|f\|_{F^\alpha_{p,2}(\bR^d,w^{1/p})} =
 \| (\sum_{j=0}^{\infty}\,2^{2j\alpha }\,|\cF^{-1}\varphi_j\cF f(\Cdot)|^2)^{1/2} \|_{L_p(\bR^d,w^{1/p})}
  \\
\leq  \| w^{1/p}\,( \sum_{j=0}^{\infty}\,2^{2j\alpha_1}\,|\cF^{-1}\varphi_j\cF f(\Cdot)|^{2}\,)^{\theta/2}\,(\sum_{j=0}^{\infty}\,2^{2j \alpha_2 }\,|\cF^{-1}\varphi_j\cF f(\Cdot)|^{2}\,)^{(1-\theta)/2} \|_{L_p(\bR^d )}
\\
\leq 
 \| w^{\frac{\theta}{p_1}}(\Cdot)\,( \sum_{j=0}^{\infty}\,2^{2j  \alpha_1 }\,|\cF^{-1}\varphi_j\cF f(\Cdot)|^{2}\,)^{\theta/2}\|_{L_{p_1/\theta}(\bR^d)}
 \\
\times\,\|w^{\frac{1-\theta}{p_2}}(\Cdot)\,(\sum_{j=0}^{\infty}\,2^{2j  \alpha_2}\,|\cF^{-1}\varphi_j\cF f(\Cdot)|)^{2}\,)^{(1-\theta)/2}  \|_{L_{p_2/(1-\theta)}(\bR^d)}
\\
\leq C\,\|f\|^\theta_{ F^{\alpha_1}_{p_1,2}(\bR^d,w^{1/p_1}) }\, \|f\|^{1-\theta}_{F^{\alpha_2}_{p_2,2}(\bR^d,w^{1/p_2})}\,.
\end{multline}
Using the fact that $\|\Cdot\|_{L^{\alpha}_p(\bR^d,w^{1/p})}$ is equivalent with $\|\Cdot\|_{F^{\alpha}_{p,2}(\bR^d,w^{1/p})}$, one concludes the proof.
\end{proof}

\begin{lem}\label{lem:Z_Psi_n}
Let $w\in L_1(\bR^2)$ be an admissible weight, $n\in\{3,4,\ldots\}$, $\delta\in(0,\infty)$ and $\kappa\in(0,2/(n-1)(n-2))$. 
Then there exists $C\in(0,\infty)$ and $p\in[1,\infty)$ such that for all \mbox{$m\in\{1,\ldots,n-1\}$} and $\Psi\in L_2^1(\bR^2,w^{1/2}) \cap L_n(\bR^2,w^{1/n})$, $Z\in L^{-\kappa}_p(\bR^2,w^{1/p})$  it holds
\begin{equation}
 |\langle Z,\Psi^m \rangle_{L_2(\bR^2,w^{1/2})}| 
 \leq
 C\,\|Z\|_{L^{-\kappa}_p(\bR^2,w^{1/p})}^p+
 \delta\,\|\vec\nabla\Psi\|_{L_2(\bR^2,w^{1/2})}^2+
 \delta\,\|\Psi\|_{L_n(\bR^2,w^{1/n})}^n+\delta.
\end{equation}
\end{lem}
\begin{proof}
Let $1/r=(1-\kappa)/n+\kappa/2$, $1/q = m/r$ and $1/\pnew=1-1/q$. By H{\"o}lder's inequality
\begin{equation}
 |\langle Z,\Psi^m \rangle_{L_2(\bR^2,w^{1/2})}| \leq C\,\|Z\|_{L^{-\kappa}_{\pnew}(\bR^2,w^{1/\pnew})}\,\|\Psi^m\|_{L^\kappa_q(\bR^2,w^{1/q})},
\end{equation}
for some $C\in(0,\infty)$. Theorem~\ref{thm:sobolev_multiplication} implies that
\begin{equation}
 \|\Psi^m\|_{L^\kappa_q(\bR^d,w^{1/q})} \leq C\, \|\Psi\|_{L^\kappa_{r}(\bR^d,w^{1/r})}^m
\end{equation}
and Theorem~\ref{thm:sobolev_interpolation} implies that
\begin{equation}
 \|\Psi\|_{L^\kappa_r(\bR^d,w^{1/r})} \leq C\,\|\Psi\|_{L^1_2(\bR^d,w^{1/2})}^\kappa\,\|\Psi\|_{L_n(\bR^d,w^{1/n})}^{1-\kappa}
\end{equation}
for some $C\in(0,\infty)$. Combining the above bounds we obtain
\begin{equation}
 |\langle Z,\Psi^n \rangle_{L^2(\bR^d,w^{1/2})}| 
 \leq C\,\|Z\|_{L^{-\kappa}_{\pnew}(\bR^d,w^{1/\pnew})}\,
 \|\Psi\|_{L^1_2(\bR^d,w^{1/2})}^{m\kappa}
 \,\|\Psi\|_{L_n(\bR^d,w^{1/n})}^{m(1-\kappa)}
\end{equation}
for some $C\in(0,\infty)$. Hence, by the Young inequality with $1/p' +(m\kappa)/2+(m(1-\kappa))/n =1$ for every $\delta\in(0,\infty)$ there is $C\in(0,\infty)$ such that
\begin{multline}\label{eq:Z_Psi_n_proof}
 |\langle Z,\Psi^m \rangle_{L_2(\bR^2,w^{1/2})}| 
 \\
 \leq C\,\|Z\|_{L^{-\kappa}_{\pnew}(\bR^2,w^{1/\pnew})}^{\pnew}\,
+\delta\,\|\Psi\|_{L^1_2(\bR^2,w^{1/2})}^2+\delta\,\|\Psi\|_{L_n(\bR^2,w^{1/n})}^n.
\end{multline}
Observe that by the H\"older inequality and the assumption $w\in L_1(\bR^2)$ for all $q,r\in[1,\infty)$ such that $q\leq r$ there exists $C\in(0,\infty)$ such that $\|\Cdot\|_{L_q(\bR^2,w^{1/q})}\leq C\,\|\Cdot\|_{L_r(\bR^2,w^{1/r})}$. Hence, the bound~\eqref{eq:Z_Psi_n_proof} implies the statement of the lemma with $1/p = (2-\kappa(n-1)(n-2))/2n$. 
\end{proof}

\begin{lem}\label{lem:embedding_sphere}
Let $p\in[2,\infty)$ and $\alpha=1-2/p$. Then there exists $C\in(0,\infty)$ such that $\|f\|_{L_p(\bS_R)}\leq C\,\|f\|_{L_2^\alpha(\bS_R)}$ for all $f\in L_2^\alpha(\bS_R)$ and all $R\in[1,\infty)$.
\end{lem}
\begin{proof} See e.g.~\cite[Thm.~6]{Be01} or~\cite[Thm.~II.2.7(ii)]{Varopoulos}.
\end{proof}

\begin{Remark}\label{rem:bessel_sobolev_psi}
Let $p\in[2,\infty)$. By Lemma~\ref{lem:embedding_sphere}, there exist $C,C' \in (0,\infty)$ such that for all $\Psi\in L_2^{2-2/p}(\bS_R)$, and all $R\in[1,\infty)$, one has
\begin{multline}
\|\vec\nabla_R\Psi\|_{L_p(\bS_R)}
\leq C\,\|\vec\nabla_R\Psi\|_{L_2^{1-2/p}(\bS_R)}
\leq C\,\|(1-\Delta_R)^{1/2}\, \Psi\|_{L_2^{1-2/p}(\bS_R)}
\leq C'\,\|\Psi\|_{L_2^{2-2/p}(\bS_R)}.
\end{multline}
\end{Remark}

\section{Functional analysis}\label{sec:functional_analysis}
Let $H$ and $E$ be two separable Hilbert spaces which are equipped with the norms $\|\Cdot\|_{H}$ and $\|\Cdot\|_E$, respectively, $\cL(H)$ denote the vector space of all bounded linear operators from $H$ to $H$ and $T\,:\,\mathrm{Dom}(T)\subset H\to H$ be a self-adjoint operator with the spectral measure $\cP_T$ such that $T=\int_{\bR}\,\lambda\,\cP_T(\rd\lambda)$.
\begin{Remark}\label{rem:commute}
The following definitions and facts are standard\,:\,\,%
\begin{itemize}
\item[(A)]
The operator $T$ commutes with a bounded operator $S\in\cL(H)$ if $R(z)= (T-z)^{-1}$ commutes with $S$.
\item[(B)]
We denote by $\Phi\,:\,B(\sigma(T),\mathbb{C}) \ni F\mapsto F(T)\in \cL(H)$ the standard functional calculus for self-adjoint   operators, where $B(\sigma(T),\mathbb{C})$ is the space of Borel functions. We recall that if $T$ commutes with some $S\in \cL(H)$ then also $F(T)$ commute with this $S$.
\item[(C)]
An operator $S\in\cL(H)$ commutes with $T$ if and only if $S$ commutes
with every spectral projection associated to $T$~\cite[Thm.~12.22]{Rudin91}.

\item[(D)] $T$ is compact if and only if its spectrum consists of eigenvalues of finite multiplicity and
forms a finite set or a sequence converging to zero.

\item[(F)]
A bounded linear operator $K\,:\, H\to E$ is called Hilbert-Schmidt operator if for some orthonormal basis $\{e_{\alpha}\}$ in $H$ it holds $\|K\|_{\mathrm{HS}}:= \left(\sum_{\alpha}\,\|K\,e_{\alpha}\|_E^2\right)^{1/2}<\infty$. We call $\|\Cdot\|_{\mathrm{HS}}$ the Hilbert-Schmidt norm of the operator $K$. Equivalently, if $K\in\mathcal{L}(H)$ is Hilbert-Schmidt, one has $\|K\|^2_{\mathrm{HS}}=\mathrm {Tr}(KK^*)<\infty$.
\end{itemize}
\end{Remark}

\begin{lem}(Schur’s test)\label{lem:schur-test}
 Let $(X,\mu)$, $(Y,\nu)$ be $\sigma-$finite measure spaces and let $K\,:\,X\times Y\to \bR$ be a measurable kernel for which there exist measurable functions $p\,:\,X\to (0,\infty)$, $q\,:\,Y\to (0,\infty)$ and constants $\alpha,\beta >0$ such that for a.e. $x\in X$ and a.e. $y\in Y$ it holds that
\begin{equation}
\begin{gathered}
\int_{Y}\,|K(x,y)|\,q(y)\,\nu(\rd y)\leq \alpha\,p(x),
\quad
\int_{X}\,|K(x,y)|\,p(x)\,\mu(\rd x)\leq \beta\,q(y)\,,
\end{gathered}
\end{equation}
respectively. Then, the mapping $T\,:\, L_2(Y,\nu)\to L_2(X,\mu)$ defined by
$$Tf(x) = \int_{Y}\,K(x,y)\,f(y)\,\nu(\rd y)$$
is a bounded operator such that its norm is bounded from above by $\sqrt{\alpha \beta}$.
\end{lem}

\begin{Remark}
For $R\in\bN_+$, and $k\in\mathbb{Z}$, it holds that
\begin{equation}
\mathrm{Dom}\,(-\Delta_R)^k=\mathrm{Dom}(1-\Delta_R)^k= L^{2k}_2(\bS_R),
\quad
\mathrm{Dom}\,(-\Delta)^k=\mathrm{Dom}(1-\Delta)^k= L^{2k}_2(\bR^2).
\end{equation}
\end{Remark}

\begin{Remark}\label{rem:taylor_semigroup}
Let $H$ be a separable Hilbert space and $S$ denote a $C_0$-semigroup on $H$ generated by the self-adjoint operator $T$, i.e., $S(t)=\re^{-t\,T}$. Writing up the Taylor formula with remainder $R_n$ results in
\begin{equation}
S(s) f= \sum_{k=0}^n\,\frac{(s-t)^k}{k!}\,\frac{\rd^k}{\rd t^k}\,S(t) f + R_n(s,t)f,
\end{equation}
where
\begin{equation}
R_n(s,t)f = \frac{1}{n!}\,\int_t^s\,(s-r)^n\,\frac{\rd^{n+1}}{\rd r^{n+1}}\,S(r)\,f \rd r,
\end{equation}
provided $f\in \mathrm{Dom}(T^{n+1})$.
\end{Remark}
\subsection*{Properties of $(-\Delta_R)$}\label{subsec:deltaR} 
Note that  $(-\Delta_R)$ has discrete spectrum accumulating to infinity as was shown in Sec.~\ref{sec:spherical_harmonics}. By Remark~\ref{rem:commute}~Item~(D) the operator $(1-\Delta_R)^{-1}$ is compact.
\begin{Remark}\label{rem:spectral}
Evoking the spectral calculus for the self-adjoint operator $(-\Delta_R)$, for any Borel measurable function $F\,:\,\mathrm{spect}(-\Delta_R)\to \bR$, one defines the operator
\begin{equation}
\begin{aligned}
F(-\Delta_R) := \sum_{l=0}^{\infty}\,(2l+1)\,F(l(l+1)/R^2)\,\cP_{R,l}
 =  \sum_{l=0}^{\infty}\,(2l+1)\,\mathrm{Tr}( F(-\Delta_R)\cP_{R,l})\,\cP_{R,l}\,,
\end{aligned}
\end{equation}
where $\cP_{R,l}$ is given in Remark~\ref{rem:spherical_projection}. In particular, $\mathrm{spec}(F(-\Delta_R)) = F(\mathrm{spec}(-\Delta_R))$. Moreover, one has
$$\|F(-\Delta_R)\|\leq \|F\|_{\infty},
\quad\text{where}\quad\|F\|_{\infty}=\mathrm{max}\{| F(l(l+1)/R^2)|\,:\,l\in\bN_0\}.$$
\end{Remark}

\begin{Remark}\label{rem:expansion_covariance_harmo}
Let $G_{R,N}:= (1-\Delta_R)^{-1}(1-\Delta_R/N^2)^{-2}$ and $f\in L_2(\bS_R)$ such that
$$(G_{R,N} f)(\rx) = \int_{\bS_R} G_{R,N}(\rx,\rx')\,f(\rx')\,\rho_R(\rd\rx').$$
Using Remark~\ref{rem:spectral}, one has
\begin{equation}\label{eq:trace_regularized_covarince}
\begin{aligned}
G_{R,N}(\rx,\rx')&=  \frac{1}{4\pi R^2}\sum_{l=0}^{\infty}\frac{(2l+1)}{(1+l(l+1)/R^2)(1+l(l+1)/(NR)^2)^2}\,P_l(\frac{\rx\cdot\rx'}{R^2}).
\end{aligned}
\end{equation}
Observe that for all $Q\in O(3)$, and all $\rx,\rx'\in\bS_R$, one obtains
$$G_{R,N}(\rx,\rx')=G_{R,N}(Q\rx,Q\rx').$$
This is due to the invariance property of the $l-$th Legendre polynomial $P_l$, which is given in Remark~\ref{rem:spherical_projection}. Furthermore, it holds that
\begin{equation}
\begin{gathered}
\mathrm{Tr}(G_{R,N})=\sum_{l=0}^{\infty}\,\frac{(2l+1)}{(1+l(l+1)/R^2)(1+l(l+1)/R^2N^2)^2}\,,
\end{gathered}
\end{equation}
where we used the facts that $P_l(1) = 1$ and $1/(4\pi R^2)\int_{\bS_R}\rho_R(\rd\rx) = 1$\,.
\end{Remark}

\begin{Remark}\label{rem:fundamental_solution}
Let $G$ denote a fundamental solution for the positive self-adjoint operator $(1-\Delta)$, where $\Delta$ is the Laplacian on $\bR^2$, i.e., $(1-\Delta)G(x,x')=\delta (x-x')$. One writes the resolvent of the operator $(1-\Delta)$ as a convolution, i.e., $(1-\Delta)^{-1} \phi = G\ast_{\bR^2}\phi$ for all $\phi\in L_2(\bR^2)$. Using the translational invariance property of $G$ one determines the mapping $G\,:\, \bR^2\to(0,\infty)$
 \begin{equation}
 G(x) := \frac{1}{4\pi}\int_0^{\infty}\re^{-\pi|x|^2/t}\,\re^{-t/(4\pi)}\,\frac{\rd t}{t}
 \end{equation}
such that $\|G\|_{L_1(\bR^2)} = 1$. Furthermore, there exist constants $C,c\in (0,\infty)$ such that for all measurable sets $A,B\subset \bR^2$ and all $\phi\in L_2(\bR^2)$ with $\supp(\phi)\subset A$, it holds~(cf.\cite[Sec.~11.1]{ISEM27})
\begin{equation}
\|(1-\Delta)^{-1}\,\phi\|_{L_2(B)} \leq C\,\exp(-c\,\mathrm{dist}(A,B))\,\|\phi\|_{L_2(A)}\,.
\end{equation}
On account of Remark~\ref{rem:Fourier_Bessel_Fundamental}, as $\|x-x'\|\to\infty$, one infers
\begin{equation}
G(x,x')
= \frac{\re^{-\|x-x'\|}}{\sqrt{8\,\pi\,\|x-x'\|}}\,(1+o(1)).
\end{equation}
Observe that $G(x,x')$ decays exponentially at infinity. Using once again Remark~\ref{rem:Fourier_Bessel_Fundamental}, as $\|x-x'\|\to 0$, it follows that
\begin{equation}\label{eq:Fourier_Bessel_diverging}
G(x,x') = -\frac{1}{2\pi}\log \|x-x'\|\,(1+ o(1)),
\end{equation}
where $\|\Cdot\|$ denotes the Euclidean norm on $\bR^2$. The above asymptotic behaviour of the integral kernel $G$ implies that the covariance of the Gaussian measure $\cN(0,(1-\Delta)^{-1})$ in Sec.~\ref{sec:measure_on_plane} diverges near the diagonal, i.e., $G(x,x')$ is $C^{\infty}(\bR^2)$ except at $x=x'$. 
\end{Remark}

\begin{Remark}\label{rem:fundamental_solution_sphere}
Note that $\bS_R$ is a Riemannian manifold, hence it behaves locally like a Euclidean space\footnote{The $d$-dimensional sphere is locally homeomorphic to $\bR^d$ at every point~\cite[Prop.~18.2.1]{Berger2}.}. Accordingly, the fundamental solution of the Laplace equation on $\bS_R$ is expected to behave locally like that on $\bR^2$. Let $G_R$ denote the fundamental solution of the positive self-adjoint operator $(1-\Delta_R)$. Using Eq.~\eqref{eq:Fourier_Bessel_diverging} together with Remark~\ref{def:sphere_distance}, one expects that $G_R$ has a logarithmic singularity of the form $\log(d_R(\rx,\rx'))$ along the diagonal. In particular, for small $\theta$, i.e., near coinciding points, $G_R$ diverges. An explicit formula for the integral kernel of $(1-\Delta_1)$ is given in~\cite[Eq.~(1.8) and Thm.~4.1]{DGR24}, and generalized to arbitrary $R>0$ in~\cite[Sec.~4.3]{DGR24}.
\end{Remark}

\begin{lem}\label{lem:trace}\cite[Lemma\,B.1]{DDJ}
Let $\kappa\in(0,\infty)$. There exists $C\in(0,\infty)$ such that for all $R,N\in[1,\infty)$, it holds that
\begin{equation}
 -C\leq \sum_{l=0}^\infty \frac{(2l+1)}{2R^2\,(1+l(l+1)/R^2)\,(1+l(l+1)/(NR)^2)^\kappa} - \log (N+1) \leq C.
\end{equation}
\end{lem}
\begin{proof}
Observe that the expression in the statement of the lemma coincides with
\begin{equation}
 \int_0^\infty \frac{(2\floor{l}+1)\,\rd l}{2R^2\,(1+\floor{l}(\floor{l}+1)/R^2)\,(1+\floor{l}(\floor{l}+1)/(NR)^2)^\kappa} - \int_0^\infty \frac{\rd l}{(1+l)(1+(1+l)/N)}.
\end{equation}
The absolute value of the above expression is bounded by
\begin{equation}
 \int_0^\infty \left|\frac{(2\floor{R l}+1)/R}{2(1+\floor{R l}(\floor{R l}+1)/R^2)\,(1+\floor{R l}(\floor{R l}+1)/(NR)^2)^\kappa} - \frac{1}{(1+l)(1+(1+l)/N)}\right|\rd l.
\end{equation}
Using $0 \leq l-\floor{R l}/R \leq 1$, we show that there exists $\hat C\in(0,\infty)$ such that the above expression is bounded by
\begin{equation}
 \hat C+ \int_0^\infty \left|\frac{1}{(1+l)(1+l^2/N^2)^\kappa} - \frac{1}{(1+l)(1+(1+l)/N)}\right|\rd l
 \leq C.
\end{equation}
This finishes the proof. 
\end{proof}

\begin{Remark}\label{rem:l_brackets}
For later use let
\begin{equation}
\begin{gathered}
\langle l\rangle= (1+l(l+1)/R^2)(1+l(l+1)/(N^2R^2))^2,
\quad
\ll l\gg:=(1+l(l+1)/R^2)\,.
\end{gathered}
\end{equation}
\end{Remark}

\begin{Remark}\label{rem:reproducing_kernel_sphere}
For all $R,N\in\bN_+$, we let $Q_{R,N}:=(1-\Delta_R)(1-\Delta_R/N^2)^2$, where $Q_{R,N}\,:\, L^6_2(\bS_R)\subset L_2(\bS_R)\to L_2(\bS_R)$ is a self-adjoint and positive operator. Thus, $(\re^{-t\,Q_{R,N}})_{t\in[0,\infty)}$ is a $C_0$-semigroup of contractions on $L^1_2(\bS_R)$. In particular, $\re^{-t\,Q_{R,N}}\,L^1_2(\bS_R)\subseteq L^1_2(\bS_R)$. 
\end{Remark}

\begin{Remark}\label{rem:trace_semigroup}
Using Remark~\ref{rem:spectral}, one obtains
\begin{equation}
\re^{-t\,Q_{R,N}}=\sum_{l=0}^{\infty}\,(2l+1)\,\re^{-\langle l\rangle t}\,\cP_{R,l},
\end{equation}
where $\langle l\rangle$ was introduced in Remark~\ref{rem:l_brackets}. Hence, for fixed $R,N\in\bN_+$ it holds that
\begin{equation}
\int_{0}^{\infty} \mathrm{Tr}(\re^{-2t\,Q_{R,N}})\, \rd t = \sum_{l=0}^{\infty} (2l+1)\,\int_{0}^{\infty} \re^{-2t\,\langle l\rangle }\, \rd t =\frac{1}{2}\,\mathrm{Tr}(G_{R,N})<\infty\,.
\end{equation}
The last bound follows from Remark~\ref{rem:expansion_covariance_harmo} and Lemma~\ref{lem:trace}.
\end{Remark}

\begin{lem}\label{lem:contraction_semigroup}
For all $R,N\in [1,\infty)$, and $t>0$, the operator $\re^{-tQ_{R,N}}\,:\, L_2(\bS_R)\to L^1_2(\bS_R)$
\begin{itemize}
\item[(A)] is bounded
\item[(B)] is a Hilbert--Schmidt operator.
\end{itemize}
\end{lem}

\begin{proof}
Let $t > 0$. Note that for each $N\in\bN_+$ there exists $C\in(0,\infty)$ such that for all $h\in L_2(\bS_R)$ it holds that
\begin{equation}
\begin{gathered}
 \|\re^{-t\,Q_{R,N}}h\|_{L^1_2(\bS_R)} = \|(1-\Delta_R)^{1/2}\,\re^{-t\,Q_{R,N}}\,h\|_{L_2(\bS_R)}
\leq C\, \|Q^{1/6}_{R,N}\,\re^{-tQ_{R,N}}\,h\|_{L_2(\bS_R)}\,.
 \end{gathered}
\end{equation}
Utilizing Remark~\ref{rem:parseval_them} one shows that for all $N\in \bN_+$
there exists $M\in (0,\infty)$ 
such that $\|\re^{-t\,Q_{R,N}}h\|_{L^1_2(\bS_R)} \leq M\,\|h\|_{L_2(\bS_R)}$
with
\begin{equation}
M:= \sup_{l\in\bN_0}\,(\langle l\rangle^{1/3}\,\re^{-2t\langle l\rangle}\,)^{1/2} \leq \sup_{x\geq 1}\,(x^{1/6}\,\re^{-tx})\leq \frac{C}{t^{1/6}}\,, \label{eq:M-bound}
\end{equation}
where $x:= \langle l \rangle$, which was introduced in Remark~\ref{rem:l_brackets} and $C\in (0,\infty)$ is a constant. The last bound follows from the fact that for all $t>0$, the map $x\mapsto x^{s}\re^{-tx}$ can be bounded uniformly for all $x>0$ by $C/t^{s}$.
This verifies that the mapping $\re^{-t\,Q_{R,N}}\,:\, L_2(\bS_R)\to L^1_2(\bS_R)$ is bounded and concludes Item~(A). To prove Item~(B) one uses Remark~\ref{rem:commute} Item~(F) with $H=L_2(\bS_R)$, $E=L_2^1(\bS_R)$ and $e_{\alpha}=Y_{lm}$. It holds
\begin{equation}
\begin{aligned}
\|\re^{-t\,Q_{R,N}}\|_{\mathrm{HS}}
=&(\sum_{l=0}^{\infty}\,\sum_{m=-l}^l\, \|\re^{-t\,Q_{R,N}}\,Y_{lm}\|^2_{L^1_2(\bS_R)})^{1/2}
\\
=&\,(\sum_{l=0}^{\infty}\,(2l+1)\,(1+\frac{l(l+1)}{R^2})\,\re^{-2t\,\langle l\rangle})^{1/2}.
\end{aligned}
\end{equation}
The sum is absolutely convergent due to exponential decay of $\re^{-2t\,\langle l\rangle}$. This finishes the proof of Item~(B) and concludes the proof.
\end{proof}
\begin{Remark}\label{rem:bounded_semigroup_L^3_2}
In a similar manner, one verifies that the mapping $\re^{-t\,Q_{R,N}}\,:\, L_2(\bS_R)\to L^3_2(\bS_R)$  is bounded, and there exists a constant $C\in (0,\infty)$ such that $\|\re^{-t\,Q_{R,N}}\|_{L_2(\bS_R)\to L^3_2(\bS_R)}\leq C/(t^{1/2})$. 
\end{Remark}

\section{Probability}\label{sec:probability}
Let $\cX$ be a complete separable metric space and $(\cX, \mathrm{Borel}(\cX))$ be the corresponding Borel measurable space. We will denote by $ \cM(\cX)$ the set of probability measures on $\mathrm{Borel}(\cX)$. Let $(\Omega,\cF,\bP)$ denote a probability space, and let $\bE$ stand for the integration with respect to the probability measure $\bP$.

\begin{dfn}\label{def:law}\cite[Eq.~(1.4)]{DaZab14}
A random variable $X$  with values  in $\cX$ is a measurable map 
$$X\,:\,(\Omega, \cF, \bP) \to (\cX, \text{Borel}(\cX))$$
i.e.,
$$X^{-1}(A) := \Big\{ \omega \in \Omega\,\,|\,\, X(\omega) \in A \Big\} \in \cF\,\,\,\mathrm{for\, \,all\,\,}A \in \mathrm{Borel}({\mathcal X}).$$
The distribution, or law, of $X$ is a probability measure $\mu$ defined on $\mathrm{Borel}(\cX)$ by
$$\mu(A) = \bP(X\in A)=\bP (X^{-1}(A)).$$
\end{dfn}

\begin{dfn}\label{def:cauchy_convergent}\cite[Def.~2.2.2]{Boga07}
Let $(f_n)_{n\in \bN_+}$ be a sequence of random variables.
\begin{itemize}
\item[(A)] The sequence $(f_n)_{n\in\bN_+}$ is called fundamental or Cauchy in probability if for every $\epsilon > 0$ one has 
\begin{equation}
\lim_{N\to \infty}\, \sup_{n,k\geq N} \, \bP\left(\omega\,:\,|f_{n}(\omega) - f_{k}(\omega)| \geq \epsilon\right) = 0\,.
\end{equation}

\item[(B)]
The sequence $(f_n)_{n\in \bN_+}$ is said to converge in probability to a random variable $f$ if for every $\epsilon > 0$ one has
\begin{equation}
\lim_{n \to \infty} \, \bP\left(\omega\,:\,|f(\omega) - f_{n}(\omega)| \geq \epsilon\right) = 0\,.
\end{equation}
\end{itemize}
\end{dfn}

\begin{rem}\label{rem:convergence}
For all $p\in[1,\infty)$ the convergence in $L_p(\Omega,\cF, \bP)$ implies convergence in probability. In particular, for $p=2$, this implies convergence in law. The former property is preserved under composition with continuous functions~\cite[Corollary~6.3.1]{Res15}~or~\cite[Corollary~2.2.6]{Boga07}. Moreover, convergence in probability implies almost sure convergence of a subsequence~\cite[Thm.~2.2.5]{Boga07}.
\end{rem}

\begin{dfn}\label{lemm:uni_inte}\cite[Def.~4.5.1]{Boga07}
A set of random variables $A\subset L_1(\Omega,\cF, \bP)$ is called uniformly integrable if
\begin{equation}
\lim_{C\to +\,\infty} \, \sup_{f\in A} \int_{\{|f(\omega)| > C\}}\,|f(\omega)|\,\bP(\rd\omega) = 0\,.
\end{equation}
\end{dfn}

\begin{Remark}\label{rem:uni_inte}
Observe that if the set $\{\,\bE\,[|f_n|^p\,]\,:\,n \in \bN_+ \}$ is bounded for some $p>1$, then $( f_n)_{n\in \bN_+}$ is uniformly integrable. To verify this claim suppose for some $p>1$ and $M>0$ that $\bE\,|f_n|^p \leq M$ for all $n \in \bN_+$\,. Note that $p-1>0$ and $t\mapsto t^{p-1}$ is increasing on $(0, \infty)$. If $|f_n(\omega)| >C$ for some $C >0$, then
$$|f_n(\omega)|^p = |f_n(\omega)|^{p -1} |f_n( \omega)| \geq C^{p -1}  |f_n(\omega)|.$$
Therefore, on the event $\{\omega\in \Omega\,:\,|f_n(\omega)| >C \}$ we obtain
$$
\bE\,\left[\,|f_n|\,\1_{\{|f_n| > C\}}\right] \leq \bE\left[|f_n|^p/C^{\,p-1}\right]\leq\,M/C^{\,p-1}\,.
$$
Observe that the last bound is independent of $n\in \bN_+$ and converges to zero as $C\to \infty$.
\end{Remark}

\begin{lem}\label{lem:E_F_X}
Let $X$ be a real-valued random variable such that $X\geq 0$. Suppose that the function $F\,:\,[0,\infty)\to[0,\infty)$ is continuously differentiable and such that $F(0)=0$ and $F'\geq 0$. Then it holds 
\begin{equation}
 \bE\,F(X)=\int_0^\infty \mathbb{P}(X>t)\, F'(t)\,\rd t.
\end{equation} 
\end{lem}
\begin{proof} From Fubini's theorem
\begin{equation}
\bE\, F(X) = \bE\,\int_{0}^X\,F'(t)\,\rd t
= \int_{0}^{\infty}\,\int_{\Omega}\,\1_{\{X>t\}}(\omega)\,\bP(\rd\omega)\,F^{\prime}(t)\,\rd t\,.
\end{equation}
This finishes the proof.
\end{proof}

\begin{thm}(Lebesgue-Vitali's theorem)\label{thm:Lebesgue-Vitali's theorem}\cite[Thm.~4.5.4]{Boga07}
Suppose that $f$ is a random variable and $(f_n)_{n\in \mb{N}_+}$ is a sequence of integrable random variables. Then, the following assertions are equivalent:
\begin{itemize}
\item[A)] The sequence $(f_n)_{n\in \mb{N}_+}$ converges to $f$ in probability and is uniformly integrable;

\item[B)] The random variable $f$ is integrable and the sequence $(f_n)_{n\in \bN_+}$ converges to $f$ in the space $L_1(\Omega,\cF, \bP)$.

\end{itemize}
\end{thm}

\begin{dfn}\cite[Def.~6.7]{Rudin87}\label{def:concentrate}
Let $\mu$ and $\nu$ be two regular Borel probability measures on $(\cX,\mathrm{Borel}(\cX))$. We say that $\nu$ is absolutely continuous with respect to $\mu$ if $\nu(A)=0$ for every $A\in\mathrm{Borel}( \cX)$ such that $\mu(A)=0$. Moreover, if there is a set $A\in\mathrm{Borel}(\cX)$ such that 
$$\mu(E) = \mu(A \cap E)$$
for all $E\in\mathrm{Borel}(\cX)$, we say that $\mu$ is concentrated on $A$. This is equivalent to saying that $\mu(E)=0$ whenever $E \cap  A = \emptyset.$
\end{dfn}

\begin{Remark}\label{rem:support_measure}
Let $\mu$ be a regular Borel probability measure on $(\cX,\mathrm{Borel}(\cX))$. The support of $\mu$ is the complement of the union of the $\mu$-null open sets, i.e., 
\begin{equation}
\supp\,\mu:=\left(\bigcup\, \{A\subset \cX\,:\,A\,\mathrm{open\,\,and}\,\, \mu(A)=0\}\right)^c.
\end{equation}
Equivalently, $\supp\,\mu$ is the smallest closed set on which $\mu$ is concentrated and if $\mu$ is concentrated on a non-empty set $B\subset \cX$, then $\supp\,\mu\subset \overline{B}$.
\end{Remark}

\begin{thm}\cite[Sec.~14.13]{Wil91}~or~\cite[Thm.~6.10]{Rudin87}~or~\cite[Thm.~3.2.2]{Boga07}
\label{Radon-Nikodym theorem}
Let $(\Omega,\cF)$ be a measure space with two probability measures $\mathbb{Q}$ and $\bP$. The probability measure $\mathbb{Q}$ is absolutely continuous with respect to $\mathbb{P}$ if there exists $V\in L_1(\Omega, \cF, \bP)$ such that for all $A\in \cF$ it holds that $\mathbb{Q}(A)= \int_A V(\omega)\,\bP\,(\rd\omega)\,$.
\end{thm}

\begin{Remark}\label{rem:inv_gau_measure}
Let $\nu_1,\nu_2$ be two probability measures on $(\cX,\mathrm{Borel}(\cX))$. We define
$(\nu_1\ast\nu_2)(A):= \int_{\cX}\, \nu_1(A-a)\,\nu_2(\rd a)$ for all $A\in \mathrm{Borel}(\cX)$.
\end{Remark}

\begin{lem}\cite[Lemma~A.7]{BarashkovSinh}~or~\cite{HairerSteele} \label{lem:mu_f_bound}
Let $(\Omega,\cF,\mu)$ be a probability space, $F:\, \Omega \to \bR$ be a measurable function such that $\exp(F) \in L_1(\Omega,\mu)$ and
\begin{equation}
\mu^F(\rd\phi) := \frac{\exp(F(\phi))\,\mu(\rd\phi)}{\int \exp(F(\phi)) \,\mu(\rd\phi)}\,.
\end{equation}
It holds
\begin{equation}
 \int \exp(F(\phi))\,\mu(d\phi) \leq \exp\left(\int F(\phi)\,\mu^F(d\phi)\right)\,.
\end{equation}
\end{lem}

\begin{proof}
To verify the above statement first notice that by assumption
\begin{equation}
\int \exp{(F(\phi))}\, \mu(\rd\phi) \int\, \exp{(-F(\phi))}\, \mu^{F}(\rd\phi) = \int \,\mu(\rd\phi) = 1.
\end{equation}
Thus, $
\int \exp{(F(\phi))}\, \mu(\rd\phi) = \left(\int\, \exp{(-F(\phi))} \mu^F(\rd\phi)\right)^{-1}.$
By the Jensen inequality
\begin{equation}
\exp{\left(-\int F(\phi)\mu^F(\rd\phi)\right)} \leq \int\, \exp{(-F(\phi))}\,\mu^F(\rd\phi).
\end{equation}
Hence,
\begin{equation}
\int \exp{(F(\phi))}\, \mu(\rd\phi)  \leq \exp{\Big(\int F(\phi)\mu^F(\rd\phi)\Big)}\,.
\end{equation}
This finishes the proof.
\end{proof}

\begin{dfn}\label{dfn:tightness_weak_convergence}\cite[Sec.~2.1]{DaZab14}
Let $(\mu_n)_{n \in \bN_+}$ be a sequence of probability measures defined on $(\cX,\mathrm{Borel}(\cX))$. The sequence $(\mu_{n})_{n \in \bN_+}$
is tight iff for every $\epsilon>0$ there exists a compact set $K_\epsilon\subset\cX$ such that $\mu_n(K_\epsilon)\geq 1-\epsilon$ for all $n\in\bN_+$. The sequence $(\mu_{n})_{n \in\bN_+}$ converges weakly if for every bounded $F\in C(\cX)$ the sequence of real numbers $(\mu_n(F))_{n\in\bN_+}$ converges.
\end{dfn}

\begin{thm}~\cite[Thm.~2.3]{DaZab14}(Prokhorov's theorem)
A~sequence of probability measures $(\mu_{n})_{n\in\bN_+}$ on $(\cX,\mathrm{Borel}(\cX))$ is tight  iff  there exists a diverging sequence of natural numbers $(a_n)_{n\in\bN_+}$ such that the sequence $(\mu_{a_n})_{n\in\bN_+}$ converges weakly.
\end{thm}

\begin{lem}\cite[Lemma B.4]{DDJ}\label{lem:tight_comapct_embedding}
Let $\imath:\cX\to\cY$ is a compact embedding and let $(\mu_{n})_{n\in\bN_+}$ be a sequence of probability measures on $(\mathcal{X},\mathrm{Borel}(\mathcal{X}))$. Assume that there exists $M\in(0,\infty)$ such that $\int_{\mathcal{X}}\|x\|_{\mathcal{X}}\,\mu_n(\rd x)\leq M$ for all $n\in\mathbb{N}_+$. Then the sequence of measures $(\nu_n)_{n\in\mathbb{N}_+}$ on $(\mathcal{Y},\mathrm{Borel}(\mathcal{Y}))$ defined by
\begin{equation}
\nu_n(A):=\mu_n(\imath^{-1}(A)),
\qquad\quad
n \in \mathbb{N}_+,
\quad
A \in \mathrm{Borel}(\mathcal{Y}),
\end{equation}
is tight.
\end{lem}
\begin{proof}
Let $\epsilon>0$, $L_\epsilon := \{x\in\cX\,|\,\|x\|_{\cX}\leq M/\epsilon\}$, which is a bounded subset of $\cX$ and $K_\epsilon:=\overline{\imath(L_{\epsilon})}$. Observe that $K_\epsilon\subset\mathcal{Y}$ is compact. It holds that
\begin{equation}
\mu_{n} \Big( \|x\|_{\cX} > \frac{M}{\epsilon}\Big) \leq \,\epsilon/M\, \int_{\cX} \,\|x\|_{\cX}\,\mu_n(\rd x) 
= \epsilon/M\, \int_{\cX} \|x\|_{\cX}\,\mu_n(\rd x)\,.
\end{equation}
Hence,
\begin{equation}
 1-\nu_n(K_{\epsilon}) \leq 1-\mu_n(L_\epsilon) = \mu_n(\|x\|_{\cX}>M/\epsilon)\leq\epsilon\,.
\end{equation}
This concludes the proof.
\end{proof}

\begin{dfn}\label{def:measure_preserving}
Let $(\cX,  \mathrm{Borel}(\cX))$, $(\cY, \mathrm{Borel}(\cY))$ be measure spaces and  $\tau: \cX\to \cY$ be measurable. 
\begin{itemize}
\item[(A)] Let $\mu\in \cM(\cX)$. Then the push-forward measure $\tau_*\mu \in \cM(\cY)$ is defined by 
$\tau_*\mu(A):=\mu(\tau^{-1}(A))$, where $A\in \mathrm{Borel}(\cY)$. 
It holds that $\int_{\cX}\,f\circ\tau\,\rd\mu=\int_{\cY}\,f\,\rd(\tau_*\mu)$ for all $f\in C(\cY)$.

\item[(B)] Fix $\mu\in \cM(\cX)$ and $\nu\in \cM(\cY)$. Then we say that $\tau$ is a measure preserving map between the
resulting probability spaces  $(\cX,  \mathrm{Borel}(\cX), \mu)$, $(\cY, \mathrm{Borel}(\cY),\nu)$ if $\tau_*\mu=\nu$

\item[(C)] Suppose that $\tau$ is invertible and the inverse is also measurable. Let $\nu\in \cM(\cY)$. Then the pullback measure $\tau^*\nu \in \cM(\cX)$ is defined by $\tau^*\nu(A)=\nu(\tau(A))$ for $A\in   \mathrm{Borel}(\cX)$.  
\end{itemize}
\end{dfn}

\section{Gaussian measures, Wiener processes and Wiener chaos} \label{Sec:Gaussian+Wiener}
Let $V$ be a locally convex topological vector space. We assume that the family
of seminorms, defining the topology, separates points, so that the space is Hausdorff. 
 In this situation the topological dual space is non-trivial and the following definition 
 of a Gaussian measure is meaningful:
\begin{dfn}\label{def:gaussian_measure_integration}
A probability measure $\nu$ on $(V,\mathrm{Borel}(V))$ is said to be a Gaussian measure if and only if the law of an arbitrary continuous linear functional $F\in V'$, considered as a random variable on $(V,\mathrm{Borel}(V),\nu)$, is a Gaussian measure on $(\bR^1,\mathrm{Borel}(\bR^1))$. 
\end{dfn}
\noindent
To answer the question of existence we need more structure. In the following two subsections we 
will discuss the nuclear spaces and Hilbert spaces. These two cases are distinct in infinite dimension.

\subsection*{Gaussian measures on nuclear spaces}
Let us move on to the definition of nuclear spaces. For every continuous semi-norm $p$ on $V$ we denote by $V_p$ the completion of $V/\mathrm{ker}p$ in the topology given by $p$, which is a Banach space.  If, in addition, for every continuous semi-norm $p$ on $V$ there exists another continuous semi-norm $q$ such that $p\leq q$\,\footnote{$p \leq q$ on $V \Longleftrightarrow \quad \exists\,C>0 \text{\,\, such\,\, that\, \,} p(v)\leq C\,q(v)\ \forall v\in V.$} and the canonical map
\begin{equation}       
\iota_{p,q}\,:\, V_q\to V_p
\end{equation}
is a nuclear operator between Banach spaces, we say that $V$ is nuclear. We recall that an operator
$T\,:\, V_1\to V_2$ between Banach spaces  is nuclear, if there exist   sequences $v^*_{1,n}\in V_1^*$ and
$v_{2,n}\in V_2$, $n\in \bN$, such that 
\begin{equation}
T v=\sum_{n\in \bN} v^*_{1,n}(v)\, v_{2,n}, \quad v\in V_1
\end{equation}
with
$$\sum_{n\in \bN} \|v^*_{1,n}\|_{V_1^*}\, \|v_{2,n}\|_{V_2}<\infty.$$ 
\begin{thm} (Bochner-Minlos) \cite[Sec.~VI.4.3.]{VTC87} \label{thm:Minlos} Let $V$ be a nuclear space.
In order that a functional $S_{\nu}: V\to \mathbb{C}$ be the characteristic 
function of a probability measure $\nu$ on $V'$, i.e.,
\begin{equation} 
S_{\nu}(v):=\int_{V'} \re^{i\langle \phi, v \rangle } \nu(\rd\phi),
\end{equation}
it is  sufficient that $S_{\nu}$ be of positive type, continuous
and $S_{\nu}(0)=1$. If the space $V$ is, in addition, metrizable, the condition is also necessary.  
\end{thm}
\noindent
If $S_{\nu}(v)=\re^{-\frac{1}{2}\mathcal{G}(v) }$, where $\mathcal{G}$ is a continuous, positive quadratic form on $V$,
then the measure $\nu$ is Gaussian in the sense of Def.~\ref{def:gaussian_measure_integration}. As the spaces of Schwartz class functions and smooth compactly  supported functions on manifolds are nuclear, using Theorem~\ref{thm:Minlos}, one can construct Gaussian measures on the corresponding spaces of distributions. However, as infinite dimensional Hilbert spaces are not nuclear, they require separate considerations.
\subsection*{Gaussian measures on Hilbert spaces}
Let $H$ be a separable Hilbert space with scalar product $\langle \,\Cdot\,,  \,\Cdot\,  \rangle_{H}$ and $G$ be a positive operator on $H$ with trivial kernel so that $G^{-1}$ is well defined
as a self-adjoint operator.
Following~\cite[p. 44]{RSII}, we define the scale of Hilbert spaces $H_{\ell}\subset H$, $\ell\in \mathbb{Z}$, as completions of $\mathrm{Dom}((G^{-1})^{\ell/2})$
in the norms coming from the scalar products $\langle \,\Cdot\,,  \,\Cdot\,  \rangle_{\ell}:=\lan (G^{-1})^{\ell/2}\,\Cdot, (G^{-1})^{\ell/2} \,  \Cdot\,     \ran_{H}$. We recall the identification $H_{-\ell}\cong H_{\ell}^*$ and set $E:=H_{-1}$
as this space will be particularly important in the following. We will identify $E^*$ with $H_1$ and denote standard pairing between $E$ and $E^*$ by $\langle\,\Cdot\, , \, \Cdot\, \rangle$. It coincides with the scalar product in $H$ on arguments from this space. We denote by $\hf:=G^{-1}f$ the isometric embedding from $E^*$ to $E$. Assuming that $G$ is trace-class on $E$, there exists a measure $\nu$ on $(E,\mathrm{Borel}(E))$ whose characteristic function
has the form
\begin{equation}\label{eq:characteristic-function-Hilbert} 
S_{\nu}( f):= \int_{E}\,\re^{i\langle \phi,  f  \rangle}\,\nu(\rd  \phi)=\int_{E}\,\re^{i\langle \phi,  \hf\rangle_E}\,\nu(\rd  \phi)=\re^{-\frac{1}{2} \langle  G  \hf,  \hf \rangle_E}=\re^{-\frac{1}{2} \|f\|^2_{H} }
\end{equation}
for $f \in E^*$. It is called the centered Gaussian measure on $E$ with covariance $G$ and it is a Gaussian measure on $E$ in the sense of Def.~\ref{def:gaussian_measure_integration}. The triple  $(H,E,\nu)$ is called an abstract Wiener space, and $H$ is called a reproducing kernel Hilbert space, cf.~\cite[Def.~1.2]{Sh04}.
\\[0.3cm]
\noindent
The existence of the measure $\nu$ above can be obtained from the following variant of the Bochner-Minlos theorem for Hilbert spaces:
\begin{thm}\cite[Thm.~2.27]{DaZab14}\label{thm:bochner}
Let $E$ be a separable Hilbert space. Let $S_{\nu}\,:\, E\to \mathbb{C}$ be the characteristic function of a probability measure $\nu$ on $(E,\mathrm{Borel}(E))$, then $S_{\nu}$ is a continuous function of positive type such that $S_{\nu}(0)=1$. Moreover, for arbitrary $\epsilon>0$ there exists a nonnegative trace-class operator $G_{\epsilon}$ such that 
$1-\mathrm{Re}\,S_{\nu}(f)\leq \epsilon$ for all $f$ such that $\langle G_{\epsilon}f, f\rangle \leq 1$.
Conversely, an arbitrary function $S_{\nu}\,:\,E\to \mathbb{C}$ satisfying the aforementioned conditions is the characteristic function of a probability measure $\nu$ on $E$.
\end{thm}

\subsection*{Wiener processes}
Now let $(\Omega, \mathcal{F}, \mathbb{P})$ be a probability space. Let $X\in L_2(\Omega, \mathcal{F}; E)$  be a  random variable, taking values in the Hilbert space $E$, whose law is given by the Gaussian measure $\nu$ of Eq.~\eqref{eq:characteristic-function-Hilbert}. Then,
\begin{equation} \label{eq:m-and-Q}
\begin{aligned}
\mathbb{E}[ \langle X, \hf \rangle_E\, \langle X, \hg\rangle_E]=\langle G \hf, \hg\rangle_E, \quad f,g\in E^*.
\end{aligned}
\end{equation}
\noindent
Now let  $\{\cF_t  \}_{t\in [0,\infty)}$ be a filtration, that is an increasing family of $\sigma$-algebras, such that $\cF_s\subset \cF_t\subset \cF$ for $s\leq t$ and let $(\Om, \mathcal{F},  \{\cF_t  \}_{t\in [0,\infty)}, \mathbb{P})$
be the resulting filtered probability space. We assume the right-continuity and completeness of the filtration, i.e., $\mathcal{F}_t=\bigcap_{s>t} \mathcal{F}_s$ and $\mathcal{F}_0$ contains all $\mathbb{P}$-null sets in $\mathcal{F}$.
\begin{dfn}\label{def:Q-Wiener}~\cite[Sec.~5.2]{DaZab96}
Let $\mu=\cN(0,G)$ be a Gaussian measure on $E$ with mean zero and a covariance operator $G$
which is positive and trace-class. Then an $E$-valued stochastic process $\{\tilde{W}(t)\}_{t\in [0,\infty)}$ on $(\Om, \mathcal{F},  \{\cF_t  \}_{t\in [0,\infty)}, \mathbb{P})$ is called a $G$-Wiener process if:
\begin{itemize}

\item[(A)] $\tilde{W}(\Cdot)$ has continuous paths and $\tilde{W}(0)=0$. 

\item[(B)] $\tilde{W}(\Cdot)$ has independent increments.

\item[(C)] For arbitrary $t\geq 0$ and $h>0$ the law of $h^{-1/2}(\tilde{W}(t+h)-\tilde{W}(t))$
equals $\cN(0,G)$.

\end{itemize}
Furthermore, it is assumed that $\tilde{W}(\Cdot)$ is adapted to the filtration, i.e., $\tilde{W}(t)$ is $\mathcal{F}_t$-measurable for every $t$.
\end{dfn}
\noindent 
It is well known that a process and filtration satisfying the requirements of Def.~\ref{def:Q-Wiener} exist. It follows from this definition and Eq.~\eqref{eq:m-and-Q} that
\begin{equation}
\mathbb{E}[ \langle \tilde{W}(t),   \hf \rangle_E \langle  \tilde{W}(s),  \hg  \rangle_E]=\langle G\hf, \hg\rangle_E\, t\wedge s 
\end{equation}
for $f,g\in E^*$. By rewriting this formula in terms of the pairing $\langle\,\Cdot\, , \, \Cdot\, \rangle$ between $E^*$ and $E$,
we get
\begin{equation}
\mathbb{E}[ \langle \tilde{W}(t), f \rangle \langle   \tilde{W}(s),  g  \rangle]=\langle f,  g\rangle_H\, t\wedge s.
\end{equation}
\noindent
The expression $W(t,f):=\langle \tilde{W}(t), f\rangle$, $f\in E^*$, gives the cylindrical Wiener process, which is suitable for modelling spacetime white noise. This process can also be defined more abstractly, without reference to $E, G$. 

\begin{dfn}\label{def:cyl_w_proc}\cite[Def.~7.11, Lemma~7.12]{PZ07}
A cylindrical Wiener process on $H$, adapted to $\{\cF_t \}_{t\in [0,\infty)}$, is a linear (in the second variable) mapping $W\,:\, [0,T]\times H\to L_2(\Omega,\cF,\bP)$ satisfying the following conditions:
\begin{itemize}

\item[(A)] For each $f\in H$, $(W (t,f))_{t \geq 0}$ is a real-valued $\cF_t$-adapted Wiener process (as in Def. \ref{def:Q-Wiener} with $E=\mathbb{R}$, $G=1$).

\item[(B)]  
 For all $t,s\geq 0$ and $f,g\in H$ it holds that $\bE\,[W(t,f)\,W(s,g)]= \langle f,g\rangle_H \,t\wedge s$\,.

\end{itemize}
\end{dfn}

\subsection*{The stochastic integral with respect to the cylindrical Wiener process}\label{sec:Ito_stochastic_integral_W}
We now briefly explain how to define the stochastic integral with respect to the cylindrical Wiener process. More details can be found in~\cite[Ch.~2]{DS26},~\cite[Sec.~2.5]{PR07},~and~\cite[Sec.~11.5]{DaPrato06}. 
\\
\\
Let $H$ be a separable Hilbert space, $(e_i)_{i\ge1}$ be an orthonormal basis of $H$, $W$ be a cylindrical Wiener process on $H$, i.e.,
$$W(t) = \sum_{i=1}^{\infty} \beta^i(t)\, e_i,\,\,\text{ where} \,\,(\beta^i)_{i\ge1}\,\,
\text{are\, \,independent\,\, standard\,\, Brownian\, motions},$$
and $f$ be a predictable $H$-valued process such that
\begin{equation}
\mathbb{E}\int_0^t \|f(s)\|_{H}^2\,\mathrm{d}s < \infty.
\end{equation}
Then, the It\^o integral is well-defined as a real-valued random variable
\begin{equation}
\int_0^t f(s)\,\mathrm{d}W(s)
= \sum_{i=1}^\infty \int_0^t \langle f(s), e_i\rangle_H \,\mathrm{d}\beta^i(s),
\end{equation}
where the series converges in $L_2(\Omega)$. Let
\begin{equation}
L_2(\Omega \times [0,t]; H)
= \left\{ f \,:\, \mathbb{E}\int_0^t \|f(s)\|_{H}^2 \,\mathrm{d}s < \infty \right\}
\end{equation}
Equivalently, the It\^o integral defines an isometry
\begin{equation}
L_2(\Omega\times[0,t]; H) \;\longrightarrow\; L_2(\Omega,\cF,\bP)
\end{equation}
such that
\begin{equation}\label{eq:Ito_iso}
\bE\Big|\int_0^t\, f(s)\,\rd W(s)\Big|^2 = \bE\int_0^t\, \|f(s)\|_{H}^2\,\rd s.
\end{equation}
To see this, one uses Parseval’s identity:
\begin{equation}
\|f(s)\|_H^2 = \sum_{i=1}^\infty |\langle f(s), e_i\rangle_H|^2.
\end{equation}
Thus,
\begin{equation}
\begin{aligned}
\mathbb{E}\left|\sum_{i=1}^\infty \int_0^t \langle f(s), e_i\rangle_H\, \mathrm{d}\beta^i(s)\right|^2
&= \sum_{i=1}^\infty \mathbb{E}\left|\int_0^t \langle f(s), e_i\rangle_H\, \mathrm{d}\beta^i(s)\right|^2 
= \sum_{i=1}^\infty \mathbb{E}\int_0^t |\langle f(s), e_i\rangle_H|^2 \mathrm{d}s \\
&= \mathbb{E}\int_0^t \sum_{i=1}^\infty |\langle f(s), e_i\rangle_H|^2 \mathrm{d}s
= \mathbb{E}\int_0^t \|f(s)\|_H^2 \,\mathrm{d}s.
\end{aligned}
\end{equation}

\begin{Remark}\label{rem:BDG_ineq}
Note that the formal expansion $W(t) = \sum_{i=1}^{\infty} \beta^i(t)\, e_i$ does not converge in $H$, since the series diverges almost surely 
\begin{equation}
\bE[|W(t)|^2]=\sum_{i=1}^\infty \mathbb{E}|\beta^i(t)|^2=\sum_{i=1}^{\infty}\,t = \infty,\,\,\,\,\, t\geq 0.
\end{equation}
Hence, the cylindrical Wiener process is understood through duality pairings. More precisely, if $W$ is a cylindrical Wiener process on $H$, then
\begin{equation}
W(t,h) = \sum_{i=1}^\infty \langle h, e_i\rangle_H \beta^i(t), \quad h\in H,
\end{equation}
converges in $L_2(\Omega,\cF,\bP)$. Generally, an $H$-valued It\^o integral can be defined with respect to a cylindrical Wiener process on a separable Hilbert space $E$, provided the integrand takes values in the space $\mathcal L_2(E,H)$ of Hilbert–Schmidt operators  defined in Remark~\ref{rem:commute} Item~(F). One has
$$\bE\Big\|\int_0^t \,f(s)\,\rd\,W(s)\Big\|_H^2
= \bE\int_0^t\, \|f(s)\|_{HS}^2\,\rd s.$$
Moreover, we have for $p\in (0,\infty)$, by the Burkholder-Davis-Gundy theorem~\cite[Thm.~5.2.4]{DaZab96}
\begin{equation} \label{eq:Burkholder-Davis-Gundy}
\bE\,\| \int_0^t \,f(s)\,\rd W(s)\,\|^p_H \leq  c_p\,\bE\,\left(\int_0^t \|f(s)\|^2_{\mathrm{HS}}\,\rd s\right)^{p/2}.
\end{equation}
We note that for $p=2$, this follows from the It\^o isometry~\eqref{eq:Ito_iso}. 
\end{Remark}
\subsection*{Wiener chaos}
Let $x\in \bR$, $m\in\bN$. Let 
$$H_m(x):=(-1)^m\,\re^{x^2}\,\frac{\rd^m}{\rd x^m}\,\re^{-x^2}$$
denote the Hermite polynomial of degree $m$ with mean zero and unit variance~\cite[Def.~8.1.1]{PT11} or~\cite[Sec.~9.2.1]{DaPrato06}. One has
\begin{equation}
H_m(x)= \sum_{k=0}^{\floor{m/2}} \frac{(-1)^k\,m!}{(m-2k)!\,k!\,2^k}\, x^{m-2k}\,.
\end{equation}
Observe that~\cite[Eq.~(18.1)~with different coefficient]{Arfken}~or~\cite[Prop.~9.3]{DaPrato06}
$$H''_m(x) -x\,H'_m(x) = -m\,H_m(x),$$
which implies that $H_m$ is an eigenfunction of the one dimensional Ornstein-Uhlenbeck operator, i.e., $L:= -1/2\, \Delta + x\circ \vec\nabla$, with eigenvalue $-m$~\cite[Ex.~3.10]{Gri09}. Moreover, the Hermite polynomials with mean zero and variance $c$ can be determined by
\begin{equation}
H_m(x,c) := (-c)^m\,\re^{\frac{x^2}{2c}}\,\frac{\rd^m}{\rd x^m}\,\re^{\frac{-x^2}{2c}}
\end{equation}
such that $H_m(x,c) = c^{m/2}\,H_m(\frac{x}{\sqrt{c}})$. It holds that
\begin{equation}\label{eq:hermite_real}
H_m(x,c) = \sum_{k=0}^{\floor{m/2}} \frac{(-1)^k\,m!}{(m-2k)!\,k!\,2^k} c^k\,x^{m-2k}, 
 \qquad c\geq 0.
\end{equation}
\noindent
Let $(\Omega,\cF, \bP)$ denote a probability space and $H$ be a real, separable Hilbert space equipped with scalar product $\langle\Cdot\,,\, \Cdot\, \rangle_H$\,. We say that a stochastic process
$$X=\{X(h) \,|\, h\in H\}$$
defined in $(\Omega,\cF, \bP)$ is a Gaussian process on $H$ if $X$ is a centered Gaussian family of random variables such that $\bE\,[X(h) X(g)]=\langle h, g\rangle_H$ for $h,g\in H$. We let
\begin{equation}\label{eq:hermite_hilbert}
H_m(X,c):=\sum_{k=0}^{\floor{m/2}} \frac{(-1)^k\,m!}{(m-2k)!\,k!\,2^k}\,c^{k}\,X^{m-2k}.
\end{equation}
\begin{dfn}\label{def:wiener_chaos}
We denote by $\cH_m$ the closed linear subspace of $L_2(\Omega,\cF,\bP)$
generated by random variables
$$\{ H_m(X(h)), h\in H, \|h\|_H=1\}$$
and call it the Wiener chaos of order $m\in\bN$ such that due to the Wiener-It\^o decomposition~\cite[Sec.~9.3]{DaPrato06} or~\cite[Thm.~2.6]{SJ97} it holds  that 
$$L_2(\Omega,\cF,\bP)=\bigoplus_{l=0}^{\infty} \hil_l.$$
We call the subspace $\bigoplus_{l=0}^{m} \hil_l$ the inhomogeneous Wiener chaos of order $m\in\bN$.
\end{dfn}
\noindent
The Hermite polynomials form an orthogonal system with respect to the Gaussian measure in Euclidean space. 
Consequently, we have the following lemma.
\begin{lem}\label{Nualart-Lemma-1-1}
Let $X,Y$ be two random variables with joint Gaussian distribution such that $\mathbb{E}(X)=\mathbb{E}(Y)=0$ and $\mathbb{E}(X^2)=\mathbb{E}(Y^2)=1$.
Then, for all $n,m$ we have
\begin{equation}
\bE\,[H_n(X)H_m(Y)]=\delta_{n,m}\,n!\,(\bE\,[XY])^n.
\end{equation}
\end{lem}
\begin{proof} 
See~\cite[Lemma 1.1.1]{Nualart} or~\cite[Lemma~9.5~with different coefficient]{DaPrato06}. 
\end{proof}

\begin{Remark}\label{rem:two_regularization} 
Let $\delta c := c_2-c_1$. From Eq.~\eqref{eq:hermite_hilbert}, one has
\begin{equation}\label{eq:double_wick_ordered_1}
\begin{aligned}
H_n(X,c_1):=&\sum_{m=0}^{\floor{n/2}} \frac{(-1)^m\,n!}{(n-2m)!\,m!\,2^m}\,c_1^{m}\,X^{n-2m}
\\
=&\sum_{m=0}^{\floor{n/2}} \frac{(-1)^m\,n!}{(n-2m)!\,m!\,2^m}\,(c_2-\delta c)^{m}\,X^{n-2m}
\\
=&\sum_{m=0}^{\floor{n/2}} \frac{(-1)^m\,n!}{(n-2m)!\,m!\,2^m}\,\sum_{h=0}^m\,\frac{m!}{(m-h)!\,h!}\,c_2^{m-h}\,(-\delta c)^{h}\,X^{m-2k}.
\end{aligned}
\end{equation}
Observe that $0\leq h\leq m\leq \floor{n/2}$. Letting $m-h=k$ with $0\leq k\leq \floor{(n-h)/2}$, one rewrites Eq.~\eqref{eq:double_wick_ordered_1}
\begin{equation}
\begin{aligned}
H_n(X,c_1)
=&\sum_{h=0}^{\floor{n/2}}\,\sum_{m=h}^{\floor{n/2}}\,\frac{n!\,(-1)^m}{(n-2m)!\,m!\,2^m}\, \frac{m!}{(m-h!)\,h!}\, \,c_2^{m-h}\,(-\delta c)^{h}\,X^{n-2m}
\\
=&
\sum_{h=0}^{\floor{n/2}}\,\frac{n!}{(n-2h)!\,2^h}\,(\delta c)^h\,\sum_{k=0}^{\floor{(n-h)/2}}\,\frac{(-1)^k\,(n-2h)!}{(n-2h-2k)!\,k!\,2^k}\, \,c_2^{k}\,X^{n-2h-2k}.
\end{aligned}
\end{equation}
This verifies that~\cite[Eq.~(9.1.12)]{GlimmJaffe}
\begin{equation}\label{eq:re_wick_ordering}
\begin{gathered}
H_n(X,c_1) = \sum_{h=0}^{\floor{n/2}}\frac{n!}{(n-2h)!\,h!\,2^h}\,(\delta c)^h\,H_{n-2h}(X, c_2).
\end{gathered}
\end{equation}
\end{Remark}

\begin{lem}(Nelson's estimate)\label{lem:nelson}
For every random variable $X$ in an inhomogeneous Wiener chaos of order $n\in\mathbb{N}_+$, cf.  \cite{Nualart}, and every $p\in[2,\infty)$ it holds
\begin{equation}
\bE \big[ |X|^{p}\big]^{\frac{1}{p}} \leq \sqrt{ n+1}\,(p - 1)^{\frac{n}{2}}\, \bE \big[ X^{2}\big]^{\frac{1}{2}},
\qquad
\bE\exp\Bigg(\frac{n\,|X|^{2/n}}{6\,\bE \big[ X^{2}\big]^{\frac{1}{n}}}\Bigg) <\infty.
\end{equation}
\end{lem}

\begin{proof}
The first bound follows from the observation that for a homogeneous Wiener chaos it holds, cf.~\cite[Eq.~(7.2)]{Hairer-Malliavin},
\begin{equation}\label{eq:Nelson_homo}
\bE\,[\,|X|^p]^{1/p} \leq \,(p-1)^{n/2}\,\bE\,[\,|X|^2]^{1/2}\,.
\end{equation}
For an inhomogeneous Wiener chaos of order $n$ we estimate
\begin{equation}
\begin{aligned}
\bE\,[\,|\sum_{i=0}^nX_i|^p]^{1/p}  \leq \sum_{i=0}^n\bE\,[\,|X_i|^p]^{1/p}& \leq (p-1)^{n/2}\,\sum_{i=0}^n\bE\,[\,|X_i|^2]^{1/2} \\
&\leq \sqrt{n+1}\, (p-1)^{n/2}\, \bE\,[\,\sum_{i=0}^{n}|X_i|^2]^{1/2},
\end{aligned}
\end{equation}
where in the first step we used the triangle inequality for the $L_p$ norm and in the last step the Cauchy-Schwarz inequality. To prove the second relation, consider
\begin{equation}
\begin{aligned}
\bE\exp\Bigg(\frac{n|X|^{2/n}}{6\bE \big[ X^{2}\big]^{\frac{1}{n}}}\Bigg) 
&= \sum_{i=0}^{\infty}\frac{n^i\bE|X|^{2i/n}}{i!\,6^i\,[\,\bE\,|X|^2]^{i/n}}
\\
&= \sum_{i=0}^{n-1}\frac{n^i\bE\,|X|^{2i/n}}{i!\,6^i\,[\,\bE\,|X|^2]^{i/n}} + \sum_{i=n}^{\infty}\frac{n^i\bE\,|X|^{2i/n}}{i!\,6^i\,[\,\bE\,|X|^2]^{i/n}}.
\end{aligned}
\end{equation}
The first sum is finite. For the second sum using the first bound in the statement with $p=(2i/n)\in[2,\infty)$, one obtains
\begin{equation}
\begin{aligned}
R:=  \sum_{i=n}^{\infty}\,\frac{n^i\bE\,|X|^{2i/n}}{i!\,6^i\,[\,\bE\,|X|^2]^{i/n}}
&\leq \sum_{i= n}^{\infty}\, \frac{n^i\, (n+1)^{i/n}\,(2i/n)^{i}\,[\,\bE\,|X|^2]^{i/n}}{i!\,6^i\,[\,\bE\,|X|^2]^{i/n}}
\\
&\leq  \sum_{i=n}^{\infty}\, \frac{1}{i!}\,\Big(\frac{(n+1)^{1/n}\,i}{3}\Big)^i.
\end{aligned}
\end{equation}
Let us calculate the ratio of $R_{i+1}/R_i$
\begin{equation}
\begin{aligned}
\lim_{i\to \infty}\frac{R_{i+1}}{R_i} 
&= \lim_{i\to\infty}\,\frac{(i+1)^{i+1}}{(1+i)!}\,\frac{i!}{i^i}\,\Big(\frac{1}{3}\Big)^{i+1}\,3^{i}\,\frac{ ((n+1)^{1/n})^{i+1}}{((n +1)^{1/n})^{i}}
\\
&= \lim_{i\to\infty}\,\frac{1}{3}\Big(1+\frac{1}{i}\Big)^{i}\, (n+1)^{1/n}
=\frac{e}{3}\, ((n+1)^{1/n})\,.
\end{aligned}
\end{equation}
We can assume without loss that $n$ is sufficiently large, because $\mathcal{H}_{n}\subset \mathcal{H}_{n'}$, $n'\geq n$. Now we use the fact that the function $f(n)=(n+1)^{1/n}$ is decreasing and tends to 1 as $n\to \infty$.  This gives convergence of the sum by the d'Alambert criterion. 
\end{proof}

\section{Stochastic differential equations and their invariant measures}\label{sec:stochastic_analysis}
\newcommand{\tH}{ {\red \tilde{H}}}
In this section we introduce, in general terms, the class of stochastic differential equations we are interested in.  Let  $H$ be the Hilbert space carrying a cylindrical Wiener process, as in Def.~\ref{def:cyl_w_proc}, and let 
$\tilde{H}\subset H$ be a Hilbert space continuously embedded in $H$ (in our case $\tilde H=L_2^1(\bS_R)$, $H=L_2(\bS_R)$, $E=L_2^{-\kappa}(\bS_R)$, for some $\kappa>0$).
\begin{dfn}\label{def:mild_solution}\cite[Sec.~11.5.2]{DaPrato06}\label{def:stochastic_mild_solution}
We denote by $C([0,T]; L^2(\Omega,\cF, \{\cF_t\}_{t\geq 0},\bP; \tilde H))$ the space consisting of all stochastic processes $X(t)\in L^2(\Omega,\cF, \{\cF_t\}_{t\geq 0},\bP; \tilde H) $ for all $t\in[0,T]$ such that the mapping $[0,T]\to L^2(\Omega,\cF,\bP;\tilde H)$ defined by $t\mapsto X(t)$ is continuous. Moreover, it is a Banach space endowed with the
norm
\begin{equation}
\|X\|_{C([0,T];L^2(\Omega,\cF, \{\cF_t\}_{t\geq 0},\bP; \tilde H))}:= \sup_{t\in[0,T]}\,\left(\bE\,\|X(t)\|_{\tilde H}^2 \right)^{1/2}\,.
\end{equation}
\end{dfn}
\noindent
Consider the following stochastic differential equation
\begin{equation}\label{eq:nonlieaner_general}
\begin{cases}
\rd X(t)=\left(-Q\,X(t)+F(X(t))\right)\,\rd t +\sqrt{2}\, \rd W(t),
\\
X(0)= \rx\in  L^2(\Omega,\cF, \{\cF_t\}_{t\geq 0},\bP; \tilde H),
\end{cases}
\end{equation}
where $Q\,:\,\text{Dom}(Q)\subset \tilde H\to \tilde H$ is the generator of a $C_0$-semigroup $(\re^{-t\,Q})_{t\geq0}$\, of contractions on $\tilde H$, the map $F\,:\,\tilde H \to H$ be a possibly non-linear function with properties specified in Lemma~\ref{lem:Hypothesis 8.4} below and $(W(t,\Cdot))_{t\geq 0}$ be the cylindrical Wiener process on $H$ adapted to 
$\{\cF_t^W\}_{t\geq 0}$ given in Def.~\ref{def:cyl_w_proc}.  A stochastic process $X\in C([0,T]; L^2(\Omega,\cF, \{\cF_t\}_{t\geq 0},\bP; \tilde H))$ 
is said to be a mild solution of Eq.~\eqref{eq:nonlieaner_general} if for every $T\geq t\geq 0$  it satisfies the following integral equation
\begin{equation} \label{eq:integrated_SPDE}
\begin{aligned}
X(t)= \re^{-t\,Q}\,\rx+\int_0^t\re^{-(t-s)\,Q}\, F(X(s))\,\rd s +\sqrt{2} \int_0^t\re^{-(t-s)\,Q} \,\rd W(s)\,.
\end{aligned}
\end{equation}
\noindent
We are going to consider such $F\,:\,\tilde H \to H$ that the corresponding integral above takes
values in $\tilde H$.
\\[0.3cm]
\noindent
Let us recall briefly the standard construction of   $(\Omega, \cF,   \{\cF_t\}_{t\geq 0}, \mathbb{P})$ used above: Let $(\Om^0, \mathcal{F}^0, \mathbb{P}^0)$
be the original probability space of the initial conditions and we complete it if necessary. (In applications the initial conditions may be defined, to start with, e.g. on the probability space $(\cD'(\bS_R), \mathrm{Borel}(\cD'(\bS_R)),  \nu_{R})$, which is incomplete). Next, let $(\Om^{\mathrm{W}}, \mathcal{F}^{W}, 
\{\mathcal{F}^{W}\}_{t\geq 0}, \mathbb{P}^{\mathrm{W}} )$ be the canonical filtered probability space of the cylindrical Wiener process. We form the product space 
\begin{equation}
(\Om, \mathcal{F}, \bP):=(\Omega^0\times \Omega^W, \mathcal{F}^0\otimes \mathcal{F}^{W}, \mathbb{P}^0\otimes \mathbb{P}^W).
\end{equation}
Proceeding to the filtration, we first define
\begin{equation}
\mathcal{G}_t:= \mathcal{F}^0 \otimes \mathcal{F}_t^W.
\end{equation}
Now, we just have to complete it to ensure right continuity. For this purpose, we set
\begin{equation}
\bar{\mathcal{G}}_t:= \sigma(\mathcal{G}_t \cup \mathcal{N}_{\mathbb{P}}), \quad \mathcal{F}_t:=\bigcap_{s>t} \bar{\mathcal{G}}_s,
\end{equation} 
where $\sigma(\mathcal{G}_t \cup \mathcal{N}_{\mathbb{P}})$ is the smallest $\sigma$-algebra generated by $\mathcal{G}_t$ and the null sets of $\mathbb{P}$. This completes the construction of the filtered probability space, which carries  the Wiener process and the initial conditions as  independent random variables. We remark that we will sometimes
write, for brevity, $X(0)\in \tilde H$, rather than $X(0)\in L^2(\Omega,\cF,\bP; \tilde H)$, even if the initial data is random.
\\[0.3cm]
\noindent
Let us now move on to the concept of an invariant measure for the above class of stochastic differential equations.
\begin{dfn}\label{def:stochastic_Markov_transition}
Let $(X(t))_{t\geq 0}$ be a continuous stochastic process on a probability space $(\Omega,\cF, \{\cF_t\}_{t\geq 0}, \bP)$ with values in $\tilde H$. For all $A\in \mathrm{Borel}(\tilde H)$, $t\geq 0$ and $X(0)=:\rx\in \tilde H$, we denote the law of $X(t)$ by
$$P(t,\rx,A):=\bP(X(t)\in A\,|\,X(0)=\rx)\in\cM(\tilde H).$$
Moreover, for all $f\in B_{\rb}(\tilde H) $ we define its corresponding transition semigroup $\cP_t$ by
\begin{equation}
\cP_tf(\rx):= \bE\big[f(X(t))\,|\,X(0)=\rx\big]=\int_{\tilde H}\, f(\ry)\,P(t,\rx,\rd \ry) \label{eq:semigroup-of-equation}
\end{equation}
such that
$$P(t,\rx,A)=\cP_t\,\mathmybb{1}_A(\rx)=\int_{\tilde H}\,\mathmybb{1}_A(\ry)\,P(t,\rx,\rd\ry)\quad\text{for\,\, all}\quad A\in\mathrm{Borel}(\tilde H).$$ 
\end{dfn}

\begin{dfn}\label{def:measure_invariant}
A probability measure $\mu\in\cM(\tilde H)$ is invariant for the semigroup $\cP_t$ if for all $f\in B_{\rb}(\tilde H)$ and all $t\geq 0$ it holds
\begin{equation}
\int_{\tilde H}\,\cP_t\,f(\rx)\,\mu(\rd\rx)
=\int_{\tilde H}\,f(\rx)\,\mu(\rd\rx)
=\int_{\tilde H}\,\bE\,[f(X(t,\rx))]\,\mu(\rd\rx).
\end{equation}
\end{dfn}
\noindent
Let us now consider $F$ in Eq.~\eqref{eq:integrated_SPDE}, which is a gradient of a certain function $U$, as described below:
\begin{lem}\label{lem:Hypothesis 8.4}~\cite[Hypothesis~8.4]{DaZab96} 
Consider Eq.~\eqref{eq:nonlieaner_general} with $F=U'$, where $U$ is a polynomial with appropriate properties and assume that
\begin{itemize}
\item[(1)] $Q$ is a self-adjoint positive definite operator on $H$ such that $Q^{-1}$ is of trace-class.
    
 \item[(2)] For any $t \geq 0$, $\tilde H$ is an invariant subspace of $S(t)=\re^{-t\,Q}$. Moreover for any $f \in \tilde H$ the mapping $S(\Cdot)f$ is continuous on $[0,+\infty)$ with the topology of $\tilde H$.
    
 \item[(3)] The Ornstein–Uhlenbeck process
 $$W_Q(t):=\sqrt{2}\,\int_0^t \re^{-(t-s)Q}\,\rd W(s),$$
 where $W$ is the cylindrical Wiener process on $H$, has a version concentrated on $\tilde H$.
    
 \item[(4)] $U\,:\,\tilde H \to \mathbb{R}$  is continuous and there exists the directional derivative $DU(x:h)$ of $U$ at any point $x \in \tilde H$ and at any direction $h \in \tilde H$. Moreover, there exists a mapping $F\,:\,  \tilde H \to H$ such that $DU(x: h) = \langle F(x), h \rangle_{H}$. In this case we write $F=U'$.
    
\item[(5)] $U$ is bounded from above, and $F\,:\,\tilde H \to H$ is a locally Lipschitz mapping: for arbitrary $r > 0$ there exists $k_r > 0$ such that
$$\|F(x) - F(y)\|_{H} \leq k_r\, \|x - y\|_{\tilde H}.$$
    
\item[(6)] For arbitrary $x \in \tilde H$ there exists an $\tilde H$-continuous solution of the integral equation
\begin{equation}
X(t) = \re^{-t\,Q}\,x + \int_0^t\,\re^{-(t - s)\,Q}\,F(X(s))\,\rd s + W_Q(t), \quad t \geq 0. 
\end{equation}
\end{itemize}
Now let $\nu=\cN(0, Q^{-1} )$. Under the above assumptions the measure
\begin{equation}
\mu_F(\rd \phi) = \re^{U(\phi)}\, \nu(\rd \phi), \quad \phi \in \tilde H
\end{equation}
is invariant for Eq.~\eqref{eq:nonlieaner_general} with $F=U'$~\cite[Thm.~8.6.3]{DaZab96} and the random field $X$ is called a gradient process. One also shows that the measure $\mu_F$ given by the above expression, is absolutely continuous with respect to the measure $\nu$. Similar result can be found in~\cite[Sec.~17.5]{PZ07}.

\end{lem}

\chapter{Finite volume measure}\label{ch:finite_volume}
We would like to construct the dynamical $P(\Phi)_2$ model of Euclidean quantum field theory, where $P\,:\,\mathbb{R}\to\mathbb{R}$ is a polynomial of the form
\begin{equation}\label{eq:polynomial_form}
P(\tau) = \frac{1}{n}\,\tau^{n} +\sum_{m=0}^{n-1}\,a_m\,\tau^{m},\,\, n\in2\bN_+,\,\,\mathrm{and}\,\,a_m\in\bR\,\,\,\,\forall\,\, m\in\{0,\cdots,n-1\}.
\end{equation}
Observe that such polynomials are bounded from below, but may not be globally Lipschitz continuous.
\section{$P(\Phi)_2$ measure on the plane}\label{sec:measure_on_plane}
The measure of the dynamical $P(\Phi)_2$ on $\sS'(\bR^2)$ is, heuristically, given by
\begin{equation}\label{eq:un_regularized_measure}
 \mu(\rd\phi):=\frac{1}{\mathcal{Z}}\exp\left(-\int_{\bR^2} \lambda\,P(\phi(x))\,\rd x\right) \nu(\rd\phi),
\end{equation}
where $\lambda\in (0,\infty)$ is the coupling constant, $\mathcal{Z}\in (0,\infty)$ is the normalization factor, which might not be finite due to the emerging dependence on the volume of the underlying space, $\nu(\rd\phi)$ indicates the free field Gaussian measure with covariance $G:=(1 -\Delta)^{-1}$, which is a bounded operator on $L_2(\bR^2)$. Moreover, $(-\Delta)$ is the Laplacian on $\bR^2$. Observe that the operator $(-\Delta)$ does not have a compact resolvent and has a purely continuous spectrum on $[0,\infty)$\footnote{The operator $(-\Delta)^{-1}$ is unbounded as a bilinear form $\sS(\bR^2)\times\sS(\bR^2)\to\bR$, since $|k|^{-2}$ is not locally integrable.}. We define $G$ on $L^{-2}_2(\bR^2)\subset \sS'(\bR^2)$.

\begin{Remark}\label{rem:characteristic_functional}
We denote the centred Gaussian measure with the covariance $G$ by $\nu:=\cN(0,G)$. The corresponding characteristic functional, i.e., $S_{\nu}(f)$ can be determined via 
\begin{equation}
\sS(\bR^2)\ni f  \mapsto S_{\nu}(f)=\int_{\sS'(\bR^2)}\re^{i\langle\phi,f\rangle}\,\nu(\rd\phi) =\re^{-\frac{1}{2} \langle Gf,f\rangle}\,,
\end{equation}
where $f\in \sS(\bR^2)$, $\langle\Cdot,\Cdot\rangle$ denote $\sS(\bR^2)-\sS'(\bR^2)$ duality pairing. The existence of the measure follows from Theorem~\ref{thm:Minlos}. Note that $\langle Gf,f\rangle_{L_2(\bR^2)}=\|(1-\Delta)^{-1/2}f\|^2_{L_2(\bR^2)}=\|f\|^2_{L^{-1}_2(\bR^2)}$. This implies that the measure $\nu$ is supported on the space of distributions, in sense of Remark~\ref{rem:support_measure}. One may find a discussion on the support property of the Gaussian measure in~\cite{RR73}.
\end{Remark}
\noindent
On account of Remark~\ref{rem:fundamental_solution}, one infers that the measure of the dynamical $P(\Phi)_{2}$ model as given in Eq.~\eqref{eq:un_regularized_measure} is ill-defined. Firstly, a typical field $\phi$ in the support of the Gaussian measure $\nu(\rd\phi)$ lacks integrability, i.e., it does not have sufficient decay at infinity. Hence, it might contain some infrared divergences, since in Eq.~\eqref{eq:un_regularized_measure} we are performing the integral over $\bR^2$. Secondly, such a field does not have enough regularity, as oftentimes it is a distribution. Hence, the product terms inside $P(\phi(x))$ might not be well-defined. To circumvent these two problems, we first construct the measure $P(\Phi)_2$ on a two-dimensional sphere $\bS_R$ with finite volume and we shall introduce an appropriate UV cutoff to regularize the covariance of the Gaussian measure $\nu(\rd\phi)$.

\begin{Remark}\label{rem:euclidean_R_two}
The isometry group of $\bR^2$ consists of all affine transformations of the form $Q(x)=Ax+b$, where $A\in O(2)$ and $b\in\bR^2$, and it is isomorphic to the Euclidean group $E(2)=T(2)\rtimes O(2)$ with the multiplication rule $(t_1,R_1)\cdot(t_2,R_2)= (t_1+R_1t_2,R_1R_2)$ for $t_i\in T(2)$ and $R_i\in O(2)$ for $i=1,2$. The orthogonal group $O(2)$ is the symmetry group of $\bR^2$ fixing the origin with $\mathrm{dim}\,O(2) = 1$. Moreover, $T(2)$ is the two-dimensional subgroup of the group $E(2)$ consisting of translations. One has $\mathrm{dim}\, E(2)= 3$, which is equal to the dimension of the symmetry group of $\bS_R$, i.e., $O(3)$, as well as the subgroup consists of rotations $SO(3)$, i.e., $SO(3)$, see Remark~\ref{rem:sphere_quotient}. 
\end{Remark}

\section{$P(\Phi)_2$ measure on the sphere}\label{sec:measure_on_sphere}
The measure of the dynamical $P(\Phi)_2$ model on $\sD'(\bS_R)$ is given by
\begin{equation}\label{eq:regularized_measure_sphere}
 \mu_R(\rd\phi):=\frac{1}{\mathcal{Z}_R}\exp\Big(-\int_{\bS_R}\, : P(\phi(\rx)) :\,\rho_R(\rd\rx)\Big) \nu_R(\rd\phi),
\end{equation}
where $\nu_R:= \mathcal{N}(0,(1-\Delta_R)^{-1})$ indicates the free field Gaussian measure with covariance $G_R:=(1-\Delta_R)^{-1}$, $\rho_R$ is the Riemannian volume measure, and $:\Cdot:$ refers to the Wick renormalization with respect to $\nu_R$\,. For all $R\in\bN_+$, the finite volume $P(\Phi)_2$ measure, i.e., $\mu_R$, is absolutely continuous with respect to the free field measure $\nu_R$, see Remark~\ref{rem:abso_con_measure}. Its construction, which we give
below, is based on Lemma~\ref{lem:nelson} (Nelson's estimate), which in turn follows from a hypercontractivity estimate~\cite{Nelson,DiGenova},~\cite[Ch.~5]{SJ97}. In contrast to the $\Phi^4$ model in three dimensions, where mass renormalization is needed, to treat the $P(\Phi)_2$ model one only needs to take into account vacuum energy renormalization via the Wick ordering\,\footnote{See~\cite[Table.~14.1]{GlimmJaffe} for a comparison between different theories.}, see~\cite[Sec.~4]{BFS83}. To this end, we work with a regularized covariance for the free field measure, which is the resolvent of an elliptic operator of degree six by introducing the UV cut-off $N\in\bN_+$, in the frequency space. Then, we shall study the limit as $N\to \infty$. Observe that such a choice of the UV cutoff breaks the reflection positivity of the Gaussian free field measure~\cite{ABRSS18}, but preserves its rotational invariance property. We restore the reflection positivity in Sec.~\ref{sec:reflection_posivity} using a different regularization.

\begin{Remark}\label{rem:measure_sphere_rotations}
For all $R\in[1,\infty)$, the measure $\mu_R$ is invariant under the action of the orthogonal group $O(3)$, see Remark~\ref{rem:sphere_O(3)}.
\end{Remark}
\section{Regularized $P(\Phi)_2$ measure}\label{sec:regularized_measure}
For all $R,N\in\bN_+$, consider the bounded operators $G_R\,,K_{R,N}\,:\,L_2(\bS_R)\to L_2(\bS_R)$ defined by
\begin{equation}
 G_R:=(1-\Delta_R)^{-1},\qquad K_{R,N}:=(1-\Delta_R/N^2)^{-1} \,.
\end{equation}
Observe that from Remark~\ref{rem:fundamental_solution_sphere}, the integral kernel of the operator $G_R= (1-\Delta_R)^{-1}$ has a logarithmic singularity in terms of the geodesic distance at the coinciding points. For all $R,N\in[1,\infty)$, we set $G_{R,N} := K_{R,N}\,G_R\,K_{R,N}$, which is a trace-class operator, see Remark~\ref{rem:expansion_covariance_harmo}, and denote its inverse by $Q_{R,N}:=(1-\Delta_R)(1-\Delta_R/N^2)^2$, where $Q_{R,N}\,:\, L^6_2(\bS_R)\subset L_2(\bS_R)\to L_2(\bS_R)$ is a self-adjoint, positive definite operator and generates a $C_0$-semigroup of contractions on $L^1_2(\bS_R)$, see Remark~\ref{rem:reproducing_kernel_sphere}.
\noindent
\\[0.3cm]
We denote the centered Gaussian measure with the covariance $G_{R,N}$ by $\nu_{R,N} := \cN(0,G_{R,N})$. The associated characteristic functional for all $f\in \sD(\bS_R)$ has the form
\begin{equation}
f\ni \sD(\bS_R) \mapsto S_{\nu_{R,N}}(f)
=\int_{\sD'(\bS_R)}\re^{i\langle\phi,f\rangle}\,\nu_{R,N}(\rd\phi) 
= \re^{-\frac{1}{2}\langle G_{R,N}f,f\rangle }\,,
\end{equation}
where $\langle\Cdot,\Cdot\rangle$ denote $\sD(\bS_R)-\sD'(\bS_R)$ duality pairing. The existence of the measure follows from Theorem~\ref{thm:Minlos}.
\\[0.3cm]
\noindent
One writes a class of regularized probability measures on $\sD'(\bS_R)$ capturing the $P(\Phi)_2$ interaction as follows
\begin{equation}\label{eq:regularized_measure}
 \mu_{R,N}(\rd\phi):=\frac{1}{\mathcal{Z}_{R,N}}\exp\left(-\int_{\bS_R} P(\phi(\rx),c_{R,N})\,\rho_R(\rd \rx)\right) \nu_{R,N}(\rd\phi),
\end{equation}
where
\begin{equation}\label{eq:counterterm}
\begin{aligned}
c_{R,N}:=\int_{\sD'(\bS_R)}\phi(\mathrm{x})^2\,\nu_{R,N}(\rd\phi)
=&\frac{1}{4\pi R^2}\int_{\bS_R}G_{R,N}(\rx,\rx)\,\rho_R(\rd\rx)
\\=&\, \mathrm{Tr}(G_{R,N})/4\pi R^2\,.
\end{aligned}
\end{equation}
We call $c_{R,N}$ a counterterm, and it is independent of the spatial variables and the metric tensor $g_R$. For all $\rx\in\bS_R$ one has,
\begin{equation}\label{eq:poly_hermit_counter}
P(\phi(\rx),c_{R,N}) :=\sum_{m=0}^n\,a_m\,H_m(\phi(\rx),c_{R,N}):= \sum_{m=0}^n\,a_m\,\phi(\rx)^{:m:}\,,
\end{equation}
where $H_m(\phi(\rx),c)$ is the $m-$th Hermite polynomial introduced in Eq.~\eqref{eq:hermite_real}. We also have the well-known Wiener-It\^o chaos decomposition~\cite[Thm.~2.6]{SJ97}
$$L^2(\sD'(\bS_R), \nu_{R,N})=\bigoplus_{l=0}^{\infty}\cH_l.$$
In particular, one has $P(\phi ,c_{R,N}) \in \bigoplus_{l=0}^{n}\cH_l$, which coincides with the inhomogeneous Wiener chaos introduced in Def.~\ref{def:wiener_chaos}. Moreover, from Lemma~\ref{Nualart-Lemma-1-1}, one obtains
$$\int_{\sD'(\bS_R)}\,P(\phi(\rx),c_{R,N})\,\nu_{R,N}(\rd\phi) = a_0.$$
Furthermore, from Remark~\ref{rem:expansion_covariance_harmo} and Lemma~\ref{lem:trace}, one deduces the existence of some $C\in(0,\infty)$ such that for all $N,R\in\bN_+$,
\begin{equation}\label{eq:bound_counterterm}
 |c_{R,N}-\frac{1}{2\pi}\,\log N|\leq C\,.
\end{equation}
\noindent
In order to control UV divergences, it is essential to establish the local regularity properties for the regularized covariance. This is the content of the next lemma. For all $N,M\in\bN_+$, define
$$G_{R,N,M}:=(1-\Delta_R/M^2)^{-1}\,(1-\Delta_R)^{-1}\,(1-\Delta_R/N^2)^{-1}.$$

\begin{lem}\label{lem:local_covariance}
Let $R\in\bN_+$. For any $n\geq 4$, there exists $C\in (0,\infty)$ such that for all $N,M\in\bN_+$ with $N\wedge M=N$, one has
\begin{itemize}
\item[(A)]
$\sup_{\rx\in\bS_R}\|G_{R,N,N}(\rx,\Cdot)\|_{L_n(\bS_R)} \leq C$.
\item[(B)]
$\|G_{R,N,N}(\Cdot,\Cdot)- G_{R,N,M}(\Cdot,\Cdot)\|_{L_n(\bS_R^2)}\leq C\,(N^{-1/n})$.
\item[(C)]
$\sup_{\rx\in\bS_R} G_{R,N,N}(\rx,\rx) \leq \,C\,(\log N)$.
\end{itemize}
\end{lem}
\begin{proof}
Note that a similar result can be found in~\cite[Thm.~1]{DiGenova}. Our proof relies on the rotational invariance property of the covariance $G_{R,N,N}$ and Lemma~\ref{lem:embedding_sphere}. There exist $\hat C,C\in(0,\infty)$ such that for all $N\in\bN_+$,
\begin{multline}
 \|G_{R,N,N}(\rx,\Cdot)\|^2_{L_n(\bS_R)}
 \leq \hat C\, 
 \|G_{R,N,N}(\rx,\Cdot)\|^2_{L_2^{(n-2)/n}(\bS_R)}
 \\
 =
 (4\pi R^2)^{-1}\,\hat C\,\big(\mathrm{Tr}\big[G_R^{1+2/n} K_{R,N}^4\big]\big).
\end{multline}
Observe that
\begin{multline}
\mathrm{Tr}\big[G_R^{(n+2)/n} K_{R,N}^4\big]
= \sum_{l=0}^{\infty}\,\frac{(2l+1)}{(1+l(l+1)/R^2)^{1+2/n}}\,\big(\frac{1}{1+l(l+1)/(RN)^2} \big)^4
 \leq C.
\end{multline}
This finishes the proof of Item~(A). To prove Item~(B), one writes
\begin{multline}
 \|(G_{R,N,N}-G_{R,N,M})(\rx,\Cdot)\|^2_{L_n(\bS_R)}
 \leq \hat C\, 
 \|(G_{R,N,N}-G_{R,N,M})(\rx,\Cdot)\|^2_{L_2^{(n-2)/n}(\bS_R)}
 \\
 =
 (4\pi R^2)^{-1}\,\hat C\,\big(\mathrm{Tr}\big[G_R^{(n+2)/n} K_{R,N}^2 (K_{R,N}-K_{R,M})^2\big]\big)
 \leq C\,(N\wedge M)^{-2/n}.
\end{multline}
See the proof of Lemma~\ref{lem:stochastic_Y} for more details. Item~(C) follows from Lemma~\ref{lem:trace} and Eq.~\eqref{eq:bound_counterterm}. This concludes the proof.
\end{proof}

\section*{Auxiliary $P(\Phi)_2$ measure on the sphere}\label{sec:auxi_measure}
Let $g\in C^\infty(\bS_R)$ such that
\begin{equation}\label{eq:bound_on_g}
\|g\|^n_{L_{n/(n-1)}(\bS_R)}\leq 1/2,\qquad\|\Delta_R\, g\|^n_{L_{n/(n-1)}(\bS_R)}\leq 1/2\,.
\end{equation}
\noindent
From now on, we would like to work deliberately with a more general class of probability measures defined by
\begin{equation}\label{eq:regularized_measure_g}
 \mu_{R,N}^g(\rd\phi):=\frac{1}{\mathcal{Z}^g_{R,N}}\exp\left(\phi(g)^n/n\right) \mu_{R,N}(\rd\phi).
\end{equation}
Observe that if $g=0$, for all $R,N\in\bN_+$, the measure $\mu_{R,N}^g$ coincides with $\mu_{R,N}$. Later our particular choice will turn out advantageous for proving the non-Gaussianity property of a $P(\Phi)_2$ measure in Ch.~\ref{ch:Osterwalder-Schrader axioms}.

\begin{lem}\label{lem:polynomial_bound}
There exists $A\in(0,\infty)$ depending only on the coefficients of the polynomial $u \mapsto P(\tau)$ such that for all $\tau\in\bR$ and $c\in(1,\infty)$ it holds $P(\tau,c)\geq \tau^n/2n-A \,c^{n/2}$.
\end{lem}
\begin{proof}
Let us start with the following consideration:
Let $\delta\in (0,1 )$, and $p,q\in[1,\infty)$ such that $1/p+1/q=1$. By the Young inequality, it holds that
\begin{equation}
\alpha\, a\,b\leq |a|\,(|\alpha|\,|b|)\leq (p\,\delta)^{1/p}\,|a|\, (p\,\delta)^{-1/p}\,(|\alpha|\,|b|)
\leq \delta\,|a|^p+ (1/q)\,(p\,\delta)^{-q/p}\,(|\alpha|\,|b|)^q\,.
\end{equation}
This implies that
$$-\delta\,|a|^p- (1/q)\,(p\,\delta)^{-q/p}\,(|\alpha|\,|b|)^q\leq \alpha\, a\,b.$$ 
Assuming that $c\in (1,\infty)$, $0\leq m\leq n-1$, $0\leq k\leq \lfloor m/2\rfloor$, $\alpha_k, \tau \in\mathbb{R}$, we apply the same reasoning with
$$1/p=(m-2k)/n,\,\text{ and}\,\,1/q=(n-m+2k)/n.$$
We obtain
\begin{multline}
\alpha_k\,\tau^{m-2k} c^k\,\leq | \tau^{m-2k}|\, (|\alpha_k| c^k)\,
\\
\leq \delta\,\tau^n+ \frac{n-m+2k}{n}\, (\frac{n}{(m-2k)}\,\delta)^{-(n-2k)/(n-m+2k)}\,|\alpha_k|^{n/(n-m+2k)}\,c^{\,kn/(n-m+2k)}.
\end{multline}
Let
$$C:=(n-m+2k)/n\, (n/(m-2k)\,\delta)^{-(n-2k)/(n-m+2k)}\,|\alpha_k|^{n/(n-m+2k)}.$$
Then, since $c^{n/2}\geq c^{\,kn/(n-m+2k)}$, we have
$$-\delta\, \tau^n- C\,c^{n/2}\leq \alpha_k\, c^k\,\tau^{m-2k}.$$
Using the above argument, one bounds all terms with powers $\tau^m$, $m<n$, by $\tau^n$: 
\begin{equation}
\alpha\,\tau^{m-2k}c^{k} \geq -\delta\,\tau^n - C\,c^{\frac{n}{2}}. \label{eq:polynomial-bound}
\end{equation}
To conclude, we apply the above bound to all terms of the polynomial $P(\tau,c)$ but the term $\tau^n/n$ and choose $\delta\in (0,1)$ sufficiently small. A similar bound is obtained in~\cite[Prop. 8.6.3]{GlimmJaffe}.
\end{proof}

\begin{Remark}
We identify implicitly a function $\phi\in L_1(\bS_R)$ with a distribution $\varphi\in\sD'(\bS_R)$
defined by
$$\varphi(f)\equiv \langle\varphi,f\rangle:=\int_{\bS_R}\,\phi(\rx)f(\rx)\,\rho_R(\rd \rx).$$
\end{Remark}

\begin{lem}\label{lem:mu_measure_well_defined}
For all $R,N\in\mathbb{N}_+$ and $g\in C^\infty(\bS_R)$ such that $\|g\|^n_{L_{n/(n-1)}(\bS_R)}\leq 1/2$, the measure $\mu_{R,N}^g$ is well-defined and both $\nu_{R,N}$ and $\mu_{R,N}^g$ are concentrated on $L^1_2(\bS_R)\subset\sD'(\bS_R)$.
\end{lem}
\begin{proof}
From Lemma~\ref{lem:stochastic_X_not_uniform} one infers that the measure $\nu_{R,N}$ is concentrated on $L^1_2(\bS_R)$ in the sense of Def.~\ref{def:concentrate}. Moreover, using Lemma~\ref{lem:embedding_sphere} it holds that $L^1_2(\bS_R)\subset L_n(\bS_R)$ for all $n\geq 2$. From the H\"older inequality and by assumption one has
\begin{equation}
\phi(g)^n= \big(\int_{\bS_R}\,\phi(\rx)\,g(\rx)\,\rho_R(\rd \rx)\big)^n\leq \|\phi\|^n_{L_n(\bS_R)}\,\|g\|^n_{L_{n/n-1}(\bS_R)}\leq \frac{1}{2}\|\phi\|^n_{L_n(\bS_R)}\,.
\end{equation}
Recall that $P(\phi(\rx),c_{R,N})=\sum_{m=0}^n\,a_m\,\phi(\rx)^{:m:}$. Applying relation~\eqref{eq:polynomial-bound} to all terms in the sum with power $m<n$, it follows that there exists $A\in(0,\infty)$ such that 
$$P(\phi(\rx),c_{R,N})\geq -\phi^n/2n-A\,c_{R,N}^{n/2}\,.$$
Hence, by linearity one obtains
\begin{equation}\label{eq:polynomail_potentail}
\frac{1}{n}\phi(g)^n-\int_{\bS_R} P(\phi(\rx),c_{R,N})\,\rho_R(\rd \rx)
\leq A\,c_{R,N}^{n/2}\,.
\end{equation}
Consequently, the map $\mathcal{U}^g_{R,N}$ defined by
\begin{equation}
 \mathcal{U}^g_{R,N}\,:\,
 L_n(\bS_R)\ni \phi \mapsto \exp\left(\frac{1}{n}\phi(g)^n-\int_{\bS_R} P(\phi(\rx),c_{R,N})\,\rho_R(\rd \rx)\right)  \in (0, \infty)
\end{equation}
is bounded and continuous, hence is measurable. Using Jensen's inequality, Fubini's theorem, and the fact that 
$\int_{\sD'(\bS_R)}\,\phi(g)^n\,\nu_{R,N} (\rd\phi)$ is positive, one gets
\begin{equation}
\begin{aligned}
\mathcal{Z}^g_{R,N}\,\mathcal{Z}_{R,N}
&=\int_{\sD'(\bS_R)}\exp \left(\frac{1}{n}\phi(g)^n-\int_{\bS_R} P(\phi(\rx),c_{R,N})\,\rho_R(\rd \rx)\right) \nu_{R,N}(\rd\phi)
\\
&\geq \exp\Big( \int_{\sD'(\bS_R)}\,\frac{1}{n}\phi(g)^n\,\nu_{R,N}(\rd\phi)-\int_{\sD'(\bS_R)}\int_{\bS_R} P(\phi(\rx),c_{R,N})\,\rho_R(\rd \rx)\,\nu_{R,N}(\rd\phi)\Big)
\\
&
\geq \exp\Big( \int_{\sD'(\bS_R)}\,\frac{1}{n}\phi(g)^n\,\nu_{R,N}(\rd\phi)-\int_{\bS_R} \,\int_{\sD'(\bS_R)}P(\phi(\rx),c_{R,N})\,\nu_{R,N}(\rd\phi)\,\rho_R(\rd \rx)\Big)
\\&
\geq \exp\Big(-a_0\,(4\pi R^2)\Big),
\end{aligned}
\end{equation}
where we evoked the fact that expectations of Wick monomials of non-zero order vanish. Hence, $\mathcal{Z}_{R,N}\mathcal{Z}_{R,N}^g = \|\mathcal{U}^g_{R,N}\|_{L_1(\sD'(\bS_R),\nu_{R,N})}\in (0,\infty)$. This finishes the proof.
\end{proof}

\section{Existence of the UV limit}\label{sec:existence_UV_limit}
Consider the probability space $(\sD'(\bS_R),\mathrm{Borel}(\sD'(\bS_R)), \nu_R)$ and define
for all $R,N\in[1,\infty)$, the random Gaussian variables $X_R$ and $X_{R,N}:=K_{R,N}X_R$ almost surely taking values in the Hilbert space $L^{1}_2(\bS_R)$. It holds that $\mathrm{Law}(X_{R})= \nu_{R}$ with covariance
$$\bE\,X_R(\rx)\otimes X_R(\rx') = G_R(\rx,\rx'),\,\,\,\, \forall\,\,\, \rx,\rx'\in\bS_R$$
and $\mathrm{Law}(X_{R,N})=\nu_{R,N}$ with covariance
$$\bE [X_{R,N}(\rx)\otimes X_{R,N}(\rx') ]= G_{R,N}(\rx,\rx'),\,\,\, \forall\,\,\, \rx,\rx'\in\bS_R.$$
On account of Lemma~\ref{lem:stochastic_X_uniform} Item~(A), one deduces that $X_R$ is of regularity $-\kappa$ for any $\kappa >0$, i.e., takes values in $L^{-\kappa}_2(\bS_R)$, almost surely with respect to the measure $\nu_{R}$.
\\[0.3cm]
Using Eq.~\eqref{eq:hermite_hilbert} one has
$$X_{R,N}^{:m:}(\rx)= c^{m/2}_{R,N}H_m(X_{R,N}(\rx)/c^{1/2}_{R,N}).$$
Note that from Lemma~\ref{lem:stochastic_X_not_uniform} Item~(A), it follows that for all $R,N\in\bN_+$ and all $m\in\{1,2\cdots,n-1\}$, the Wick-ordered random variables $X_{R,N}^{:m:}$ can be bounded almost surely in a space with regularity $2-\kappa$ for any $\kappa>0$. Now, for all $h\in C^\infty(\bS_R)$ we set
\begin{equation}\label{eq:pairing_gaussian_process}
X_{R,N}^{:m:}(h):=\int_{\bS_R} X_{R,N}^{:m:}(\rx)\,h(\rx)\,\rho_R(\rd \rx)\,.
\end{equation}
Using the preceding expression with $h=\1_{\bS_R}$, one defines
\begin{equation}\label{dfn:X_Y}
\begin{aligned}
Y_{R,N}&:=\sum_{m=0}^n\, a_m\, X_{R,N}^{:m:}(\mathmybb{1}_{\bS_R}),
\qquad
Y_{R,N}^g:=Y_{R,N}-X_{R,N}(g)^n/n\,,
\\
Y_{R}&:=\sum_{m=0}^n\, a_m\, X_{R}^{:m:}(\mathmybb{1}_{\bS_R}),
\qquad
Y_{R}^g:=Y_{R}-X_{R}(g)^n/n\,.
\end{aligned}
\end{equation}
Combining Lemma~\ref{lem:stochastic_X_not_uniform} and Lemma~\ref{lem:embedding_sphere}, one verifies that almost surely $X_{R,N}\in L^{2-\kappa}_2(\bS_R)\subset L_n(\bS_R)$. Consequently, $Y_{R,N}$ and $Y^g_{R,N}$ are well-defined.

\begin{Remark}
Note that $\nu_{R,N},\, \mu_{R,N}\in \cM(L^1_2(\bS_R))$ in the sense explained in Sec.~\ref{sec:probability}.
This implies that for all bounded and continuous $F\,:\, \sD'(\bS_R)\to\bR$, we have
\begin{equation}
\begin{gathered}
\int F(\phi)\,\nu_{R,N}(\rd\phi)=\bE\,F(X_{R,N}),
\quad
 \int F(\phi)\, \mu^g_{R,N}(\rd\phi)
 =
 \frac{\bE\,F(X_{R,N})\exp(-Y^g_{R,N})}{\bE\exp(-Y^g_{R,N})}\,.
 \end{gathered}
\end{equation}
It follows form the fact that for all $R,N\in\bN_+$, one has $\mathrm{Law}(X_{R,N})=\nu_{R,N}$.
\end{Remark}

\begin{lem}\label{lem:prob_of_density}
For every $R\in \bN_+$, there exist constants $C,\,c\in (0,\infty)$ such that for all $N\in \bN_+$ and all $t> 0$ it holds
\begin{equation}
 \mathbb{P}(-\,Y^g_{R,N}> t) \leq C\,\exp(-\exp(c\,t^{2/n})).  \label{eq:prob_of_density}
\end{equation}
\end{lem}

\begin{Remark}\label{rem:prob_of_density}
On account of the preceding lemma, one deduces that for every $R\in\bN_+$, there exist $C,c\in(0,\infty)$ such that for all $N\in\mathbb N_+$, $p>0$, and $t>0$ it holds
\begin{equation}
\bP(-p\,Y^g_{R,N}>t) =\bP(-Y^g_{R,N}>\frac{t}{p})
\leq C\exp[-\exp(c\, p^{-2/n}t^{2/n})].
\end{equation}
\end{Remark}

\begin{proof}[Proof of Lemma~\ref{lem:prob_of_density}]
By Lemma~\ref{lem:polynomial_bound} for every $R\in\mathbb{N}_+$ there exists $A\in(0,\infty)$ such that for all $M\in\mathbb{N}_+$ it holds 
\begin{equation}
Y_{R,M}^g\geq -A\,c_{R,M}^{n/2}. \label{eq:polynomial_bound} 
\end{equation}
Consequently, one has
\begin{equation}
\begin{aligned}
 \bP(-Y^g_{R,N}> 2A\,c_{R,M}^{n/2}) 
& = \bP(-Y^g_{R,M}-(Y^g_{R,N} -Y^g_{R,M})> 2A\,c_{R,M}^{n/2}) 
\\
 &\leq  \mathbb{P}(|Y^g_{R,N}-Y^g_{R,M}|> A\, c_{R,M}^{n/2} )\,.
\end{aligned}
\end{equation}
Let $(N\wedge M)=M$ and $c\in (0,\infty)$. One writes
\begin{multline}\label{eq:uniform_uv_1}
\bP(|Y^g_{R,N}-Y^g_{R,M}|> A\, c_{R,M}^{n/2} )
=  \bP\Big(\exp(c\,|Y^g_{R,N}-Y^g_{R,M}|^{2/n})> \exp(c\,A^{2/n}\, c_{R,M})\Big)
\\
 =\bP\Big(\exp(c|Y^g_{R,N}-Y^g_{R,M}|^{2/n}\,M^{1/n^2})> \exp(c\,A^{2/n}\, c_{R,M}\,M^{1/n^2})\Big)
\\
\leq \frac{\bE\,\exp(c\,|Y^g_{R,N}-Y^g_{R,M}|^{2/n}\,M^{1/n^2}) }{\exp(c\,A^{2/n}\, c_{R,M}\,M^{1/n^2})}
\\
\leq \exp(-c\,A^{2/n}\, c_{R,M}\,M^{1/n^2})\, \bE\,\exp\Big(c\,M^{1/n^2}\,|Y^g_{R,N}-Y^g_{R,M}|^{2/n}\Big),
\end{multline}
where we used the Markov inequality in the third step. Invoking Lemma~\ref{lem:nelson} with $X=|Y^g_{R,N}-Y^g_{R,M}|$ and letting $c=n/ (6\,C^{2/n})$, where $C$ is the constant from Lemma~\ref{lem:stochastic_Y} Item~(B), for some $C'\in(0,\infty)$ depending on $R\in\bN_+$, one obtains
\begin{multline}\label{eq:uniform_uv_2}
\bE\,\exp\Big(c\,M^{1/n^2}\,|Y^g_{R,N}-Y^g_{R,M}|^{2/n}\Big)
\\
\leq\bE\,\exp\Big(\frac{n\,|Y^g_{R,N}-Y^g_{R,M}|^{2/n}\,(\bE|Y^g_{R,N} - Y^g_{R,M}|^2)^{1/n}\,M^{1/n^2}}{6\,\bE\,(|Y^g_{R,N} - Y^g_{R,M}|^2\,C^{2/n} )^{1/n}}\Big) < C'\,.
\end{multline}
We used Lemma~\ref{lem:stochastic_Y} Item~(B), which gives $\bE\,|Y^g_{R,N} - Y^g_{R,M}|^2)^{1/n}\,M^{1/n^2} \leq C^{2/n}$, where $C^{2/n}\in(0,\infty)$ does depend on $R\in\bN_+$. Now from Eq.~\eqref{eq:bound_counterterm} one has
$$-\,C+(2\pi)^{-1}\,\log M\leq c_{R,M}\leq C+(2\pi)^{-1}\,\log M.$$
For now we restrict attention to $M\geq M_0$ such that $-\,C+(2\pi)^{-1}\,\log M> 0$. This implies that
\begin{itemize}
\item[(A)]
$\log M\leq 2\pi(c_{R,M}+C)$. Hence, $M\leq \exp{(2\pi(c_{R,M}+C))}.$
\item[(B)]
$c_{R,M}\leq C+\frac{1}{2\pi}\,\log M$. Thus, $\exp{(2\pi(c_{R,M}-C))}\leq M$.
\end{itemize}
By Item~(B) above, one rewrites Eq.~\eqref{eq:uniform_uv_1} with the help of Eq.~\eqref{eq:uniform_uv_2} 
\begin{equation}\label{eq:uniform_uv_4}
\begin{aligned}
 \bP(-Y^g_{R,N}> 2A\,c_{R,M}^{n/2}) 
&\leq\, C'\,\exp(-c\,A^{2/n}\, c_{R,M}M^{1/n^2})\\
&\leq C'\,\exp(-c\,A^{2/n}\, c_{R,M}\exp(\frac{2\pi\,c_{R,M}-C}{n^2}\,))\,.
\end{aligned}
\end{equation}
Recall that $A\in (0,\infty)$ was chosen so as to ensure  $Y_{R,M}^g\geq -A\,c_{R,M}^{n/2}$. This latter bound remains valid if we increase $A$. Therefore, without loss, we can assume $cA^{2/n} \exp(-C/n^2)\geq 1$.  By increasing $M_0$, if necessary, we can also ensure $c_{R,M}\geq 1$ and write
\begin{equation}
\bP(-Y^g_{R,N}> 2A\,c_{R,M}^{n/2}) \leq C'\,\exp(- \exp(\frac{2\pi\,c_{R,M}}{n^2}\,)).
\end{equation}
Setting $t_M:=2A\,c_{R,M}^{n/2}$, we obtain for some $C''>0$, $N\geq M\geq M_0$ 
\begin{equation}
\bP(-Y^g_{R,N}> t_M )\leq   C'\,\exp(-\exp( C'' t_M^{2/n} \,)). \label{eq:central-relation}
\end{equation}
Let us now eliminate the restriction $N\geq M$.
Suppose that $(N\wedge M)=N$, which implies $\log N\leq \log M$. It holds that
\begin{multline}\label{eq:uniform_uv_5}
\{-Y^{g}_{R,N}> 2A\,c_{R,M}^{n/2}\}\subseteq \{-Y^{g}_{R,N}> 2A\,(-C+(2\pi)^{-1}\,\log M)^{n/2}\}
\\
\subseteq \{-Y^{g}_{R,N}> 2A\,(-C+(2\pi)^{-1}\,\log N)^{n/2}\}\subseteq \{-Y^{g}_{R,N}> A\,c_{R,N}^{n/2}\}
\end{multline}
To justify the last inclusion, let us first write
\begin{equation}\label{eq:uniform_uv_6}
-3C+(2\pi)^{-1}\log(N)\leq c_{R,N}-2C \leq -C+(2\pi)^{-1}\log(N).
\end{equation}
Then, by the second inequality above,
\begin{multline}\label{eq:uniform_uv_7}
 \{-Y^{g}_{R,N}> 2A\,(-C+(2\pi)^{-1}\,\log N)^{n/2}\}
 \\
 \subseteq \{-Y^{g}_{R,N}>2A(c_{R,N}-2C)^{n/2}\}
 \subseteq \{-Y^{g}_{R,N}> A\,c_{R,N}^{n/2}\},
\end{multline}
provided that $N\geq N_0$, $N_0$ sufficiently large such that $2A(c_{R,N}-2C)^{n/2}\geq A\,c_{R,N}^{n/2}$. This justifies~\eqref{eq:uniform_uv_5}. Using this relation and
 the fact that $B_1\subseteq B_2$ implies $\bP(B_1)\leq\bP(B_2)$ one infers that
\begin{equation}
\bP\,(-Y^{g}_{R,N}>2\,A\,c_{R,M}^{n/2} )\leq \bP\,(-Y^{g}_{R,N}>\,A\,c_{R,N}^{n/2})=0,\,
\end{equation}
where in the last step we used Eq.~\eqref{eq:polynomial_bound}. We conclude that Eq.~\eqref{eq:central-relation} holds for $N\geq N_0$ and all $M\in \mathbb{N}_+$, and note that 
it suffices to prove the lemma under these conditions. In fact, for a finite number of cases, i.e., $N=1,\ldots N_0-1$,
estimate~\eqref{eq:prob_of_density} trivially follows from the fact that $\bP\,(-Y^{g}_{R,N}>t )=0$ for $t$ larger than some 
$t^{(N)}$, which is a consequence of Eq.~\eqref{eq:polynomial_bound}. To extend Eq.~\eqref{eq:central-relation} from the sequence $\{t_M\}_{M\in \mathbb{N}_+}$ to all $t\in \mathbb{R}_+$\, we proceed as follows. For a given $t$, let $t_{M_-}$ be the largest element of this sequence such that $t_{M_-}\leq t$. 
Let us show that for $\ell=\lceil \exp(4\pi C) \rceil$ we have $t_{M_-} \leq t_{\ell M_-}$, hence $t \leq t_{\ell M_-}$.  This follows from
\begin{equation}
c_{R, M_-}\leq C+(2\pi)^{-1}\,\log M_-=(2\pi)^{-1}\,\log \ell M_--C+(2C- (2\pi)^{-1}\,\log \ell)\leq c_{R, \ell M_-}.
\end{equation}
Now, we can write, using Eq.~\eqref{eq:central-relation}
\begin{equation}
\begin{aligned}
\bP(-Y^g_{R,N}> t )\leq  \bP(-Y^g_{R,N}> t_{M_-}) & \leq    C'\,\exp(-\exp( C''  ({t_{M_-}} / {t_{\ell M_-}})^{2/n}  t_{\ell M_-}^{2/n}\,)) \\
& \leq C'\,\exp(-\exp( C''  ({t_{M_-}} / {t_{\ell M_-}})^{2/n}  t^{2/n}\,)).
\end{aligned}
\end{equation}
Finally, we note that
\begin{equation}
({t_{M_-}} / {t_{\ell M_-}})^{2/n}=(c_{R,M_-}/ c_{R,\ell M_-} )\geq \frac{(2\pi)^{-1}\,\log  M_- -C }{(2\pi)^{-1}\,(\log \ell +\log M_-) +C  }\,.
\end{equation}
As we restricted attention to $M\geq M_0$ such that $(2\pi)^{-1}\,\log  M_- -C>0$, the last expression can be bounded from below by 
a strictly positive constant, independent of $M_-$. Thus, we conclude the proof for $t\geq t_0$, $t_0$ independent of $N$, enforced by the condition $M_-\geq M_0$. For $t\leq t_0$ we can ensure
(\ref{eq:prob_of_density}) by adjusting the constants so that the right hand side is larger that $1$ on this set.
\end{proof}

\begin{lem}\label{lem:X_Y_covergence}
Let $F: \sD'(\bS_R)\to \mathbb{R}$ be continuous. Then, the sequence
$$(F(X_{R,N}) \exp(-\,Y_{R,N}^g))_{N\in\bN_+}$$
converges in probability to $F(X_R) \exp(-\,Y_R^g)$.
\end{lem}

\begin{proof}
On account of Remark~\ref{rem:convergence}, it suffices to verify the $L_2(\Omega,\cF,\bP)$ convergence. In fact, for $\epsilon > 0$, $\kappa\in(0,\infty)$, and $\delta\in[0,2]$, one has
\begin{equation}
\bP\,(\|X_R- X_{R,N}\|^2_{L_2^{-\kappa-\delta}(\bS_R)}\geq \epsilon )
\leq\,\epsilon^{-1}\,\bE\,\|X_R- X_{R,N}\|_{L_2^{-\kappa-\delta}(\bS_R)}^2
\leq \epsilon^{-1}\,R^2\,C^2\,N^{-2\delta}.
\end{equation}
Similarly,
\begin{equation}
\mathbb{P}\,( |Y_R-Y_{R,N}|^2\geq \epsilon) \leq\,\epsilon^{-1}\,C^2\,N^{-1/n}\,,
\end{equation}
where we used Lemma~\ref{lem:stochastic_X_uniform} and Lemma~\ref{lem:stochastic_Y}.
From the above bounds one concludes that the sequence $(X_{R,N})_{N\in\bN_+}$ and $(Y_{R,N}^g)_{N\in\bN_+}$ converge in probability as $N$ approaches infinity to $X_R$ and $-Y^g_R$, respectively. The statement follows from the fact that
$$F\,:\,\mathscr{D}^{\prime}(\bS_R)\to\bR\,\, \text{and}\,\,t\mapsto\exp(t)\,\,\text{ are\,\, continuous}.$$
Hence, by Remark~\ref{rem:convergence}, the sequences $(F(X_{R,N}))_{N\in\bN_+}$ and $(\exp(-\,Y_{R,N}^g))_{N\in\bN_+}$ also converge in probability to $F(X_R)$ and $\exp(-Y^g_R)$. This concludes the proof.
\end{proof}

\begin{lem}\label{lem:uniform_integrable}
Let $F\,:\, \sD'(\bS_R)\to \mathbb{R}$ be bounded and continuous. Then,
the sequence
$(F(X_{R,N}) \,\exp(-\,Y_{R,N}^g))_{N\in\bN_+}$ is uniformly integrable.
\end{lem}
\begin{proof}
By Remark~\ref{rem:uni_inte}, it suffices to verify that the set
$$\{\bE\,|\exp(-Y^g_{R,N})|^p~:~N\in\bN_+\}$$
is bounded for some $p>1$. Observe that
\begin{equation}
\begin{aligned}
\bE\, \exp(-p\,Y^g_{R,N})
& = \bE\, \exp(-p\,(Y^g_{R,N}\wedge 0))\, \1_{\{Y^g_{R,N}\leq 0\}}
+\bE\, \exp(-p\,(Y^g_{R,N}\vee 0))\, \1_{\{Y^g_{R,N}> 0\}}
\\
&\leq 1+\bE\, \exp(-p\,(Y^g_{R,N}\wedge 0))\, \1_{\{Y^g_{R,N}\leq 0\}}\,.
\end{aligned}
\end{equation}
Hence, using Lemma~\ref{lem:E_F_X} with $F(X) =\exp(X)-1 $ and $X = -p\,(Y^g_{R,N}\wedge 0)$ yields
\begin{equation}\label{eq:density_bound}
\begin{aligned}
 \bE\,\exp(-p\,Y_{R,N}^g) &\leq 2+\bE\,( \exp(-p\,(Y^g_{R,N}\wedge0))-1) 
 \\
 &=2+\int_0^\infty \bP(-p\,(Y^g_{R,N}\wedge 0)>t)\,\exp(t) \,\rd t
 \\
&= 2+\int_0^\infty \bP(-p\,Y^g_{R,N}>t)\, \exp(t) \,\rd t\,.
 \end{aligned}
\end{equation}
One utilizes Lemma~\ref{lem:prob_of_density} to ensure that the expression $\bP(-p\,Y^g_{R,N}>t)$ has sufficient decay to counteract the growing behaviour of $\exp(t)$. This makes the last integral convergent and concludes the proof.
\end{proof}

\begin{prop}\label{prop:uv_limit}
Let $R\in\bN_+$ and $g\in C^\infty(\bS_R)$ satisfy the bounds~\eqref{eq:bound_on_g}.  There exist random variables $X_R$ and $Y_R^g$ such that $\bE\,\exp(-Y_R^g)<\infty$ and for all bounded and continuous $F\,:\,\sD'(\bS_R)\to \bR$ and all $p\in (0,\infty)$ it holds that
\begin{equation}
 \lim_{N\to\infty} \bE\, F(X_{R,N}) \,\exp(-p\,Y_{R,N}^g) = \bE\, F(X_R)\, \exp(-p\,Y_R^g).
\end{equation}
\end{prop}
\begin{proof}
On account of Lemmas~\ref{lem:X_Y_covergence}, and~\ref{lem:uniform_integrable}, one infers the assumptions stated in Item~(A) of the Vitali theorem~\ref{thm:Lebesgue-Vitali's theorem}. This concludes the statement.
\end{proof}

\begin{Remark}\label{rem:abso_con_measure}
From the Radon-Nikodym theorem~\ref{Radon-Nikodym theorem}, one infers that the finite volume measure $\mu^g_{R}$ is absolutely continuous with respect to the finite volume Gaussian measure $\nu_{R}$, e.g., they have the same null sets, i.e., for all $A\in\mathrm{Borel}(\sD'(\bS_R))$ such that $\nu_R(A)=0$ one has $\mu^g_R(A)$=0. However, in the limit as $R\to \infty$, the infinite volume $P(\Phi)_2$ measure on $\sS'(\bR^2)$ might not be absolutely continuous with respect to Gaussian measure, see~\cite[p.~25]{RZZ17}.
\end{Remark}

\chapter{Stochastic quantization}\label{ch:stochastic_quantization}
Stochastic quantization is a quantization method introduced by Parisi and Wu~\cite{parisi1981}. It is based on the fact that the Euclidean Green functions can be seen as correlation functions of a statistical system in equilibrium. The equivalence of the stochastic quantization and the path integral quantization for a variety of field theories, in a heuristic sense, has been shown in~\cite{BZ81,FL83}. In particular, the equivalence of the perturbative stochastic quantization with the path integral quantization for a scalar field is demonstrated in~\cite[Sec.~3.3]{DH87}. 

\section{Stochastic quantization on the sphere}\label{sec:stochastic_quantization}
Implementing the setting described in Ch.~\ref{ch:introduction}, the parabolic stochastic quantization associated to the probability measure $\mu^g_{R,N}(\rd\phi)$ can be read off as the following non-linear SPDE 
\begin{equation}\label{eq: SPDE on Sphere}
(\partial_t  +Q_{R,N})\varPhi^g_{R,N}(t,\rx)
= \sqrt{2}\, \xi(t,\rx) -\,P'(\varPhi^g_{R,N}(t,\rx), c_{R,N}) +(\varPhi^g_{R,N}(t,g))^{n-1}\,g\,,
\end{equation}
where $(t,\rx)\in \bR_+\times \bS_R$, $Q_{R,N}:=(1-\Delta_R)(1-\Delta_{R} / N^2)^2$,  $P'(\tau,c) := \partial_{\tau}P(\tau,c)$, and $\xi(t,\rx):= \frac{\rd W_R(t,\rx)}{\rd t}$, where $W_R$ is the
cylindrical Wiener process  on $H=L_2(\bS_R)$, cf. Def.~\ref{def:cyl_w_proc}. (By rescaling $t$, the factors $1/2$ appearing in Eq.~\eqref{eq:sq_pphi} are traded for $\sqrt{2}$ in front of $\xi$).
\\[0.3cm]
\noindent
We aim to prove that the measure $\mu^g_{R,N}(\rd\phi)$ is invariant under the dynamics generated by Eq.~\eqref{eq: SPDE on Sphere}, see Thm.~\ref{thm:inv_measu_non_gauss}.
This implies that if $\varPhi^g_{R,N}$ is a stationary solution of Eq.~\eqref{eq: SPDE on Sphere} with the initial condition $\varPhi^g_{R,N}(0,\Cdot) = \phi^g_{R,N}$ distributed according to the measure $\mu_{R,N}^g(\rd\phi)$, then for all $t\in \bR_+$ the random field $\varPhi^g_{R,N}(t,\Cdot)$ is also distributed according to $\mu^g_{R,N}(\rd\phi)$. 
\\[0.3cm]
The $P(\Phi)_2$ model is subcritical, i.e., the non-linear part in the RHS of Eq.~\eqref{eq: SPDE on Sphere} can be viewed as a perturbation of the linear equation. This implies that the renormalized non-linearity should have better regularity than the Gaussian space-time white noise~\cite[p.~417]{Hai14}. As a first step toward constructing the $P(\Phi)_2$ measure on $\sD'(\bS_R)$, we need to verify the existence of the global in time solution to Eq.~\eqref{eq: SPDE on Sphere} for any fixed $N\in\bN_+$. To this end,
we first prove the existence of the local in time stationary solution of Eq.~\eqref{eq: SPDE on Sphere} using the so-called Da Prato-Debussche decomposition~\cite{daprato2003} combined with a fixed point argument. Then, we shall prove the existence of the global in time stationary solution of Eq.~\eqref{eq: SPDE on Sphere} using the energy method. Note that the negative sign in Eq.~\eqref{eq: SPDE on Sphere} is essential for proving the long time existence of the solution in Sec.~\ref{subsec:global_solution_remainder}.

\begin{Remark}
Here the extra (stochastic) time dimension is just a parameter and it differs from the physical time in Minkowski setting. It enables us to compute the correlation functions of the $d$–dimensional quantum field theory after the $(d+1)-$dimensional system has reached its equilibrium.
\end{Remark}

\begin{Remark}\label{rem:regularity_noise_fields}
As a first step, let us analyze Eq.~\eqref{eq: SPDE on Sphere} heuristically
to identify function spaces in which one can try to solve it. Let us consider the corresponding linear equation
\begin{equation}
\big(\partial_t  +Q_{R,N}\big)Z_{R,N}(t,\rx)= \sqrt{2}\,\xi(t,\rx).
\end{equation}
The solution, if it exists, has the form
\begin{equation}\label{eq:lin-sol}
Z_{R,N}(t,\Cdot)=\re^{-tQ_{R,N}}\,Z_{R,N}(0)+\sqrt{2} \int_0^t \re^{-(t-s)\,Q_{R,N}}\,\xi(s,\Cdot)\,\rd s\,. 
\end{equation}
Let us set $Z_{R,N}(0)\equiv 0$ for simplicity. The second term on the right hand side of Eq.~\eqref{eq:lin-sol} above is called a stochastic convolution 
and will be analyzed rigorously in Sec.~\ref{sec:linear} below. For now, let us provide a formal computation:
\begin{align}
\mathbb{E} \| Z_{R,N}(t,\Cdot)\|^2_{L_2^{\alpha}(\bS_R)}&=2\int_0^t \int_0^t\mathbb{E} \lan  \re^{-(t-s)Q_{R,N}}\xi(s, \Cdot), \re^{-(t-s')Q_{R,N}}  \xi(s',\Cdot)\ran_{L_2^{\alpha}(\bS_R)}\, \rd s\,  \rd s' \\
&=2\int_0^t \int_0^t  \mathbb{E} \lan \xi(s, \Cdot),\, \re^{-(2t-s-s')Q_{R,N}}\,(1-\Delta_{R})^{\alpha} \xi(s' , \Cdot)\ran_{L_2(\bS_R)}\, \rd s\, \rd s' \\
&=2\int_0^t \int_0^t \sum_{l=0}^{\infty}\sum_{m=-l}^l\,  \mathbb{E} \lan \xi(s, \Cdot) , Y_{l,m}\ran_{L_2(\bS_R)}  \lan Y_{l,m}, \xi(s', \Cdot) \ran_{L_2(\bS_R)} \times\\
&\phantom{444444444444444444}\times \re^{-(2t-s-s')\langle l\rangle }\big(1+ \frac{l(l+1)}{R^2} \big)^{\alpha}  \rd s\,  \rd s',
\end{align}
where $\langle l \rangle \simeq l^6$ was introduced in Remark~\ref{rem:l_brackets}. Using Remark~\ref{rem:spherical_projection} one obtains
\begin{multline}
\sum_{m=-l}^l\, \mathbb{E} \lan \xi(s, \Cdot) ,Y_{l,m}\ran_{L_2(\bS_R)}  \lan Y_{l,m}, \xi(s', \Cdot) \ran_{L_2(\bS_R)} 
\\
=  \sum_{m=-l}^l\,  \mathbb{E} \int_{\bS_R}\, \xi(s, \rx) Y_{l,m}(\rx)\rho_R(\rd \rx) \,  \int_{\bS_R} \overline Y_{l,m}(\rx') \xi(s', \rx') \rho_R(\rd \rx')
\\
= \sum_{m=-l}^l\, \int_{\bS_R}\, \int_{\bS_R}\,  \mathbb{E}[ \xi(s, \rx)\xi(s', \rx')]\,  Y_{l,m}(\rx)\, \overline Y_{l,m}(\rx') \rho_R(\rd \rx')\rho_R(\rd \rx)
\\
= \delta(s-s')\,\sum_{m=-l}^l\, \int_{\bS_R}\, \int_{\bS_R}\, \delta_{R}( \rx, \rx')  Y_{l,m}(\rx) \overline Y_{l,m}(\rx') \rho_R(\rd \rx')\rho_R(\rd \rx)
\\
= \delta(s-s')\, \int_{\bS_R}\,\sum_{m=-l}^l Y_{l,m}(\rx)\, \overline Y_{l,m}(\rx)\, \rho_R(\rd \rx)
= \delta(s-s')\,\frac{(2l+1)}{ 4\pi R^2}\, \int_{\bS_R}\, \rho_R(\rd \rx)
\\
= \delta(s-s')\,(2l+1),
\end{multline}
where we have used $\mathbb{E}[\xi(s,x)\xi(s',x')] = \delta_R(x,x')\delta(s-s')$ such that
$$
\int_{\bS_R}\delta_R(\rx,\rx')\,\rho_R(\rd \rx') = 1,\,\text{and}\,\,\int_{\bS_R}\delta_R(\rx,\rx')\,Y_{l,m}(\rx')\rho_R(\rd \rx') = Y_{l,m}(\rx).$$
We obtain
\begin{align} \label{ell-sum}
\mathbb{E}\,\|Z_{R,N}(t)\|^2_{L_2^{\alpha}(\bS_R)}=\sum_{l=0}^{\infty} (2l+1)  \frac{1-\re^{- 2t\langle l\rangle }}{ \langle l\rangle}\big(1+ \frac{l(l+1)}{R^2} \big)^{\alpha}.
\end{align}
The above sum converges provided that $6-(2\alpha+1)>1$, i.e. $\alpha <2$. (This is consistent with Lemma~\ref{lem:stochastic_X_not_uniform}). Accordingly, we will look for mild solutions of 
 the regularized equation~(\ref{eq: SPDE on Sphere}) in the space $L^1_2(\bS_R)$, for which the non-linear terms in the equation make sense. On the other hand, if we removed the UV regularization, that is, set  $N\to \infty$, the corresponding condition for convergence of the sum in Eq.~\eqref{ell-sum} would be $ 2-(2\alpha+1)>1$, i.e., $\alpha<0$. (This is consistent with 
 Lemma~\ref{lem:stochastic_X_uniform}). As the corresponding 
 spaces $L^{\alpha}_2(\bS_R)$ contain distributions whose powers are undefined, the non-linear Eq.~\eqref{eq: SPDE on Sphere} would be singular without regularization.
\end{Remark}
\noindent
Throughout this chapter, we fix a stochastic basis $(\Omega, \cF,\{\cF_t\}_{t\in[0,\infty)},\bP)$ and we consider the cylindrical Wiener process $(W_R(t,\Cdot))_{t\geq 0}$ on $L_2(\bS_R)$ defined in Def.~\ref{def:cyl_w_proc}. 
\section{Linear equation}\label{sec:linear}
Let us first look at the linear equation corresponding to~\eqref{eq: SPDE on Sphere},
\begin{equation}\label{eq:weak_free}
 \rd Z_{R,N}(t,\Cdot) 
 =
 \sqrt{2}\,\rd W_R(t,\Cdot)- Q_{R,N} Z_{R,N}(t,\Cdot)\,\rd t\,,
\end{equation}
whose invariant measure is $\nu_{R,N}$.
\begin{lem}\label{lem:free_solution}
The unique mild solution $Z_{R,N}\in C([0,\infty),L^1_2(\bS_R))$
of Eq.~\eqref{eq:weak_free} with the initial condition $Z_{R,N}(0,\Cdot)=z_{R,N}\in  L^1_2(\bS_R)$ exists in the sense of Def.~\ref{def:mild_solution}. It has the form
\begin{equation}\label{eq:mild_free}
 Z_{R,N}(t,\Cdot) = \re^{-t\,Q_{R,N}} z_{R,N}+ \sqrt{2}\,\int_0^t\re^{-(t-s)\,Q_{R,N}}\,\rd W_R(s,\Cdot)\,.
\end{equation}
\end{lem}
\begin{Remark} \label{Remark:gen-Wick-powers} Actually, the lemma remains true if we replace  $L^1_2(\bS_R)$ with $L_2^{2-\kappa}(\bS_R)$
for some $0<\kappa<1$. This can be checked by replacing in Eq.~\eqref{eq:counting-powers} the 
factor $(1+\frac{l(l+1)}{R^2})$ with $(1+\frac{l(l+1)}{R^2})^{2-\kappa}$ and choosing $\alpha<\kappa/6$.
We will need this generalization in Remark~\ref{Remark:bound-on-Wick-powers} below.
\end{Remark}
\begin{proof}[Proof of Lemma~\ref{lem:free_solution}]
The problem of existence of the mild solution~\eqref{eq:mild_free} amounts to the existence of the stochastic integral. 
It exists almost surely, in particular, for (deterministic) Borel functions $\Phi\,:\,\mathbb{R}_+\to \mathcal{L}_2(L_2(\bS_{R}), L_2^1(\bS_{R}))$, where the latter denotes the space of Hilbert-Schmidt operators from $L_2(\bS_{R})$ to $L_2^1(\bS_{R})$, with norm $\|\,\cdot\,\|_{\mathrm{HS}}$. Moreover, we have for $p\in (0,\infty)$, by the Burkholder-Davis-Gundy theorem
\begin{equation} \label{eq:Burkholder-Davis-Gundy}
\bE\,\| \int_0^t\, \Phi(s)\,\rd W(s)\,\|^p_{L^1_2(\bS_{R})} \leq  c_p\,\bE\,\left(\int_0^t \|\Phi(s)\|^2_{\mathrm{HS}}\,\rd s\right)^{p/2}.
\end{equation}
See Remark.~\ref{rem:BDG_ineq}. For $p\geq 2$ it is a consequence of the Nelson's estimate (Lemma~\ref{lem:nelson} above), since we consider deterministic $\Phi$. We aim to modify~\cite[Thm.~6.10]{Hai23}, which is based on~\cite[Thm.~5.2.6]{DaZab96} for our problem at hand. (We recall that~\cite[Thm.~6.10]{Hai23} considers stochastic convolutions of the form $ \int_0^t \,S(t-s)\,Q\, \rd W(s) $,
where $W$ is cylindrical Wiener process on $\mathcal{K}$, the operator $Q: \mathcal{K}\to \mathcal{H}$ is bounded and $S$ is a semigroup on $\mathcal{H}$. We will check below that the argument still works if $Q$ is omitted and $S$ is considered as
a family of operators from $\mathcal{K}$ to $\mathcal{H}$). Let $t\in [0,T]$ and define for $\alpha\in(0,1/6)$
\begin{equation}
y(t,\Cdot):=\int_0^t\,(t-s)^{-\alpha}\,\re^{-(t-s)Q_{R,N}}\,\rd W_R(s,\Cdot)\,.
\end{equation}
To check if the integral is well defined, we estimate the Hilbert-Schmidt norm of the integrand:
From Lemma~\ref{lem:contraction_semigroup} for arbitrary $T>0$, one shows 
\begin{equation}\label{eq:linear_equation_1}
\int_0^T s^{-2\alpha} \|\re^{-sQ_{R,N}}\|^2_{\mathrm{HS}}\, \rd s
=\sum_{l=0}^{\infty}\,(2l+1)\,(1+\frac{l(l+1)}{R^2})\,\int_0^T s^{-2\alpha}\,e^{-2s\langle l\rangle}\,\rd s<\infty\,,
\end{equation}
where $\langle l\rangle \simeq l^6$ was introduced in Remark~\ref{rem:l_brackets}. The preceding bound follows from the fact that for all $s\in[0,T]$ it holds that
\begin{equation}
\int_0^T\,s^{-2\alpha}\,\re^{-2s\langle l\rangle}\rd s\leq \int_0^{\infty}\,s^{-2\alpha}\,\re^{-2s\langle l\rangle}\rd s\,=\, (2\langle l\rangle)^{2\alpha-1}\,\Gamma(1-2\alpha),
\end{equation}
which converges for fixed $l\in\bN_0$ since $1-2\alpha>0$. Consequently, Eq.~\eqref{eq:linear_equation_1} is finite if $\alpha<1/6$. This follows from  
\begin{multline} \label{eq:counting-powers}
\sum_{l=0}^{\infty}\,(2l+1)\,(1+\frac{l(l+1)}{R^2})\,\int_0^T s^{-2\alpha}\,e^{-2s\langle l\rangle}\,\rd s
\\
\leq 
\,
\Gamma(1-2\alpha)\,\sum_{l=0}^{\infty}\,(2l+1)\,(1+\frac{l(l+1)}{R^2})\, (2\langle l\rangle)^{2\alpha-1}\,.
\end{multline}
Comnining Eqs.~\eqref{eq:Burkholder-Davis-Gundy},~\eqref{eq:linear_equation_1}, we write
\begin{multline}\label{eq:factorization_1}
\bE\,\int_0^t\|y(s,\Cdot)\|_{L^1_2(\bS_R)}^p\rd s\leq C_p\,\int_0^t\,\left|\int_0^s\,\|(s-r)^{-\alpha}\re^{-(s-r)Q_{R,N}}\|^2_{\mathrm{HS}}\,\rd r \right|^{p/2}\,\rd s <\infty,
\end{multline}
Let $c_{\alpha}= \frac{\sin(\pi\alpha)}{\pi}$ such that $c_{\alpha}\,\int_s^t\,(t-r)^{\alpha-1}\,(r-s)^{-\alpha}\,\rd r=1$
holds for every $0<s<t$. Note that the integral over the set $\{ (s,r)\,:\, 0<s<r<t\}$ can be written in two equivalent ways 
$$\int_0^t\,\int_s^t\,\ldots \rd r\,\rd W(s)=\int_0^t\,\int_0^r \,\ldots  \rd W(s)\,\rd r.$$
Hence, implementing the preceding identity into Eq.~\eqref{eq:weak_free} one gets
\begin{equation}
\begin{aligned}
Z_{R,N}(t,\Cdot)
&=\re^{-tQ_{R,N}}z_{R,N}+c_{\alpha}\,\int_0^t\,\int_s^t\,(t-r)^{\alpha-1}\,(r-s)^{-\alpha}\,\re^{-(t-s)Q_{R,N}}\,\rd r\,\rd W_R(s,\Cdot)
\\
&=\re^{-tQ_{R,N}}z_{R,N}+c_{\alpha}\,\int_0^t\,\int_0^r\,(t-r)^{\alpha-1}\,(r-s)^{-\alpha}\,\re^{-(t-s)Q_{R,N}}\,\rd W_R(s,\Cdot)\,\rd r.
\end{aligned}
\end{equation}
Using $\re^{-(t-s)Q_{R,N}}= \re^{-(t-r)Q_{R,N}}\,\re^{-(r-s)Q_{R,N}}$ one has
\begin{multline}
\int_0^t\,\int_0^r\,(t-r)^{\alpha-1}\,(r-s)^{-\alpha}\,\re^{-(t-s)Q_{R,N}}\,\rd W_R(s,\Cdot)\,\rd r
\\
=\int_0^t\,\re^{-(t-r)Q_{R,N}}\,\int_0^r\,(r-s)^{-\alpha}\,\re^{-(r-s)Q_{R,N}}\,\rd W_R(s,\Cdot)\,(t-r)^{\alpha-1}\,\rd r.
\end{multline}
Hence,
\begin{equation}
\begin{aligned}
Z_{R,N}(t,\Cdot)
&=\re^{-tQ_{R,N}}z_{R,N}+c_{\alpha}\,\int_0^t\,\re^{-(t-r)Q_{R,N}}\,y(r,\Cdot)\,(t-r)^{\alpha-1}\,\rd r\,.
\end{aligned}
\end{equation}
To proceed, consider the map $y\mapsto F_y$, where
$$F_y(t,\Cdot):=\int_0^t\,\re^{-(t-r)Q_{R,N}}\,(t-r)^{\alpha-1}\,y(r,\Cdot)\,\rd r.$$
The statement will follow from Eq.~\eqref{eq:factorization_1} if we show that for every $\alpha\in (0,1/6)$ there exists $p>0$ such that $ y\mapsto F_y$ maps $L_p([0,\infty),L^1_2(\bS_R))$ into $C([0,\infty),L^1_2(\bS_R))$. Note that the semigroup $t\mapsto \re^{-tQ_{R,N}}$  is uniformly bounded in the operator norm on $L_2^1(\bS_R)$ on any time interval due to the positivity of spectrum of the elliptic operator $Q_{R,N}$ and the map $t\mapsto (t-r)^{\alpha-1}$ belongs to $L_q([0,t))$ for $q\in[1,1/(1-\alpha) )$. From the H\"older inequality, one infers that there exists a constant $C_T>0$
such that for all $p>1/\alpha$ one has
\begin{equation}
\begin{aligned}
\|F_y(t,\Cdot)\|_{L^1_2(\bS_R)}
&
\leq \int_0^t\,(t-r)^{\alpha-1}\,\|y(r,\Cdot)\|_{L^1_2(\bS_R)}\,\rd r
\\&\leq \left(\int_0^t\, (t-r)^{q(\alpha-1)}\,\rd r\right)^{1/q}\left(\int_0^t\,\|y(r,\Cdot)\|_{L^1_2(\bS_R)}^p\,\rd r\right)^{1/p}.
\end{aligned}
\end{equation}
Consequently,
\begin{equation}
\sup_{t\in [0,T]}\,\|F_y(t,\Cdot)\|_{L^1_2(\bS_R)}^p\leq C_T\,\int_0^T\,\|y(r,\Cdot)\|_{L^1_2(\bS_R)}^p\,\rd r\,,
\end{equation}
where $C_T:=(\frac{T^{q(\alpha-1)+1} }{q(\alpha-1)+1})^{p/q}$. To proceed, we verify that the map $ t\mapsto F_y(t)$ is continuous for every continuous function $y$ satisfying $y(0)=0$. Fix $t>0$. Using the semigroup property of $\re^{-(t+\Delta t-r)\,Q_{R,N}}$ one has
\begin{multline}\label{eq:free_field_continuity_1}
\|F_y(t+\Delta t,\Cdot)-F_y(t,\Cdot)\|_{L^1_2(\bS_R)}
\\
=\| \int_0^t\,\left[(t+\Delta t -r)^{\alpha-1}\,\re^{-\Delta tQ_{R,N}}- (t-r)^{\alpha-1} \right]\,\re^{-(t-r)Q_{R,N}}\,y(r,\Cdot)\,\rd r
\\
+\int_t^{t+\Delta t}\, (t+\Delta t -r)^{\alpha-1}\,\re^{-(t+\Delta t-r)Q_{R,N}}\,y(r,\Cdot)\,\rd r\|_{L^1_2(\bS_R)}.
\end{multline}
Consider the first integrand on the RHS of Eq.~\eqref{eq:free_field_continuity_1}. It converges to zero as $\Delta t\to 0$ thanks to the strong continuity of the semigroup $\re^{-(t-r)Q_{R,N}}$. Moreover, it is bounded by $C(t-r)^{\alpha-1}\,\| y(r,\Cdot)\|_{L^1_2(\bS_R)}$, which is integrable.
Hence, by the dominated convergence theorem, as $\Delta t\to 0$ this integral tends to zero. Consider the second integrand on the RHS of Eq.~\eqref{eq:free_field_continuity_1}. One has
\begin{multline}
\|\int_t^{t+\Delta t}\, (t+\Delta t -r)^{\alpha-1}\,\re^{-(t+\Delta t-r)\,Q_{R,N}}\,y(r,\Cdot)\,\rd r\|_{L^1_2(\bS_R)}
\\
=\|\int_0^{\infty}\, \chi_{[0,\Delta t]}(r)\, r^{\alpha-1}\,\re^{-r\,Q_{R,N}}\,y(t+\Delta t -r,\Cdot)\,\rd r\|_{L^1_2(\bS_R)}\,.
\end{multline}
Observe that as $\Delta t\to 0$, the integrand tends to zero, since $\chi_{[0,\Delta t]}(r)$ goes to zero. Moreover, 
\begin{align}
\|\chi_{[0,\Delta t]}(r)\, r^{\alpha-1}\,\re^{-rQ_{R,N}}\,y(t+\Delta t -r,\Cdot)\|_{L^1_2(\bS_R)}
&
\leq  C\,\chi_{[0,1]}(r)\, r^{\alpha-1}\,\|y( t+\Delta t -r,\Cdot)\|_{L^1_2(\bS_R)} 
\\
&
\leq C'\,\chi_{[0,1]}(r)\, r^{\alpha-1}\,(1+\|y(t-r,\Cdot)\|_{L^1_2(\bS_R)}), 
\end{align}
where we made use of the continuity of $y$. As the last expression is integrable, by the dominated convergence theorem, this integral tends to zero as $\Delta t\to 0$. This shows that $ y\mapsto F_y$ is right-continuous for every continuous function $y$ with $y(0)=0$. For the proof of left-continuity, which is similar, we refer to  \cite[Thm.~6.10]{Hai23}. Noting that continuous functions are dense in $L_p$-spaces for all $1\leq p<\infty$, one concludes the proof of continuity of the paths for almost all $\omega\in \Omega$. It now follows by dominated convergence and Eq.~\eqref{eq:Burkholder-Davis-Gundy} that the solution belongs to $L_2(\Omega, \mathcal{F}, \mathbb{P})$.
\end{proof}

\section{Non-linear equation}\label{sec:non-linear}
Let us come back to the stochastic quantization equation~\eqref{eq: SPDE on Sphere}  whose invariant measure is $\mu^g_{R,N}$. It holds that
\begin{multline}\label{eq:weak_interacting}
 \rd \varPhi_{R,N}^g(t,\Cdot)
 =
 \sqrt{2}\, \rd W_R(t,\Cdot)- Q_{R,N}\,\varPhi_{R,N}^g(t,\Cdot)\,\rd t 
 \\
 -P'(\varPhi_{R,N}^g(t,\Cdot),c_{R,N})\,\rd t + (\varPhi_{R,N}^g(t,g))^{n-1}g\,\rd t\,,
\end{multline}
which is well-defined for fixed $R,N\in\bN_+$. However, it becomes singular in the limit $N\to\infty$ due to Remark~\ref{rem:regularity_noise_fields}, i.e., its solution, if it exists, will be almost surely bounded in a Bessel potential space of negative regularity. Invoking the Da Prato-Debussche trick~\cite{daprato2003}, we write, for all $t\in [0,\infty)$, the random field $\varPhi^g_{R,N}(t,\Cdot)$ as
\begin{equation}\label{eq:DDD}
\varPhi^g_{R,N}(t,\Cdot) = \varPsi^g_{R,N}(t,\Cdot) + Z_{R,N}(t,\Cdot),
\end{equation}
where $Z_{R,N}(t,\Cdot)$ is given in Eq.~\eqref{eq:weak_free} and $\varPsi^g_{R,N}(t,\Cdot)$ is the mild solution to the following PDE
\begin{equation}\label{eq:varPsi_pde}
\partial_t\varPsi^g_{R,N}  = -Q_{R,N}\varPsi^g_{R,N}
- P'(\varPsi^g_{R,N}+Z_{R,N},c_{R,N}) +((\varPsi^g_{R,N}+Z_{R,N})(\Cdot,g))^{n-1}g\,.
\end{equation}
We choose the initial data $z_{R,N}\in L^1_2(\bS_R)$, $\phi^g_{R,N}\in L^1_2(\bS_R)$ of the solutions $\{ Z_{R,N}(t,\Cdot)\}_{t\geq 0}$, $\{\varPhi_{R,N}^g(t,\Cdot) \}_{t\geq 0}$ to Eqs.~\eqref{eq:weak_free} and~\eqref{eq:weak_interacting} as random variables with laws given by the respective  measures $\nu_{R,N}$, $\mu^g_{R,N}$. This is legitimate because we verified in  Lemma~\ref{lem:mu_measure_well_defined} 
 that these measures are concentrated on $L^1_2(\bS_R)\subset \sD'(\bS_R)$. We will verify in Sec.~\ref{sec:invariant_measures} that these are invariant measures for the respective equations. Consequently, the resulting solutions are stationary if they exist, cf. Lemmas~\ref{lem:free_solution}, \ref{lem:mild_solution_interacting}. However, $\varPsi_{R,N}(t,\Cdot)$ might not be stationary.

\begin{Remark}\label{rem:regularity_free_fields}
On account of Remark~\ref{rem:regularity_noise_fields}, $Z_{R,N}(t,\Cdot)$ is of regularity $2-\kappa$ for all $\kappa >0$. Using the smoothing property of the heat kernel associated to the operator $Q_{R,N}$, we expect better regularity for the remainder $\varPsi^g_{R,N}(t,\Cdot)$. This property is of importance in Sec.~\ref{subsec:global_solution_remainder}, where we apply the energy method to Eq.~\eqref{eq:varPsi_pde} for proving the existence of the global-in-time solution corresponding to Eq.~\eqref{eq:varPsi_pde}.
\end{Remark}

\begin{lem}\label{lem:mild_solution_interacting}
The mild solution $\varPhi^g_{R,N}(t,\Cdot)\in C ([0, \infty), L^1_2 (\bS_R ))$ to Eq.~\eqref{eq:weak_interacting} with initial condition $\phi_{R,N}^g\in L^1_2(\bS_R)$ exists in the sense of Def.~\ref{def:mild_solution}.
\end{lem}

\begin{proof}
From Eq.~\eqref{eq:DDD} it suffices to verify the existence and the uniqueness of random fields $Z_{R,N}(t,\Cdot)$ and $\varPsi^g_{R,N}(t,\Cdot)$. By Lemma~\ref{lem:free_solution} the mild solution $Z_{R,N}(t,\Cdot)\in C([0,\infty), L^1_2(\bS_R))$ exists and is unique. The existence and uniqueness of the local- and global-in-time mild solution to Eq.~\eqref{eq:varPsi_pde}, i.e., $\varPsi^g_{R,N}(t,\Cdot)\in C([0,\infty), L^1_2(\bS_R))$ are shown in Sec.~\ref{subsec:local_solution_remainder}. and Sec.~\ref{subsec:global_solution_remainder}. This finishes the proof.
\end{proof}

\subsection{Local existence in time}\label{subsec:local_solution_remainder}
This subsection is devoted to proving the existence of the local-in-time mild solution to Eq.~\eqref{eq:varPsi_pde} via a fixed point argument in the space $L^1_2(\bS_R)$. To this end, we adapt the framework presented in~\cite[Sec.~15.1]{TaylorIII}. A similar idea is implemented in~\cite[Sec.~3.2]{Zab89}. The proof is deterministic and relies mainly on two facts: first, for all $m \in \{1,2,\cdots, n-1\}$, $p \in[2,\infty)$ and fixed $R,N\in\bN_+$, the renormalized free fields $Z_{R,N}^{:m:}(t,\Cdot)$ are uniformly bounded in $L_p(\bS_R)$ for $t\in [0,T]$, see Remark~\ref{Remark:bound-on-Wick-powers} below. 
Second, the non-linearity
$$P'(\varPhi_{R,N}^g(t,\Cdot),c_{R,N})+((\varPhi^g_{R,N})(t,g))^{n-1}\,g(\Cdot)\in C([0,\infty),L_2(\bS_R))$$
defines, for all $t\geq 0$, a locally Lipschitz continuous map from $L^1_2(\bS_R)$ to $L_2(\bS_R)$, which follows from a suitable Sobolev embedding and H\"older's inequality.
\\[0.3cm]
In Ch.~\ref{ch:finite_volume} we introduced a set of real coefficient $(a_m)_{m\in\{0,\ldots,n\}}$, which are fixed throughout the thesis.
Now, we define a new set of coefficients
$$a_{m,l}:=-\,\frac{a_{m+1}\,(m+1)!}{(m-l)!l!}\quad \text{ for}\quad l\in\{0,\ldots,n-2\}\quad\text{ and } \quad m\in\{l,\ldots,n-1\}.$$
Note that for some $C\in (0,\infty)$ and for all $0\leq l\leq n-2$, $l\leq m\leq n-1$ it holds that $|a_{m,l}|\leq C$.
Moreover, from Eq.~\eqref{eq:hermite_hilbert} for all $m\in \{1,\cdots,n-1\}$, one has
\begin{equation}
Z_{R,N}^{:m:}= \sum_{k=0}^{\lfloor m/2\rfloor} \frac{(-1)^k\,m!}{(m-2k)!\, k!\, 2^k}\,(c_{R,N})^k\, Z_{R,N}^{m-2k}\,.
\end{equation}
By definition $Z_{R,N}^{:0:}=1$. Recall that $P'(\tau,c) := \partial_{\tau}P(\tau,c)$. One has
\begin{equation}
P'(\varPhi^g_{R,N},c_{R,N})
=P'(\varPsi^g_{R,N}+Z_{R,N},c_{R,N}) 
 = (\varPsi^g_{R,N})^{n-1} - \sum_{l=0}^{n-2}\sum_{m=l}^{n-1}\,a_{m,l}\, Z_{R,N}^{:m-l:}\,(\varPsi^g_{R,N})^l.
\end{equation}
\noindent
Let
\begin{multline}\label{eq:non-linear-Psi}
F( \varPsi^{g}_{R,N}(s,\Cdot))
\\
= -P'((\varPsi_{R,N}^g+Z_{R,N})(s,\Cdot),c_{R,N})+ ((\varPsi_{R,N}^g +Z_{R,N})(s,g))^{n-1}\,g(\Cdot)
\end{multline}
and extend this definition from $s\mapsto \varPsi^{g}_{R,N}(s,\Cdot)$ to arbitrary $\varPsi\in C([0,T] ; L^1_2(\bS_R))$.

\begin{lem}\label{lem:lipschitz}
Consider the mapping $\cK$ defined by
\begin{equation}
\cK(\varPsi)(t,\Cdot):= \re^{-t\,Q_{R,N}}\,\psi^g_{R,N} - \Gamma(\varPsi)(t,\Cdot),
\quad
\Gamma(\varPsi)(t,\Cdot) = \int_0^t\,\re^{-(t-s)\,Q_{R,N}}\,F(\varPsi(s,\Cdot))\,\rd s,
\end{equation}
where $\psi_{R,N}^g:=\phi^g_{R,N}- z_{R,N}\in L^1_2(\bS_R)$. It holds that
\begin{equation}
\cK\,:\,C([0,T] ;L^1_2(\bS_R)) \to C([0,T] ;L^1_2(\bS_R))
\end{equation}
and there exists a constant $C\in (0,\infty)$ such that for every $\varPsi_1,\varPsi_2\in C([0,T];L^1_2(\bS_R))$ and all $g\in C^{\infty}(\bS_R)$, it holds that
\begin{equation}\label{eq:lip_bound_non_lineararity}
\|\cK(\varPsi_1) - \cK(\varPsi_2)\|_{C([0,T]; L^1_2(\bS_R))}
\leq C\,T^{\,5/6}\,D(\varPsi_1,\varPsi_2)\,\,\|\varPsi_1 -\varPsi_2\|_{C([0,T];L^1_2(\bS_R))},
\end{equation}
where 
\begin{equation}
D(\varPsi_1,\varPsi_2):=\Big[2+ \sum_{l=1}^{n-2}\,\left(\|\varPsi_1\|^{l}_{C([0,T];L^1_2(\bS_R))} + \|\varPsi_2\|^{l}_{C([0,T];L^1_2(\bS_R))}\right)\Big].
\end{equation}
\end{lem}
\noindent
\begin{Remark} \label{Remark:bound-on-Wick-powers}
The constant $C$ on the RHS of Eq.~\eqref{eq:lip_bound_non_lineararity} depends on the quantities $\sup_{0\leq s\leq T}\|Z_{R,N}^{:m-l:}(s,\Cdot)\|_{L_p(\bS_R)}$.
To justify that they  are finite for any $T\geq 0$, we recall that we are now looking for solutions at a fixed UV cut-off $N$, thus we can express the Wick powers as polynomials in $Z_{R,N}$.
For any monomial $Z_{R,N}^{\tilde m}$, $\tilde{m}\in \mathbb{N}$, we can write
\begin{equation} \label{eq:auxiliary}
\|Z_{R,N}^{\tilde{m}}(s,\Cdot)\|_{L_p(\bS_R)}=\|Z_{R,N} (s,\Cdot)\|^{\tilde m}_{L_{\tilde{m} p}(\bS_R)} \leq c\|Z_{R,N} (s,\Cdot)\|^{\tilde m}_{L_{2}^1(\bS_R)},
\end{equation}
where we applied the Sobolev inequality given in Lemma~\ref{lem:embedding_sphere}. 
(We assumed here that $p\geq 2$, which can be ensured using the H\"older inequality).
For almost all $\omega\in\Omega$,  the RHS of Eq.~\eqref{eq:auxiliary} is bounded in $s\in [0,T]$ on compact sets by the continuity statement in Lemma~\ref{lem:free_solution}. 
 The argument can be generalized to treat $\|\vec \nabla_R Z_{R,N}^{:m-l:}(0,\Cdot)\|_{L_{p}(\bS_R)} $. Instead of (\ref{eq:auxiliary}), we are led to
\begin{equation} \label{eq:bound-on-Wick-powers}
\|\vec\nabla_R Z_{R,N}(s,\Cdot) Z_{R,N}^{\tilde{m}}(s,\Cdot)\|_{L_p(\bS_R)}\leq \|Z_{R,N}(s,\Cdot)\|_{L_2^{2-\kappa}(\bS_R)} \|Z_{R,N}^{\tilde{m}} (s,\Cdot)\|_{L_{p'}(\bS_R)}
\end{equation}
with some $0<\kappa<1$,
where we applied the H\"older inequality with $1/p=1/q'+1/p'$, $q'\geq 2$, and then Remark~\ref{rem:bessel_sobolev_psi}.
The second factor on the RHS is now estimated as in Eq.~\eqref{eq:auxiliary}. For the finiteness of the first factor  we refer to Remark~\ref{Remark:gen-Wick-powers}.
 \end{Remark}
\begin{proof}[Proof of Lemma~\ref{lem:lipschitz}]
Observe that if $\psi_{R,N}^g\in L^1_2(\bS_R)$, then by Remark~\ref{rem:reproducing_kernel_sphere} it follows that $\re^{-t\,Q_{R,N}}\psi^g_{R,N}$ also belongs to $L^1_2(\bS_R)$. Furthermore,  
by estimate~\eqref{eq:M-bound}, for all $t\in [0,T]$, one has
\begin{equation}\label{eq:non-linearity_second_1}
\begin{aligned}
\|\Gamma(\varPsi)(t,\Cdot)\|_{L^1_2(\bS_R)}
&
\leq C\,\int_0^t\,\frac{1}{|t-s|^{1/6}}\,\|F(\varPsi)(s,\Cdot)\|_{L_2(\bS_R)}\,\rd s,
\\
&
\leq C\,\sup_{s\in [0,T]}\,\,\|F(\varPsi)(s,\Cdot)\|_{L_2(\bS_R)}\,\int_0^t\,\frac{\rd s}{|t-s|^{1/6}}\,,
\end{aligned}
\end{equation}
where
\begin{multline}\label{eq:non_linearity_contra_1}
\|F(\varPsi)(s,\Cdot)\|_{L_2(\bS_R)}
\leq\,\, \| P'\big((\varPsi+Z_{R,N})(s,\Cdot),c_{R,N}\big)\|_{L_2(\bS_R)}
\\
+ \|\big((\varPsi+Z_{R,N})(s,g)\big)^{n-1}g(\Cdot)\|_{L_2(\bS_R)}
:=R^{(1)}(s)+R^{(2)}(s)\,.
\end{multline}
By the triangle inequality, one obtains
\begin{multline}
R^{(1)}(s)=
\| P'(( \varPsi +Z_{R,N})(s,\Cdot),c_{R,N})(s,\Cdot)\,\|_{L_2(\bS_R)}
\\
\leq
\sum_{l=0}^{n-1}\sum_{m=l}^{n-1}\,|a_{m,l}|\,\| Z_{R,N}^{:m-l:}(s,\Cdot)\,\varPsi(s,\Cdot)^l\|_{L_2(\bS_R)}\,.
\end{multline}
From Remark~\ref{Remark:bound-on-Wick-powers} above, one infers that the random fields $Z^{:m-l:}_{R,N}(s,\Cdot)$ are $L_p$-integrable with probability 1 for all $1\leq m\leq n-1$, and for all $p\in[2,\infty)$. Using H\"older's inequality with $1/p+1/q=1$, one gets
$$\|Z_{R,N}^{:m-l:}\,\varPsi^l\|_{L_2(\bS_R)}
 \leq \,\|\varPsi^l\|_{L_{2q}(\bS_R)} \,\|Z_{R,N}^{:m-l:}\|_{L_{2p}(\bS_R)}.$$
Hence, from the assumptions on the coefficients $|a_{m,l}|$ one verifies that there are $C,\hat C,\check C\in(0,\infty)$ such that
\begin{equation}\label{eq:R^1_local}
\begin{aligned}
R^{(1)}(s)\leq  
C+\,\hat C\,\sum_{l=1}^{n-1} \| \varPsi(s,\Cdot)\|^l_{L_{2ql}(\bS_R)}\leq 
C\, +\check C\,\sum_{l=1}^{n-1}\,\|\varPsi(s,\Cdot)\|^l_{L^1_2(\bS_R)}\,,
\end{aligned}
\end{equation}
where $C\in (0,\infty)$ depends on $\|Z_{R,N}^{:m-l:}\|_{L_{2p}(\bS_R)}$, cf. Remark~\ref{Remark:bound-on-Wick-powers}. Moving on to $R^{(2)}(s)$, from the H\"older inequality with $1/n+(n-1)/n=1$ one gets
\begin{equation}
\begin{aligned}
|\varPsi(s,g)|&\leq \| \varPsi(s,\Cdot)\|_{L_n(\bS_R)}\, \|g\|_{L_{n/(n-1)}(\bS_R)},
\\
|Z_{R,N}(s,g)|& \leq \|Z_{R,N}(s,\Cdot)\|_{L_n(\bS_R)}\, \|g\|_{L_{n/(n-1)}(\bS_R)}.
\end{aligned}
\end{equation}
Hence, for $C,\tilde C,C'\in (0,\infty)$ one has 
\begin{equation} \label{eq:R^2_local}
\begin{aligned}
R^{(2)}(s)=&\,\|(( \varPsi+Z_{R,N})(s,g))^{n-1}g\|_{L_2(\bS_R)}=\|g\|_{L_2(\bS_R)}\,|(\Psi+Z_{R,N})(s,g)^{n-1}|
\\
\leq &\,\,\|g\|_{L_{\infty}(\bS_R)}\,\sum_{k=0}^{n-1}\,{n-1\choose k}\,|\varPsi(s,g)|^k\,|Z_{R,N}(s,g)|^{n-1-k}
\\
\leq &\,
C+C'\,\sum_{k=1}^{n-1}\,{n-1\choose k}\,\| \varPsi(s,\Cdot)\|^k_{L_{n}(\bS_R)}\, \|Z_{R,N}(s,\Cdot)\|_{L_n(\bS_R)}^{n-1-k}
\\
\leq &\,C+
 \tilde C\,\sum_{k=1}^{n-1}\,\| \varPsi(s,\Cdot)\|^{k}_{L^1_2(\bS_R)},
\end{aligned}
\end{equation}
where in the last step we used Remark~\ref{Remark:bound-on-Wick-powers} and Lemma~\ref{lem:embedding_sphere}. Combining Eqs.~\eqref{eq:R^1_local} and~\eqref{eq:R^2_local}, one infers that $\Gamma(\varPsi)(t,\Cdot)$ takes values in $L^1_2(\bS_R)$ if $\varPsi(t,\Cdot)\in L^1_2(\bS_R)$ for all $t\in[0,T]$. Now, we aim to verify that $\cK(\varPsi)(t,\Cdot)\in C([0,T]; L^1_2(\bS_R))$ if $\varPsi\in C([0,T] ; L^1_2(\bS_R))$. Using Remark~\ref{rem:reproducing_kernel_sphere} it holds
$\|\re^{-tQ_{R,N}}\,\psi_{R,N}\|_{L^1_2(\bS_R)}\leq\,\|\psi_{R,N}\|_{L^1_2(\bS_R)}$, which is due to positivity of the spectrum of the operator $Q_{R,N}$. It holds that
\begin{equation}\label{eq:Psi_continuity_local_0}
\lim_{\Delta t\to 0}\|\re^{-t\,Q_{R,N}} (\re^{-\Delta t\,Q_{R,N}} - 1)\, \psi_{R,N}^g\|_{L^1_2(\bS_R)} = 0,
\end{equation}
which is due to the strong continuity of the semigroup $t\mapsto \re^{-tQ_{R,N}}$.
It remains to prove that $\Gamma(\varPsi)\in C([0,T] ; L^1_2(\bS_R))$ if $\varPsi \in C([0,T]; L^1_2(\bS_R))$. To this end, we need to verify that the limit of
\begin{multline}\label{eq:Psi_continuity_local_1}
\|\Gamma(\varPsi)(t+\Delta t,\Cdot)-\Gamma(\varPsi)(t,\Cdot)\|_{L^1_2(\bS_R)}
=\| \int_0^{t}\,[\re^{-(t+\Delta t-s)\,Q_{R,N}} -\re^{-(t-s)\,Q_{R,N}}\,]\,F(\varPsi(s,\Cdot))\,\rd s
\\
+ \int_t^{t+\Delta t}\,\re^{-(t+\Delta t-s)\,Q_{R,N}}\, \,F(\varPsi(s,\Cdot))\,\rd s\,\|_{L^1_2(\bS_R)}
\end{multline}
tends to zero as $\Delta t\to 0$. Consider the first integral on the RHS of Eq.~\eqref{eq:Psi_continuity_local_1}. Observe that $\|Q_{R,N}^{1/6}\,\re^{-(t-s)Q_{R,N}}\,(\re^{-\Delta t Q_{R,N}} - 1)\,  F( \varPsi(s,\Cdot))\|_{L_2(\bS_R)}$ is bounded by $(t-s)^{-1/6}\,\|F(\Psi(s,\Cdot))\|_{L_2(\bS_R)}$, which in integrable for any fixed for $s\neq t$. (The expression  $\|F( \Psi(s,\Cdot))\|_{L_2(\bS_R)}$ is estimated as in the discussion  \eqref{eq:non_linearity_contra_1}--\eqref{eq:R^2_local} above).
 Moreover, for any fixed $s\neq t $ the following limit holds
\begin{equation}
\lim_{\Delta t\to 0}\,\|Q_{R,N}^{1/6}\,\re^{-(t-s)Q_{R,N}}\,(\re^{-\Delta t Q_{R,N}} - 1)\,F(\varPsi(s,\Cdot))\|_{L_2(\bS_R)}=0.
\end{equation}
Thus, by dominated convergence theorem, the first integral on the RHS of Eq.~\eqref{eq:Psi_continuity_local_1} converges to zero as $\Delta t\to 0$.
Now consider the second integral on the RHS of Eq.~\eqref{eq:Psi_continuity_local_1}. One has
\begin{multline}
\|\int_t^{t+\Delta t}\re^{-(t+\Delta t-s)Q_{R,N}}\,F(\varPsi(s,\Cdot))\,\rd s\|_{L^{1}_2(\bS_R)}
\\
\leq \|\int_0^{\infty}\,\chi_{[0,\Delta t]}(s)\,Q^{1/6}_{R,N}\,\re^{-sQ_{R,N}}\,F(\varPsi(t+\Delta t -s,\Cdot))\, \rd s\|_{L_2(\bS_R)}\,.
\end{multline}
Note that 
\begin{multline}
\|\chi_{[0,\Delta t]}(s)\,Q_{R,N}^{1/6}\,\re^{-s\,Q_{R,N}}\,F( \Psi(t+\Delta t -s,\Cdot))\|_{L_2(\bS_R)}
\\
\leq\,\chi_{[0,1]}(s)\,s^{-1/6}\,\|F(\Psi( t+\Delta t -s,\Cdot))\|_{L_2(\bS_R)}, 
\end{multline}
where the expression $\|F(\Psi(t+\Delta t -s,\Cdot))\|_{L_2(\bS_R)}$ can be bounded as in the discussion \eqref{eq:non_linearity_contra_1}--\eqref{eq:R^2_local} above by a constant independent of $\Delta t$.
Furthermore, one has
\begin{equation}
\lim_{\Delta t\to 0}\,\|\chi_{[0,\Delta t]}(s)\,Q_{R,N}^{1/6}\,\re^{-s\,Q_{R,N}}\,F( \Psi(t+\Delta t -s,\Cdot))\|_{L_2(\bS_R)}= 0\,.
\end{equation}
Having verified the requirements of the dominated convergence theorem, it follows that
the second integral on the RHS of Eq.~\eqref{eq:Psi_continuity_local_1} converges to zero as $\Delta t\to 0$. This concludes that $\Gamma(\varPsi)\in C((0,\infty);L^1_2(\bS_R))$ if $\Psi\in C((0,\infty); L^1_2(\bS_R))$. Combining the above results with Eq.~\eqref{eq:Psi_continuity_local_0}, one shows that 
$\cK( \varPsi) \in C([0,T] ;L^1_2(\bS_R))$ if $\varPsi \in C((0,\infty) ; L^1_2(\bS_R))$.
\\
\\
To verify the bound~\eqref{eq:lip_bound_non_lineararity}, let $\varPsi_1,\varPsi_2\in C([0,T] ; L^1_2(\bS_R))$ and $t\in [0,T]$. Proceeding as in \eqref{eq:non-linearity_second_1} above, one obtains
\begin{multline}
\|\Gamma(\varPsi_1)(t,\Cdot) - \Gamma(\varPsi_2)(t,\Cdot)\|_{L^1_2(\bS_R)}
\\ \leq C\,\int_0^t\,\frac{1}{|t-s|^{1/6}}\|F(\varPsi_1)(s,\Cdot) - F(\varPsi_2)(s,\Cdot)\|_{L_2(\bS_R)}\,\rd s
\\
\leq C\,\sup_{s\in[0,T]}\,\|F(\varPsi_1)(s,\Cdot) -F(\varPsi_2)(s,\Cdot)\|_{L_2(\bS_R)} \int_0^t\,\frac{\rd s}{|t-s|^{1/6}}\,.
\end{multline}
Performing the integral over $s\in[0,t]$ and taking the supremum over $0\leq t\leq T$ of the resulting expression culminates in
\begin{equation}\label{eq:lip_semi_group}
\|\Gamma(\varPsi_1) - \Gamma(\varPsi_2)\|_{C([0,T] ; L^1_2(\bS_R))}
\leq C\,T^{5/6}\,\|F(\varPsi_1) -F(\varPsi_2)\|_{C([0,T] ; L_2(\bS_R))}\,,
\end{equation}
where
\begin{multline}\label{eq:non_linearity_1_psi}
\|F(\varPsi_1(s,\Cdot)) -F(\varPsi_2(s,\Cdot))\|_{L_2(\bS_R)}
\\
\leq \| P^{\prime}((\varPsi_1 +Z_{R,N})(s,\Cdot),c_{R,N})-P^{\prime}((\varPsi_2+Z_{R,N})(s,\Cdot),c_{R,N})\|_{L_2(\bS_R)}
\\
+ \|((\varPsi_1 +Z_{R,N})(s,g))^{n-1}g(\Cdot) - ((\varPsi_2 +Z_{R,N})(s,g))^{n-1}g(\Cdot) \|_{L_2(\bS_R)}
\\
:=R^{(11)}(s)+R^{(22)}(s)\,.
\end{multline}
Using Remark~\ref{rem:mean_young_bernoulli}, combined with the H\"older inequality and with Remark~\ref{Remark:bound-on-Wick-powers}, there exist $\check C,C\in(0,\infty)$ such that 
\begin{equation}\label{eq:non_linearity_2_psi}
\begin{aligned}
R^{(11)}(s)
& \leq \sum_{l=0}^{n-1}\sum_{m=l}^{n-1}\,|a_{m,l}|\,\|Z^{:m-l:}_{R,N}(s,\Cdot)\|_{L_{2p}(\bS_R)}\, \|(\varPsi^{l}_1 - \varPsi^{l}_2)(s,\Cdot)\|_{L_{2q}(\bS_R)}
\\
&\leq  C\,\|(\varPsi_1 - \varPsi_2)(s,\Cdot)\|_{L_{4q}(\bS_R)}\,{\color{black}{ \big[2+ \sum_{l=1}^{n-2}\,\left(\|\varPsi_1(s,\Cdot)\|^{l}_{L_{4ql}(\bS_R)} + \|\varPsi_2(s,\Cdot)\|^{l}_{L_{4ql}(\bS_R)}\right)\big]\,}}
\\
&\leq \check C\,\|(\varPsi_1 - \varPsi_2)(s,\Cdot)\|_{L^1_{2}(\bS_R)}
 \big[2+ \sum_{l=1}^{n-2}\,\left(\|\varPsi_1(s,\Cdot)\|^{l}_{L_2^1(\bS_R)} + \|\varPsi_2(s,\Cdot)\|^{l}_{L_2^1(\bS_R)}\right)\big],
\end{aligned}
\end{equation}
where $p,q\geq 1$ with $1/p+1/q=1$ and in the last step we made use of Lemma~\ref{lem:embedding_sphere}. Utilizing H\"older's inequality, Remark~\ref{Remark:bound-on-Wick-powers}, Remark~\ref{rem:mean_young_bernoulli}, and the triangle inequality, we obtain for some $C,\hat C\in (0,\infty)$ that
\begin{equation}\label{eq:non_linearity_3_psi}
\begin{aligned}
R^{(22)}(s)&
\leq\,  \|g\|_{L_{2}(\bS_R)} \,\sum_{k=0}^{n-1}\,{n-1\choose k}\,|Z_{R,N}(s,g)|^{n-1-k}\,\,|\varPsi_1(s,g)^k-\varPsi_2(s,g)^k|
\\&
\leq  C | \varPsi_1(s,g)-\varPsi_2(s,g)| \,\big[2+\sum_{k=0}^{n-1} \big(|\varPsi_1(s,g)|^k+|\varPsi_2(s,g)|^k \big)\big]
\\&
\leq C\,\|(\varPsi_1-\varPsi_2)(s,\Cdot)\|_{L_{2}(\bS_R)}\,\big[2+\sum_{k=1}^{n-2}\,\big(\|\varPsi_1(s,\Cdot)\|^k_{L_{2}(\bS_R)}+\|\varPsi_2(s,\Cdot)\|^k_{L_{2}(\bS_R)}\big)\big]
\\&
\leq \hat C\,\|(\varPsi_1-\varPsi_{2})(s,\Cdot)\|_{L^1_2(\bS_R)}\,\big[2+\sum_{k=1}^{n-2}\,\big(\|\varPsi_1(s,\Cdot)\|^k_{L^1_2(\bS_R)}+\|\varPsi_2(s,\Cdot)\|^k_{L^1_2(\bS_R)}\big)\big].
\end{aligned}
\end{equation}
Combining Eqs.~\eqref{eq:non_linearity_2_psi} and \eqref{eq:non_linearity_3_psi} and letting $C=2\,(\check C\vee\hat C)$ culminates in the following inequality
\begin{multline}\label{eq:non_linearity_4}
\|F(\varPsi_1(s,\Cdot)) -F(\varPsi_2(s,\Cdot))\|_{L_2(\bS_R)}
\leq C\,\|(\varPsi_1-\varPsi_{2})(s,\Cdot)\|_{L^1_2(\bS_R)}
\\
\times \Big[2+ \sum_{l=1}^{n-2}\,\left(\|\varPsi_1(s,\Cdot)\|^{l}_{L_2^1(\bS_R)} + \|\varPsi_2(s,\Cdot)\|^{l}_{L_2^1(\bS_R)}\right)\Big].
\end{multline}
The above bound implies that $F \,:\,L^1_2(\bS_R)\to L_2(\bS_R)$ is locally Lipschitz continuous. One verifies Eq.~\eqref{eq:lip_bound_non_lineararity} by taking the supremum over $s\in[0,T]$ of \eqref{eq:non_linearity_4} and combining the resulting expression with Eq.~\eqref{eq:lip_semi_group}. This finishes the proof.
\end{proof}

\begin{thm}\label{thm:mild_slotion_non-linear}
There exists $T>0$ depending on $\|\psi^g_{R,N}\|_{L^1_2(\bS_R)}$ such that Eq.~\eqref{eq:varPsi_pde} has a unique mild solution $\varPsi^g_{R,N}\in C([0,T] ;L^1_2(\bS_R))$.
\end{thm}

\begin{proof}
The proof is an application of the standard contraction mapping argument. Let $\|\psi^g_{R,N}\|_{L^1_2(\bS_R)} = L$. Thanks to the positivity of the spectrum of the elliptic operator $Q_{R,N}$, one obtains
$$\|\re^{-tQ_{R,N}}\,\psi^g_{R,N}\|_{L^1_2(\bS_R)} \leq L\quad \text{ for\, all} \quad 0\leq t\leq T.$$
We choose
\begin{equation}
H:= \{ \varPsi^g_{R,N} \in C([0,T] ; L^1_2(\bS_R))~:~\|\varPsi^g_{R,N} \|_{C([0,T] ; L^1_2(\bS_R))}\leq 2L \}\,.
\end{equation}
Using Lemma~\ref{lem:lipschitz} one infers that $\cK\,:\, H\to H$ if we choose $T> 0$ such that
\begin{equation}\label{eq:time_existence}
C\,T^{5/6}\,(2+2\sum_{l=1}^{n-2}\,L^l)   =: \gamma\,,
\end{equation}
where $0<\gamma<1$. Hence, for every $\varPsi_1, \varPsi_2\in H$ it holds that
\begin{equation}
\|\cK(\varPsi_1) - \cK(\varPsi_2)\|_{C([0,T] ; L^1_2(\bS_R))} \leq \gamma\, \|\varPsi_1 - \varPsi_2\|_{C([0,T] ; L^1_2(\bS_R))}\,.
\end{equation}
Now, the existence and uniqueness of the mild solution $\varPsi^g_{R,N}\in H\subset C([0,T]  ; L^1_2(\bS_R))$ follows from the contraction mapping theorem. Observe that Eq.~\eqref{eq:time_existence} implies that the existence time depends on the $L^1_2(\bS_R)$-norm of the initial data $\psi^g_{R,N}$\,. This concludes the proof.
\end{proof}

\subsection{Global existence in time}\label{subsec:global_solution_remainder}
In this section, we prove the existence of a global-in-time mild solution to Eq.~\eqref{eq:varPsi_pde} in the space $L_2^1(\bS_R)$. In general, the $L_2^1(\bS_R)$-norm of the solution might explode in finite time. However, if one has an a priori $L_2^1(\bS_R)$ estimate for $\varPsi^g_{R,N}(t,\Cdot)$, that is global-in-time, then global existence follows from the local existence~\cite[p.~8]{Gub21}. To obtain such a priori estimate, we apply the energy method  in  $L^1_2(\bS_R)$ to Eq.~\eqref{eq:varPsi_pde2}. The proof relies mainly on the fact that for all $m \in \{1,2,\cdots, n-1\}$, $p\in[2,\infty)$, and all $R,N\in\bN_+$, the renormalized free fields $Z_{R,N}^{:m:}(t,\Cdot)$, and their derivatives $\vec\nabla_R\,Z_{R,N}^{:m:}(t,\Cdot)$ are uniformly bounded in $L_p(\bS_R)$ for $t\in [0,T]$, see Remark~\ref{Remark:bound-on-Wick-powers}. Moreover, we also make use of the negative sign on the RHS of Eq.~\eqref{eq: SPDE on Sphere}, which is included in the definition of the coefficients $a_{m,l}$ introduced in the beginning of Sec.~\ref{subsec:local_solution_remainder}. 
\\
\\
One rewrites Eq.~\eqref{eq:varPsi_pde} as follows
\begin{multline}\label{eq:varPsi_pde2}
 (\partial_t + Q_{R,N})\,\varPsi^g_{R,N} 
 + (\varPsi^g_{R,N})^{n-1}
 \\
 =
 \sum_{l=0}^{n-2}\sum_{m=l}^{n-1} a_{m,l}\, Z_{R,N}^{:m-l:} 
 \,(\varPsi^g_{R,N})^l
 +
 ((\varPsi^g_{R,N}+Z_{R,N})(\Cdot,g))^{n-1}g\,.
\end{multline}

\begin{prop}\label{prop:global_in_time_existence}
For all $R,N\in\bN_+$, there exist $p\in[2,\infty)$ and $C\in(0,\infty)$ such that for all $t\in (0,\infty)$ and all $g\in C^\infty(\bS_R)$, it holds that
\begin{equation}
\begin{aligned}
\partial_t\|\varPsi^g_{R,N}(t,\Cdot)\|^2_{L^1_2(\bS_R)} 
&
+\|\varPsi^g_{R,N}(t,\Cdot)\|^2_{L^4_2(\bS_R)} 
+\|\varPsi^g_{R,N}(t,\Cdot)\|^n_{L_n(\bS_R)}
\\
&
+\int_{\bS_R}\,(\vec\nabla_R\varPsi^g_{R,N})(t,\Cdot)^2\,(\varPsi^g_{R,N})(t,\Cdot)^{n-2}\,\rho_R(\rd\rx)
\\
&\leq C\,\sum_{k=0}^{n-1}\,\left(\|Z^{:k:}_{R,N}(t,\Cdot)\|^p_{L_p(\bS_R)} +\|\vec\nabla_R Z_{R,N}^{:k:}(t,\Cdot)\|^p_{L_p(\bS_R)}\right)\,.
\end{aligned}
\end{equation}
\end{prop}

\begin{Remark}
On account of Remark~\ref{rem:regularity_free_fields}, it follows that for all $m \in\bN_+$ and all $R,N\in\bN_+$, the renormalized free fields $Z_{R,N}^{:m:}(t,\Cdot)$ are of regularity $2-\kappa$ for all $\kappa>0$. Moreover, from the smoothing property of the heat kernel associated to the operator $Q_{R,N}$, we expect the solution to Eq.~\eqref{eq:varPsi_pde2} to be of regularity $8-\kappa$ for all $\kappa>0$. In particular, $\|\varPsi^g_{R,N}(t,\Cdot)\|^2_{L^4_2(\bS_R)}$ is well defined. The implies that all product terms present in Eq.~\eqref{eq:varPsi_pde2} can be bounded in some Sobolev spaces with positive regularity cf.~\cite[Secs.~2.7--8]{BCD11}. Accordingly, one can apply the energy method to Eq.~\eqref{eq:varPsi_pde2}. 
\end{Remark}

\begin{proof}[Proof of Prop.~\ref{prop:global_in_time_existence}]
Let 
\begin{equation}
\begin{gathered}
R_1 :=\sum_{l=0}^{n-2}\sum_{m=l}^{n-1}a_{m,l}\,Z_{R,N}^{:m-l:}(\varPsi_{R,N}^g)^l,
\quad
R_2:= ((\varPsi_{R,N}^g+Z_{R,N})(\Cdot,g))^{n-1}g(\Cdot)\,.
\end{gathered}
\end{equation}
To lighten the notation, let us denote $\varPsi^g_{R,N}(t,\Cdot) := \varPsi^g(\Cdot)$ and $Z^{:m:}_{R,N}(t,\Cdot):= Z^{:m:}(\Cdot)$. Multiplying both sides of Eq.~\eqref{eq:varPsi_pde2} by $(1-\Delta_R)\varPsi^g$ and integrating over $\bS_R$ using the measure $\rho_R(\rd\rx)$, we obtain the following terms
\begin{itemize}

\item[(A)] $\langle (1-\Delta_R)\varPsi^g, \partial_{t}\varPsi^g\rangle_{L_2(\bS_R)}= \frac{1}{2} \partial_t\,\|\varPsi^g\|^2_{L^1_2(\bS_R)}$\,.

\item[(B)] $\langle (1-\Delta_R)\varPsi^g, Q_{R,N}\,\varPsi^g \rangle_{L_2(\bS_R)} \geq \frac{1}{N^4}\|\varPsi^g \|^2_{L^4_2(\bS_R)}$\,.

\item[(C)] $\langle (1-\Delta_R)\varPsi^g,  (\varPsi^g)^{n-1}\rangle_{L_2(\bS_R)} \leq \|\varPsi^g\|^n_{L_n(\bS_R)}
+(n-1)\,\int_{\bS_R} |\nabla_R\varPsi^g(x)|^2  |\varPsi^g(x) |^{n-2}\rho_R(\rd\rx)$\,.

\item[(D)] $\langle (1-\Delta_R)\varPsi^g, R_1 \rangle_{L_2(\bS_R)} = \sum_{l=0}^{n-2}\sum_{m=l}^{n-1}a_{m,l}\,\langle (1-\Delta_R)\varPsi^g,\,Z^{:m-l:}(\varPsi^g)^l \rangle_{L_2(\bS_R)}$\,.

\item[(E)] $\langle (1-\Delta_R)\varPsi^g,R_2\rangle_{L_2(\bS_R)} =\langle (1-\Delta_R)\varPsi^g,((\varPsi^g+Z)(\Cdot,g))^{n-1}g(\Cdot)\rangle_{L_2(\bS_R)}$\,.
\end{itemize}
The general strategy of the proof is to estimate Items~(D) and~(E) in terms of Items~(B) and~(C) in order to find an upper bound in terms of the $L_p$-norms of $Z^{:m:}$ and $\vec\nabla_R\,Z^{:m:}$. To this end, we shall repeatedly use the fact that $(1-\Delta_R)$ is a positive, self-adjoint operator, see Remark~\ref{rem:lap_bel_identities}. Consider Item~(D). One has
\begin{multline}\label{eq:item_D}
\Big|\sum_{l=0}^{n-2}\sum_{m=l}^{n-1}a_{m,l}\,\langle (1-\Delta_R)\varPsi^g,Z^{:m-l:}(\varPsi^g)^l\rangle_{L_2(\bS_R)}\Big|
\leq \Big |\sum_{l=0}^{n-2}\sum_{m=l}^{n-1}a_{m,l}\,\langle \varPsi^g,Z^{:m-l:}(\varPsi^g)^l \rangle_{L_2(\bS_R)} \Big|
\\
+\Big | \sum_{l=0}^{n-2}\sum_{m=l}^{n-1}a_{m,l}\,\langle\,(\vec\nabla_R\varPsi^g),(\vec\nabla_R Z^{:m-l:})(\varPsi^g)^l \rangle_{L_2(\bS_R)}\,\Big|
\\
+ \Big |\sum_{l=0}^{n-2}\sum_{m=l}^{n-1}a_{m,l}\, \langle (\vec\nabla_R\varPsi^g),(Z^{:m-l:})(\vec \nabla_R(\varPsi^g)^l)\rangle_{L_2(\bS_R)}\,\Big|\,.
\end{multline}
The upper bounds follow from the fact that the $L_p$-norms of $Z^{:m-l:}$ and $\vec \nabla_R\, Z^{:m-l:}$ are uniformly bounded on $t\in [0,T]$ with probability one for all $m \in \{1,2,\cdots, n-1\}$, $p\in[2,\infty)$ and all $R,N\in\bN_+$, see Remark~\ref{Remark:bound-on-Wick-powers}. As a result, we apply H\"older's and Young's inequalities with $1/\tilde{q} +1/p = 1$, where $p\in[2,\infty)$ is chosen large and $\tilde{q}\in(1,\infty)$ is close to one. 
\\[0.3cm]
Consider the first sum in the above expression. Setting $(l+1)/q +1/p= 1$ with $q> n-1$, one infers that for all sufficiently small $\delta_{d_1}>0$, there exists $C_{d_1}\in (0,\infty)$ such that
\begin{multline}\label{eq:item_D_1}
\Big|\sum_{l=0}^{n-2}\sum_{m=l}^{n-1}\,a_{m,l}\,\langle \varPsi^g,Z^{:m-l:}(\varPsi^g)^{l} \rangle_{L_2(\bS_R)}\Big|
\leq \sum_{l=0}^{n-2}\sum_{m=l}^{n-1}\,|a_{m,l}|\,\|\varPsi^g\|^{l+1}_{L_{q}(\bS_R)} \,\|Z^{:m-l:}\|_{L_p(\bS_R)}
\\
\leq \sum_{l=0}^{n-2}\sum_{m=l}^{n-1}\,|a_{m,l}|\,\big(\delta'_{d_1}\|\varPsi^g\|^{q}_{L_{q}(\bS_R)}+ \,C'_{d_1}\|Z^{:m-l:}\|^p_{L_p(\bS_R)} \big)\\
\leq\,\delta_{d_1}\|\varPsi^g\|^n_{L_n(\bS_R)}+ C_{d_1}\, \sum_{k=0}^{n-1}\,\|Z^{:k:}\|^p_{L_p(\bS_R)},
 \end{multline}
where we set $q=n$ in the last line. Consider the second sum in Eq.~\eqref{eq:item_D}. Let $1/q +1/p = 1$, where $p\in[2,\infty)$ is chosen large and $q\in(1,\infty)$ is close to one. Applying H\"older's and Young's inequalities, for every sufficiently small $\delta_{d_2}>0$, there exists $C_{d_2}\in (0,\infty)$ such that
\begin{multline}\label{eq:item_D_2}
\Big | \sum_{l=0}^{n-2}\sum_{m=l}^{n-1}\,|a_{m,l}|\,\langle\,(\vec\nabla_R\varPsi^g),(\vec\nabla_R Z^{:m-l:})(\varPsi^g)^l \rangle_{L_2(\bS_R)}\,\Big|
\\
\leq \sum_{l=0}^{n-2}\sum_{m=l}^{n-1}\,|a_{m,l}|\,\Big(\int_{\bS_R}\,|\vec\nabla_R\varPsi^g(\rx)|^q\,|\varPsi^g(\rx)|^{lq} \rho_R(\rd\rx)\Big)^{1/q}\Big|\,\|(\vec\nabla_R Z^{:m-l:})\|_{L_p(\bS_R)}
\\
\leq
 \delta_{d_2}\,\sum_{l=0}^{n-2}\,\int_{\bS_R}\,|\vec\nabla_R\varPsi^g(\rx)|^q\,|\varPsi^g(\rx)|^{lq} \rho_R(\rd\rx)
+C_{d_2}\,\sum_{k=1}^{n-1}\|(\vec\nabla_R Z^{:k:})\|^p_{L_p(\bS_R)}\,.
\end{multline}
To proceed, let us fix a sufficiently large $p\in[2,\infty)$. Then $q\in(1,\infty)$ is determined via $1/p+1/q=1$, and is therefore close to one. We distinguish the following two cases:
\begin{itemize}
\item[I)] $0\leq l < n/2$. Let $1/2 +l/n =1/q$. Using H\"older's and Young's inequalities with $q/2 +ql/n =1$, one obtains
\begin{multline}\label{eq:item_D_22}
\int_{\bS_R}\,|\vec\nabla_R\varPsi^g(\rx)|^q\,|\varPsi^g(\rx)|^{lq}\,\rho_R(\rd\rx)
\leq \|\vec\nabla_R\varPsi^g\|^{q}_{L_2(\bS_R)} \|\varPsi^g\|^{lq}_{L_n(\bS_R)}
\\
\leq \hat \delta\,\|\vec\nabla_R\varPsi^g\|^{2}_{L_2(\bS_R)} +\check \delta\,\|\varPsi^g\|^{n}_{L_n(\bS_R)}
\leq \tilde \delta\,\|\varPsi^g\|^2_{L_2^1(\bS_R)}+\check \delta\,\|\varPsi^g\|^{n}_{L_n(\bS_R)}\,.
 \end{multline}
The last bound follows from Remark~\ref{rem:bessel_sobolev_psi}.

\item[II)] $n/2\leq l< n-2$. One has $1\leq (2l-(n-2))\,q/(2-q) < n$ for $q\in (1,\infty)$ sufficiently close to one. For some $\delta',\delta''>0$, one obtains
\begin{multline}\label{eq:item_D_23}
\int_{\bS_R}\,|\vec\nabla_R\varPsi^g(\rx)|^q\,|\varPsi^g(\rx)|^{lq} \rho_R(\rd\rx)
\leq
 \Big(\int_{\bS_R}\,|\vec \nabla_R\varPsi^g(\rx)^2|\,|\varPsi^g(\rx)|^{(n-2)}\,\rho_R(\rd\rx)\Big)^{q/2}\,
 \\
\times \Big(\int_{\bS_R}\,|\varPsi^g(\rx)|^{(2l-(n-2)) q/(2-q)}\,\rho_R(\rd\rx)\Big)^{1-q/2}
 \\
 \leq \delta'\,\Big(\int_{\bS_R}\,|\vec\nabla_R\varPsi^g(\rx)|^2\,|\varPsi^g(\rx)|^{(n-2)}\,\rho_R(\rd\rx)\Big)
+ \delta''\,\|\varPsi^g\|^n_{L_n(\bS_R)} + \delta''.
\end{multline}
\end{itemize}
The above follows from H\"older's and Young's inequalities with $(1-\frac{q}{2})+\frac{q}{2}=1$ combined with the fact that for all $1\leq\gamma< n$, there are $C,\hat C\in (0,\infty)$ such that
\begin{multline}\label{eq:increasing-power}
\|\varPsi^g\|^{\gamma}_{L_{\gamma}(\bS_R)}
\leq C\,\|\varPsi^g\|^{\gamma}_{L_{n}(\bS_R)}
\\
\leq C (1+\|\varPsi^g\|_{L_{n}(\bS_R)})^{\gamma}\leq C (1+\|\varPsi^g\|_{L_{n}(\bS_R)})^{n}\leq \hat C(1+\|\varPsi^g\|^{n}_{L_{n}(\bS_R)}).
\end{multline}
Consider the third sum in Eq.~\eqref{eq:item_D}. Let $1/q +1/p = 1$, where $p\in[2,\infty)$ is chosen large and $q\in (1,\infty)$ sufficiently close to one. Note that $\vec\nabla_R(\varPsi)^l=l\,\varPsi^{l-1}(\vec\nabla_R\varPsi)$ and recall that
$$(\sum_{l=0}^n x^l)'= \sum_{l=1}^{n} l\,x^{l-l}= \sum_{l=0}^{n-1} (l+1)\,x^{l}.$$
For every small enough $\delta_{d_3}>0$, there exists some $C_{d_3}\in (0,\infty)$ such that one has
\begin{multline}\label{eq:item_D_3}
 \Big |\, \sum_{l=0}^{n-2}\sum_{m=l}^{n-1}a_{m,l}\, \langle (\vec\nabla_R\varPsi^g),(Z^{:m-l:})(\vec \nabla_R(\varPsi^g)^l\rangle_{L_2(\bS_R)}\,\Big|
\\
\leq
\delta_{d_3}\sum_{l=1}^{n-2}\,\int_{\bS_R} |\vec\nabla_R\varPsi^g(\rx)|^{2q}\,|\varPsi^g(\rx)|^{q(l-1)} 
\,\rho_R(\rd\rx)+C_{d_3}\,\sum_{k=0}^{n-1}\|Z^{:k:}\|^p_{L_p(\bS_R)}\,.
\end{multline}
The last estimate follows from H\"older's and Young's inequalities. Let
$$q= \frac{q(l-1)}{n-2} +r,\quad \text{where} \quad r =\frac{q(n-2-(l-1))}{n-2},$$
then for all $n\geq 4$ and for all $q\in(1,\infty)$ sufficiently close to one, it holds that $r\geq 1$. Thus, for all $\hat \delta>0$ there exists $\check C_{d_3}\in (0,\infty)$ such that
\begin{multline}\label{eq:item_D_31}
\int_{\bS_R} \,|\vec\nabla_R\varPsi^g(\rx)|^{2q}
\,|\varPsi^g(\rx)|^{q(l-1)} \,\rho_R(\rd\rx)
\leq
\\
\big(\int_{\bS_R} |\vec\nabla_R\varPsi^g(\rx)|^2\,|\varPsi^g(\rx)|^{n-2}\,\rho_R(\rd\rx) \big)^{q(l-1)/(n-2)}
\\
 \times
\big(\int_{\bS_R} 
|\vec\nabla_R\varPsi^g(\rx)|^{2q(n-2-(l-1))/((n-2)-q(l-1))}\rho_R(\rd\rx)\big)^{1-q(l-1)/(n-2)}
\\ \leq \hat \delta\,\int_{\bS_R} |\vec\nabla_R\varPsi^g(\rx)|^2\,|\varPsi^g(\rx)|^{n-2}\,\rho_R(\rd\rx)
+\check C_{d_3}\int_{\bS_R} |\vec\nabla_R\varPsi^g(\rx)|^{2\tilde{r}}\,\rho_R(\rd\rx)\,.
\end{multline}
where 
$$\tilde{r}:=\frac{q(n-2-(l-1))}{(n-2)-q(l-1)}.$$
Above, we used H\"older's and Young's inequalities for $1/p'+1/q'=1$ with
$$1/p'=\frac{q(l-1)}{n-2} \quad \text{and} \quad 1/q'=1-\frac{q(l-1)}{n-2}.$$
To proceed, one uses Remark~\ref{rem:bessel_sobolev_psi} to write
$$\int_{\bS_R}\,|\vec\nabla_R\varPsi^g)|^{2\tilde{r}}\,\rho_R(\rd\rx)
\leq C_{\tilde{r}}\,\|\varPsi^g\|^{2\tilde{r} }_{L^2_2(\bS_R)}$$ 
for some $C_{\tilde{r}}\in (0,\infty)$. On the account of Remark~\ref{rem:parseval_them}, one gets
$$\|\varPsi^g\|^2_{L_2(\bS_R)}
=\sum_{l=0}^{\infty}\,\|\varPsi^g_l\|^2_{L_2(\bS_R)}=\sum_{l=0}^{\infty}\sum_{m=-l}^l\,|\varPsi^g_{lm}|^2.$$
Using the notation introduced in Remark~\ref{rem:l_brackets}, one obtains
\begin{multline}\label{eq:item_D_32}
\begin{aligned}
\|\varPsi^g\|^{2}_{L^2_2(\bS_R)}
&=
\sum_{l=0}^{\infty}\,\ll l\gg\,\|\varPsi^g_l\|^2_{L_2(\bS_R)} =\sum_{l=0}^{\infty}\sum_{m=-l}^l\,\ll l\gg^2\,|\varPsi^g_{l,m}|^2
\\
&\leq \left(\sum_{l=0}^{\infty}\,\sum_{m=-l}^l\,\ll l\gg^3\,|\varPsi_{l,m}^g|^2\right)^{2/3}\,
\left(\sum_{l=0}^{\infty}\,\sum_{m=-l}^l\, |\varPsi^g_{l,m}|^2\right)^{1/3}
\\
&\leq \|\varPsi^g\|_{L^3_2(\bS_R)}^{4/3}\,\|\varPsi^g\|_{L_2(\bS_R)}^{2/3}\leq
 \tilde C_{d_3} \|\varPsi^g\|_{L^3_2(\bS_R)}^{4/3}\,\|\varPsi^g\|_{L_n(\bS_R)}^{2/3},
\end{aligned}
\end{multline}
where we used the H\"older inequality for sequences. Now using Young's inequality and the argument in Eq.~\eqref{eq:increasing-power} above, one obtains
\begin{multline}\label{eq:item_D_33}
\int_{\bS_R}\,|\vec\nabla_R\varPsi^g)|^{2\tilde{r}}\,\rho_R(\rd\rx)
\leq\,
C_{\tilde{r}}\,\|\varPsi^g\|^{2\tilde{r}}_{L^2_2(\bS_R)}
\leq  C_{\tilde{r}}\, \tilde C_{d_3}^{\tilde{r}} \|\varPsi^g\|_{L^3_2(\bS_R)}^{(4/3) \tilde{r} }\,\|\varPsi^g\|^{(2/3)\tilde {r}} _{L_n(\bS_R)}
\\ 
\leq \delta'\|\varPsi\|^2_{L^3_2(\bS_R)} +\,\overline{\delta}\|\varPsi\|^n_{L_n(\bS_R)}+\overline{\delta}.
\end{multline}
Consider Item~(E). One gets
\begin{multline}
\langle (1-\Delta_R)\varPsi^g,R_2\rangle_{L_2(\bS_R)}
\\= \langle\varPsi^g,((\varPsi^g+Z)(g))^{n-1}g\rangle_{L_2(\bS_R)}
 + \langle (-\Delta_R)\varPsi^g,((\varPsi^g+Z)(g))^{n-1}g\rangle_{L_2(\bS_R)}. 
\end{multline}
From H\"older's and Young's inequalities, combined with the relation $\|g\|^n_{L_{n/(n-1)}}<1/2$ given in Eq.~\eqref{eq:bound_on_g} and with the fact that
$$\forall\,\, k> 1\quad \text {and} \quad f(a_i)\in\bR_+,\quad\text{ it\,\, holds\,\, that}\quad (\sum_{i=1}^n\,f(a_i))^k\leq n^{k-1}\,\sum_{i=1}^n\,f(a_i)^k,$$
one obtains for some $\delta_{e_1},\delta_{e_2}>0$
\begin{multline}\label{eq:item_E_1}
 \langle\varPsi^g,((\varPsi^g+Z)(g))^{n-1}g\rangle_{L_2(\bS_R)}
\leq
| \langle\varPsi^g,g\rangle_{L_2(\bS_R)}\,((\varPsi^g+Z)(g))^{n-1}|
\\
\leq \,\|\varPsi^g\|_{L_n(\bS_R)}\,\|g\|_{L_{n/(n-1)}(\bS_R)}\,\,\big[2^{n-2}\,\|g\|^{n-1}_{L_{n/(n-1)}(\bS_R)}\big]\,\big[\|\varPsi^g\|^{n-1}_{L_n(\bS_R)} +\|Z\|^{n-1}_{L_{n}(\bS_R)}\big]
\\
\leq \delta_{e_1} \|\varPsi^g\|^{n}_{L_n(\bS_R)} +\delta_{e_2}\|Z\|^n_{L_{n}(\bS_R)}.
\end{multline}
Furthermore, from H\"older's and Young's inequalities, combined with the assumption $\|\Delta_R\, g\|^n_{L_{n/(n-1)}}<1/2$ given in Eq.~\eqref{eq:bound_on_g} and with the fact that $(-\Delta_R)$ is self-adjoint one has
\begin{multline}\label{eq:item_E_2_2}
|\langle -\Delta_R\varPsi^g,((\varPsi^g+Z)(g))^{n-1}g\rangle_{L_2(\bS_R)}|
=
|\langle \varPsi^g,-\Delta_R\,g\rangle_{L_2(\bS_R)}|\, \, \left|((\varPsi^g+Z)(g))^{n-1}\right|
\\
\leq \|\varPsi^g\|_{L_n(\bS_R)}\,\|\Delta_R\,g\|_{L_{n/(n-1)}(\bS_R)}\,\big[2^{n-2}\,\|g\|^{n-1}_{L_{n/(n-1)}(\bS_R)}\big]\,\big[\|\varPsi^g\|^{n-1}_{L_n(\bS_R)} +\|Z\|^{n-1}_{L_{n}(\bS_R)}\big]
\\
\leq 
\delta_{e_4} \|\varPsi^g\|^{n}_{L_n(\bS_R)} +\delta_{e_5}\,\|Z\|^{n}_{L_n(\bS_R)},
\end{multline}
where $\delta_{e_4},\delta_{e_5}>0$. Putting~Items~(A),~(B),~(C), and Eqs.~\eqref{eq:item_D_1},~\eqref{eq:item_D_2},~\eqref{eq:item_D_22},\\
~\eqref{eq:item_D_23},~\eqref{eq:item_D_3},~\eqref{eq:item_D_31},~\eqref{eq:item_D_32},~\eqref{eq:item_D_33},~\eqref{eq:item_E_1}~and~\eqref{eq:item_E_2_2} together, one demonstrates that for all $R,N\in\bN_+$, there exist $C\in(0,\infty)$ and $p\in[2,\infty)$ such that for all $t\in (0,\infty)$ and all $g\in C^\infty(\bS_R)$
\begin{multline}
\frac{1}{2}\,\partial_t\|\varPsi^g_{R,N}(t,\Cdot)\|^2_{L^1_2(\bS_R)}
+(1/N^4-\delta_{d_2}\tilde\delta\,C_{\sigma}-\delta_{d_3}\check C_{d_3}\delta'\,C_{\sigma})\,\|\varPsi^g_{R,N}(t,\Cdot)\|^2_{L^4_2(\bS_R)}
\\ +(1-\delta_{d_1}-\delta_{d_2}\delta'' -\delta_{d_2}\check\delta-\delta_{d_3}\check C_{d_3}\bar\delta -\delta_{e_1}-\delta_{e_4})\,\,\|\varPsi^g_{R,N}(t,\Cdot)\|^n_{L_n(\bS_R)}
\\
+(n-1-\delta_{d_2}\delta' -\delta_{d_3}\hat\delta)\,\,\int_{\bS_R}\,|\vec\nabla_R\varPsi^g_{R,N}(t,\rx)|^2\,|\varPsi^g_{R,N}(t,\rx)|^{n-2}\,\rho_R(\rd\rx)
\\
\leq \,(\delta_{d_2}\delta''+\delta_{d_3}\check C_{d_3}\overline{\delta})
+ C\,\sum_{k=0}^{n-1}\,\Big(\|Z^{:k:}_{R,N}(t,\Cdot)\|^p_{L_p(\bS_R)} +\|\vec\nabla_R Z_{R,N}^{:k:}(t,\Cdot)\|^p_{L_p(\bS_R)}\Big)\,,
\end{multline}
where we used the fact that for $\sigma\in\{1,2,3\}$, there is $C_{\sigma}\in (0,\infty)$ such that $\|\Cdot\|_{L^{\sigma}_2(\bS_R)} \leq C_{\sigma}\,\|\Cdot\|_{L^4_2(\bS_R)}$ and we set $C= (C_{d_1}\vee C_{d_2}\vee C_{d_3}\vee \delta_{e_2}\vee \delta_{e_5})$ in the last line. Observe that the constant term $(\delta_{d_2}\delta''+\delta_{d_3}\check C_{d_3}\overline{\delta})$ can be absorbed in the sum with $k=0$. This finishes the proof.
\end{proof}

\begin{thm}\label{lem:global_solution_Psi}
For all initial condition $\psi^g_{R,N} \in L^1_2(\bS_R)$ and all $g\in C^{\infty}(\bS_R)$, the global-in-time solution $\varPsi^g_{R,N}\in C([0,\infty), L^1_2(\bS_R))$ to Eq.~\eqref{eq:varPsi_pde2} exists in the sense of Def.~\ref{def:mild_solution}.
\end{thm}
\begin{proof}
Using Prop.~\ref{prop:global_in_time_existence}, it follows that for all $R,N\in \bN_+$, there exist $C\in (0,\infty)$ and $p\in[2,\infty)$ such that for all $T\in (0,\infty)$ and all $g\in C^{\infty}(\bS_R)$,  it holds that
\begin{multline}
\|\varPsi^g_{R,N}(T,\Cdot)\|^2_{L^1_2(\bS_R)} 
+ \int_0^T \|\varPsi^g_{R,N}(t,\Cdot)\|^n_{L_n(\bS_R)}\, \rd t\,
\\
+\int_0^T\int_{\bS_R}\,(\vec\nabla_R\varPsi^g_{R,N})(t,\Cdot)^2\,(\varPsi^g_{R,N})(t,\Cdot)^{n-2}\,\rho_R(\rd\rx)\,\rd t\,
\leq \,\|\psi^g_{R,N}\|^2_{L^1_2(\bS_R)}
\\ 
+ C\,\sum_{k=0}^{n-1}\, \int_0^T\Big(\|Z^{:k:}_{R,N}(t,\Cdot)\|^p_{L_p(\bS_R)} +\|\vec\nabla_R Z_{R,N}^{:k:}(t,\Cdot)\|^p_{L_p(\bS_R)}\Big)\,\rd t\,.
\end{multline}
From the preceding bound one infers that the solution to Eq.~\eqref{eq:varPsi_pde2} cannot blow up in finite time, and we have an explicit upper bound on its growth in terms of $\|\psi^g_{R,N}\|^2_{L^1_2(\bS_R)}$ and the $L_p$-norms of $Z^{:m:}_{R,N}(t,\Cdot)$, and $\vec\nabla_R Z^{:m:}_{R,N}(t,\Cdot)$. Noting that for all $m \in \{1,2,\cdots, n-1\}$, $p\in[2,\infty)$ and all $R,N\in\bN_+$, the renormalized free fields $Z_{R,N}^{:m:}(t,\Cdot)$ and their derivatives $\vec\nabla_R\,Z_{R,N}^{:m:}(t,\Cdot)$ are uniformly bounded in $L_p(\bS_R)$ for $t\in [0,T]$, see Remark~\ref{Remark:bound-on-Wick-powers}, one concludes the proof.
\end{proof}

\begin{lem}\label{lem:mild_regularity_time_space}
It holds that
\begin{equation}
 \varPsi^g_{R,N}\in C([0,\infty), L^1_2(\bS_R))\cap C((0,\infty),L^3_2(\bS_R))\cap C^1((0,\infty),L^{-3}_2(\bS_R))
\end{equation}
and the following equality
$$
\partial_t\varPsi^g_{R,N} 
= -Q_{R,N}\varPsi^g_{R,N}- P'(\varPhi^g_{R,N},c_{R,N}) +(\varPhi^g_{R,N})(t,g)^{n-1}g$$
is satisfied in $C((0,\infty),L^{-3}_2(\bS_R))$. 
\end{lem}
\begin{proof}
See Appendix~\ref{sec:mild_regularity_time_space}.
\end{proof}

\section{Invariant measures}\label{sec:invariant_measures}
The aim of this section is to study the long-time behaviour of the semigroups associated to Eqs.~\eqref{eq:weak_free} and~\eqref{eq:weak_interacting} in the sense of Def.~\ref{def:stochastic_Markov_transition}. We shall verify that for all $t\geq 0$ and all $R,N\in\bN_+$, one has
$$\mathrm{Law}(Z_{R,N}(t,\Cdot))=\mathrm{Law}(z_{R,N})=\nu_{R,N}\quad\text{ and }\quad \mathrm{Law}(\varPhi^g_{R,N}(t,\Cdot))=\mathrm{Law}(\phi^g_{R,N})=\mu^g_{R,N}$$
on $\sD'(\bS_R)$ as anticipated in Sec.~\ref{sec:non-linear}. Recall that we set $\nu_{R,N}:= \mathcal{N}(0,G_{R,N})$ and $\mu^g_{R,N}$ is given in Eq.~\eqref{eq:regularized_measure_g}.
Combining the intertwining property of the map $\jmath_R^*$ with the preceding results, it follows that for all $R,N\in\bN_+$, and all $t\geq 0$,
$$\mathrm{Law}(\jmath_R^*Z_{R,N}(t,\Cdot))=\mathrm{Law}(\jmath_R^*z_{R,N})=\jmath^*_R\sharp\nu_{R,N}$$
and
$$\quad \mathrm{Law}(\jmath_R^*\varPhi^g_{R,N}(t,\Cdot))=\mathrm{Law}(\jmath_R^*\phi_{R,N})=\jmath^*_R\sharp\mu^g_{R,N}$$
on $\sS'(\bR^2)$. All Hilbert space scalar products $\langle \Cdot, \Cdot\rangle$ in this section
are in $\tilde{H}=L^1_2(\bS_R)$.

\subsection*{Gaussian invariant measure}\label{subsec:Gaussian}

\newcommand{\lambdan}{ { f} }

\noindent
For all $R,N\in\bN_+$, consider the linear operator
$$Q_t\,:\, L^1_2(\bS_R)\to L^1_2(\bS_R)\quad \text{defined\,\, by}\quad 
Q_t \phi:=2\,\int_0^t\,\re^{-2s\,Q_{R,N}}\,\phi\,\rd s\quad
\text{for}\quad \phi\in L^1_2(\bS_R).$$
On account of Remark~\ref{rem:trace_semigroup}, the operator $Q_{\infty}$ is of trace-class, i.e., $\mathrm{Tr}(Q_{\infty})=\mathrm{Tr}(G_{R,N})<\infty$. Let $z\in L^1_2(\bS_R)$ and $F\in B_{\rb}(L^1_2(\bS_R))$
be a bounded Borel function. One denotes the semigroup~\eqref{eq:semigroup-of-equation} associated to Eq.~\eqref{eq:weak_free} by
\begin{equation}\label{eq:O_U_operator}
\begin{aligned}
\cR_{R,t}\,F(z):=&\, \bE\big[F(Z_{R,N}(t,\Cdot))\,|\,Z_{R,N}(0,\Cdot)=z\big]
\\
=&\int_{L^1_2(\bS_R)}F( \phi)\,\cN(\re^{-t\,Q_{R,N}}\,z, Q_t)(\rd \phi).
\end{aligned}
\end{equation}
A probability measure $\nu$ in $\cM(L^1_2(\bS_R))$ is invariant for the semigroup $\cR_{R,t}$, in the sense of Def.~\ref{def:measure_invariant}, if and only if
$$\widehat{\nu}(\lambdan)= \widehat{\nu}(\re^{-t\,Q_{R,N}}\lambdan)\,\re^{-\frac{1}{2}\langle Q_t\lambdan,\lambdan\rangle}\quad \text{for\,\, all}\quad t\geq 0\quad\text{and}\quad \lambdan\in L^1_2(\bS_R).$$
To see this, let $F_{\lambdan}(\phi)=\re^{i\langle \lambdan, \phi\rangle}\in B_{\rb}(L^1_2(\bS_R))$. Using Def.~\ref{def:gaussian_measure_integration}, one has
\begin{equation}
\cR_{R,t}\,F_{\lambdan}(z)= \int_{L^1_2(\bS_R)}\,\re^{i\langle \lambdan, \phi\rangle}\,\cN(\re^{-t\,Q_{R,N}}\,z, Q_t)(\rd\phi)
=
\re^{i\langle \lambdan, \re^{-t\,Q_{R,N}} z\rangle}\,\re^{-\frac{1}{2}\langle Q_t\lambdan,\lambdan\rangle}\,.
\end{equation}
The preceding expression implies that $\cR_{R,t}\,F_{f}(z )$ is the characteristic function of the measure $\cN(\re^{-tQ_{R,N}}\,z, Q_t)$. It holds that
\begin{equation}\label{eq:semigroup_gaussian_2}
\int_{L^1_2(\bS_R)}\,\cR_{R,t}\,F_{\lambdan}( \phi)\,\nu(\rd \phi)
= \widehat{\nu}(\re^{-t\,Q_{R,N}}\lambdan)\,\re^{-\frac{1}{2}\langle Q_t\lambdan,\lambdan\rangle}.
\end{equation}
Now from Def.~\ref{def:measure_invariant} a measure $\nu$ is invariant for Eq.~\eqref{eq:weak_free} if Eq.~\eqref{eq:semigroup_gaussian_2} is equal to
$$\widehat{\nu}(\lambdan)=\int_{L^1_2(\bS_R)}\,\re^{i\langle \lambdan,\phi\rangle}\nu(\rd \phi).$$
Observe that the requirement $\widehat{\nu}(\lambdan)= \widehat{\nu}(\re^{-tQ_{R,N}}\lambdan)\,\re^{-\frac{1}{2}\langle Q_t\lambdan,\lambdan\rangle}$ is equivalent to saying that $\nu$ is invariant if and only if
$\nu=\left(\re^{-t\,Q_{R,N}}\nu\right)\ast \mathcal{N}(0,Q_t)$ for all $t\geq 0$. It follows from the fact that $\mathcal{F}(f\ast g)= \mathcal{F}(f)\,\times\mathcal{F}(g)$, see Remark~\ref{rem:inv_gau_measure}.

\begin{thm}\label{thm:invariant_gaussuan_measuer_equi}
The following conditions are equivalent:
\begin{itemize}

\item[(A)] $\sup_{t\geq 0}\,\mathrm{Tr}\,Q_t<\infty$

\item[(B)] $\nu=\cN(0,Q_{\infty})$ is an invariant measure for Eq.~\eqref{eq:weak_free}, i.e., if $z_{R,N}$ has a distribution $\nu$ and is independent of the cylindrical Wiener process $W_R$, then the $L^1_2(\bS_R)$-valued stochastic process $Z_{R,N}(t)$ has distribution $\nu$ for each fixed $t\geq 0$. 

\end{itemize}
\end{thm}

\begin{proof}
Similar result can be found in~\cite[Thm.~6.22]{Hai23} and~\cite[Thm.~6.2.1]{DaZab96}. Recall Eq.~\eqref{eq:weak_free}. One has
\begin{equation}
 Z_{R,N}(t,\Cdot) 
 = \re^{-tQ_{R,N}}\,z_{R,N}+ \sqrt{2}\,\int_0^t\,\re^{-sQ_{R,N}}\,\rd W_R(s,\Cdot)\,.
\end{equation}
Using the fact that the distribution of $\sqrt{2}\,\int_0^t\,\re^{-sQ_{R,N}}\,\rd W_R(s,\Cdot)$ is $\cN(0,Q_t)$, it follows that $\nu$ is invariant for Eq.~\eqref{eq:weak_free} if
$$\nu=(\re^{-tQ_{R,N}}\nu)\ast \cN(0,Q_t)\quad \text{for\,\, all}\quad t\geq 0.$$
Taking the Fourier transform of both sides leads to
\begin{equation}\label{eq:inva_gauss_measure_1}
\widehat{\nu}(\lambdan)=\widehat{(\re^{-tQ_{R,N}}\nu)}(\lambdan)\,\times \widehat{\mathcal{N}(0,Q_t)}(\lambdan)
\quad \text{for}\quad \lambdan\in L^1_2(\bS_R).
\end{equation}
It holds that
\begin{equation}\label{eq:inv_gauss_measure_2}
\begin{aligned}
\widehat{ (\re^{-tQ_{R,N}}\nu) }(\lambdan)
=&
\int_{L^1_2(\bS_R)}\,\re^{i\langle \lambdan, \phi\rangle}\,\left(\re^{-tQ_{R,N}}\nu\right)(\rd \phi)
=\int_{L^1_2(\bS_R)}\,\re^{i\langle \lambdan, \re^{-tQ_{R,N}}\phi\rangle}\,\nu(\rd \phi)
\\
=&
\int_{L^1_2(\bS_R)}\,\re^{i\langle \re^{-tQ_{R,N}} \lambdan,\phi\rangle}\,\nu(\rd \phi)
=
\widehat{\nu}(\re^{-tQ_{R,N}} \lambdan).
\end{aligned}
\end{equation}
Note that 
$$\widehat{\mathcal{N}(0,Q_t)}(\lambdan)= \re^{-\frac{1}{2}\langle Q_t\lambdan,\lambdan\rangle}.$$
Combining Eqs.~\eqref{eq:inva_gauss_measure_1} and~\eqref{eq:inv_gauss_measure_2}, one obtains
\begin{equation}
\hat{\nu}(\lambdan)= \hat{\nu}(\re^{-t\,Q_{R,N}} \lambdan)\times \re^{-\frac{1}{2}\langle Q_t\lambdan,\lambdan\rangle}\quad \forall \lambdan\in L^1_2(\bS_R)\,.
\end{equation}
This implies that
\begin{equation}
\re^{\frac{1}{2}\langle Q_t\lambdan,\lambdan\rangle}\,\textrm{Re}\,\hat{\nu}(\lambdan)=\textrm{Re}\,\hat{\nu}(\re^{-t\,Q_{R,N}} \lambdan)\leq 1\,.
\end{equation}
Hence, for all $\lambdan\in L^1_2(\bS_R)$ one has
\begin{equation}
\langle Q_t\lambdan,\lambdan \rangle \leq 2\log\,\left(\frac{1}{\textrm{Re}\,\hat{\nu} (\lambdan)}\right).
\end{equation}
By the Bochner theorem (Thm.~\ref{thm:bochner}) for $\epsilon=1/2$, there exists a trace-class operator $G$ such that
$$\mathrm{Re}\,\hat{\nu}(\lambdan)\geq 1/2\quad \text{for\,\, all} \quad \lambdan \in L^1_2(\bS_R)\quad\text{with}\quad \langle G\lambdan,\lambdan\rangle\leq 1.$$
This gives rise to $\langle Q_t\lambdan,\lambdan\rangle\leq (2\log 2)$, which amounts to writing 
$$\sup_{t\geq 0}\mathrm{Tr}\,Q_t\leq 2\log\,2\,\mathrm{Tr}\,G<\infty.$$
This concludes that item~$(B)$ implies item~(A). For the converse direction, recall that 
$$Q_t\,:\, L^1_2(\bS_R)\to L^1_2(\bS_R)\quad \text{is\,\, given\,\, by}\quad 
Q_t:=2\int_0^t\,\re^{-2sQ_{R,N}}\,\rd s.$$
Let $\nu=\cN(0,Q_{\infty})$, i.e., $\widehat{\nu}(\lambdan)=\re^{-\frac{1}{2}\langle Q_{\infty}\lambdan,\lambdan\rangle}$. This yields
$$
\widehat{\nu}(\re^{-tQ_{R,N}}\lambdan)=\re^{-\frac{1}{2}\langle \re^{-tQ_{R,N}}\,Q_{\infty}\,\re^{-tQ_{R,N}}\lambdan,\lambdan\rangle}.$$
From the semigroup property, one has
\begin{equation}
\begin{aligned}
 \re^{-t\,Q_{R,N}}Q_{\infty}\re^{-t\,Q_{R,N}}=&
2\int_0^{\infty}\,\re^{-2(t+s)\,Q_{R,N}}\rd s
\\
&=2\int_0^{\infty}\,\re^{-2r\,Q_{R,N}}\rd r-2\int_{0}^t\,\re^{-2r\,Q_{R,N}}\rd r
=Q_{\infty}-Q_t\,.
\end{aligned}
\end{equation}
This implies that $\hat{\nu}(\lambdan)= \hat{\nu}(\re^{-t\,Q_{R,N}}\lambdan)\,\re^{-\frac{1}{2}\langle Q_t\lambdan,\lambdan\rangle}$ and finishes the proof of the converse direction.
Finally,
$$\re^{-\frac{1}{2}\langle Q_t\lambdan,\lambdan\rangle}\to \re^{-\frac{1}{2}\langle Q_{\infty}\lambdan,\lambdan\rangle},\quad\text{as}\quad t\to\infty,$$ 
which is the characteristic function of $\cN(0,Q_\infty)$.
\end{proof}

\begin{lem}\label{lem:inv_measu_gauss}
For all $R,N\in\bN_+$, the measures $\nu_{R,N}=\cN(0,G_{R,N})$
concentrated on $L^1_2(\bS_R)$ are invariant for Eq.~\eqref{eq:weak_free}. In particular, for all $t\in[0,\infty)$ it holds that
$$\mathrm{Law}(Z_{R,N}(t,\Cdot))=\mathrm{Law}(z_{R,N})=\nu_{R,N}.$$
\end{lem}
\begin{proof}
Noting that $\cN(0,Q_{\infty})=\cN(0,G_{R,N})=\nu_{R,N}$, the statement is immediate consequence of Thm.~\ref{thm:invariant_gaussuan_measuer_equi}.
\end{proof}

\subsection*{Interacting measure}\label{subse:non-Gaussian}
We denote by $\cP_{R,t}$ the semigroup associated to Eq.~\eqref{eq:weak_interacting}
\begin{equation}\label{eq:non-linear_eq}
\begin{aligned}
\cP_{R,t}\,F(\phi ):=&\, \bE\left[F(\varPhi^g_{R,N}(t,\Cdot))\,|\,\varPhi^g_{R,N}(0,\Cdot)=\phi\right]
\quad \text{for\,\,all}\quad F\in B_{\rb}(L^1_2(\bS_R)).
\end{aligned}
\end{equation}
We aim to verify that for all $R,N\in\bN_+$, the measures $\mu^g_{R,N}$ is invariant for $\cP_{R,t}$ in the sense of Def.~\ref{def:measure_invariant}.

\begin{thm}\label{thm:inv_measu_non_gauss}
Let $g\in C^{\infty}(\bS_R)$. For all $R,N\in\bN_+$, there exists a measure of the form 
\begin{equation}
\mu_{R,N}^g(\rd \phi) = \frac{1}{\mathcal{Z}^g_{R,N}}\,\frac{1}{\mathcal{Z}_{R,N}} \exp\Big(\frac{1}{n}\phi(g)^n-\int_{\bS_R}\,P(\phi(\rx),c_{R,N})\, \rho_R(\rd\rx)\Big)\,\nu_{R,N}(\rd\phi)\,,
\end{equation}
concentrated on $L^1_2(\bS_R)$, which is invariant for Eq.~\eqref{eq:weak_interacting}. In particular, 
$$\mathrm{Law}(\varPhi^g_{R,N}(t,\Cdot))=\mu^g_{R,N}\quad \text{for\,\, all}\quad t\geq 0 .$$
\end{thm}

\begin{proof}
We shall verify the assumptions of Lemma~\ref{lem:Hypothesis 8.4} with $\tilde{H}= L^1_2(\bS_R)$ and $ H= L_2(\bS_R)$.

\begin{itemize}
\item[1)] For all $R,N\in\bN_+$, the elliptic operator $Q_{R,N}\,:\,L^6_2(\bS_R)\subset L_2(\bS_R)\to L_2(\bS_R)$ is a self-adjoint and positive definite. In particular, for all $R,N\in \bN_+$,
$$\mathrm{Tr}(Q^{-1}_{R,N}) =\mathrm{Tr}(G_{R,N})<\infty.$$
This follows from Remark~\ref{rem:trace_semigroup}.

\item[2)]
On account of Remark~\ref{rem:reproducing_kernel_sphere}, for all $R,N\in\bN_+$ the operator $Q_{R,N}$ generates a $C_0$-semigroup of contractions on $L^1_2(\bS_R)$. This implies that
$$\|\re^{-t\,Q_{R,N}} f\|_{L_2^1(\bS_R)} \leq \,\|f\|_{L_2^1(\bS_R)}\quad\forall \quad f\in L_2^1(\bS_R).$$
Consequently,
$$\re^{-tQ_{R,N}}\,L^1_2(\bS_R)\subseteq L^1_2(\bS_R)\quad \forall\quad t\geq 0.$$
This verifies that $L^1_2(\bS_R)$ is invariant under $\re^{-tQ_{R,N}}$ for all $R,N\in\bN_+$. To proceed, fix $t>0$ and set $\Delta t>0$. Using the semigroup property of $\re^{-t\,Q_{R,N}}$, one obtains
\begin{equation}
\lim_{\Delta t\to 0}\,\|\re^{-(t+\Delta t)\,Q_{R,N}} f -\,\re^{-t\,Q_{R,N}}f\|_{L^1_2(\bS_R)} =\lim_{\Delta t\to 0}\,\|\re^{-t\,Q_{R,N}} (\re^{-\Delta t \,Q_{R,N}} -1) f\|_{L^1_2(\bS_R)}=0.
\end{equation}
It follows from the strong continuity of the semigroup $\re^{-tQ_{R,N}}$. This implies that the map $\re^{-t\,Q_{R,N}}\,f$ is right continuous. Similarly, one deduces that
\begin{equation}
\lim_{\Delta t\to 0}\,\|\re^{-t\,Q_{R,N}} f -\,\re^{-(t-\Delta t)\,Q_{R,N}}f\|_{L^1_2(\bS_R)} = 0\,,
\end{equation}
which implies that the map $\re^{-t\,Q_{R,N}}\,f$ is left continuous. Hence, for all $f\in L_2^1(\bS_R)$ the map $t\mapsto \re^{-t\,Q_{R,N}}f$ is continuous in $L^1_2(\bS_R)$.

\item[3)]
Lemma~\ref{lem:free_solution} implies that the Ornstein-Uhlenbeck process $Z_{R,N}(t,\Cdot)\in C([0,\infty),L_2^1(\bS_R))$ exists. 

\item[4)]
Let $U_{R,N}^g(\phi):=\frac{1}{n}\,\phi(g)^n-\,\int_{\bS_R}\,P(\phi(\rx),c_{R,N})\, \rho_R(\rd\rx)$. It is obvious that $U^g_{R,N}$ is continuous as it is a polynomial. Let $\psi\in L_2^1(\bS_R)$ and $D_{\epsilon}$ denote the derivative with respect to $\varepsilon>0$. By definition one gets
\begin{multline}
\left.D_{\varepsilon}\,U_{R,N}^g(\phi+\varepsilon\, \psi)\right|_{\varepsilon=0} 
\\
=\left.D_{\varepsilon}\left(-\int_{\bS_R}\, P(\phi(\rx)+\varepsilon\, \psi(\rx),c_{R,N})\,\rho_R(\rd\rx)
+\frac{1}{n}\,((\phi+\varepsilon\,\psi)(g))^n\right)\right|_{\varepsilon=0}
\\
=\,-\,\int_{\bS_R}\,P'(\phi(\rx),c_{R,N})\,\psi(\rx)\,\rho_R(\rd\rx)\,+\,\phi(g)^{n-1}\,\psi(g)\,.
\end{multline}
Observe that $D_{\varepsilon}\,U_{R,N}^g(\phi)$ with $\phi\mapsto\varPhi_{R,N}^g(t,\Cdot)$ coincides with the nonlinearity $F$ given by Eq.~\eqref{eq:non-linear-Psi} with $\varPhi_{R,N}^g(t,\Cdot)= \varPsi_{R,N}^g(t,\Cdot)+Z_{R,N}(t,\Cdot)$.

\item[5)]
The proof of Lemma~\ref{lem:mu_measure_well_defined} implies that the mapping $U^g_{R,N}$ is bounded from above. Moreover, proof of Lemma~\ref{lem:lipschitz} shows that $F\,:\,L^1_2(\bS_R)\to L_2(\bS_R)$ is locally Lipschitz continuous.

\item[6)]
The mild solution $\varPhi^g_{R,N} \in C ([0, \infty), L^1_2 (\bS_R ))$ of Eq.~\eqref{eq:weak_interacting} with the initial condition $\varPhi_{R,N}^{g}(0,\Cdot)=\phi_{R,N}^g\in L^1_2(\bS_R)$ is given by 
\begin{multline}\label{eq:mild_interacting}
 \varPhi_{R,N}^g(t,\Cdot) 
= \re^{-t\,Q_{R,N}} \phi_{R,N}^g
\\ +\int_0^t\re^{-(t-s)\,Q_{R,N}}\, \left[\sqrt{2}\,\rd W_R(s,\Cdot)
 -P'(\varPhi_{R,N}^g(s,\Cdot),c_{R,N})\,\rd s + (\varPhi_{R,N}^g(s,g))^{n-1}g(\Cdot)\,\rd s\right]. 
\end{multline}
This is the content of Lemma~\ref{lem:mild_solution_interacting}.
\end{itemize}
Combining Items~(1)-(6) obtained above yields the existence of an invariant measure $\mu_{R,N}^g$ for Eq.~\eqref{eq:weak_interacting} on $L^1_2(\bS_R)$, which is absolutely continuous with respect to the measure $\nu_{R,N}$, see Remark~\ref{rem:abso_con_measure}. This completes the proof.
\end{proof}

\chapter{Infinite volume measure}\label{ch:infinite-measure}
In this section, we shall show the existence of a weak limit of a subsequence of the sequence of measures $(\jmath^*_R\sharp \mu^g_{R,N})_{R,N\in\bN_+}$ on $\sS'(\bR^2)$ as $R\to \infty$ using tightness arguments. We call such a limit a $P(\Phi)_2$ measure on $\sS'(\bR^2)$. To this end, we consider the limits as $R\to \infty$ and $N\to \infty$ for the sequence of measures $(\jmath^*_R\sharp \mu^g_{R,N})_{R,N\in\bN_+}$ on $\sS'(\bR^2)$. The existence of the UV limit as $N\to \infty$ of $(\mu^g_{R,N})_{R,N\in\bN_+}$ was established in Sec.~\ref{sec:existence_UV_limit}. This result also holds for $(\jmath^*_R\sharp \mu^g_{R,N})_{R,N\in\bN_+}$ on $\sS'(\bR^2)$.
\\
\\
To prove the existence of the limit as $R\to \infty$, we first obtain an a priori bound in Prop.~\ref{prop:energy_method_plane} using the energy method in $L_2(\bR^2,v^{1/2}_L)$. The function $\nu_L\in L_1(\bR^2)$ is an admissible weight introduced in Def.~\ref{def:weight_plane}. Next, we prove tightness of $(\jmath_R^*\sharp \mu^g_{R,N})_{R,N\in\bN_+}$ by averaging in time the expectation of $\jmath_R^*\varPhi^g_{R,N}(t,\Cdot)$, using the fact that
$$\mathrm{Law}(\jmath_R^*\varPhi^g_{R,N}(t,\Cdot))=\jmath_R^*\sharp\mu^g_{R,N} \quad \forall \quad t\in[0,\infty).$$
This implies the existence of a subsequential limit as $R\to\infty$ of the sequence of measures $(\jmath_R^*\sharp \mu^g_{R,N})_{R,N\in\bN_+}$. Indeed, a $P(\Phi)_2$ measure on $\sS'(\bR^2)$ can be viewed as the marginal distribution of the stationary solution associated with the counterpart of Eq.~\eqref{eq: SPDE on Sphere} on $\bR^2$. A similar argument was first used in~\cite{AK17,AK21,GH18}, where the existence of a $\Phi^4$ measure on $\sS'(\bR^3)$ is proven.

\section{A priori bound}\label{sec:a_priori_bound}
\begin{dfn}\label{def:weight_plane}
Fix  $L\in[1,\infty)$ as in Lemma~\ref{lem:wieght_plane_laplacian} and let 
$$v_L:= 1/(4\pi L^2)\,w_L^8,\quad \text{with}\quad w_L(x):=16L^4/(4L^2+x_1^2+x_2^2)^2 \in C^{\infty}(\bR^2).$$
Observe that using polar coordinates one gets
\begin{equation}
\|v_L\|_{L_1(\bR^2)} = \frac{1}{4 L^2}\,\int_{0}^{\infty}\, (\frac{16 L^4}{(4L^2 +r^2)^2})^8\,(2r)\,\rd r = \frac{4^{16} L^{32}}{4L^2}\,\frac{1}{15(8^{15})L^{30}} =\frac{1}{15\,(2)^{15}}<1\,.
\end{equation}
\end{dfn}

\begin{Remark}
The precise choice of the weight $v_L$ is not of much importance. It is convenient to use a weight that decays polynomially and express it in terms of the function $w_R(x)$ introduced below Eq.~\eqref{eq:conformal_factor}. The prefactor $1/(4\pi L^2)$ guarantees that the $L_1(\bR^2)$ norm of the weight is bounded by $1$ and the decay rate is chosen so that the estimate stated in Remark~\ref{rem:L_two_interms_of_L_n} is true.
\end{Remark}
\noindent
Utilizing Remarks~\ref{rem:norm_stereo} and~\ref{rem:comf_laplacian}, one rewrites Eq.~\eqref{eq:varPsi_pde2} in the stereographic coordinates as follows
\begin{multline}\label{eq:spde_psi_stereographic}
 (\partial_t + (1- w_R^{-1}(x)\Delta)(1-w_R^{-1}(x)\Delta/N^2)^2)\jmath_R^*\varPsi^g_{R,N}(t,x) + (\jmath_R^*\varPsi^g_{R,N}(t,x))^{n-1}
 \\
 = 
 \sum_{l=0}^{n-2}\sum_{m=l}^{n-1} a_{m,l}\, \jmath_R^* Z_{R,N}^{:m-l:}(t,x)
 \,(\jmath_R^*\varPsi^g_{R,N}(t,x))^l
 \\
 +
 (\jmath_R^*(\varPsi^g_{R,N}+Z_{R,N})(t,w_R\jmath_R^* g))^{n-1}\jmath_R^*g(x)\,.
\end{multline}
We shall apply the energy method in $L_2(\bR^2,v^{1/2}_L)$ in order to obtain an a priori bound. To this end, we multiply both sides of the above equation by $v_L\jmath_R^*\varPsi^g_{R,N}$, and integrate over $\bR^2$. The resulting bound, which is uniform in $R,N\in\bN_+$, will be used to prove tightness of the family of measures $(\jmath_R^*\sharp\mu^g_{R,N})_{R,N\in\bN_+}$. Note that we have already utilized the energy method in $L^1_2(\bS_R)$ in order to prove the global existence of the mild solution in Sec.~\ref{subsec:global_solution_remainder}.

\begin{lem}\label{lem:wieght_plane_laplacian}
There exists $L\in[1,\infty)$ such that for all $R\in[L,\infty)$ it holds that
\begin{enumerate}
\item[(A)] $\langle\Psi, v_L (-w^{-1}_R \Delta)\Psi\rangle_{L_2(\bR^2)} 
\geq
1/2\,\|\vec\nabla\Psi\|_{L_2(\bR^2,w^{-1/2}_Rv_L^{1/2})}^2
-1/8\,\|\Psi\|^2_{L_2(\bR^2,w_R^{-1/2} v_L^{1/2})}$,

\item[(B)] $\langle\Psi, v_L (-w^{-1}_R\Delta)^2 \Psi\rangle_{L_2(\bR^2)}
\geq 
1/2\,\|\Delta\Psi\|_{L_2(\bR^2,w^{-1}_Rv_L^{1/2})}^2
-1/8\,\|\Psi\|^2_{L_2(\bR^2,w_R^{-1} v_L^{1/2})}$,

\item[(C)] $\langle\Psi, v_L (-w^{-1}_R\Delta)^3 \Psi\rangle_{L_2(\bR^2)}
\geq 
1/2\,\|\vec\nabla\Delta\Psi\|_{L_2(\bR^2,w^{-3/2}_Rv_L^{1/2})}^2
-1/8\,\|\Psi\|^2_{L_2(\bR^2,w_R^{-3/2} v_L^{1/2})}$.
\end{enumerate}
\end{lem}
\begin{proof}
There exists $C\in(0,\infty)$ such that for all $L\in[1,\infty)$, $R\in[L,\infty)$
\begin{equation}
|\vec\nabla w^{-1/2}_R|\leq (C/L)\,w_R^{-1/2},
\qquad
|\vec\nabla v^{1/2}_L|\leq (C/L)\,v_L^{1/2}.
\end{equation}
This gives readily~(A) by integrating by parts in $\|\vec\nabla\Psi\|_{L_2(\bR^2,w^{-1/2}_Rv_L^{1/2})}^2$, applying the Leibniz rule and Young's inequality. Estimates (B)~and~(C) are obtained analogously, with the help of the following auxiliary inequalities
\begin{equation}
 \|\vec\nabla\Psi\|_{L_2(\bR^2,w^{-1}_Rv_L^{1/2})}^2
 \leq 
 2\|\Delta\Psi\|_{L_2(\bR^2,w^{-1}_Rv_L^{1/2})}^2
 +
 2\|\Psi\|^2_{L_2(\bR^2,w_R^{-1} v_L^{1/2})}\,,
\end{equation}
\begin{multline}
 \|\vec\nabla\Psi\|_{L_2(\bR^2,w^{-3/2}_Rv_L^{1/2})}^2
 +
 \|\Delta\Psi\|_{L_2(\bR^2,w^{-3/2}_Rv_L^{1/2})}^2
\\
 \leq 
 4\|\vec\nabla\Delta\Psi\|_{L_2(\bR^2,w^{-3/2}_Rv_L^{1/2})}^2
 +
 4\|\Psi\|^2_{L_2(\bR^2,w_R^{-3/2} v_L^{1/2})}
\end{multline}
valid for sufficiently large $L\in[1,\infty)$ and all $R\in[L,\infty)$. The latter inequalities are proven by the same token as Item~(A). For the complete proof see Appendix~\ref{sec:proof_wieght_plane_laplacian}.
\end{proof}

\begin{Remark}\label{rem:L_two_interms_of_L_n}
For all $L\in[1,\infty)$, $R\in[L,\infty)$, and $p\in\{1,2,3\}$, one has
\begin{equation}
\begin{aligned}
 \|\Psi\|_{L_2(\bR^2,w_R^{-p/2} v_L^{1/2})}
 &= \|w_R^{-p/2}\,v_L^{1/2}\Psi\|_{L_2(\bR^2)}
 \\
 &\leq
 \|w_R^{-p/2} v_L^{(n-2)/2n}\|_{L_{2n/(n-2)}(\bR^2)}
 \|\Psi\|_{L_n(\bR^2,v_L^{1/n})}\,.
 \end{aligned}
\end{equation} 
It follows from H\"older's inequality. Observe that there exists $C\in (0,\infty)$ such that
\begin{equation}
 \|w_R^{-p/2} v_L^{(n-2)/2n}\|_{L_{2n/(n-2)}(\bR^2)}\leq C\,.
\end{equation}
It follows from the facts that $w_R\ge w_L$ and the scaling relation $w_L(x)=w_1(x/L)$, the change of variables. The resulting integral is finite for $n\ge4$ and $p\in{1,2,3}$.
\end{Remark}

\begin{prop}\label{prop:energy_method_plane}
There exist $\kappa\in(0,\infty)$, $C\in(0,\infty)$, $p\in[1,\infty)$ and a ball $B\subset\sS(\bR^2)$ with respect to some Schwartz semi-norm centered at the origin such that for all $t\in(0,\infty)$ and $R,N\in\mathbb{N}_+$, $R\geq L$, as well as all $g\in C^\infty(\bS_R)$, $w_R\jmath_R^* g\in B$, it holds that
\begin{equation}
8\,\partial_t\|\jmath_R^* \varPsi^g_{R,N}(t,\Cdot)\|^2_{L_2(\bR^2,v_L^{1/2})}
 +
 \|\jmath_R^* \varPsi^g_{R,N}(t,\Cdot)\|^n_{L_n(\bR^2,v_L^{1/n})}
 \\
\leq C \sum_{k=0}^{n-1}\,
\|\jmath_R^*Z_{R,N}^{:k:}(t,\Cdot)\|^p_{L_p^{-\kappa}(\bR^2, v_L^{1/p})}\,.
\end{equation}
\end{prop}
\begin{proof}
After multiplying both sides of Eq.~\eqref{eq:spde_psi_stereographic} by $v_L\jmath_R^*\varPsi^g_{R,N}$, integrating over $\bR^2$ and applying Lemma~\ref{lem:wieght_plane_laplacian} and Remark~\ref{rem:L_two_interms_of_L_n}, we obtain
\begin{multline}
\frac{1}{2}\partial_t\|\jmath_R^*\varPsi^g_{R,N}(t,\Cdot)\|^2_{L_2(\bR^2, v_L^{\frac{1}{2}})} +\|\jmath_R^*\varPsi^g_{R,N}(t,\Cdot)\|^2_{L_2(\mathbb{R}^2, v_L^{\frac{1}{2}})}
+ (\frac{1}{N^2}+ \frac{1}{2})\|\vec\nabla\jmath_R^*\varPsi^g_{R,N}(t,\Cdot)\|^2_{L_2(\bR^2,v_L^{\frac{1}{2}}w_R^{-\frac{1}{2}})}
\\
+(1-\frac{1}{4N^2} -\frac{1}{8N^4} -\frac{1}{8} -\frac{1}{4N^2} -\frac{1}{8N^4}) \|\jmath_R^*\varPsi^g_{R,N}(t,\Cdot)\|^n_{L_n(\mathbb{R}^2, v_L^{\frac{1}{n}})} 
\\
+ (\frac{1}{2N^4} +\frac{1}{N^2}) \|\Delta\jmath^*_R\Psi^g_{R,N}(t,\Cdot)\|^2_{L_2(w_R^{-\frac{1}{2}} v_L^{\frac{1}{2}})}
\\
+ \frac{1}{2N^4} \|\vec\nabla \Delta\jmath^*_R\Psi^g_{R,N}(t,\Cdot) \|^2_{L_2(w_R^{-\frac{3}{2}} v_L^{\frac{1}{2}})}
 \leq R^{(1)}_{R,N}(t) +R^{(2)}_{R,N}(t) ,
\end{multline}
where
\begin{equation}
\begin{aligned}
R_{R,N}^{(1)}(t)&= \sum_{l=0}^{n-2}\sum_{m=l}^{n-1} \int_{\bR^2} a_{m,l}\,v_L(x)\,(\jmath_R^*Z_{R,N}^{:m-l:})
 \,(\jmath_R^*\varPsi^g_{R,N}(t, x))^{l+1} \rd x,
\\
R_{R,N}^{(2)}(t)&=(\jmath_R^*(\varPsi^g_{R,N}+Z_{R,N})(t,w_R\jmath_R^* g))^{n-1}\,(\jmath_R^*\varPsi^g_{R,N})(t,v_L\jmath_R^*g)\,.
\end{aligned}
\end{equation}
Observe that $w_R^{-1}(x) \geq 1 $ and $N\in\bN_+$. Hence, 
\begin{multline}\label{eq: energy method}
\frac{1}{2}\partial_t\|\jmath_R^*\varPsi^g_{R,N}(t,\Cdot)\|^2_{L_2(\bR^2, v_L^{\frac{1}{2}})} 
+\frac{1}{2}\|\nabla\jmath_R^*\varPsi^g_{R,N}(t,\Cdot)\|^2_{L_2(\bR^2,v_L^{\frac{1}{2}})} 
\\ + \frac{1}{8} \|\jmath_R^*\varPsi^g_{R,N}(t,\Cdot)\|^n_{L_n(\mathbb{R}^2, v_L^{\frac{1}{n}})}
\leq \,|R_{R,N}^{(1)}(t) + R^{(2)}_{R, N}(t)|.
\end{multline}
Moreover, one has
\begin{equation}
 |R^{(1)}_{N,R}(t) |\leq \sum_{l=0}^{n-2}\sum_{m=l}^{n-1}|a_{m,l}|\,|\langle \jmath_R^*Z^{:m-l:}_{R,N}(t,x), (\jmath_R^*\Psi_{R,N}^g(t,x))^{l+1}\rangle_{L_2(\bR^2,v_L)}|\,.
\end{equation}
By Lemma~\ref{lem:Z_Psi_n}, for every $\delta_1\in(0,1)$ there exists $\hat C\in(0,\infty)$ such that
\begin{multline}
|\langle \jmath_R^*Z_{R,N}^{:m-l:}, (\jmath_R^*\Psi^g_{R,N})^{l+1}\rangle_{L_2(\bR^2,v_L^{1/2})}|
\leq \hat C \|\jmath_R^*Z_{R,N}^{:m-l:}(t,\Cdot)\|^{p}_{L_p^{-\kappa}(\bR^2,v_L^{1/p})}
\\
+\delta_1\,\|\nabla\jmath_R^*\varPsi^g_{R,N}(t,\Cdot)\|^2_{L_2(\bR^2, v_L^{1/2})}
+ \delta_1\,\|\jmath_R^*\varPsi^g_{R,N}(t,\Cdot) \|^n_{L_n(\bR^2, v_L^{1/n})}+\delta_1\,.
\end{multline}
By the assumption made in Sec.~\ref{subsec:local_solution_remainder} on the coefficients $a_{m,l}$, there exists $\check C\in(0,\infty)$ such that $|a_{m,l}|\leq \check C$. Thus, 
\begin{equation}
\begin{aligned}
 \sum_{l=0}^{n-2}\sum_{m=l}^{n-1}\hat C\,|a_{m,l}|\,\|\jmath_R^*Z_{R,N}^{:m-l:}(t,\Cdot)\|^{p}_{L_p^{-\kappa}(\bR^2,v_L^{1/p})}
 &=  \sum_{l=0}^{n-2}\sum_{k=0}^{n-1-l}\,C\, \|\jmath_R^*Z_{R,N}^{:k:}(t,\Cdot)\|^{p}_{L_p^{-\kappa}(\bR^2,v_L^{1/p})}
 \\&
 =
 \sum_{k=0}^{n-1}\sum_{l=0}^{(n-1-l)\wedge (n-2)}\, C\, \|\jmath_R^*Z_{R,N}^{:k:}(t,\Cdot)\|^{p}_{L_p^{-\kappa}(\bR^2,v_L^{1/p})}
 \\
 &
\leq  C_1\,\sum_{k=0}^{n-1} \|\jmath_R^*Z_{R,N}^{:k:}(t,\Cdot)\|^{p}_{L_p^{-\kappa}(\bR^2,v_L^{1/p})}
 \end{aligned}
 \end{equation}
for some $C, C_1\in (0,\infty)$. Combining all the above results culminates in
\begin{multline}\label{eq:remainder_1}
|R_{R,N}^{(1)}(t)|\leq C_1 \sum_{k=0}^{n-1}\,\|\jmath_R^*Z_{R,N}^{:k:}(t,\Cdot)\|^{p}_{L_p^{-\kappa}(\bR^2,v_L^{1/p})}
\\
+\delta_1\,\|\nabla\jmath_R^*\varPsi^g_{R,N}(t,\Cdot)\|^2_{L_2(\bR^2, v_L^{1/2})}
+ \delta_1\,\|\jmath_R^*\varPsi^g_{R,N}(t,\Cdot) \|^n_{L_n(\bR^2, v_L^{1/n})},\,
\end{multline}
where the $k=0$ term of the sum above is a constant. Using the fact that
$$\forall \quad k> 1, \quad f(a_i)\in\bR_+\quad :\quad (\sum_{i=1}^n\,f(a_i))^k\leq n^{k-1}\,\sum_{i=1}^n\,f(a_i)^k,$$
one obtains
\begin{equation}
\begin{aligned}
|R^{(2)}_{R, N}(t)|
\leq& 2^{n-2}\,|(\jmath_R^*(\varPsi^g_{R,N})(t,w_R\jmath_R^* g))^{n-1}(\jmath_R^*\varPsi^g_{R,N})(t,v_L\jmath_R^*g)|
\\& +2^{n-2}\,|(\jmath_R^*Z_{R,N})(t,w_R\jmath_R^* g))^{n-1}\,
(\jmath_R^*\varPsi^g_{R,N})(t,v_L\jmath_R^*g)|.
\end{aligned}
\end{equation}
From H\"older's inequality with $1/n+(n-1)/n=1$, one gets
\begin{equation}
\begin{aligned}
|(\jmath_R^*\varPsi^g_{R,N})(t,w_R\jmath_R^* g)) |\leq& \|\jmath_R^*(\varPsi^g_{R,N}(t,\Cdot)\|_{L_n(\bR^2,v_L^{1/n})}\, \|v_L^{-1/n}w_R\jmath^*_Rg\|_{L_{n/(n-1)}(\bR^2)}
\\
|(\jmath_R^*Z_{R,N})(t,w_R\jmath_R^* g))|\leq& \|\jmath^*_RZ_{R,N}(t,\Cdot)\|_{L^{-\kappa}_n(\bR^2, v_L^{1/n})}\, \|v_L^{-1/n}w_R\jmath_R^*g\|_{L^\kappa_{n/(n-1)}(\bR^2)}
\\
|(\jmath_R^*\varPsi^g_{R,N})(t,v_L\jmath_R^*g)|\leq& \|\jmath_R^*(\varPsi^g_{R,N}(t,\Cdot)\|_{L_n(\bR^2,v_L^{1/n})}\, \|v_L^{(n-1)/n}\jmath^*_Rg\|_{L_{n/(n-1)}(\bR^2)}\,.
\end{aligned}
\end{equation}
The above expressions imply that there exists $C_2\in(0,\infty)$ such that for all $\delta_2\in(0,1)$,
\begin{equation}\label{eq:remainder_2}
|R^{(2)}_{R,N}(t)| \leq C_2\,\delta_2^n\,
\|\jmath_R^*\varPsi^g_{R,N}(t,\Cdot)\|^n_{L_n(\bR^2, v_L^{1/n})}
+
C_2\,\delta_2^n\,\|\jmath_R^*Z_{R,N}\|^n_{L_n^{-\kappa}(\bR^2, v_L^{1/n})}
\end{equation}
provided that
\begin{equation}
 \|v_L^{(n-1)/n}\jmath_R^*g\|_{L_{n/(n-1)}(\bR^2)}
 \vee
 \|v_L^{-1/n}w_R\jmath_R^*g\|_{L_{n/(n-1)}(\bR^2)}
 \vee
 \|v_L^{-1/n}w_R\jmath_R^*g\|_{L^\kappa_{n/(n-1)}(\bR^2)}
 \leq \delta_2.
\end{equation}
Now the statement follows by substituting Eqs.~\eqref{eq:remainder_1} and~\eqref{eq:remainder_2} in Eq.~\eqref{eq: energy method} and choosing $\delta_1,\delta_2$ such that $\delta_1\leq1/2$ and $\delta_1+C_2\,\delta_2^n\leq1/16$. This guarantees that all coefficients in the resulting expression are positive and finishes the proof.
\end{proof}

\section{Tightness}\label{sec:tightness}
\begin{prop}\label{prop:tightness}
Let $\kappa\in(0,\infty)$. There exists a ball $B\subset\sS(\bR^2)$, with respect to some Schwartz seminorm centered at the origin and a constant $C\in(0,\infty)$ such that for all $R\in\bN_+$, $R\geq L$,  $N\in \bN_+$ and all $g\in C^\infty(\bS_R)$, $w_R\jmath_R^* g\in B$, it holds that
\begin{equation}
 \int\|\jmath^*_R\phi\|^n_{L_n^{-\kappa}(\bR^2,v_L^{1/n})}\,\mu^g_{R,N}(\rd\phi)\leq C.
\end{equation}
\end{prop}

\begin{proof}
On account of Thm.~\ref{thm:inv_measu_non_gauss}, one has 
$$\mathrm{Law}(\jmath_R^*\varPhi^g_{R,N}(t,\Cdot))=\jmath_R^*\sharp\mu^g_{R,N}\quad \text{for\,\,all} \quad t\geq 0.$$ Hence,
\begin{multline}
\int\|\phi\|^n_{L_n^{-\kappa}(\bR^2,v_L^{1/n})}\,(\jmath_R^*\sharp\mu^g_{R,N})(\rd\phi)
\\=
\int\|\jmath^*_R\phi\|^n_{L_n^{-\kappa}(\bR^2,v_L^{1/n})}\,\mu^g_{R,N}(\rd\phi)
=\bE\|\jmath_R^*\varPhi^g_{R,N}(t,\Cdot)\|^n_{L_n^{-\kappa}(\bR^2,v_L^{1/n})}\,.
\end{multline}
Using Lemma~\ref{lem:inv_measu_gauss}, it follows that $\mathrm{Law}(\jmath_R^*X_{R, N})=\mathrm{Law}(\jmath_R^*Z_{R,N}(t,\Cdot))$ for all $t\geq 0$. Thus, from Lemma~\ref{lem:stochastic_bound_infinite_volume} and
Proposition~\ref{prop:energy_method_plane}, we have
\begin{equation}\label{eq: first priori bound}
8\,\partial_t\bE\|\jmath_R^* \varPsi^g_{R,N}(t,\Cdot)\|^2_{L_2(\bR^2,v_L^{1/2})}
 +
\bE \|\jmath_R^* \varPsi^g_{R,N}(t,\Cdot)\|^n_{L_n(\bR^2,v_L^{1/n})}
\leq C_1 \,
\end{equation}
for some constant $C_1\in(0,\infty)$ independent of $g$, $R,N$ and $t$. The above inequality implies that for all $T\in(0,\infty)$, one has
\begin{multline}\label{eq:tightness_time_average}
\frac{1}{T}\int_0^T\bE \|\jmath_R^* \varPsi^g_{R,N}(t,\Cdot)\|^n_{L_n(\bR^2,v_L^{1/n})}\rd t
\\
\leq C_1-\frac{8}{T}\bE\,\|\jmath_R^*\Psi^g_{R,N}(T,\Cdot)\|^2_{L_2(\bR^2, v_L^{1/2})}
+\frac{8}{T}\bE\|\jmath_R^*\Psi^g_{R,N}(0,\Cdot)\|^2_{L_2(\bR^2, v_L^{1/2})}
\leq C_1 + \frac{C_{R,N}}{T}\,,
\end{multline}
where
\begin{equation}
 C_{R,N}:=8\,\bE\|\jmath_R^*\Psi^g_{R,N}(0,\Cdot)\|^2_{L_2(\bR^2, v_L^{1/2})}\leq 8\,\bE\|\Psi^g_{R,N}(0,\Cdot)\|^2_{L_2(\bS_R)} < \infty
\end{equation}
for every $R,N\in\bN_+$ and $R\geq L$. The last bound follows from Remark~\ref{rem:norm_stereo} and the fact that $w_R\geq w_L$. From the fact that $\jmath_R^*\varPhi^g_{R,N}$ and $\jmath_R^*Z_{R,N}$ are stationary in time, it follows that
\begin{multline}\label{eq: bound for expectation phi}
\bE\|\jmath_R^*\varPhi^g_{R,N}(0, \Cdot)\|^n_{L^{-\kappa}_n(\bR^2, v_L^{1/n})}
=
\frac{1}{T}\int_0^T \bE\|\jmath_R^*\varPhi^g_{R,N}(t, \Cdot)\|^n_{L^{-\kappa}_n(\bR^2, v_L^{1/n})}\rd t
\\
\leq 
{c}\,\bE\|\jmath_R^*Z_{R,N}(0,\Cdot)\|^n_{L^{-\kappa}_n(\bR^2, v_L^{1/n})}+ \frac{c}{T}\int_0^T\bE\|\jmath_R^*\Psi^g_{R,N}(t, \Cdot)\|^n_{L^{-\kappa}_n(\bR^2, v_L^{1/n})}\rd t\,,
\end{multline}
where $c=2^{n-1}$ and we used the fact that for all $n>1$ and $a_i\in\bR_+,\, i\in\{1,\cdots,k\}$, one has $(\sum_{i=1}^k\,a_i)^n\leq k^{n-1}\sum_{i=1}^k\,a_i^n.$
By Lemma~\ref{lem:stochastic_bound_infinite_volume}, there exists $C_2\in(0,\infty)$ such that for all $R,N\in\bN_+$,
\begin{equation}
 \bE\|\jmath_R^*Z_{R,N}(0,\Cdot)\|^n_{L^{-\kappa}_n(\bR^2, v_L^{1/n})}\leq C_2.
\end{equation}
Combining the bounds proved above yields 
\begin{equation}
 \bE\|\jmath_R^*\varPhi^g_{R,N}(0, \Cdot)\|^n_{L^{-\kappa}_n(\bR^2, v_L^{1/n})}\leq c\,C_1+c\,C_2 + \frac{c\,C_{R,N}}{T}
\end{equation}
for all $T\in(0,\infty)$. Choosing $T=C_{R,N}$ concludes the proof.
\end{proof}

\begin{Remark}\label{rem:tightness}
Let $a,b\in\bR$, $n\geq 1$ and $c_n\in (0,\infty)$. From the inequality
$$
\bigl| |a|^n-|b|^n\bigr| \leq c_n\bigl(|a|^{n-1}+|b|^{n-1}\bigr)|a-b|
$$ it follows that
\begin{multline}
\left|
\|\jmath_R^*X_{R,N}\|_{L_n^{-\kappa}(\bR^2,v_L^{1/n})}^n-\|\jmath_R^*X_R\|_{L_n^{-\kappa}(\bR^2,v_L^{1/n})}^n\right|
\\
\leq c_n\left(\|\jmath_R^*X_{R,N}\|_{L_n^{-\kappa}(\bR^2,v_L^{1/n})}^{n-1} +\|\jmath_R^*X_R\|_{L_n^{-\kappa}(\bR^2,v_L^{1/n})}^{n-1}\right) \|\jmath_R^*(X_{R,N}-X_R)\|_{L_n^{-\kappa}(\bR^2,v_L^{1/n})}.
\end{multline}
It holds that
\begin{multline}
\Bigl|\bE\bigl[
\|\jmath_R^*X_{R,N}\|_{L_n^{-\kappa}(\bR^2,v_L^{1/n})}^n \re^{-Y_{R,N}^g}\bigr]
-\bE\bigl[\|\jmath_R^*X_R\|_{L_n^{-\kappa}(\bR^2,v_L^{1/n})}^n \re^{-Y_{R,N}^g}\bigr]\Bigr| 
\\
\leq
\bE\Bigl[
\bigl|\|\jmath_R^*X_{R,N}\|_{L_n^{-\kappa}(\bR^2,v_L^{1/n})}^n-\|\jmath_R^*X_R\|_{L_n^{-\kappa}(\bR^2,v_L^{1/n})}^n\bigr|
\re^{-Y_{R,N}^g}\Bigr]
\\
\leq c_n\,\bE\Bigl[
\bigl(\|\jmath_R^*X_{R,N}\|_{L_n^{-\kappa}(\bR^2,v_L^{1/n})}^{n-1}
+\|\jmath_R^*X_R\|_{L_n^{-\kappa}(\bR^2,v_L^{1/n})}^{n-1}\bigr)
\|\jmath_R^*(X_{R,N}-X_R)\|_{L_n^{-\kappa}(\bR^2,v_L^{1/n})}
\re^{-Y_{R,N}^g}\Bigr].
\end{multline}
From H\"older's inequality combined with Minkowski's inequality one obtains an upper bound for the last expression 
\begin{multline}\label{bound_monet}
c_n
\Bigl[
\bigl(\bE\|\jmath_R^*X_{R,N}\|_{L_n^{-\kappa}(\bR^2,v_L^{1/n})}^{3(n-1)}\bigr)^{1/3}
+\bigl(\bE\|\jmath_R^*X_R\|_{L_n^{-\kappa}(\bR^2,v_L^{1/n})}^{3(n-1)}\bigr)^{1/3}
\Bigr]
\\
\times
\bigl(\bE\|\jmath_R^*(X_{R,N}-X_R)\|_{L_n^{-\kappa}(\bR^2,v_L^{1/n})}^3\bigr)^{1/3}
\bigl(\bE\re^{-3Y_{R,N}^g}\bigr)^{1/3}.
\end{multline}
By Lemma~\ref{lem:stochastic_bound_infinite_volume}, the first two factors in Eq.~\eqref{bound_monet} are uniformly bounded in the limit $N\to\infty$. It follows from the fact that for $q> n\geq 4$, H\"older's inequality with $\frac{n}{q}+\frac{q-n}{q}=1$ in the weighted measure $v_L(x)\,\rd x$ gives
\begin{equation}\label{eq:embedding_weighted}
\|f\|_{L_n^{-\kappa}(\bR^2,v_L^{1/n})} \leq C \|f\|_{L_q^{-\kappa}(\bR^2, v_L^{1/q})}.
\end{equation}
for some $C\in (0,\infty)$. Then let $q=3(n-1)$ and use Lemma~\ref{lem:stochastic_bound_infinite_volume} with $p=q$. Applying
Jensen's inequality for the third term in Eq.~\eqref{bound_monet}, one obtains
\begin{equation}
\bE\left\|\jmath_R^*(X_{R,N}-X_R)\right\|_{L_n^{-\kappa}(\mathbb R^2,v_L^{1/n})}^{3}\\
\leq
(\bE\|\jmath_R^*(X_{R,N}-X_R)\|_{L_n^{-\kappa}(\mathbb R^2,v_L^{1/n})}^{n})^{3/n},
\end{equation}
which tends to zero as $N$ approaches infinity by the second limit in Lemma~\ref{lem:stochastic_bound_infinite_volume}. By Remark~\ref{rem:prob_of_density} the last factor in Eq.~\eqref{bound_monet} is uniformly bounded. Consequently,
\begin{multline}\label{bound_monet_1}
\bE\bigl[\|\jmath_R^*X_{R,N}\|_{L_n^{-\kappa}(\bR^2,v_L^{1/n})}^n\exp(-Y_{R,N}^g)\bigr]-
\bE\bigl[
\|\jmath_R^*X_R\|_{L_n^{-\kappa}(\bR^2,v_L^{1/n})}^n\exp(-Y_{R,N}^g)\bigr] 
\end{multline}
tends to zero in the limit $N\to\infty$. Furthermore, Lemma~\ref{lem:stochastic_Y} implies that $Y_{R,N}^g\to Y_R^g$ in probability. By Lemma~\ref{lem:stochastic_bound_infinite_volume} combined with Eq.~\eqref{eq:embedding_weighted} and Remark~\ref{rem:prob_of_density} one obtains
\begin{multline}
\sup_{N\in\bN_+}\bE\big[
(\|\jmath_R^*X_R\|_{L_n^{-\kappa}(\bR^2,v_L^{1/n})}^n\exp(-Y_{R,N}^g))^2\big]
\\
\leq
\bigl(\bE\|\jmath_R^*X_R\|_{L_n^{-\kappa}(\bR^2,v_L^{1/n})}^{4n}\bigr)^{1/2}
\sup_{N\in\bN_+}\bigl(\bE\exp(-4Y_{R,N}^g)\bigr)^{1/2} <\infty.
\end{multline}
Hence the sequence $\bigl(\|\jmath_R^*X_R\|^n\exp(-Y_{R,N}^g)\bigr)_{N\in\bN_+}$ is uniformly integrable, see Remark~\ref{rem:uni_inte} and observe that by Minkowski's inequality in $L_n(\Omega,\cF,\bP)$ one has
\begin{multline}
\big(\mathbb E\|\jmath_R^*X_R\|_{L_n^{-\kappa}(\mathbb R^2,v_L^{1/n})}^n\big)^{1/n}
\\
\leq
\big(\mathbb E\|\jmath_R^*(X_R-X_{R,N})\|_{L_n^{-\kappa}(\mathbb R^2,v_L^{1/n})}^n\big)^{1/n}+
\big(\mathbb E\|\jmath_R^*X_{R,N}\|_{L_n^{-\kappa}(\mathbb R^2,v_L^{1/n})}^n\big)^{1/n}.
\end{multline}
Now taking the limit $N\to\infty$, Lemma~\ref{lem:stochastic_bound_infinite_volume} implies
\begin{equation}
\bE\|\jmath_R^*X_R\|_{L_n^{-\kappa}(\bR^2,v_L^{1/n})}^{n}
\leq C.
\end{equation}
Note that the $4n$-moment of $\jmath_R^*X_R$ can be shown to be bounded using Lemma~\ref{lem:stochastic_bound_infinite_volume} combined with Eq.~\eqref{eq:embedding_weighted}. From the Vitali theorem~\ref{thm:Lebesgue-Vitali's theorem} one obtains
\begin{equation}\label{bound_monet_2}
\bE\bigl[
\|\jmath_R^*X_R\|_{L_n^{-\kappa}(\bR^2,v_L^{1/n})}^n\exp(-Y_{R,N}^g)\bigr]
\longrightarrow
\bE\bigl[\|\jmath_R^*X_R\|_{L_n^{-\kappa}(\bR^2,v_L^{1/n})}^n\exp(-Y_R^g)\bigr].
\end{equation}
Combining Eqs.~\eqref{bound_monet_1} and~\eqref{bound_monet_2} results in
\begin{equation}
\bE\bigl[
\|\jmath_R^*X_{R,N}\|_{L_n^{-\kappa}(\bR^2,v_L^{1/n})}^n\exp(-Y_{R,N}^g)\bigr] \longrightarrow
\bE\bigl[\|\jmath_R^*X_R\|_{L_n^{-\kappa}(\bR^2,v_L^{1/n})}^n\exp(-Y_R^g)\bigr].
\end{equation}
From Prop~\ref{prop:uv_limit} with $F=1$ one has
\begin{equation}
\bE[\exp(-Y_{R,N}^g)] \longrightarrow \bE[\exp(-Y_R^g)]>0.
\end{equation}
Consequently,
\begin{equation}
\begin{aligned}
\int
\|\jmath_R^*\phi\|^n_{L_n^{-\kappa}(\bR^2,v_L^{1/n})}
\,\mu_R^g(\rd\phi) =
\lim_{N\to\infty}\int\|\jmath_R^*\phi\|^n_{L_n^{-\kappa}(\bR^2,v_L^{1/n})}\,\mu_{R,N}^g(\rd\phi)
\leq C.
\end{aligned}
\end{equation}
On account of Thm.~\ref{thm:embedding} Item~(C), the embedding ${L^{-\kappa}_n(\bR^2, v_L^{1/n})}\to L^{-2\kappa}_n(\bR^2, v_L^{2/n})$ is compact. As a result, by Lemma~\ref{lem:tight_comapct_embedding}, the sequence of measures $(\jmath_R^*\sharp\mu^g_R)_{R\in\bN_+}$ on $L^{-2\kappa}_n(\bR^2, v_L^{2/n})$ is tight. Hence, by the Prokhorov theorem it admits a weakly convergent subsequence, i.e., there exists some sequence $(R_l)_{l\in\bN_+}$ with $R_l\to\infty$ such that for every bounded and continuous functional $F\,:\,L^{-2\kappa}_n(\bR^2, v_L^{2/n})\to\bR$, it holds that
\begin{equation}
 \mu^g(F) =\lim_{l\to\infty}\jmath_{R_l}^*\sharp\mu^g_{R_l}(F)=\lim_{l\to\infty}\lim_{N\to\infty}\jmath_{R_l}^*\sharp\mu^g_{R_l,N}(F)  = \lim_{l\to\infty}\lim_{N\to\infty}\mu^g_{R_l,N}(F\circ\jmath_{R_l}^*)\,.
\end{equation}
Observe that the uniform bound stated in Prop.~\ref{prop:tightness} does not depend on $g\in C^{\infty}(\bS_R)$. Thus, using the preceding proposition with $g=0$, one infers the existence of a weak limit for the family of the measures $(\jmath_R^*\sharp\mu_{R,N})_{R,N\in\bN_+}$, i.e., $\mu$.
\end{Remark}

\begin{dfn}\label{def:PPhi}
Any weak limit of a subsequence of the sequence of measures $(\jmath^*_R\sharp \mu_{R,N})_{R,N\in\bN_+}$ on $\sS'(\bR^2)$ is called a $P(\Phi)_2$ measure.
\end{dfn}

\begin{Remark}\label{embedding_on_plane}
From H\"older's inequality, one obtains
\begin{equation}
\begin{aligned}
\|\jmath_R^*\phi\|_{L^{-1}_2(\bR^2,v_L^{1/2})}
\leq
\|v_L^{1/2-1/n}\|_{L_{2n/(n-2)}(\bR^2)}\,\|\jmath_R^*\phi\|_{L^{-1}_n(\bR^2,v_L^{1/n})}.
\end{aligned}
\end{equation}
This implies that there exists $C\in (0,\infty)$ such that for all $R\in\bN_+$,
\begin{equation}
\int\,\|\jmath_R^*\phi\|^n_{L^{-1}_2(\bR^2,v_L^{1/2})}\,\mu_R(\rd \phi)\leq \int\,\|\jmath_R^*\phi\|^n_{L^{-1}_n(\bR^2,v_L^{1/n})}\,\mu_R(\rd\phi) <C\,.
\end{equation}
It follows from Prop.~\ref{prop:tightness}, combined with the fact that the norm
$$\|v_L^{1/2-1/n}\|_{L_{2n/(n-2)}(\bR^2)}$$
is bounded, since $n\geq 4$ and $v_L\leq 1$.
\end{Remark}

\chapter{Integrability and Osterwalder--Schrader axioms}\label{ch:Osterwalder-Schrader axioms}
Let $\mu$ be any weak limit of the sequence of measures $(\jmath_R^*\sharp\mu_{R,N})_{R,N\in\bN_+}$ as is determined in Def.~\ref{def:PPhi}. Let $B\subset \sS(\bR^2) $ be a ball with respect to some Schwartz seminorm, and $f \in B$ be a test function. In Prop.~\ref{prop:integrability}, we first verify that $(\phi(f))^n$ is exponentially integrable with respect to $\mu$, where $n\geq 4$ is the degree of the polynomial $P(\phi)$ fixed in Eq.~\eqref{eq:polynomial_form}. Note that exponential integrability of $(\phi(f))^n$ manifests the non-Gaussianity property of the measure $\mu$. Next, we shall verify the Osterwalder--Schrader axioms, i.e., regularity, reflection positivity and Euclidean invariance properties, cf. Def~\ref{def:OS_axioms}. Once we verify the Osterwalder--Schrader axioms, the existence of a local relativistic QFT on Minkowski space-time satisfying the Wightman axioms, will follow from the reconstruction theorem~\cite[Chs. 6,\,19]{GlimmJaffe}. 
\section{Integrability}
\begin{prop}\label{prop:integrability}
There exists a ball $B\subset\sS(\bR^2)$ with respect to some Schwartz seminorm centered at the origin such that for all $f\in B$,
\begin{equation}\label{eq:main_thm_bound}
 \int\exp(\phi(f)^n)\,\mu(\rd\phi) 
 \leq 2\,.
\end{equation}
\end{prop}

\begin{proof}
By Remark~\ref{rem:norm_stereo}, for all $f\in\sS(\bR^2)$ and $R\in\bN_+$, there exists $g_R\in C^\infty(\bS_R)$ such that $w_R\jmath_R^* g_R=f$. Now let $B$ be contained in the ball from the statement of Prop.~\ref{prop:tightness} and suppose that $f\in B$. For arbitrary $\phi\in\sD'(\bS_R)$, one has
\begin{equation}\label{eq:integrability_g_identity}
 \phi(g_R)=(\jmath^*_R\phi)(w_R\jmath^*_Rg_R) = (\jmath^*_R\phi)(f).
\end{equation}  
Then by Lemma~\ref{lem:mu_f_bound}, one obtains
\begin{equation}
 \int \exp\left(\phi(g_R)^n/n\right)\,\mu_{R,N}(\rd\phi) 
 \\
 \leq 
 \exp\left(\frac{1}{n}\int \phi(g_R)^n\,\mu^{g_R}_{R,N}(\rd\phi)\right)\,.
\end{equation}
Lemma~\ref{lem:mu_measure_well_defined} implies that the expression on the LHS is integrable. The identity~\eqref{eq:integrability_g_identity}, H\"older's inequality, and Proposition~\ref{prop:tightness} yield
\begin{multline}
\int \phi(g_R)^n\, \mu_{R,N}^{g_R}(\rd\phi) 
\leq  \hat C\,
\|v_L^{-1/n}w_R\jmath^*_Rg_R\|^n_{L^{\kappa}_{n/(n-1)}(\bR^2)}
\,
\int\|\jmath_R^*\phi\|^n_{L_n^{-\kappa}(\bR^2, v_L^{1/n})}\,
\mu_R^{g_R}(\rd\phi)
\\
\leq\,C\,\, \|v_L^{-1/n}w_R\jmath^*_Rg_R\|^n_{L^{\kappa}_{n/(n-1)}(\bR^2)}
\end{multline}
for some constants $\hat C,C\in(0,\infty)$ independent of $R,N$ and $g_R$. Choosing the ball $B$ so that $\|v_L^{-1/n}f\|^n_{L^{\kappa}_{n/(n-1)}(\bR^2)}\leq n/2C$ for all $f\in B$ by the above inequalities and Proposition~\ref{prop:uv_limit}, we obtain
\begin{equation}
 \int \exp\left((\jmath^*_R\phi)(f)^n/n\right)\,\mu_{R}(\rd\phi)
 =
 \lim_{N\to\infty}\int \exp\left(\phi(g_R)^n/n\right)\,\mu_{R,N}(\rd\phi)
 \leq 2\,.
\end{equation}
This concludes the proof.
\end{proof}
\noindent
Recall that $n\geq 4$ in Eq.~\eqref{eq:main_thm_bound}. Prop.~\ref{prop:integrability} implies that $\mu(\rd\phi)$ is non-Gaussian, since Gaussian measures do not integrate functions growing so fast. Furthermore, one has, for $f\in \sS(\bR^2)$,
\begin{equation}\label{eq:first_moment}
\int \exp(|\phi(f)|)\,\mu(\rd\phi) <\infty
\end{equation}
without the restriction $f\in B$. It follows from the fact that for all $\varepsilon>0$,
$$|\phi(f)| = \varepsilon^{-1}|\phi(\varepsilon f)|
= (n^{-1/n}\varepsilon^{-1})\,(n^{1/n}|\phi(\varepsilon f)|).$$
Apply Young’s inequality for $\frac{1}{p}+\frac{1}{n}=1$ with $a = n^{-1/n}\varepsilon^{-1}$, and $b = n^{1/n}|\phi(\varepsilon f)|$
$$|\phi(f)|
\le \frac{1}{p} \,n^{-p/n}\varepsilon^{-p} + |\phi(\varepsilon f)|^n.$$
For any given $f$ and all sufficiently small $\varepsilon>0$, we have $\varepsilon f\in B$. Note that Eq.~\eqref{eq:first_moment} implies that $\mu$ has moments of all orders, since the convergence of the integrals $\int \exp(\phi(f))\,\mu(\rd\phi)$ ensures that $\int |\phi(f)|^m\mu(\rd\phi)$ is finite for all $m\in\bN_+$~\cite[Prop.~6.1.4]{GlimmJaffe}.

\section{Regularity}\label{sec:regularity}
\begin{prop} There holds the  regularity axiom OS0: There is a Schwartz norm $\|\Cdot\|_{\mathrm{s}}$
and $K,\beta>0$ such that for all $m\in \mathbb{N}$ and $f_1,\ldots, f_m\in \sS(\mathbb{R}^2)$
\begin{equation}
\big|\int \phi(f_1)\ldots \phi(f_m)\, \mu(\rd\phi)\big|\leq (m!)^{\beta}\, K^m\, \prod_{i=1}^m \|f_i\|_{\mathrm{s}}.
\end{equation}
\end{prop}
\begin{proof}
Since $\frac{x^m}{m!} \leq \re^{x}$ for all $x\ge 0$, we obtain for $f_0\in B$
\begin{align}
\int |\phi(f_0)|^m\, \mu(\rd \phi)
&\leq m! \int \exp(|\phi(f_0)|)\, \mu(\rd \phi).
\end{align}
Using Eq.~\eqref{eq:main_thm_bound} and the elementary bound $
\exp(|x|)\leq \re\,\exp(|x|^n),\quad n\geq 1$, it follows that
\begin{equation}\label{eq:regularity_2}
\begin{aligned}
\int |\phi(f_0)|^m\, \mu(\rd \phi)
\leq m!\,\re \int\, \exp(|\phi(f_0)|^n)\, \mu(\rd \phi) \leq m!\,(2\,\re).
\end{aligned}
\end{equation}
Now let $f_{1,0},\cdots,f_{m,0}\in B$. By H\"older's inequality,
\begin{equation}\label{eq:regularity_3}
\begin{aligned}
\left|\int \prod_{i=1}^m \phi(f_{i,0})\, \mu(\rd \phi)\right|
&\leq \prod_{i=1}^m \left(\int |\phi(f_{i,0})|^m\, \mu(\rd \phi)\right)^{1/m}.
\end{aligned}
\end{equation}
From Eq.~\eqref{eq:regularity_2}, we obtain
\begin{equation}\label{eq:regularity_4}
\begin{aligned}
\left|\int \prod_{i=1}^m \phi(f_{i,0})\, \mu(\rd \phi)\right|
\le m!\,(2\,\re).
\end{aligned}
\end{equation}
One extends this estimate to arbitrary non-zero $f_1,\cdots,f_m\in\sS(\bR^2)$. Let
$$
f_{i,0}=\varepsilon\, \frac{f_i}{\|f_i\|_{\mathrm{s}}},
$$
where $\varepsilon>0$ is chosen independently of $f_1,\cdots,f_m$ sufficiently small so that $f_{i,0}\in B$ for all $i$. By linearity of $\phi$
$$
\phi(f_{i,0})=\varepsilon \|f_i\|_{\mathrm{s}}^{-1}\phi(f_i),
$$
and hence
\begin{equation}
\prod_{i=1}^m \phi(f_{i,0})
=
\varepsilon^m\, \left(\prod_{i=1}^m \|f_i\|_{\mathrm{s}}^{-1}\right)
\prod_{i=1}^m \phi(f_i).
\end{equation}
Inserting this into~\eqref{eq:regularity_4} yields
\begin{equation}
\varepsilon^m\, \left(\prod_{i=1}^m \|f_i\|_{\mathrm{s}}^{-1}\right)
\left|\int \prod_{i=1}^m \phi(f_i)\,\mu(\rd\phi)\right|
\le m!\,(2\,\re).
\end{equation}
Rearranging gives
\begin{equation}
\left|\int \prod_{i=1}^m\, \phi(f_i)\,\mu(\rd\phi)\right|
\leq m!\,2\re\,\varepsilon^{-m}\, \prod_{i=1}^m \|f_i\|_{\mathrm{s}}.
\end{equation}
Note that $2\re\leq (2\re)^m$. This completes the proof.
\end{proof}

\section{Reflection positivity}\label{sec:reflection_posivity}
In this section we verify that any accumulation point of the family of measures $(\jmath_R^*\sharp\mu_R)_{R\in\bN_+}$ on $\sS'(\bR^2)$ is reflection positive. To this end, wel utilize the fact that for all $R\in\bN_+$, the Gaussian measure $\nu_R$ with covariance $G_R=(1-\Delta_R)^{-1}$ is reflection positive. We then introduce a UV regularization, which preserves the reflection positivity property of the Gaussian measure $\nu_R$. In general, to construct a nonlinear quantum field, it is more convenient to establish reflection positivity for some measures occurring as intermediate steps in the construction~\cite[p.~196]{GlimmJaffe}. Accordingly, we shall first show that for all $R\in\bN_+$, the finite volume measure $\mu_R$ on $\sD'(\bS_R)$ is reflection positive. In Lemma~\ref{lem:reflection_positivity}, we verify the statement for any weak limit of the family of measures $(\jmath_R^*\sharp\mu_R)_{R\in\bN_+}$ on $\sS'(\bR^2)$. This relies on the fact that reflection positivity is preserved under the weak convergence considered here.

\begin{Remark}
The significance of reflection positivity relies on the fact that it yields the existence of the Hilbert space $H$ of quantum states. In particular, it guarantees that the induced inner product is positive definite~\cite[Prop.~6.1.1]{GlimmJaffe}. Moreover, it gives rise to the existence of a positive-energy Hamiltonian, which enables an analytic continuation from Euclidean space to Minkowski space. 
Reflection positivity also establishes a correspondence between the unitary representations of the Euclidean group on the Euclidean side and the Poincar\'e group on the Minkowski side,~cf.~\cite[Sec.~8.2]{NO18}. 
\end{Remark}
\noindent
By Remark~\ref{rem:fundamental_solution_sphere}, the integral kernel of the covariance $G_R=(1-\Delta_R)^{-1}$ has a singularity at the coinciding points. To regularize $G_R$, we use an approximate Dirac delta defined in terms of the geodesic distance. Similar regularization on $\bR^2$ is used in~\cite[Sec.~8.5]{GlimmJaffe}.

\begin{dfn}\label{def:hat-K-g}
Fix $h\in C^\infty(\bR)$ such that $\supp\,h\subset(-1,1)$, and $h=1$ on $[-1/2,1/2]$ with $2\pi\int_0^\infty \theta\,h(\theta)\,\rd\theta=1$. For $R,N\in\bN_+$, consider the  measurable function $\hat K_{R,N}\,:\, \bS_R\times\bS_R\to \mathbb{R}$ defined by 
\begin{equation}
 \hat K_{R,N}(\rx,\ry):=N^2 h(N\,d_R(\rx,\ry)),
 \end{equation}
which is a smooth function with support in a ball of radius $1/N$. For all $R,N\in\bN_+$, one defines the bounded operator $\hat K_{R,N}\,:\, L_2(\bS_R)\to L_2(\bS_R)$ by
\begin{equation}
(\hat K_{R,N}f)(\rx)=\int_{\bS_R}\hat K_{R,N}(\rx,\ry)\,f(\ry)\rho_R(\rd\ry)\,.
\end{equation}
\end{dfn}

\begin{Remark}
The normalization condition on $h$ ensures, by Schur’s test~\ref{lem:schur-test} that $\hat K_{R,N}$ definea a bounded operator on $L_2(\bS_R)$. In particular,
\begin{equation}
\begin{gathered}
\|\hat K_{R,N}\|\leq \sup_{\rx\in\bS_R}\,\int |\hat K_{R,N}(\rx,\ry)|\,\rho_R(\rd\ry),
\quad
\|\hat K_{R,N}\|\leq \sup_{\ry\in\bS_R}\,\int |\hat K_{R,N}(\rx,\ry)|\,\rho_R(\rd\rx)\leq C\,.
\end{gathered}
\end{equation}
Using a change of variable, i.e., $\theta' =N\,R\,\theta$, $\sin\theta\leq \theta$ and the support property of the function $h$, it follows that
\begin{equation}
\begin{aligned}
\|\hat K_{R,N}\|
\leq \sup_{\rx\in\bS_R}\int_{\bS_R}\,N^2 |h(Nd_R(\rx,\ry))|\,\rho_R(\rd \ry)
\leq
 2\pi\int_0^\infty \theta\,|h(\theta)|\,\rd\theta=1\,.
\end{aligned}
\end{equation}
\end{Remark}

\begin{Remark}\label{rem:funk-hecke}
By the Funk-Hecke formula~\cite[Thm.~2.22]{AH12}, one gets
\begin{equation}
\begin{gathered}
(\hat{K}_{R,N}Y_{l,m})(\rx)
= \int_{\bS_R} N^2h(N\,R\,\theta(\rx,\ry))\,Y_{l,m}(\ry)\,\rho_R(\rd\ry)
= h_l \,Y_{l,m}(\rx),
\end{gathered}
\end{equation}
where
\begin{equation}
h_l= \int_{\bS_R}\,N^2h(N\,R\,\theta(\rx,\ry))\,P_l(\rx\cdot\ry/R^2)\,\rho_R(\rd\ry).
\end{equation}

\end{Remark}

\begin{dfn}\label{def:parametrized_sphere}
For all $R,N\in[1,\infty)$, we set
\begin{equation}
\bS_{R,N}^\pm:=\{(\rx_1,\rx_2,\rx_3)\in\bS_R\,|\,\pm\rx_1>1/N\},
\quad
\bS_{R,N} := \bS_{R,N}^+\,\cup\, \bS_{R,N}^-,
\quad
\bS_R^\pm:=\cup_{N\in[1,\infty)}\, \bS_{R,N}^\pm\,.
\end{equation}
\end{dfn}

\begin{dfn}\label{dfn:cylindircal_functionals}
Let $R\in[1,\infty)$. A functional $F\,:\sD'(\bS_R) \to \bC$ is called cylindrical if and only if there exists $k\in\bN_+$, $G\in C^\infty_{\mathrm{b}}(\bR^k)$ and $f_l\in  C^\infty_\rc(\bS_R):= C^\infty(\bS_R)$, $l\in\{1,\ldots,k\}$, such that
\begin{equation}\label{eq:cylindircal}
 F(\phi) = G(\phi(f_1),\ldots,\phi(f_k)).
\end{equation}
The algebra of cylindrical functions is denoted by $\cF_R$. The subalgebras of $\cF_R$ consisting of functionals of the form~\eqref{eq:cylindircal} with $\supp\,f_l\subset \bS_R^\pm$, $l\in\{1,\ldots,k\}$, or such that $\supp\,f_l\subset \bS_{R,N}^\pm$, $l\in\{1,\ldots,k\}$, are denoted by $\cF^\pm_R$ and $\cF^\pm_{R,N}$, respectively. The definitions of $\cF$ and $\cF^\pm$ are analogous to the definitions of $\cF_R$ and $\cF_R^\pm$ with $\bS_R$ and $\bS_R^\pm$ replaced by $\bR^2$ and the half-plane $\{(x_1,x_2)\in\bR^2|\pm x_1>0\}$, respectively. 
\end{dfn}

\begin{dfn}\label{def:invo_sphere}
Let $\Theta_R:\bS_R\to\bS_R$ be the reflection defined by
$$\Theta_R(\rx_1,\rx_2,\rx_3):=(-\rx_1,\rx_2,\rx_3).$$
Observe that $\Theta_R\in O(3)$ and $\det\Theta_R=-1$. For $f\in C^\infty(\bS_R)$, define $\Theta_R f\in C^\infty(\bS_R)$ by $(\Theta_R f)(\rx):=f(\Theta_R\rx).$ For $\phi\in\sD'(\bS_R)$, define $\Theta_R\phi\in\sD'(\bS_R)$ by
$$\langle\Theta_R\phi,f\rangle:=\langle\phi,\Theta_R f\rangle$$
for all $f\in C^\infty(\bS_R)$. Observe that $\Theta_R^2=\mathmybb{1}$, i.e., $\Theta_R$ is an involution.
\end{dfn}

\begin{dfn}\label{def:invo_plane}
Let $\Theta:\bR^2\to\bR^2$ be the reflection defined by
$$\Theta(x_1,x_2):=(-x_1,x_2).$$
Observe that $\Theta\in O(2)$ and $\det\Theta=-1$. For $f\in C^\infty(\bR^2)$, define $\Theta f\in C^\infty(\bR^2)$ by $(\Theta f)(x):=f(\Theta x).$ For $\phi\in\sS'(\bR^2)$, define $\Theta\phi\in\sS'(\bR^2)$ by
$$\langle\Theta\phi,f\rangle:=\langle\phi,\Theta f\rangle$$
for all $f\in\sS(\bR^2)$. Observe that $\Theta^2=\mathmybb{1}$, i.e., $\Theta$ is an involution.
\end{dfn}

\begin{Remark}\label{rem:Z2group}
It holds that $\jmath_R^*\circ \Theta_R = \Theta\circ \jmath_R^*$ on $\sD'(\bS_R)$. Equivalently, for all $\phi\in\sD'(\bS_R)$,
$$\jmath_R^*(\Theta_R\phi)=\Theta(\jmath_R^*\phi).$$
\end{Remark}

\begin{prop}\label{prop:reflection_positivity}
Let $\mu$ be a weak limit of a subsequence of the sequence of measures $(\jmath_R^*\sharp\mu_R)_{R\in\bN_+}$ on $\sS'(\bR^2)$. For all $F\in\cF^+$ it holds
$
\int \overline{F(\Theta\phi)} F(\phi)\,\mu(\rd\phi) \geq 0
$.
\end{prop}
\begin{Remark}
We call this property the Glimm-Jaffe (GJ) reflection positivity, as it is similar to the variant stated in~\cite{GlimmJaffe}. We will show how to obtain the OS reflection positivity from it (cf.~Def.~\ref{def:OS_axioms}) at the end of this subsection. 
\end{Remark}
\begin{proof}
It is enough to prove that 
\begin{equation}
\begin{aligned}
 \int \overline{F(\Theta\phi)} F(\phi)\,(\jmath_R^*\sharp\mu_R)(\rd\phi) 
& =
 \int \overline{F(\Theta \jmath_R^*\phi)} F(\jmath_R^*\phi)\,\mu_R(\rd\phi)
 \\
&=
 \int \overline{F(\jmath_R^*\Theta_R\phi)} F(\jmath_R^*\phi)\,\mu_R(\rd\phi)
 \geq 0
\end{aligned}
\end{equation} 
for all $R\in\bN_+$ and $F\in\cF^+$. By Def.~\ref{dfn:cylindircal_functionals} and Remark~\ref{rem:norm_stereo}, for every $F\in\cF^+$ it holds $F\circ\jmath_R^*\in\cF^+_R$. Hence, the last bound above follows from the reflection positivity of the measure $\mu_R$, which is given in Lemma~\ref{lem:reflection_positivity}~Item~(D).
\end{proof}
\noindent
We introduce a regularized free field $\hat X_{R,N}$ using the bounded operator $\hat K_{R,N}$ such that its law enjoys reflection positivity, see Lemma~\ref{lem:reflection_positivity}~(B). Next, we show that the measure $\mu_R$ can be approximated using the corresponding regularized measure associated with $\hat X_{R,N}$, see Lemma~\ref{lem:reflection_positivity}~(C).

\begin{dfn}\label{dfn:hat_X_hat_Y}
Let $\hat c_{R,N}:=\mathrm{Tr}(\hat K_{R,N}G_R\hat K_{R,N})/4\pi R^2$. We define $\hat X_{R,N}:=\hat K_{R,N}X_R$,
\begin{equation}
\begin{gathered}
 \hat X_{R,N}^{:m:}:=\sum_{k=0}^{\floor{m/2}} \frac{(-1)^k m!}{(m-2k)!k!2^k} (\hat c_{R,N})^k \hat X_{R,N}^{m-2k},
 \qquad
 \hat X_{R,N}^{:m:}(h):=\int_{\bS_R} \hat X_{R,N}^{:m:}(\rx)\,h(\rx)\,\rho_R(\rd \rx),
 \\
 \hat Y_{R,N}:=\sum_{m=0}^n a_m \hat X_{R,N}^{:m:}(\mathmybb{1}_{\bS_R}),
 \qquad
 \tilde Y_{R,N}^\pm:=\sum_{m=0}^n a_m \hat X_{R,N}^{:m:}(\mathmybb{1}_{\bS^\pm_{R,N}}),
 \qquad
 \tilde Y_{R,N} = \tilde Y^+_{R,N} + \tilde Y^-_{R,N},
\end{gathered} 
\end{equation}
where $h\in L_\infty(\bS_R)$ and $\mathmybb{1}_B$ denotes the characteristic function of the set $B$.
\end{dfn}

\begin{Remark}
Note that $\hat X_{R,N}$ introduced above and $X_{R,N}=K_{R,N}X_R$ introduced in Sec.~\ref{sec:existence_UV_limit} are free fields on $\bS_R$ with different UV cutoffs. We use the same symbol $N\in\bN_+$ to denote both cutoffs.
\end{Remark}

\begin{Remark}
By Lemma~\ref{lem:stochastic_X_not_uniform} Item~(B) and Lemma~\ref{lem:embedding_sphere}, it holds $\hat X_{R,N}\in L^1_2(\bS_R)\subset L_n(\bS_R)$ almost surely. In particular, $\hat Y_{R,N}$, $\tilde Y_{R,N}$ are well-defined. The same conclusion follows from Remark~\ref{rem:two_regularization} and Remark~\ref{rem:hat_c}, combined with the fact that $Y_{R,N}$ given in Eq.~\eqref{dfn:X_Y} is well-defined. Moreover, there exists $C\in(0,\infty)$ such that for all $N,R\in\bN_+$, it holds that
$$|\hat c_{R,N}-\frac{1}{2\pi}\,\log N|\leq C$$
by the bound~\eqref{eq:bound_counterterm} and Remark~\ref{rem:hat_c}.
\end{Remark}

\begin{Remark}\label{rem:K_G_K}
Remark~\ref{rem:funk-hecke} implies that the operators $\hat K_{R,N}$ and $(-\Delta_R)$ both have the same set of eigenfunctions, i.e., they are simultaneously diagonalizable in the spherical harmonics basis. Now using Remark~\ref{rem:commute} Item~(A) and Item~(B) with
$$T=(-\Delta_R),\quad F(T)=(1-\Delta_R/N^2)^{-1}=K_{R,N},\quad \text{and}\quad S=\hat K_{R,N},$$
it follows that for all $R,N\in\bN_+$, the operators $\hat K_{R,N}$ and $K_{R,N}$ commute, i.e., for all $R,N\in\bN_+$, the bounded operator $\hat K_{R,N}$ commutes with every orthogonal projection $(2l+1)\cP_{R,l}$ for all $l\in\bN_+$, see Remark~\ref{rem:commute} Item~(C)\,.
\end{Remark}

\begin{Remark}
Assume that $A,B$ are commuting bounded self-adjoint operators, i.e., $A^*=A$ and $B^*=B$. By the joint spectral theorem, there exists a common measure space and a unitary transformation under which $A$ and $B$ are realized as multiplication operators by real-valued functions $a$ and $b$ on a common $L_2$-space. Hence, $
|a+b|\le |a|+|b|$ almost everywhere, which implies the operator inequality 
$$
\langle f, |A+B| f\rangle \le \langle f, (|A|+|B|) f\rangle,
\quad \forall f\in L_2,
$$
where $|A|=(A^*A)^{1/2}$.
\end{Remark}

\begin{Remark}\label{rem:hat_c}
Recall that $K_{R,N}=(1-\Delta_R/N^2)^{-1}$, $G_R=(1-\Delta_R)^{-1}$ and the counterterms $c_{R,N}$, $\hat c_{R,N}$ were introduced in Eq.~\eqref{eq:counterterm} and Def.~\ref{dfn:hat_X_hat_Y}, respectively. Note that by Remark~\ref{rem:K_G_K}, the operators $G_R,K_{R,N},\hat K_{R,N}$ commute. 
Using Lemma~\ref{lem:K_hat} and the fact that for commuting self-adjoint operators $A,B$, one has
$$|A+B|\leq |A|+|B| \quad \text{and}\quad |A-B|=|A-1-(B-1)|\leq |A-1|+|B-1|,$$
we obtain
\begin{equation}
 |\hat K_{R,N}^2-K_{R,N}^2| \leq |\hat K_{R,N}-K_{R,N}|\,|\hat K_{R,N}+K_{R,N}| \leq 2\,C\,(C+1)\,\frac{(1-\Delta_R)/N^2}{(1-\Delta_R/N^2)^2}\,,
\end{equation}
where we used the fact that
\begin{equation}
|1-K_{R,N}|+|1-\hat K_{R,N}|\leq (C+1)\,\frac{(1-\Delta_R)/N^2}{1-\Delta_R/N^2}\,.
\end{equation}
Observe that  
\begin{equation}
\frac{(1-\Delta_R)/N^2}{(1-\Delta_R/N^2)^2} G_R= N^{-2}\frac{1}{(1-\Delta_R/N^2)^2}= 
\frac{(1-\Delta_R)}{N^{2}}\frac{1}{(1-\Delta_R/N^2)^2} \frac{1}{(1-\Delta_R)}
\end{equation}
and decompose $(1-\Delta_R)/N^2= (1-\Delta_R/N^2)-1+1/N^2$. Consequently, it holds that
\begin{multline}
 |\hat c_{R,N}-c_{R,N}|
 \leq
 \mathrm{Tr}(|\hat K_{R,N}^2-K_{R,N}^2|\,G_R)/4\pi R^2
\\
\leq
\frac{2C(C+1)}{4\pi R^2}\,
 \bigg[\mathrm{Tr}((1-\Delta_R/N^2)^{-1} (1-\Delta_R)^{-1})
 - \mathrm{Tr}((1-\Delta_R/N^2)^{-2} (1-\Delta_R)^{-1})
 \\
 +
 \mathrm{Tr}((1-\Delta_R/N^2)^{-2} (1-\Delta_R)^{-1})/N^2\bigg].
\end{multline}
By Lemma~\ref{lem:trace}, the RHS of the last inequality above is bounded by a constant independent of $R,N\in\bN_+$. Note that all inequalities are understood in the joint spectral representation of commuting self-adjoint operators.
\end{Remark}

\begin{lem}\label{lem:X_Y_covergence_tilde}
The sequence
$(F(\hat X_{R,N}) \exp(-\,\tilde Y_{R,N}))_{N\in\bN_+}$ converges in probability to $F(X_R) \exp(-\,Y_R)$.
\end{lem}
\begin{proof}
It follows from Lemma~\ref{lem:stochastic_X_uniform} that for all $\kappa\in(0,\infty)$ and $\delta\in[0,2]$, one has
\begin{equation}
\bP\big(\|X_R- \hat X_{R,N}\|^2_{L_2^{-\kappa-\delta}(\bS_R)}\geq \epsilon\big)
\leq\,\epsilon^{-1}\,\bE\|X_R- \hat X_{R,N}\|_{L_2^{-\kappa-\delta}(\bS_R)}^2\leq \epsilon^{-1}\,R^2\,C^2\,N^{-2\delta}
\end{equation}
and by Lemma~\ref{lem:stochastic_Y} Item~(C)~and~Item~(D), one obtains
\begin{equation}
\begin{gathered}
\bP\big( |Y_R-\tilde{Y}_{R,N}|^2\geq \epsilon\big) 
\leq
\bP\big( |Y_R-\hat{Y}_{R,N}|^2 \geq \epsilon/2 \big) + \mathbb{P}\big( |\hat{Y}_{R,N} -\tilde{Y}_{R,N}|^2\geq \epsilon/2 \big)
\\
\leq \,\epsilon^{-1}\,C^2\,\big(N^{-1/n}+N^{-1}\big)\,.
\end{gathered}
\end{equation}
Using the above bounds, one shows that the sequence $(\hat X_{R,N})_{N\in\bN_+}$ and $(\tilde Y_{R,N})_{N\in\bN_+}$ converge in probability as $N$ approaches infinity to the random variables $X_R$ and $Y_R$. One concludes the statement using the fact that $F\,:\,\sD'(\bS_R)\to\bR$ and $t\mapsto\exp(t)$ are continuous. Hence, by Remark~\ref{rem:convergence} the sequences $(F(X_{R,N}))_{N\in\bN_+}$ and $(\exp(-\,\,\tilde Y_{R,N}))_{N\in\bN_+}$ converge in probability to $F(X_R)$ and $\exp(-Y_R)$. This finishes the proof.
\end{proof}

\begin{lem}\label{lem:UV_convergence_tilde}
For all $R\in\bN_+$, and all bounded and continuous $F\,:\,\sD'(\bS_R)\to \bR$, it holds that
\begin{equation}
 \lim_{N\to\infty} \bE F(\hat X_{R,N})\, \exp(-\tilde Y_{R,N}) = \lim_{N\to\infty} \bE F(X_{R,N})\, \exp(-Y_{R,N}). 
\end{equation}
\end{lem}
\begin{proof}
The proof follows from the same strategy as in the proof of the proposition~\ref{prop:uv_limit}.
By Lemma~\ref{lem:X_Y_covergence_tilde}, the sequences
$$(F(\hat X_{R,N}) \exp(-\tilde Y_{R,N}))_{N\in\bN_+}\quad\text{and}\quad (F(X_{R,N}) \exp(-Y_{R,N}))_{N\in\bN_+}$$
converge in probability to $F(X_R) \exp(-Y_R)$. To conclude we show that the aforementioned sequences are uniformly integrable by repeating verbatim the argument from the proof of Proposition~\ref{prop:uv_limit}.
\end{proof}

\begin{lem}\label{lem:reflection_positivity}
The following statements hold true for all $R,N\in\bN_+$:
\begin{enumerate}
 \item[(A)] If $F\in\cF_R^+$, then $\bE \overline{F(\Theta_R X_R)} F(X_R)\geq 0$.
 \item[(B)] If $F\in\cF_{R,N}^+$, then $\bE \overline{F(\Theta_R \hat X_{R,N})} F(\hat X_{R,N})\geq 0$.
 \item[(C)] If $F\in\cF_{R,N}^+$, then $\bE \overline{F(\Theta_R \hat X_{R,N})} F(\hat X_{R,N})\exp(-\tilde Y_{R,N})\geq 0$.
 \item[(D)] For all $F\in\cF^+_R$ it holds
$
\int \overline{F(\Theta_{R}\phi)} F(\phi)\,\mu_R(\rd\phi) \geq 0
$.
\end{enumerate}
\end{lem}
\begin{proof}
For the proof of Item~(A) see~\cite[Thm.~2]{dimock2004} or~\cite[Thm.~8.3]{And13}. To prove Item~(B), denote by $\rx^-$ the reflection of $\rx\in \bS_R$ corresponding to the action of $\Theta_R$. Observe that $d_R(\rx^-,\ry) = d_R(\rx, \ry^-)$, which implies that $\hat K_{R,N}(\rx^-,\ry)= \hat K_{R,N}(\rx,\ry^-)$. Hence, by Remark~\ref{rem:commute}, one deduces that $[\Theta_R,\hat K_{R,N}]=0$, which yields
$$\Theta_R \hat K_{R,N} X_R = \hat K_{R,N} \Theta_R X_R.$$
By assumption $F\in\cF^+_{R,N}$. Observe that for all $f\in C^\infty(\bS_R)$ it holds
$$(\hat K_{R,N}X_R)(f) = X_R(\hat K_{R,N}f),$$
which implies that
\begin{equation}\label{eq:cylindircal_2}
 F(\hat K_{R,N}X_R) = G(X_R(\hat K_{R,N}f_1),\ldots,X_R(\hat K_{R,N}f_k)).
\end{equation}
Utilizing the fact that  $f$ is supported in $\bS_{R,N}^+$ and the integral kernel of $K_{R,N}$ is supported in a $1/N$-neighbourhood of the diagonal, $\hat K_{R,N}f$ belongs to $\cF_R^+$. In fact, positivity is preserved because the support restriction on $\hat K_{R,N}$ ensures that convolution maps test functions supported in $\bS_{R,N}^{+}$ to test functions supported in $\bS_R^{+}$. Consequently, the statement follows from Item~(A). To prove Item~(C), note that
\begin{equation}
\bE \overline{F(\Theta_R \hat X_{R,N})} F(\hat X_{R,N})\exp(-\tilde Y_{R,N})
=
\bE \overline{F(\Theta_R \hat X_{R,N})\exp(-\tilde Y^-_{R,N})}F(\hat X_{R,N})\exp(-\tilde Y^+_{R,N})\,.
\end{equation}
Denote $H(\hat X_{R,N}) := F(\hat X_{R,N})\,\exp(-\tilde Y_{R,N}^+)$. Using $\tilde{Y}^{-}_{R,N}(X_{R,N}) = \tilde{Y}^{+}_{R,N}(\Theta_R X_{R,N})$, one infers
\begin{equation}
\bE \overline{F(\Theta_R \hat X_{R,N})} \,F(\hat X_{R,N})\,\exp(-\tilde Y_{R,N})=\bE\overline{ H(\Theta_R\hat X_{R,N})} \,H(\hat X_{R,N})\,.
\end{equation}
Let us now show that (C) follows from (A). Denote $H(X_{R}):=F(\hK X_{R}) \exp(-\ti{Y}^+_{R,N})$. We want to approximate $H$ by $H_k\in \mcF_R^+$
in the topology of $L_2(\sD'(\bS_R),\nu_R)$. We recall the definition of $\ti{Y}^+_{R,N}$:
\beqa
\ti{Y}^+_{R,N}=\int \, P((\hK X_{R}(\rx)), \hat{c}_{R,N}) \one_{\bS^+_{R,N}}(\rx)\,\rho_R(\rd\rx).
\eeqa
By continuity of $\hK X_{R,N}$ a.e. in $\om \in \Om$,  we can approximate the integral by a Riemann sum
\beqa
(\ti{Y}^+_{R,N})_k:=\sum_{\rx_j\in \mrm{Part}_k} (\De \rx)_j\, P((\hK X_{R}(\rx_j)), \hat{c}_{R,N}) \,\one_{\bS^+_{R,N}}(\rx_j),
\eeqa
which converges pointwise in $\om$. Partitions of the sphere can be obtained, e.g., in spherical coordinates, in which case
$(\De x)_j=\sin(\theta_i)\De \theta_i \De\varphi_i$. Exploiting again the convergence of the Riemann sums, we
can assume   $\sum_{\rx_j\in \mrm{Part}_k} (\De \rx)_j \leq 2\,\mrm{vol}(\bS_R)$, for sufficiently fine partitions. 
Then, by Lemma~\ref{lem:polynomial_bound},
\beqa
(\ti{Y}^+_{R,N})_k \2 \geq \2 
- \sum_{x_j\in \mrm{Part}_k} (\De \rx)_j\,A\,c^{n/2}   \one_{\bS^+_{R,N}}(\rx_j)
 \geq  - A\,c^{n/2}2\, \mrm{vol}(\bS_R). 
\eeqa
Hence
\begin{equation} \label{bound-DC}
\exp(- (\ti{Y}^+_{R,N})_k) \leq \exp( A\,c^{n/2}2 \mathrm{Vol}(\bS_R) ).
\end{equation}
We note that $H_{0,k}$ is cylindrical due to smearing of  $X_R$ with $\hK$.
Now we want to show that the cylindrical functions 
\beqa
H_k(X_{R}):= F(\hK X_{R})\exp(-(\ti{Y}^+_{R,N})_k)
\eeqa
approximate $H$ in  $L_2(\sD'(\bS_R),\nu_R)$, i.e.,
\begin{equation}\label{convergence-to-zero} 
\lim_{k\to \infty}\int\, |H(\phi) -H_k(\phi)|^2\, \nu_R(\rd\phi)=0.
\end{equation}
We note that by Eq.~\eqref{bound-DC}
\beqa
|H_k(\phi)|:= |F(\hat K \phi )| \exp(-H_{0,k}(\phi)) \leq C,
\eeqa
uniformly in $k$, thus dominated convergence gives~\eqref{convergence-to-zero}. This concludes the proof of (C). Let us turn to the proof of Item~(D). First note that for any $F\in\cF^+_R$ there exists $M\in\bN_+$ such that $F\in\cF^+_{R,M}$. Hence, it suffices to show that $\int \overline{F(\Theta_R\phi)} F(\phi)\,\mu_R(\rd\phi) \geq 0$ for all $R,M\in\bN_+$ and $F\in\cF^+_{R,M}$. To establish this claim we note that by Lemma~\ref{lem:UV_convergence_tilde}
\begin{equation}\label{eq:positivity_mu_R_approximation}
\int \overline{F(\Theta_R\phi)} F(\phi)\,\mu_R(\rd\phi)
=
\lim_{N\to\infty} \frac{\bE\overline{F(\Theta_R \hat X_{R,N})} F(\hat X_{R,N})\exp(-\tilde Y_{R,N})}{\bE \exp(-\tilde Y_{R,N}) }
\end{equation}
and use Item~(C) together with the fact that $\cF^+_{R,M}\subset\cF^+_{R,N}$ for all $N\geq M$.
\end{proof}
\subsubsection*{The Glimm-Jaffe reflection positivity}
Recall that cylindrical functions $F \in \mathcal{F}$ have the form 
\begin{equation}
F(\phi)= G(\phi(f_1), \ldots, \phi(f_m)),
\end{equation}
where $f_1, \ldots, f_m\in \sS(\mathbb{R}^2)$ and $G$ is a smooth, bounded, complex-valued function. 
We have the following
\begin{equation}
 \int \overline{F(\Theta\phi)} F(\phi)\,(\jmath_R^*\sharp\mu_R)(\rd\phi) \geq 0.
\end{equation}
We consider the Schwinger functions $S_m\in \sS'(\bR^{2m})$. We aim to verify the OS
reflection positivity which we recall here for convenience:
\begin{itemize}
\item[OS2:]  Let $\mathbb{R}^{dm}_+:=\{ (x^{(1)}, \ldots, x^{(m)})\in \mathbb{R}^{dm}\, : x_0^{(j)}>0,\, j=1, \ldots, m   \}$ and 
\begin{equation}
\sS(\mathbb{R}_+^{dm}; \mathbb{C}):=\{\, f\in \sS(\mathbb{R}^{dm}; \mathbb{C})\,:\, \mathrm{supp}(f)\subset  \mathbb{R}^{dm}_+ \}.
\end{equation}
For all sequences $f^{(m)}\in \sS(\mathbb{R}^{dm}_+; \mathbb{C})$, $m\in \mathbb{N}_0$, with finitely many non-zero elements
\begin{equation}
\sum_{\ell,m \in  \mathbb{N}_0} S_{\ell+m}( \overline{\Theta f^{(\ell)}} \otimes f^{(m)})\geq 0.
\end{equation}
\end{itemize}

\begin{lem}\label{lem:reflection_density}
It suffices to show reflection positivity for $f^{(m)}_{*}\in \sS(\mathbb{R}^{dm}_+; \mathbb{C})$ of the form
\begin{equation}
f_*^{(m)}=\sum_{\alpha\in \mathcal{I}_m, |\alpha|=m} f^{\alpha}, \label{eq:factorization}
\end{equation}
where  $f^{\alpha}(x^{(1)}, \ldots, x^{(m)} )=f_{\alpha_1}(x^{(1)}) \ldots  f_{\alpha_m}(x^{(m)})$,  $f_{\alpha_i}\in \sS(\mathbb{R}^{d}_+)$
and $\mathcal{I}_m$ is a finite set of multiindices.
\end{lem}

\begin{proof} By standard density arguments, given any $f^{(m)} \in \sS(\mathbb{R}_+^{dm}; \mathbb{C})$ we can find a sequence of
 $f^{(m)}_{*,\epsilon}$, $\epsilon>0$, of functions of the form~\eqref{eq:factorization}, possibly without the support restriction, such that
\begin{equation}
f^{(m)}=\lim_{\epsilon\to 0} f^{(m)}_{*,\epsilon} \label{eq:Schwartz-approximation}
\end{equation}
in the Schwartz topology. To ensure the support property of the functions on the RHS, we choose, for our given $f^{(m)}$, a bounded smooth function $\eta$ supported in $\mathbb{R}_+^{d}$, such that 
\begin{equation}
(\eta(x^{(1)})\ldots \eta(x^{(m)}) )f^{(m)}(x^{(1)}, \ldots, x^{(m)})= f^{(m)}(x^{(1)}, \ldots, x^{(m)} ).
\end{equation}
By multiplying both sides of Eq.~\eqref{eq:Schwartz-approximation} with such products of functions $\eta$ and making use of the fact 
that this operation is continuous in the Schwartz topology we ensure that the RHS of Eq.~\eqref{eq:Schwartz-approximation} consists of
functions from $\sS(\bR^{dm}_+; \bC)$. Using that the Schwinger functions are tempered distributions and positivity
is preserved in the limit, we obtain the claim.
\end{proof}
\noindent
Now we define $\phi(f)^{\alpha}=\phi(f_{\alpha_1}) \ldots \phi(f_{\alpha_{|\alpha|} })$. We also set
\begin{equation}
\phi(f)_{\varepsilon}^{\alpha}:=\phi(f_{\alpha_1}) \re^{-\varepsilon \phi(f_{\alpha_1})^2} \ldots \phi(f_{\alpha_{|\alpha|} }) \re^{-\varepsilon \phi(f_{\alpha_m})^2}
\end{equation}
and note that this is a cylindrical function. We write
\begin{align}
&\sum_{\ell, m \in  \mathbb{N}_0} S_{\ell+m}( \overline{\Theta f_*^{(\ell)}} \otimes f_*^{(m)})\\
&=\int \overline{ 
\sum_{\ell}\sum_{\alpha\in \mathcal{I}_\ell, |\alpha|=\ell} \phi(\Theta f)^{\alpha_\ell}}  \sum_{m} \sum_{\beta\in \mathcal{I}_m, |\beta|=m}    \phi(f)^{\beta_m}\, \mu(\rd\phi) \\
&=\lim_{\varepsilon\to 0 }\int  \overline{ \sum_{\ell}\sum_{\alpha\in \mathcal{I}_\ell, |\alpha|=\ell} \phi(\Theta f)^{\alpha_\ell}_{\varepsilon}  } 
\sum_{m} \sum_{\beta\in \mathcal{I}_m, |\beta|=m}  \phi(f)^{\beta_m}_{\varepsilon} \,\mu(\rd\phi)\geq 0, 
\end{align}
where we made use of dominated convergence, the existence of the moments of the measure and the GJ-reflection positivity to conclude
the OS reflection positivity.


\section{Euclidean invariance}\label{sec:Euclidean} 
In this section we demonstrate that every accumulation point $\mu$ of a subsequence of the sequence of measures $(\jmath_R^*\sharp \mu_R)_{R\in\bN_+}$ on $\sS'(\bR^2)$ is invariant under the action of  the Euclidean group $T(2)\rtimes O(2)$.
To this end, we appeal to the fact that the symmetry groups of the plane and the sphere have the same dimension, see Remark~\ref{rem:euclidean_R_two}. This suggests a correspondence between the corresponding group actions. We further use that for all $R\in\bN_+$, the measure $\mu_R$ is invariant under the action of the orthogonal group $O(3)$, which follows from the invariance of the Gaussian measure $\nu_{R}$ with covariance $G_R=(1-\Delta_R)^{-1}$, see Lemma~\ref{rem:sphere_O(3)}.
\\
\\
Observe that rotations of $\bS_R$ around the $\rx_3$-axis fixing the north pole $N = (0,0,R)$ stabilize the $x_1x_2$-plane and form a copy of $SO(2)\subset SO(3)$, see Remark~\ref{rem:sphere_quotient}. Consequently, these rotations are mapped by the stereographic projection to rotations of the plane by arbitrary angles. This yields rotational invariance of $(\jmath_R^*\sharp \mu_R)_{R\in\bN+}$ on $\sS'(\bR^2)$. The translational invariance holds only in an asymptotic sense: it emerges in the limit $R\to\infty$ when one considers rotations around $\rx_i$, $i=1,2$, with angles proportional to $1/R$, see Remark~\ref{rem:sphere_mobius_plan}. Furthermore, reflection invariance follows from Remark~\ref{rem:O(3)_splitting} using the involutions $\Theta_R$ and $\Theta$ defined in Definitions~\ref{def:invo_sphere}~and~\ref{def:invo_plane}, respectively.

\begin{Remark}
The proof of the Euclidean invariance is independent of the value of the coupling constant $\lambda\in (0,\infty)$ and holds true for any positive coupling constant. 
\end{Remark}
\noindent
In the remainder of this section, we first establish a one-to-one correspondence between the action of $SO(3)$ on $\bS_R$ and the action of $T(2)\rtimes SO(2)$ on $\bR^2$, and then show that weak limit points of the sequence $(\jmath^*_R\sharp\mu_R)_{R\in\bN_+}$ on $\sS'(\bR^2)$ are invariant under the action of $T(2)\rtimes SO(2)$. We next establish the reflection invariance property of these weak limit points. Combining the preceding results yields the Euclidean invariance of weak limit points of $(\jmath^*_R\sharp\mu_R)_{R\in\bN_+}$.

\begin{dfn}
For $R\in\bN_+$, $\alpha\in\bR$ and $\rx=(\rx_1,\rx_2,\rx_3)\in\bS_R$, we specify the actions of the generators of $SO(3)$ on $\bS_R$ by the maps 
$$\cR_{R,\alpha},\cT_{R,\alpha},\cT'_{R,\alpha}\,:\,\bS_R\to\bS_R$$
defined by
\begin{equation}
\begin{gathered}
 \cR_{R,\alpha}(\rx) 
 = (\rx_1\cos\alpha+\rx_2\sin\alpha,\rx_1\sin\alpha-\rx_2\cos\alpha,\rx_3),
 \\
 \cT_{R,\alpha}(\rx) 
 = (\rx_1\cos(\alpha/R)-\rx_3\sin(\alpha/R),\rx_2,\rx_1\sin(\alpha/R)+\rx_3\cos(\alpha/R)),
\\
\cT'_{R,\alpha}(\rx) 
= (\rx_1,\rx_2\cos(\alpha/R)-\rx_3\sin(\alpha/R),\rx_2\sin(\alpha/R)+\rx_3\cos(\alpha/R))\,.
\end{gathered} 
\end{equation}
\end{dfn}
\begin{Remark}
Any rotation $Q\in SO(3)$ can be represented using the generators of $SO(3)$,~cf.~\cite[Lemma~A.2.2]{DX13}.
\end{Remark}

\begin{dfn}\label{def:def_Q_sphere}
Let $Q_{R,\alpha}\in\{ \cR_{R,\alpha},  \cT_{R,\alpha},  \cT'_{R,\alpha}\}$ be defined by $Q_{R,\alpha}\,:\, \bS_R\to\bS_R$. Note that $Q_{R,\alpha}$ acts orthogonally, transitively and in an orientation-preserving manner on $\bS_R$, and give rise to a group of unitary operators on $L_2(\bS_R)$. Let $Q^*_{R,\alpha}\in\{\cR^*_{R,\alpha},  \cT^*_{R,\alpha},  {\cT'}^*_{R,\alpha}\}$ be maps $Q^*_{R,\alpha}\,:\, \sD'(\bS_R)\to\sD'(\bS_R)$ defined as follows: first, for all $f\in C^\infty(\bS_R)$ 
 we set $Q_{R,\alpha}^*f:=f\circ Q_{R,\alpha}\in C^\infty(\bS_R)$. Next, for $\phi\in\sD'(\bS_R)$ we set $Q_{R,\alpha}^*\phi:=\phi\circ Q_{R,-\alpha}^*\in\sD'(\bS_R).$ It holds that
 $$\langle Q_{R,\alpha}^*\phi,f\rangle=\langle \phi,Q_{R,-\alpha}^*f\rangle_ =\langle\phi,f\circ Q_{R,-\alpha}\rangle,$$ which is due to the rotational invariance of the Riemannian volume form $\rho_R(\rd y)$. 
\end{dfn}

\begin{Remark}\label{rem:det_SO(3)}
Let $Q_{R,\alpha}\in\{ \cR_{R,\alpha},  \cT_{R,\alpha},  \cT'_{R,\alpha}\}$ be as defined in Def.~\ref{def:def_Q_sphere}. It holds that
$$\mathrm{T}_p\,Q_{R,\alpha}\,:\, \mathrm{T}_p\,\bS_R\to\mathrm{T}_{Q_{R,\alpha}(p)}\,\bS_R\quad\text{ for\,\, any}\quad p\in \bS_R.$$
Observe that $\mathrm{T}_p\,\bS_R=p^{\perp}$ and $\mathrm{T}_{Q_{R,\alpha}(p)}\,\bS_R=(Q_{R,\alpha}(p))^{\perp}$ both have dimension two, i.e., they are isomorphic to $\bR^2$ as vector spaces. This is due to the fact that $Q_{R,\alpha}$ is a diffeomorphism and it preserves angles and lengths and its differential is an isomorphism so $\det(\mathrm{T}_p\,Q_{R,\alpha})\neq 0$~\footnote{Taking the derivative of $Q^{-1}_{Q(p)}\circ Q(p)=\mathrm{Id}_{\bS_R}$ at $p$, one obtains $\mathrm{T}_{Q(p)}\,Q^{-1}\circ \mathrm{T}_p\,Q=\mathrm{Id}_{\mathrm{T}_p\,\bS_R}$. Hence, $\det(\mathrm{T}_{Q(p)}\, Q^{-1})\,\det(\mathrm{T}_p\,Q)=\det(\mathrm{Id}_{\mathrm{T}_p\,\bS_R})=1$.}. In particular, since $Q_{R,\alpha}$ preserves orientation on $\bS_R$, its differential map $\mathrm{T}_p\,Q_{R,\alpha}$ acts as orientation-preserving isometry of the corresponding tangent spaces such that $\mathrm{T}_p\,Q_{R,\alpha}\in SO(2)$ with $\det(\mathrm{T}_p\,Q_{R,\alpha})=1$ for all $p\in\bS_R$, see Remark~\ref{rem:tangent_lei_algebra}.
\end{Remark}

\begin{Remark}\label{rem:cylindrical}
The algebra of cylindrical functionals $\cF$ separates points in
$$L^{-1}_2(\bR^2,v_L^{1/2})\subset \sD'(\bR^2).$$
Hence, if $\mu_j$, $j=1,2$, are Borel probability measures on $L^{-1}_2(\bR^2,v_L^{1/2})$ such that $\mu_1(F)=\mu_2(F)$ for all $F\in\cF$, then $\mu_1=\mu_2$ by~\cite[Thm.~4.5(a), Ch. 3]{ethier2005}. The same statement holds if $\mu_j$, $j=1,2$ are Borel probability measures on $L^{1}_2(\bR^2,v_L^{1/2})\subset \sD'(\bR^2)$.
\end{Remark}

\begin{lem}\label{lem:rotational_invariant_measure}
Let $Q^*_{R,\alpha}\in\{ \cR^*_{R,\alpha},  \cT^*_{R,\alpha},  {\cT'}^*_{R,\alpha}\}$ denote the maps $Q_{R,\alpha}^*\,:\,\sD'(\bS_R)\to \sD'(\bS_R)$ defined as above. For all $R,N\in\bN_+$,
\begin{itemize}
\item[(A)]
The regularized Gaussian measure $\nu_{R,N}$ has rotational invariance.

\item[(B)]
The measure $\mu_R$ is rotationally invariant, i.e.,
$$\mu_{R}(F) = \mu_{R}(F\circ Q_{R,\alpha}^*)$$
for all bounded and measurable $F\,:\,\sD'(\bS_R)\to\bR$ and all $\alpha\in\bR$.
\end{itemize}
\end{lem}

\begin{proof}
To prove Item~(A), it suffices to show that for all $R\in\bN_+$, the covariance $G_{R,N}$ is rotationally invariant and refer to the Bochner-Minlos theorem (Thm.~\ref{thm:Minlos}). Utilizing Remark~\ref{rem:expansion_covariance_harmo}, it follows that for all $R,N\in\bN_+$,
$$G_{R,N}(\rx,\ry)= G_{R,N}(Q_{R,\alpha}\rx,Q_{R,\alpha}\ry)= \int_{\sD'(\bS_R)}\phi(\rx)\,\phi(\ry)\,\nu_{R,N}(\rd\phi)\quad \text{for\,\, all}\quad \rx,\ry\in\bS_R.$$
Thus, from the preceding relation for all $f,g\in\sD(\bS_R)$, one obtains
\begin{equation}
G_{R,N}(f,g)=\int_{\bS_R\times\bS_R}\,G_{R,N}(\rx,\ry)\,f(\rx)\,g(\ry)\,\rho_R(\rd\rx)\,\rho_R(\rd\ry)
=G_{R,N}(Q^*_{R,-\alpha}f,Q^*_{R,-\alpha}g),
\end{equation}
where we used the facts that for all $R\in\bN_+$, the Riemannian volume form $\rho_R$ is rotationally invariant, $Q_{R,\alpha}$ acts transitively on $\bS_R$ and Def.~\ref{def:def_Q_sphere}. This finishes the proof of Item~(A). To prove Item~(B), recall that for all $R,N\in\bN_+$, the measure $\nu_{R,N}$ is concentrated on $L^1_2(\bS_R)$, see Lemma~\ref{lem:mu_measure_well_defined}. From Remark~\ref{rem:cylindrical}, it suffices to prove that for all $\alpha\in\bR$ and all cylindrical functionals $F\in\cF_{R,N}$, it holds that $\mu_{R,N}(F\,\circ Q^*_{R,\alpha})=\mu_{R,N}(F)$. One obtains
\begin{equation}
\begin{aligned}
\int_{L^1_2(\bS_R)}\, F(Q^*_{R,\alpha}\phi)\,\mu_{R,N}(\rd\phi)
&
=\frac{1}{\mathcal{Z}_{R,N}}\,\int_{L^1_2(\bS_R)}\, F(Q^*_{R,\alpha}\phi)\,\cU_{R,N}(Q^*_{R,\alpha}\phi)\,\nu_{R,N}(\rd\phi)
\\
&
=\frac{1}{\mathcal{Z}_{R,N}}\,\int_{L^1_2(\bS_R)}\, F(\phi)\,\cU_{R,N}(\phi)\,(Q^*_{R,\alpha}\sharp\nu_{R,N})(\rd\phi)
\\
&
=\int_{L^1_2(\bS_R)}\, F(\phi)\,\mu_{R,N}(\rd\phi),
\end{aligned}
\end{equation}
where in the third line, we used Item~(A), i.e., $(Q^*_{R,\alpha}\sharp\nu_{R,N})(A)=\nu_{R,N}(A)$ for all $A\in\mathrm{Borel}(L^1_2(\bS_R))$. Taking the limit $N\to\infty$ concludes the lemma.
\end{proof}

\begin{dfn}
For $\alpha\in\bR$ and $x:=(x_1,x_2)\in\bR^2$, we define the actions of the generators of the group $T(2)\rtimes SO(2)$ on $\bR^2$ by the maps $\cR_{\alpha},\cT_{\alpha},\, \cT'_{\alpha}\,:\,\bR^2\to\bR^2$ given by
\begin{equation}
\begin{gathered}
 \cR_{\alpha}(x_1,x_2) := (x_1\cos\alpha +x_2\sin\alpha, x_1\sin\alpha-x_2\cos\alpha),
 \\
 \cT_{\alpha}(x_1,x_2) := (x_1+\alpha, x_2),\quad
 \cT'_{\alpha}(x_1,x_2) := (x_1, x_2+\alpha).
\end{gathered}
\end{equation}
\end{dfn}

\begin{dfn}
Let $Q_{\alpha}\in\{ \cR_{\alpha},  \cT_{\alpha},  \cT'_{\alpha}\}$ and  $Q^*_{\alpha}\in\{ \cR^*_{\alpha},  \cT^*_{\alpha},  {\cT'}^*_{\alpha}\}$ defined by $Q_{\alpha}\,:\,\bR^2\to\bR^2$ and by $Q^*_{\alpha}\,:\,\sD'(\bR^2)\to\sD'(\bR^2)$, respectively. For $f\in C_{\rc}^\infty(\bR)$ and $\phi\in\sD'(\bR^2)$ we set
\begin{equation}
Q_{\alpha}^*f:=f\circ Q_{\alpha}\in C^\infty_{\rc}(\bR^2),
 \quad
Q_{\alpha}^*\phi:=\phi\circ Q_{-\alpha}^*\in\sD'(\bR^2).
\end{equation}
This implies that
$$\langle Q_{\alpha}^*\phi,f\rangle =\langle \phi,Q_{-\alpha}^*f\rangle=\langle\phi,f\circ Q_{-\alpha}\rangle,$$
which follows from the invariance of the Lebesgue measure under Euclidean transformations.
\end{dfn}

\begin{dfn}
Let $R\in\bN_+$, $\alpha\in (-R,R)$ and consider the ball $B_R:=\{x\in\bR^2\,|\,|x|< R\}\subset \bR^2$. We define the map $\cS_{R,\alpha}\,:\,B_R\to \bR^2$ by
\begin{equation}
 \cS_{R,\alpha}(x_1,x_2):= 
 \frac{2(R\sin (\alpha/R) (1-(x_1^2+x_2^2)/4R^2)+x_1 \cos(\alpha/R),x_2)}{1+\cos(\alpha/R)+ (x_1^2+x_2^2)/4R^2\,(1-\cos(\alpha/R)) -(x_1/R)\,\sin(\alpha/R)}\,.
\end{equation}
Note that for all $x=(x_1,x_2)\in B_R$ formally, as $R\to \infty$, one recovers
$$\lim_{R\to \infty}\,\cS_{R,\alpha}(x)= (x_1+\alpha, x_2)= \cT_{\alpha}(x_1,x_2).$$
For fixed $\alpha\in\bR$, it follows from L'H\^opital's rule 
\begin{equation}
\begin{gathered}
\lim_{R\to\infty}\,2R\sin(\alpha/R)\Big(1 -\frac{(x_1^2 +x_2^2)}{4R^2} \Big) = \lim_{R\to \infty}\,2R\sin(\alpha/R) = 2\alpha\,.
\\
\lim_{R\to\infty}\,(1+\cos(\alpha/R)) +(\frac{x_1^2+x_2^2}{4R^2})(1-\cos(\alpha/R)) - \frac{x_1}{R}\sin(\alpha/R) = 2\,.
\end{gathered}
\end{equation}
\end{dfn}

\begin{lem}\label{rem:rot_trans}
\begin{itemize}
\item[(A)]
For all $R\in\bN_+$ and $\alpha\in (-R,R)$ on the ball $B_R$, one has
$$\cT_{R,\alpha}\circ \jmath_R = \jmath_R\circ \cS_{R,\alpha}.$$
\item[(B)]For all $R\in\bN_+$ and $\alpha\in\bR$, it holds that
$$\mathcal{R}_{R,\alpha}\circ \jmath_R= \jmath_R\circ \mathcal{R}_{\alpha}.$$
\end{itemize} 
\end{lem}
\begin{proof}
Both items can be verified using direct calculations and are similar. Let us prove Item~(A). For all $x=(x_1,x_2)\in B_R\subset\bR^2$, one obtains
$$\mathcal{T}_{R,\alpha}\circ\jmath_R(x_1,x_2) =(\rx_1',\rx_2',\rx_3'),$$
where
\begin{equation}
\begin{aligned}
\rx_1':=& \frac{4R^2x_1}{4R^2+x_1^2+x_2^2}\cos(\alpha/R) - \frac{R(x_1^2+x_2^2-4R^2)}{4R^2+x_1^2+x_2^2}\sin(\alpha/R)\,,
\\
\rx_2':=& \frac{4R^2\,x_2}{4R^2+x_1^2+x_2^2}\,,
\\
\rx_3':=& \frac{4R^2x_1}{4R^2+x_1^2+x_2^2}\sin(\alpha/R) + \frac{R(x_1^2+x_2^2-4R^2)}{4R^2+x_1^2+x_2^2}\cos(\alpha/R)\,.
\end{aligned}
\end{equation}
Using the map $\jmath^{-1}_R(\rx_1,\rx_2,\rx_3)= 2R(\rx_1,\rx_2)/(R-\rx_3)$ given in Def.~\ref{def:stereoraphic}, one writes
\begin{equation}
\begin{aligned}
R-\rx_3^{\prime}
= &\,4R^2\,\frac{R}{4R^2+x_1^2+x_2^2}\,\Big(1+\cos(\alpha/R) +(\frac{x_1^2+x_2^2}{4R^2})(1-\cos(\alpha/R)) - \frac{x_1}{R}\sin(\alpha/R) \Big)
\\
2R\,\rx_1'=&\, 2R\,\frac{4R^2}{4R^2+x_1^2+x_2^2} \Bigg( x_1\cos(\alpha/R) +R\sin(\alpha/R)\Big(1 -\frac{(x_1^2 +x_2^2)}{4R^2}   \Big)  \Bigg)\,.
\end{aligned}
\end{equation}
Combining the obtained results, one verifies the identity
$$\cS_{R,\alpha}(x)= \jmath_R^{-1}\circ\mathcal{T}_{R,\alpha}\circ\jmath_R(x)\,\quad \forall \,x\in B_R.$$
This finishes the proof of item~(A). In a similar manner, one could prove Item~(B).
\end{proof}

\begin{Remark}\label{rem:pull-back-rotations}
Lemma~\ref{rem:rot_trans} indicates that the stereographic projection $\jmath_R$ intertwines the actions of the generators on $\bS_R$ and $\bR^2$. By the general property that pullbacks reverse composition, this immediately implies that the induced actions on test functions and distributions also intertwine, namely
$$\cS_{R,\alpha}^*\circ \jmath_R^*= \jmath_R^*\circ \cT_{R,\alpha}^*\quad \text{and} \quad \cR_{\alpha}^*\circ \jmath_R^*= \jmath_R^*\circ \cR_{R,\alpha}^*.$$
\end{Remark}

\begin{Remark}\label{rem:change_of_variable_pairing}
Let $U,V$ be two non-empty open sets in $\bR^2$ and $F\,:\,U\to V$ be a $C^{\infty}-$ diffeomorphism and $\varphi\in\sD'(V)$. We define $\varphi\circ F\in \sD'(U)$ by
\begin{equation}
\langle \varphi\circ F,f\rangle := \int_U\varphi(F(x))\,f(x)\,\rd x=\int_V\varphi(y)\,f(F^{-1}(y))\,|\det(\mathrm{T}(F^{-1}))|\,\rd y\,,
\end{equation}
where $\det(\mathrm{T} F^{-1})$ denote the determinant of the Jacobian of the map $F^{-1}$. The informal expression $\varphi(F(x))$ 
in the first step serves only to motivate the second step. The expression on the RHS is well defined provided
$$f(F^{-1}(\Cdot))\,|\det(\mathrm{T}(F^{-1}))|$$ is a test-function supported in $V$.
\end{Remark}

\begin{Remark}\label{rem:tran_plane_1}
As $R\to\infty$, by the property of the determinant function combined with Remark~\ref{rem:det_SO(3)}, we expect that
$$\det(\mathrm{T}\cS_{R,\alpha}) =\det(\mathrm{T} (\jmath_R^{-1}\circ\cT_{R,\alpha}\circ \jmath_R))\to 1.$$
Let $R\in\bN_+$, $\alpha\in (-R,R)$. For $f\in C^\infty_\rc(\bR^2)$ and all $\phi\in\sD'(\bR^2)$ with $\supp\,\phi\subset B_R$, we set
$$\cS_{R,\alpha}^* f:= f\circ \cS_{R,\alpha}\in C^\infty(B_R)\quad \text{and} \quad \cS_{R,\alpha}^*\phi= \phi\circ \cS_{R,-\alpha}^*\in\sD'(\bR^2).$$
Utilizing Remark~\ref{rem:change_of_variable_pairing}, one has
$$\langle\cS_{R,\alpha}^*\phi,f\rangle:=
 \langle\phi,\det(\mathrm T\cS_{R,-\alpha})\,\cS_{R,-\alpha}^*f\rangle \quad \forall\, f\in C^\infty_\rc(\bR^2).$$
\end{Remark}

\begin{Remark}\label{rem:translation}
Let $\alpha\in\bR$, $a\in\bN_0^2$ and $M\in(0,\infty)$. By direct calculations one verifies that there exists $C\in(0,\infty)$ such that for all $x\in B_M$ and $R\in(|\alpha|\vee M,\infty)$ it holds 
\begin{equation}
 \cT_{R,\alpha}\circ \jmath_R(x) = \jmath_R\circ \cS_{R,\alpha}(x)
 \qquad
 \textrm{and}
 \qquad
 |\partial^a\cS_{R,\alpha}(x)-\partial^a\cT_\alpha(x)| \leq C/R\,.
\end{equation}
On account of Remark~\ref{rem:sphere_mobius_plan}, a concise proof of the above bound is given in the Appendix~\ref{sec:translation}.
\end{Remark}

\begin{Remark}\label{rem:determinant_Q}
Let $Q$ be an $n\times n$ positive definite matrix such that $Q=I+E$, where $\|E\|<1$ sufficiently small. Then, $I+E$ is invertible and the matrix logarithm is well-defined by its power series,
$$\log(I+E)=\sum_{p=1}^{\infty}\frac{(-1)^{p-1}}{p}\,E^p = E-\frac{1}{2}\, E^2+O(\|E\|^3).$$
Taking the trace yields
$$\mathrm{Tr}\,\log(I+E)
=
\mathrm{Tr}(E)-\frac12\mathrm{Tr}(E^2)+O(\|E\|^3).$$
Using the identity
$$\det(I+E)=\exp(\mathrm{Tr}\,\log(I+E)),$$
which holds for all invertible matrices, together with the expansion
$$\exp(x)=1+x+O(x^2),
\qquad x\to 0,$$
one obtains
$$\det(I+E)=1+\mathrm{Tr}(E)+O(\|E\|^2).$$
In particular, there exists a constant $C_n>0$ such that
\begin{equation}
|\det(I+E)-1| \leq C_n\,\|E\|
\end{equation}
for $\|E\|$ sufficiently small. More generally, 
$$\det(I+\varepsilon E)=1+\varepsilon\,\mathrm{Tr}\,(E)+O(\varepsilon^2),\quad \text{as}\quad \varepsilon\to 0.$$
\end{Remark}

\begin{Remark}\label{rem:map_transaltion}
By definition, one has
$$\cS_{R,-\alpha}^*\,f(x)-\cT_{-\alpha}^*\,f(x)= f(\cS_{R,-\alpha}(x))- f(\cT_{-\alpha}(x)).$$
Moreover, from the Taylor expansion, one obtains
\begin{equation}\label{eq:taylor_1}
f(\cS_{R,-\alpha}(x))= f(\cT_{-\alpha}(x)) +\vec{\nabla} f(\cT_{-\alpha}(x))^{\mathrm{T}}\times (\cS_{R,-\alpha}(x)-\cT_{-\alpha}(x)) + R(x),
\end{equation}
where $R(x)$ is the higher order remainder, i.e.,
\begin{equation}
R(x) = \frac{1}{2} (\cS_{R,-\alpha}(x)-\cT_{-\alpha}(x))^{\mathrm{T}}\, H_f\,(\cS_{R,-\alpha}(x)-\cT_{-\alpha}(x)),
\end{equation}
where $(H_f)_{i,j}=\partial^2 f/\partial x_i\partial x_j$ is the Hessian matrix of $f$ evaluated at some point between $\cS_{R,-\alpha}(x)$ 
and $\cT_{-\alpha}(x)$.  By Remark~\ref{rem:translation}, it holds that
\begin{equation}
\begin{gathered}
|R(x)|\leq  C^2/2R^2\,\sup_{x\in B_M }\,|H_f(x)|,
\\
|\vec{\nabla} f(\cT_{-\alpha}(x))^{\mathrm{T}}\,(\cS_{R,-\alpha}(x)-\cT_{-\alpha}(x))\,|\leq C/R\,\sup_{x\in B_M }\,|\vec{\nabla}f(x)|\,.
\end{gathered}
\end{equation}
Combining the above results, one rewrites Eq.~\eqref{eq:taylor_1}
\begin{equation}\label{eq:taylor_2}
|f(\cS_{R,-\alpha}(x))-f(\cT_{-\alpha}(x))|
\leq\,C/R\,\sup_{x\in B_M }\,\left(|\vec{\nabla}f(x)| +C/2R\, |H_f(x)| \right).
\end{equation}
Noting that $\det\,(\mathrm{T}\,\cT_{-\alpha})=1$, now from Remark~\ref{rem:determinant_Q} and the second bound in Remark~\ref{rem:translation} with $a=1$, one has 
$$|\det(\mathrm{T}\cS_{R,-\alpha}) - 1)|\leq C/R.$$
To see this, we use $\mathrm{T} T_{-\alpha}=1$ to write
\begin{equation}
\det(\mathrm{T}\cS_{R,-\alpha})=\det(1+ \mathrm{T}(\cS_{R,-\alpha} - T_{-\alpha}))
\end{equation}
and observe that the tangent map $\mathrm{T}$ amounts to first derivatives of $\cS_{R,-\alpha} - T_{-\alpha}$.
Hence, for all $\alpha\in\bR$ and $f\in C^\infty_\rc(\bR^2)$, one infers that there exists $C\in (0,\infty)$ such that for all sufficiently large $R\in\bN_+$,
\begin{multline}
\|\det(\mathrm{T}\cS_{R,-\alpha}) \cS_{R,-\alpha}^* f-\cT^*_{-\alpha} f\|_{L^1_2(\bR^2,v_L^{-1/2})} 
\\
\leq
\|(\det(\mathrm{T}\cS_{R,-\alpha}) - 1) f\circ \cS_{R,-\alpha})\|_{L^1_2(\bR^2,v_L^{-1/2})}
+\| f\circ\cS_{R,-\alpha}-f\circ\cT_{-\alpha}\|_{L^1_2(\bR^2,v_L^{-1/2})} 
\\
\leq C/R\, \bigg(\|f\,\circ\,\cS_{R,-\alpha}\|_{L^1_2(\bR^2,v_L^{-1/2})}
+\|\vec{\nabla}\,f|\|_{L^1_2(\bR^2,v_L^{-1/2})} +(C/2R)\,\|H_f\|_{L^1_2(\bR^2,v_L^{-1/2})}   \bigg).
\end{multline}
Hence, for all $\phi\in\sD'(\bR^2)$ with $\supp\,\phi\subset B_R$ and all $f\in C^{\infty}_{\rc}(\bR^2)$, one has
\begin{equation}
\begin{aligned}
| \langle\cS_{R,\alpha}^*\phi,f\rangle - \langle\cT_{\alpha}^*\phi,f\rangle| 
&\leq \|\det(\mathrm T\cS_{R,-\alpha})\cS_{R,-\alpha}^*f-\cT_{-\alpha}^*f\|_{L^1_2(\bR^2,v_L^{-1/2})}\,\|\phi\|_{L^{-1}_2(\bR^2,v_L^{1/2})}
\\&
\leq\,C/R\, \|\phi\|_{L^{-1}_2(\bR^2,v_L^{1/2})}\,.
\end{aligned}
\end{equation}
\end{Remark}

\begin{prop}\label{prop:euclidean_inv_plane}
Let $\mu$ be a weak limit of a subsequence of the sequence of measures $(\jmath_R^*\sharp\mu_R)_{R\in\bN_+}$ on $\sS'(\bR^2)$. It holds that
$$\mu(F) = \mu(F\circ\cR_\alpha^*)\quad \text{and} \quad \mu(F) = \mu(F\circ\cT_\alpha^*)$$
for all bounded and measurable $F\,:\,\sS'(\bR^2)\to\bR$ and all $\alpha\in\bR$.
\end{prop}

\begin{proof}
Suppose that the sequence of measures $\jmath_R^*\sharp\mu_R$ on $\sS'(\bR^2)$ converges to $\mu$ along the subsequence $(R_M)_{M\in\bN_+}$. By Remarks~\ref{rem:tightness} and~\ref{embedding_on_plane}, the measure $\mu$ is concentrated on $L^{-1}_2(\bR^2,v_L^{1/2})$.
For every $F\in\cF$ we obtain 
\begin{multline}
 \mu(F\circ \cR^*_\alpha-F) =
 \lim_{M\to\infty} 
 \jmath_{R_M}^*\sharp\mu_{R_M}(F\circ \cR^*_\alpha-F) =
 \lim_{M\to\infty} 
 \mu_{R_M}(F\circ \cR^*_\alpha\circ \jmath_{R_M}^* -F\circ\jmath_{R_M}^*)
 \\
 =
 \lim_{M\to\infty} 
 \mu_{R_M}(F\circ\jmath_{R_M}^*\circ \cR_{R_M,\alpha}^*-F\circ\jmath_{R_M}^*) = 
 \lim_{M\to\infty} 
 (\jmath_{R_M}^*\sharp\mu_{R_M}-\mu)(F)
 =
 0.
\end{multline}
The result follows from Lemma~\ref{lem:rotational_invariant_measure} and Remark~\ref{rem:pull-back-rotations}. By Remark~\ref{rem:cylindrical}, it holds for arbitrary cylindrical functionals, which completes the proof of rotational invariance.
\\[0.3cm]
Let us turn to the proof of the translational invariance. Note that by Remark~\ref{rem:pull-back-rotations} for every $F\in\cF$ and all sufficiently large $R\in\bN_+$, one has
\begin{equation}
 \jmath_R^*\sharp\mu_R(F) 
 =
 \mu_R(F\circ\jmath_R^*) 
 =
 \mu_R(F\circ\jmath_R^*\circ\cT_{R,\alpha}^*)
 =
 \mu_R(F\circ\cS_{R,\alpha}^*\circ\jmath_R^*)
 =
 \jmath_R^*\sharp\mu_R(F\circ\cS_{R,\alpha}^*).
\end{equation}
Using Remark~\ref{rem:map_transaltion}, one deduces that for every $\alpha\in\bR$ and $F\in\cF$, there exists $C\in(0,\infty)$ such that for all $\psi\in L^{-1}_2(\bR^2,v_L^{1/2})$ and all sufficiently large $R\in\bN_+$, it holds that
\begin{equation}
 |F(\cS_{R,\alpha}^*\psi)-F(\cT_\alpha^*\psi)| \leq (C/R)~\|\psi\|_{L^{-1}_2(\bR^2,v_L^{1/2})}.
 \end{equation}
By Remark~\ref{embedding_on_plane} for all $\alpha\in\bR$ and $F\in\cF$, one has
\begin{equation}
 \lim_{R\to\infty} \jmath_R^*\sharp\mu_R(F\circ\cS_{R,\alpha}^*-F\circ\cT_\alpha^*) =0.
\end{equation}
Consequently,
\begin{multline}
 \mu(F\circ\cT^*_\alpha-F) 
 =
 \lim_{M\to\infty} 
 \jmath_{R_M}^*\sharp\mu_{R_M}(F\circ\cT^*_\alpha-F)
 \\
 =
 \lim_{M\to\infty} 
 \jmath_{R_M}^*\sharp\mu_{R_M}(F\circ\cS_{R_M,\alpha}^*-F)
 =
 \lim_{M\to\infty} 
 (\jmath_{R_M}^*\sharp\mu_{R_M}-\mu)(F)
 =
 0.
\end{multline}
This finishes the proof.
\end{proof}

\begin{lem}\label{lem:rotational_invariant_measure_reflection}
Let $\Theta_R\,:\,\sD'(\bS_R)\to\sD'(\bS_R)$ be defined as in Def.~\ref{def:invo_sphere}. For all $R,N\in\bN_+$
\begin{itemize}
\item[(A)]
The regularized Gaussian measure $\nu_{R,N}$ is $\Theta_R$-invariant.

\item[(B)]
The measure $\mu_R$ is $\Theta_R$-invariant, i.e., 
$$\mu_{R}(F) = \mu_{R}(F\circ \Theta_R)$$
for all bounded and measurable $F\,:\,\sD'(\bS_R)\to\bR$.
\end{itemize}
\end{lem}

\begin{proof}
The proof follows from the same strategy as in the proof of the lemma.~\ref{lem:rotational_invariant_measure}. To prove Item~(A), it suffices to show that for all $R,N\in\bN_+$, the covariance $G_{R,N}$ is $\Theta_R$-invariant and to invoke the Bochner-Minlos theorem (Thm.~\ref{thm:Minlos}). Utilizing Remark~\ref{rem:expansion_covariance_harmo}, it follows that for all $R,N\in\bN_+$, and for all $\rx,\ry\in\bS_R$,
$$G_{R,N}(\rx,\ry)=
G_{R,N}(\Theta_{R}\rx,\Theta_{R}\ry).$$
Thus, for all $f,g\in\sD(\bS_R)$, one obtains
\begin{equation}
G_{R,N}(f,g)=\int_{\bS_R\times\bS_R}\,G_{R,N}(\rx,\ry)\,f(\rx)\,g(\ry)\,\rho_R(\rd\rx)\,\rho_R(\rd\ry)
=G_{R,N}(\Theta_{R}f,\Theta_{R}g).
\end{equation}
It follows from the fact that $\Theta_R^2=\1$, combined with 
\begin{equation}
(\Theta_R\times \Theta_R)(\rho_R(\rd \rx)\otimes\rho_R(\rd\ry))= (\rho_R(\rd \rx))\otimes(\rho_R(\rd \ry)),
\end{equation}
where the change in orientation cancels out. This finishes the proof of Item~(A). To prove Item~(B), recall that for all $R,N\in\bN_+$, the measure $\nu_{R,N}$ is concentrated on $L^1_2(\bS_R)$. Thus, by Remark~\ref{rem:cylindrical}, it suffices to prove that for all cylindrical functionals $F\in\cF_{R,N}$, one has
$$\mu_{R,N}(F\,\circ \Theta_{R})=\mu_{R,N}(F).$$
It holds that
\begin{equation}
\begin{aligned}
\int_{L^1_2(\bS_R)}\, F(\Theta_{R}\phi)\,\mu_{R,N}(\rd\phi)
&=\frac{1}{\mathcal{Z}_{R,N}}\,\int_{L^1_2(\bS_R)}\, F(\phi)\,\cU_{R,N}(\phi)\,(\Theta_{R}\sharp\nu_{R,N})(\rd\phi)
\\&=\int_{L^1_2(\bS_R)}\, F(\phi)\,\mu_{R,N}(\rd\phi),
\end{aligned}
\end{equation}
where we used the fact that $\cU_{R,N}(\Theta_{R}\,\phi)=\cU_{R,N}(\phi)$, combined with Item~(A), i.e., $(\Theta_{R}\sharp\nu_{R,N})(A)=\nu_{R,N}(A)$ for all $A\in\mathrm{Borel}(L^1_2(\bS_R))$. Taking the limit $N\to\infty$, we conclude the proof.
\end{proof}

\begin{Remark}\label{rem:sphere_O(3)}
Combining Lemma~\ref{lem:rotational_invariant_measure} with Lemma~\ref{lem:rotational_invariant_measure_reflection}, one infers the invariance of the measure $\mu_R$ on $\sD'(\bS_R)$ under the action of the orthogonal group $O(3)$. This follows by an argument analogous to that used below the proof of Prop.~\ref{prop:euclidean_inv_plane_reflection}, see also Remark~\ref{rem:O(3)_splitting}.
\end{Remark}

\begin{prop}\label{prop:euclidean_inv_plane_reflection}
Let $\mu$ be a weak limit of a subsequence of the sequence of measures $(\jmath_R^*\sharp\mu_R)_{R\in\bN_+}$ on $\sS'(\bR^2)$. It holds that $\mu(F) = \mu(F\circ\Theta)$ and for all
bounded and measurable $F\,:\,\sS'(\bR^2)\to\bR$.
\end{prop}

\begin{proof}
Suppose that the sequence of measures $\jmath_R^*\sharp\mu_R$ on $\sS'(\bR^2)$ converges to $\mu$ along the subsequence $(R_M)_{M\in\bN_+}$. By Remarks~\ref{rem:tightness} and~\ref{embedding_on_plane}, the measure $\mu$ is concentrated on $L^{-1}_2(\bR^2,v_L^{1/2})$. For every $F\in\cF$, we obtain 
\begin{multline}
 \mu(F\circ \Theta-F) =
 \lim_{M\to\infty} 
 \jmath_{R_M}^*\sharp\mu_{R_M}(F\circ \Theta-F) =
 \lim_{M\to\infty} 
 \mu_{R_M}(F\circ \Theta\circ \jmath_{R_M}^* -F\circ\jmath_{R_M}^*)
 \\
 =
 \lim_{M\to\infty} 
 \mu_{R_M}(F\circ\jmath_{R_M}^*\circ \Theta_R-F\circ\jmath_{R_M}^*) = 
 \lim_{M\to\infty} 
 (\jmath_{R_M}^*\sharp\mu_{R_M}-\mu)(F)
 =
 0,
\end{multline}
where in the second line, we first used Remark~\ref{rem:Z2group}, then Lemma~\ref{lem:rotational_invariant_measure_reflection}. By Remark~\ref{rem:cylindrical}, it holds for arbitrary cylindrical functionals, which completes the proof of the reflection invariance.
\end{proof}
\noindent
Let us now deduce the invariance of any weak limit point of the sequence of measures $(\jmath_R^*\sharp\mu_R)_{R\in\bN_+}$ on $\sS'(\bR^2)$ under the action of an arbitrary element of $E(2)=T(2)\rtimes O(2)$:
On account of Remark~\ref{rem:O(3)_splitting}, any element $g\in O(2)$ either belongs to $SO(2)$ or satisfies $\mathrm{det}(g)=-1$. 
In the former case the invariance follows directly from Prop.~\ref{prop:euclidean_inv_plane}. In the latter case
we write $g=\Theta (\Theta g)$, so that $\Theta g\in SO(2)$. Then, the invariance follows by combining  Prop.~\ref{prop:euclidean_inv_plane} with Prop.~\ref{prop:euclidean_inv_plane_reflection}.
The invariance under reflections is a special property of scalar polynomial theories.
\noindent
\\[0.3cm]
Euclidean invariance of the Schwinger functions, as required by the axiom OS3, can now be checked by similar arguments as in the last part of Sec.~\ref{sec:reflection_posivity}.

\appendix
\chapter*{Appendices}

\addcontentsline{toc}{chapter}{Appendices}
\renewcommand{\thesection}{\Alph{section}}

\section{Bessel potentials}\label{sec:bessel}
The modified Bessel function satisfies~\cite[Eq.~(14.97)]{Arfken} 
\begin{equation}
\rho^2\frac{\rd^2}{\rd \rho^2}\,P_{m}(k\rho)+\rho\frac{\rd}{\rd\rho}\,P_m(k\rho)-(k^2\rho^2+m^2)\,P_m(k\rho)=0\,.
\end{equation}
Let $k\rho=z$. One obtains
\begin{equation}
\left( \frac{\rd^2}{\rd z^2} +\frac{1}{z}\,\frac{\rd}{\rd z}-\frac{m^2}{z^2}-1\right)\,P_m(z)=0\,.
\end{equation}
Indeed, the modified Bessel equation is the eigenvalue equation of $\left(-\frac{\rd^2}{\rd z^2} -\frac{1}{z}\,\frac{\rd}{\rd z}+\frac{m^2}{z^2}\right)$ with eigenvalue $-1$. Let us restrict ourselves to the case where $\mathrm{Re}(z)>0$ and $m\in\mathbb{Z}$. The preceding equation has two solutions: the first is given by~\cite[Eq.~(14.99)]{Arfken}
\begin{equation}
I_m(z)=\sum_{n=0}^{\infty}\,\frac{(\frac{z}{2})^{2n+m}}{n!\,\Gamma(m+n+1)}
\end{equation}
such that $I_m(z)=I_{-m}(z)$ for $m\in\mathbb{Z}$~\cite[Eq.~(14.101)]{Arfken}. The second solution is given by the integral representation, 
\begin{equation}
K_m(z):=\frac{1}{2}\,\int_0^{\infty}\,\exp{\big( -\frac{z}{2}(s+s^{-1}) \big)}\,s^{m -1}\,\rd s\,,
\end{equation}
which is absolutely convergent and is known as the second modified Bessel function. Alternatively, for $\mathrm{Re}(z)>0$ one has
\begin{equation}\label{Bessel_alter_1}
K_m(z)= \frac{1}{2}\,(\frac{z}{2})^m\,\int_0^{\infty}\,s^{-m-1}\,\exp{(-s-\frac{z^2}{4s})}\,\rd s\,.
\end{equation}
Starting from Eq.~\eqref{Bessel_alter_1}, letting $a,b\in\bC$ with $\mathrm{Re}(a)>0$ and $\mathrm{Re}(b)>0$, and using a change of variable $s=\sqrt{\frac{a}{b}}\,t$, one gets
\begin{equation}\label{Bessel_alter_2}
\int_0^{\infty}\,s^{m-1}\,\exp{(-\frac{a}{s}-bs)}\,\rd s=2\,(\frac{a}{b})^{\frac{m}{2}}\,K_m(2\sqrt{ab}).
\end{equation}
Note that $K_n(z)=K_{-n}(z)$ for all $n\in\mathbb{Z}$. It holds that~\cite[Eq.~(14.107)]{Arfken}
\begin{equation}
K_n(z)=\lim_{m\to n}\frac{\pi}{2\,\sin(m\pi)}\,\left( I_{-m}(z)-I_m(z) \right).
\end{equation}
Applying L'H\^opital's rule to this limit produces the logarithmic terms characteristic of the second solution to the modified Bessel equation for integer orders.

\subsection*{Limiting behaviour}\label{Limiting behaviour}
In what follows we write $f(z)\sim g(z)$ as $z\to a$ if $\lim_{z\to a}\,\frac{f(z)}{g(z)}=1$. Using the little-o notation, this can be rewritten as $f(z)= g(z)\left(1+o(1)\right)$. Using this convention, 
for $z\to 0$ and for $m\neq -1,-2,\cdots$ the limiting behaviour of $I_m$ will be of the form~\cite[Eq.~(14.100)]{Arfken} 
$$I_m(z)=z^m\,2^{-m}\,[\Gamma(m+1)]^{-1}(1+o(1))$$
and as $z\to\infty$ one has
$$I_m(z)= \frac{1}{\sqrt{2\pi z}}\,\re^{z}(1+o(1)).$$
Furthermore, one also determines the lowest-order terms of $K_m(z)$ as $z\to 0$ via~\cite[Eqs.~(14.110--111)]{Arfken}
$$K_0(z)=-\log(z)-\gamma+\log(2)+\cdots \quad\text{and}\quad K_{m}(z)=2^{m-1}\Gamma(m)\,z^{-m}+\cdots.$$ Using the introduced convention, one rewrites the above limiting behaviours
\begin{equation}\label{eq:asy_bessel_modified}
\begin{aligned}
K_0(z)= -\log(z)(1+o(1)),
\qquad
K_m(z)= 2^{m-1}\Gamma(z)\,z^{-m}(1+o(1)).
\end{aligned}
\end{equation}
This indicates that $K_0(z)$ is singular at $z=0$. Moreover, one shows that $K_m(z)$ decays exponentially~\cite[Eq.~(14.127)]{Arfken}
$$K_m(z)=\sqrt{\frac{\pi}{2z}}\,\re^{-z}(1+o(1)).$$
More precisely, for all $z\in\bC$ such that $|\mathrm{arg}z|<\pi-\varepsilon$, one has
\begin{equation}
\lim_{|z|\to\infty}\,\frac{K_m(z)}{\re^{-z}\sqrt{\frac{\pi}{2z}}}=1\,.
\end{equation}

 \subsection*{Bessel potential using the Fourier transform}
Let $(1-\Delta)^{\alpha/2}\,G(x)=\delta(x)$, where $\Delta$ is the Laplacian on $\bR^d$.  Let $p^2=\vec p\cdot \vec p$. Using the Fourier transform one gets
\begin{equation}
G(x)=\mathcal{F}^{-1}\left[\frac{1}{(1+p^2)^{\alpha/2}}\right](x)=\int\,\frac{\re^{i\vec p\cdot \vec x}}{(1+p^2)^{\alpha/2}}\,\frac{\rd^d p}{(2\pi)^d}\,.
\end{equation}

\begin{Remark}\label{eq:power_gamma}
Let $A$ be a positive self-adjoint operator. It holds that
\begin{equation}
A^{-\alpha}=\frac{1}{\Gamma(\alpha)}\,\int_0^{\infty}\,\re^{-sA}\,s^{\alpha-1}\,\rd s,
\qquad
\int\re^{-sp^2}\,\re^{i\vec p\cdot \vec x}\,\frac{\rd^d p}{(2\pi)^d}= \frac{1}{(4\pi\,s)^{d/2} }\,\re^{-\frac{x^2}{4s}}\,.
\end{equation}
\end{Remark}

\begin{lem}\label{thm:Fourier_Bessel_Fundamental}
Let $|x|$ denote the Euclidean norm on $\bR^d$ and $\mathrm{Re}(\alpha)>d/2$. One has
\begin{equation}
\int\,\frac{\re^{i\vec p\cdot\vec x}}{(1+p^2)^{\alpha/2}}\,\frac{\rd^d p}{(2\pi)^d}
= \frac{1}{\Gamma(\frac{\alpha}{2})\,\pi^{d/2}\, 2^{\frac{d+\alpha}{2}-1}}\,|x|^{(\alpha-d)/2}\,K_{(d-\alpha)/2}(|x|)\,.
\end{equation}
\end{lem}
\begin{proof}
Combining Remark~\ref{eq:power_gamma} with Eqs.~\eqref{Bessel_alter_2} and~\eqref{Bessel_alter_1}, one concludes the statement. See also~\cite[Thm.~2.1]{DGR24}.
\end{proof}

\begin{Remark}\label{rem:Fourier_Bessel_Fundamental}
Let $G$ denote a fundamental solution for the positive self-adjoint operator $(1-\Delta)$, where $\Delta$ is the Laplacian on $\bR^2$, i.e., $(1-\Delta)G(x,x')=\delta (x-x')$. Invoking Lemma.~\ref{thm:Fourier_Bessel_Fundamental} with $d=2$, and $\alpha=2$ and taking into account Eq.~\eqref{eq:asy_bessel_modified}, as $\|x-x'\|\to 0$, one obtains
\begin{equation}
G(x,x') = -\frac{1}{2\pi}\log \|x-x'\|\,(1+ o(1)).
\end{equation}
Moreover, using the fact that $K_m(z)$ decays exponentially as $\|x-x'\|\to\infty$, one has
\begin{equation}
G(x,x')
= \frac{\re^{-\|x-x'\|}}{\sqrt{8\,\pi\,\|x-x'\|}}\,(1+o(1)).
\end{equation}
\end{Remark}

\section{Stochastic estimates}\label{sec:stochastic_estimates}
In this appendix, we work with the probability space $(\sD'(\bS_R),\mathrm{Borel}(\sD'(\bS_R)),\nu_R)$ and $\bE$ denotes expectation with respect to the Gaussian probability measure $\nu_R$ with covariance $G_R=(1-\Delta_R)^{-1}$. Moreover, for all $R,N\in\bN_+$, we consider the stochastic processes
$$X_R,X_{R,N}\,:\,\sD'(\bS_R)\to L^{-1}_2(\bS_R).$$
\begin{lem}\label{lem:stochastic_X_not_uniform}
For every $\kappa\in(0,\infty)$, $N\in\bN_+$ there exists $C\in(0,\infty)$ such that for all $R\in\bN_+$ it holds
\begin{enumerate}
\item[(A)] $\bE\|X_{R,N}\|_{L^{2-\kappa}_2(\bS_R)}^2 \leq \, R^2\,C^2$.
\item[(B)] $\bE\|\hat X_{R,N}\|_{L^{2-\kappa}_2(\bS_R)}^2 \leq \, R^2\,C^2$.
\end{enumerate}
\end{lem}
\begin{proof}
It suffices to just prove Item~(A).
By Lemma~\ref{lem:K_hat} the proof of Item~(B) is identical to the proof of Item~(A). One starts off by writing
\begin{equation}\label{eq:stochastic_not_uiform_1}
\begin{aligned}
\bE\|X_{R,N}\|^2_{L_2^{2-\kappa}(\bS_R)} 
&= \bE\|(1-\Delta_R)^{(2-\kappa)/2}X_{R,N}\|^2_{L_2(\bS_R)}
\\
&=\bE \int_{\bS_R}\,((1-\Delta_R)^{(2-\kappa)/2}X_{R,N})(\rx)\, \overline{((1-\Delta_R)^{(2-\kappa)/2}X_{R,N})(\rx)}\,\rho_R(\rd \rx)\,.
\end{aligned}
\end{equation}
Utilizing Remark~\ref{rem:parseval_them} and the notations introduced in Remark~\ref{rem:l_brackets}, one gets
\begin{equation}\label{eq:main}
\bE\|X_{R,N}\|^2_{L_2^{2-\kappa}(\bS_R)} = \bE \,\sum_{l=0}^{\infty}\sum_{m=-l}^{l}\ll l\gg^{(2-\kappa)}
(X_{R,N})_{l,m}\,\overline{(X_{R,N})_{l,m}}\,.
\end{equation}
Using Fubini’s theorem and from the fact that $X_{R,N}$ is real-valued, it follows that
\begin{equation}\label{eq:stochastic_not_uiform_3}
\begin{aligned}
\bE (X_{R,N})_{l,m}\, \overline{(X_{R,N})_{l,m}}
&= \int_{\bS_R}\int_{\bS_R} \bE (X_{R,N}(\ry)X_{R,N}(\rz))\overline{Y_{l,m}}(\ry)Y_{l,m}(\rz)\,\rho_R(\rd\ry)\rho_R(\rd\rz)
\\
&=  \int_{\bS_R}\int_{\bS_R} G_{R,N}(\ry,\rz)\,\overline{Y_{l,m}}(\ry)\,Y_{l,m}(\rz)\,\rho_R(\rd\ry)\,\rho_R(\rd\rz)\,,
\end{aligned}
\end{equation}
where we used Lemma~\ref{Nualart-Lemma-1-1} in the second line. Applying Remarks~\ref{rem:expansion_covariance_harmo} and~\ref{rem:spherical_projection} to the preceding expression, culminates in
\begin{equation}
\begin{gathered}
\bE\|X_{R,N}\|^2_{L_2^{2-\kappa}(\bS_R)}
=  \frac{1}{(4\pi R^2)^{2} }\sum_{l=0}^{\infty}\, \ll l\gg^{(2-\kappa)}(2l+1) \sum_{l_1=0}^{\infty}\,
 \langle l_1\rangle^{-1}\,(2l_1+1)
\\ \times\,\int_{\bS_R}\int_{\bS_R}\,P_{l_1}(\frac{\ry\cdot\rz}{R^2})\,P_{l}(\frac{\ry\cdot\rz}{R^2})\,\rho_R(\rd\ry)\,\rho_R(\rd\rz)\,.
\end{gathered}
\end{equation}
Let $t:=(\ry\cdot\rz)/R^2= \cos\theta, \,\rd t = \sin\theta \rd \theta$ and recall that $\rho_R(\rd\ry)= 2\pi R^2\sin\theta\, \rd \theta$. From rotation invariance property of $P_l$ and Remark~\ref{rem:legendre_ortho}, one obtains
\begin{equation}
\begin{aligned}
 \int_{\bS_R}\int_{\bS_R}\, P_{l_1}(\frac{\ry\cdot\rz}{R^2})\,P_{l}(\frac{\ry\cdot\rz}{R^2})\,\rho_R(\rd\ry)\,\rho_R(\rd\rz)
= (4\pi R^2\times 2\pi R^2)\int_{-1}^1 \,P_{l_1}(t) P_{l}(t)\,\rd t\\
=(4\pi R^2\times 2\pi R^2) \frac{2}{2l_1+1}\,\delta_{l_1,l}.
\end{aligned}
\end{equation}
Substituting this to the previous equation and referring to Lemma~\ref{lem:trace}  
for the $R^2$-dependence of the constants, we conclude the proof. 
\end{proof}

\begin{lem}\label{lem:stochastic_X_uniform}
For every $\kappa\in(0,\infty)$, $\delta\in[0,2]$, there exists $C\in(0,\infty)$ such that for all $R,N\in\bN_+$ it holds that
\begin{enumerate}
 \item[(A)] $\bE\|X_R\|_{L^{-\kappa}_2(\bS_R)}^2 
  \leq R^2\,C^2$,
 \item[(B)] $\bE\|X_R- X_{R,N}\|_{L_2^{-\kappa-\delta}(\bS_R)}^2 
 \leq R^2\,C^2\, N^{-2\delta}$, 
 \item[(C)] $\bE\|X_R-\hat X_{R,N}\|_{L_2^{-\kappa-\delta}(\bS_R)}^2  
 \leq R^2\,C^2\, N^{-2\delta}$. 
\end{enumerate}
\end{lem}

\begin{proof}
By definition
\begin{equation}
\begin{aligned}
\bE\|X_{R,N}\|^2_{L_2^{-\kappa}(\bS_R)} 
&= \sum_{l=0}^{\infty}(1+\frac{l(l+1)}{R^2})^{-\kappa}\,(2l+1)\,(1+\frac{l(l+1)}{R^2})^{-1}(1+\frac{(l(l+1)}{N^2R^2})^{-2}
\\
& = \sum_{l=0}^{\infty} (2l+1)\,\Big(1+\frac{l(l+1)}{R^2}\Big)^{-1-\kappa}\,\Big(1+\frac{l(l+1)}{R^2N^2}\Big)^{-2}
\\
& \leq \sum_{l=0}^{\infty} (2l+1)\,\Big(1+\frac{l(l+1)}{R^2}\Big)^{-1}\,\Big(1+\frac{l(l+1)}{R^2N^2}\Big)^{-2}\,.
\end{aligned}
\end{equation}
Now, Item~(A) follows from Lemma~\ref{lem:trace}. To prove Item~(B), we start off by writing
\begin{equation}
\begin{aligned}
\bE\|X_R - X_{R,N}\|^2_{L^{-\kappa -\delta}_{2}(\bS_R)} 
= \bE\|(1 - \Delta_R)^{\frac{-\kappa -\delta}{2}} (X_R - X_{R, N}) \|^2_{L_2(\bS_R)}\,.
\end{aligned}
\end{equation}
The above expression can be expanded by
\begin{multline}
 \mathrm{Tr}\left( (1-\Delta_R)^{-1-\kappa-\delta}\right)
- 2\,\mathrm{Tr}\left((1-\Delta_R)^{-1-\kappa-\delta} (1 -\frac{\Delta_R}{N^2})^{-1}\right)
\\
+\mathrm{Tr}\left(1-\Delta_R)^{-1-\kappa-\delta}(1 -\frac{\Delta_R}{N^2})^{-2}\right),
\end{multline}
which is equal to
\begin{multline}
 \text{Tr}\Big((1-\Delta_R)^{-1-\kappa-\delta} (1- (1 -\frac{\Delta_R}{N^2})^{-1})^2 \Big)\,
 \\
= \sum_{l=0}^{\infty} (2l+1)\,\left(1+\frac{l(l+1)}{R^2}\right)^{-\kappa-\delta-1}\left(1 - (1+\frac{l(l+1)}{R^2N^2})^{-1}\right)^{2}\,.
\end{multline}
It follows from the fact that $\mathrm{Tr}(A) +\mathrm{Tr}(B)=\mathrm{Tr}(A+B)$. Observe that for all $\delta\in [0,2]$, one has
\begin{equation}
\left(1- \left(1+ \frac{l(l+1)}{R^2N^2} \right)^{-1}\right)^{2} \leq N^{-2\delta}\,\Big(\frac{l(l+1)}{R^2}\Big)^{\delta}\,.
\end{equation}
Hence,
\begin{multline}
\bE \|X_R - X_{R,N}\|^2_{L^{-\kappa -\delta}_{2}(\bS_R)} 
= 
\mathrm{Tr}\big((1-\Delta_R)^{-1-\kappa-\delta} (1- (1 -\Delta_R/N^2)^{-1})^2 \big)
\\
\leq N^{-2\delta}\,\mathrm{Tr}\big((1-\Delta_R)^{-1-\kappa}\big)
= 
\sum_{l=0}^{\infty} \frac{N^{-2\delta}\,(2l+1)}{(1+l(l+1)/R^2)^{1+\kappa}}.
\end{multline}
Now, Item~(B) follows from Lemma~\ref{lem:trace}. Thanks to Lemma~\ref{lem:K_hat}, the proof of Item~(C) is the same as the proof of Item~(B).
\end{proof}

\begin{lem}\label{lem:process_measure_convergence}
Let $(X_n)_{n\in\bN_+}$ be a sequence of Gaussian random variables with values in a Hilbert space $H$ with distributions $\cN(m_n, Q_n)$. Assume that $X_n\to X$ in $L_2(\Omega,\mathcal{F},\bP; H)$. Then $X$ is a Gaussian random variable with mean $m=\lim_{n\to \infty} m_n$ and covariance operator $\langle Q h,h\rangle := \lim_{n\to\infty}\langle Q_n h,h\rangle $. 
\end{lem}
\begin{proof} 
From Remark~\ref{rem:convergence} convergence in $L_2(\Omega,\cF,\bP)$ implies convergence in law, hence point-wise convergence of the characteristic
functions. Considering that
\begin{equation}
\bE\, \re^{i\langle X_n,f\rangle} =\re^{i\langle m_n,f \rangle-\frac{1}{2}\, \langle Q_nf,f\rangle}.
\end{equation}
We observe that $\langle m_n,f \rangle $ and $ \langle Q_nf,f\rangle$ must converge for any fixed $f$ and
thus determine the mean $m$ and covariance $Q$ of the limiting characteristic function. From its form we see that $X$ is Gaussian.
\end{proof}

\begin{Remark}\label{rem:process_measure_convergence}
Combing Lemma~\ref{lem:process_measure_convergence} and Lemma~\ref{lem:stochastic_X_uniform}, one infers that $X_R$ is a Gaussian random variable with mean zero and covariance operator
$$\langle G_R h,h\rangle := \lim_{N\to\infty}\langle G_{R,N} h,h\rangle.$$ 
\end{Remark}
\begin{Remark}\label{rem:mean_young_bernoulli}
For all $a,b\in\bR$ and integers $p\geq 1$, one has
\begin{equation}\label{eq}
|a^p-b^p|\leq \frac{p}{2}\,|a-b|\,
\bigl(|a|^{p-1}+|b|^{p-1}\bigr).
\end{equation}
The case $p=1$ is obvious. Let $p>1$. Using
\begin{equation}
a^p-b^p=(a-b)\sum_{k=0}^{p-1}a^{p-1-k}b^k,
\end{equation}
we obtain
\begin{equation}
|a^p-b^p| \leq |a-b|\sum_{k=0}^{p-1}|a|^{p-1-k}|b|^k.
\end{equation}
For $1\leq k\leq p-2$, Young's inequality with exponents
\begin{equation}
\widetilde p=\frac{p-1}{p-1-k},
\qquad
\widetilde q=\frac{p-1}{k}
\end{equation}
gives
\begin{equation}
|a|^{p-1-k}|b|^k \leq \frac{p-1-k}{p-1}|a|^{p-1}+\frac{k}{p-1}|b|^{p-1}.
\end{equation}
The same inequality holds directly for the endpoint terms $k=0$ and $k=p-1$. Summing over $k=0,\ldots,p-1$ and using
\begin{equation}
\sum_{k=0}^{p-1}\frac{p-1-k}{p-1}=\sum_{k=0}^{p-1}\frac{k}{p-1}=\frac{p}{2},
\end{equation}
we conclude that
\begin{equation}
|a^p-b^p|
\leq \frac{p}{2}|a-b|
\bigl(|a|^{p-1}+|b|^{p-1}\bigr).
\end{equation}

\end{Remark}

\begin{lem}\label{lem:stochastic_Y}
Let $R\in\bN_+$. There exists a real-valued random variable $Y_R$ and $C\in(0,\infty)$ such that for all $N\in\bN_+$ it holds
\begin{enumerate}
 \item[(A)] $\bE Y_R^2\leq C^2$,
 \item[(B)] $\bE (Y_R-Y_{R,N})^2\leq C^2\,N^{-1/n}$,
 \item[(C)] $\bE (Y_R-\hat Y_{R,N})^2\leq C^2\,N^{-1/n}$,
 \item[(D)] $\bE (\hat Y_{R,N}-\tilde Y_{R,N})^2\leq C^2\,N^{-1}$.
\end{enumerate}
\end{lem}
\begin{Remark}
Recall that $n\in 2\bN_+$, $n\geq 4$, is the degree of the polynomial $P$ given in Eq.~\eqref{eq:polynomial_form} and the random variables $Y_{R,N}$ and $\hat Y_{R,N},\tilde Y_{R,N}$ were introduced in Eq.~\eqref{dfn:X_Y} and Def.~\ref{dfn:hat_X_hat_Y}, respectively.
\end{Remark}
\begin{proof}
To prove Items~(A)~and~(B), it is enough to show that for every $m\in\{1,\ldots,n\}$ there exists $C\in(0,\infty)$ such that for all $N,M\in\bN_+$, it holds
\begin{equation}
 \bE X^{:m:}_{R,N}(\mathmybb{1}_{\bS_R})(X^{:m:}_{R,N}-X^{:m:}_{R,M})(\mathmybb{1}_{\bS_R})\leq C^2\,(N\wedge M)^{-1}.
\end{equation}
Let $G_{R,N,M}:=K_{R,N}G_R K_{R,M}$. By Lemma~\ref{Nualart-Lemma-1-1} 
\begin{multline}
 \bE X^{:m:}_{R,N}(1_{\bS_R})(X^{:m:}_{R,N}-X^{:m:}_{R,M})(\mathmybb{1}_{\bS_R}) \\
 =m! \int_{\bS_R}\int_{\bS_R}\,(G^m_{R,N,N}(\rx,\ry)-G^m_{R,N,M}(\rx,\ry))\,\rho_R(\rd\rx)\rho_R(\rd\ry).
\end{multline}
On account of Remark~\ref{rem:mean_young_bernoulli}, there exists $C\in (0,\infty)$ such that
\begin{equation}
|G_{R,N,N}^m-G_{R,N,M}^m|
\leq\,C\,|G_{R,N,N}-G_{R,N,M}|\,\left | |G_{R,N,N}|^{m-1}+|G_{R,N,M}|^{m-1}\right |.
\end{equation}
For every $m\in\{1,\ldots,n\}$ from H{\"o}lder's inequality with $1/m+(m-1)/m=1$, we obtain that  there exists $C\in(0,\infty)$ such that for all $N,M\in\bN_+$ and all $\rx\in\bS_R$ it holds
\begin{multline}
 |\bE X^{:m:}_{R,N}(\mathmybb{1}_{\bS_R})(X^{:m:}_{R,N}-X^{:m:}_{R,M})(\mathmybb{1}_{\bS_R})| 
 \\
 \leq C\,\|(G_{R,N,N}-G_{R,N,M})(\Cdot,\Cdot)\|_{L_m(\bS_R^2)}
 \big(\|G_{R,N,N}(\Cdot,\Cdot)\|_{L_m(\bS_R^2)}^{m-1}
 +
 \|G_{R,N,M}(\Cdot,\Cdot)\|_{L_m(\bS_R^2)}^{m-1}\big)
 \\
 \leq \hat C\,\|(G_{R,N,N}-G_{R,N,M})(\Cdot,\Cdot)\|_{L_n(\bS_R^2)}
 \big(\|G_{R,N,N}(\Cdot,\Cdot)\|_{L_n(\bS_R^2)}^{m-1}
 +
 \|G_{R,N,M}(\Cdot,\Cdot)\|_{L_n(\bS_R^2)}^{m-1}\big)
 \\=
 \check C\,\|(G_{R,N,N}-G_{R,N,M})(\rx,\Cdot)\|_{L_n(\bS_R)}
 \big(\|G_{R,N,N}(\rx,\Cdot)\|_{L_n(\bS_R)}^{m-1}
 +
 \|G_{R,N,M}(\rx,\Cdot)\|_{L_n(\bS_R)}^{m-1}\big),
\end{multline}
where in the last step above we used the rotational invariance property of $G_{R,N,N}$ and $G_{R,N,M}$. Note that $\hat C=(4\pi R^2)^{( 2-2m/n)} C$, $\check C=(4\pi R^2)^{m/n}\hat C$. Evoking the rotational invariance property combined with Lemma~\ref{lem:embedding_sphere} there exist $\hat C,C\in(0,\infty)$ such that for all $N\in\bN_+$ it holds
\begin{multline}
 \|(G_{R,N,N}-G_{R,N,M})(\rx,\Cdot)\|^2_{L_n(\bS_R)}
 \leq \hat C\, 
 \|(G_{R,N,N}-G_{R,N,M})(\rx,\Cdot)\|^2_{L_2^{(n-2)/n}(\bS_R)}
 \\
 =
 (4\pi R^2)^{-1}\,\hat C\,\big(\mathrm{Tr}\big[G_R^{(n+2)/n} K_{R,N}^2 (K_{R,N}-K_{R,M})^2\big]\big)
 \leq C\,(N\wedge M)^{-2/n}.
\end{multline}
The last bound follows from the fact that
\begin{multline}
\text{Tr}\left(G_R^{(n+2)/n}\,\left(K^2_{R,N}\,(K_{R,N}-K_{R,M})^2\right) \right)
\\
\leq \sum_{l=0}^{\infty}\,\frac{(2l+1)}{\left(1+l(l+1)/R^2\right)^{1+2/n}}\,\left( \frac{1}{1+l(l+1)/(RN)^2} -\frac{1}{1+l(l+1)/(RM)^2}\right)^2
\\
\leq \sum_{l=0}^{\infty}\,\frac{(2l+1)}{\left(1+l(l+1)/R^2\right)^{1+2/n}}\,\bigg( (1+l(l+1)/R^2)^{1/n}\,\left( M^{-2/n}+N^{-2/n} \right)\bigg)
\\
\leq \,2\,(N\wedge M)^{-2/n}\,\sum_{l=0}^{\infty}\,\frac{(2l+1)}{\left(1+l(l+1)/R^2\right)^{1+1/n}}\,.
\end{multline}
The last sum is bounded as
$$(1+(l(l+1))/(RN)^2)\geq1,\,\,n\geq 4.$$
By an analogous reasoning, we obtain
\begin{multline}\label{eq:prove_Y_bound_G}
 \|G_{R,N,N}(\rx,\Cdot)\|^2_{L_n(\bS_R)}
 \leq \hat C\, 
 \|G_{R,N,N}(\rx,\Cdot)\|^2_{L_2^{(n-2)/n}(\bS_R)}
 \\
 =
 (4\pi R^2)^{-1}\,\hat C\,\big(\mathrm{Tr}\big[G_R^{(n+2)/n} K_{R,N}^4\big]\big).
\end{multline}
Observe that
\begin{multline}
\mathrm{Tr}\left[G_R^{(n+2)/n} K_{R,N}^4\right]
= \sum_{l=0}^{\infty}\,\frac{(2l+1)}{(1+l(l+1)/R^2)^{1+2/n}}\left(\frac{1}{1+l(l+1)/(RN)^2} \right)^4
 \leq C
\end{multline}
for some constants $C,\hat C$ independent of $N$ and $m$. This proves~(A) and~(B). Thanks to Lemma~\ref{lem:K_hat}, the above estimates are also valid when $X_{R,N}$ is replaced with $\hat X_{R,N}$ and $G_{R,N,M}$ is replaced with $\hat G_{R,N,M}:=\hat K_{R,N}G_R\hat K_{R,M}$. Hence,~(C) follows. To prove Item~(D) note that for every $m\in\{1,\ldots,n\}$, there exists $C\in(0,\infty)$ such that for all $N\in\bN_+$ and $\rx\in\bS_R$, one has
\begin{multline}
 \bE \hat X^{:m:}_{R,N}(1_{\bS_R\setminus\bS_{R,N}})\hat X^{:m:}_{R,N}(1_{\bS_R\setminus\bS_{R,N}})
 \leq \|\hat G_{R,N,N}(\Cdot,\Cdot)\|^m_{L_m(\bS_R\setminus\bS_{R,N}\times \bS_R\setminus\bS_{R,N})}
 \\
 \leq 
 \|\hat G_{R,N,N}(\Cdot,\Cdot)\|^m_{L_m(\bS_R\setminus\bS_{R,N}\times \bS_R)}
 \leq C/N\,
 \|\hat G_{R,N,N}(\rx,\Cdot)\|^m_{L_m(\bS_R)},
\end{multline}
where in the last step we used the rotational invariance of $\hat G_{R,N,N}$ and the fact that the volume of $\bS_R\setminus\bS_{R,N}$ is bounded by $C/N$. It holds that
\begin{multline}\label{eq:prove_Y_bound_G_hat}
 \|\hat G_{R,N,N}(\rx,\Cdot)\|^2_{L_m(\bS_R)}\leq \check C\, \|\hat G_{R,N,N}(\rx,\Cdot)\|^2_{L_n(\bS_R)}
 \leq \hat C\, 
 \|\hat G_{R,N,N}(\rx,\Cdot)\|^2_{L_2^{(n-2)/n}(\bS_R)}
 \\
 =
 (4\pi R^2)^{-1}\,\hat C\,\big(\mathrm{Tr}\big[\hat G_R^{1+2/n} \hat K_{R,N}^4\big]\big)
 \leq C.
\end{multline}
This finishes the proof.
\end{proof}

\begin{lem}\label{lem:stochastic_bound_infinite_volume}
Let $m\in\bN_+$, $p\in[1,\infty)$, $\kappa\in(0,\infty)$ and $L\in[1,\infty)$. There exists $C\in(0,\infty)$ such that for all $R,N\in\bN_+$, $R\geq L$, it holds that
\begin{itemize}
\item[(A)]
$ \bE\|\jmath_R^* X^{:m:}_{R,N}\|_{L^{-\kappa}_p(\bR^2,v_L^{1/p})}^p \leq C$,
\item[(B)]
$\bE\,\|\jmath_R^* (X_R-X_{R,N})\|_{L^{-\kappa}_p(\bR^2,v_L^{1/p})}^p \leq CN^{-\kappa}\,.$
\end{itemize}
\end{lem}
\begin{proof}
By Jensen's inequality it suffices to prove the statement for $p\in2\bN_+$. Let $q=(4/\kappa) \vee 4$. There exists $C\in(0,\infty)$ depending on $p$ and $\kappa$ such that for all $R,N\in\bN_+$, one has
\begin{multline}
 \bE\|\jmath_R^* X^{:m:}_{R,N}\|_{L^{-\kappa}_p(\bR^2,v_L^{1/p})}^p \leq 
 \|v_L w_L^{-1/q}\|_{L_1(\bR^2)}\,\|\bE ((1-\Delta)^{-\kappa/2}\jmath_R^* X^{:m:}_{R,N}(\Cdot))^p\|_{L_\infty(\bR^2,w_L^{1/q})}
 \\
 \leq C\,\|\bE((1-\Delta)^{-\kappa/2}\jmath_R^* X^{:m:}_{R,N}(\Cdot))^2\|_{L_\infty(\bR^2,w_L^{1/q})}^{p/2}\,,
\end{multline}
where the last bound is a consequence of Lemma~\ref{lem:nelson}. Recall that $\bE X_{R,N}\otimes X_{R,N} = G_{R,N}(\Cdot,\Cdot)$, where $G_{R,N}=K_{R,N}G_RK_{R,N}$. By Lemma~\ref{Nualart-Lemma-1-1}
\begin{equation}
 \bE \jmath_R^*X^{:m:}_{R,N}\otimes \jmath_R^*X^{:m:}_{R,N} = m!\, \tilde G_{R,N}^m,
 \qquad 
 \tilde G_{R,N}:=(\jmath_R^*\otimes\jmath_R^*)G_{R,N}(\Cdot,\Cdot)\,.
\end{equation}
Hence, by Fubini's theorem and explicit formula for the kernel in terms of spherical harmonics
\begin{multline}
\bE(1-\Delta)^{-\kappa/2}\jmath_R^* X^{:m:}_{R,N}\otimes\,
(1-\Delta)^{-\kappa/2}\jmath_R^* X^{:m:}_{R,N}
\\
= m!\,
\big((1-\Delta)^{-\kappa/2} \otimes (1-\Delta)^{-\kappa/2}\big)\, \tilde G_{R,N}^m \in C(\bR^2\times\bR^2).
\end{multline}
Since for $F\in C(\bR^2\times\bR^2)$ it holds $\sup_{x\in\bR^2} F(x,x)\leq \sup_{y\in\bR^2}\sup_{x\in\bR^2} F(x,y)$ we obtain
\begin{multline}
\|\bE((1-\Delta)^{-\kappa/2}\jmath_R^* X^{:m:}_{R,N}(\Cdot))^2\|_{L_\infty(\bR^2,w_L^{1/q})}
\\
\leq
m!\,\sup_y w_L^{1/q}(y)\,\|((1-\Delta)^{-\kappa/2} \otimes 1)\big(1 \otimes (1-\Delta)^{-\kappa/2}) \tilde G_{R,N}^m\big)(\Cdot,y)\|_{L_\infty(\bR^2)}
\\
=
m!\,\sup_y w_L^{1/q}(y)\,\|\big( 1 \otimes (1-\Delta)^{-\kappa/2}) \tilde G_{R,N}^m\big)(\Cdot,y)\|_{L^{-\kappa}_\infty(\bR^2)}.
\end{multline}
From Thm.~\ref{thm:embedding}~(B), there exists $C\in(0,\infty)$ such that for all $R,N\in\bN_+$, the above expression is bounded by
\begin{multline}
\sup_{y\in\bR^2} w_L^{1/q}(y)\,\|\big(1\otimes(1-\Delta)^{-\kappa/2}) \tilde G_{R,N}^m\big)(\Cdot,y)\|_{L_\infty(\bR^2)}
\\
=
\sup_{x\in\bR^2}  \|\big(  1\otimes(1-\Delta)^{-\kappa/2}) \tilde G_{R,N}^m\big)(x,\Cdot)\|_{L_\infty(\bR^2,w_L^{1/q})}
\\
=
\sup_{x\in\bR^2} \|\tilde G_{R,N}^m(x,\Cdot)\|_{L_\infty^{-\kappa}(\bR^2,w_L^{1/q})}
\end{multline}
up to a multiplicative constant $C$, which depends on $m$. The first equality above follows from the fact that for 
\mbox{$F\in C(\bR^2\times\bR^2)$} it holds
$$\sup_{x\in\bR^2}\sup_{y\in\bR^2} F(x,y)=\sup_{y\in\bR^2}\sup_{x\in\bR^2} F(x,y).$$
By Thm.~\ref{thm:embedding}~(B), since $q>2/\kappa$, the above expression is bounded by
\begin{multline}
\sup_{x\in\bR^2} \|\tilde G_{R,N}^m(x,\Cdot)\|_{L_q(\bR^2,w_L^{1/q})}
\leq
\sup_{x\in\bR^2} \|\tilde G_{R,N}^m(x,\Cdot)\|_{L_q(\bR^2,w_R^{1/q})}
=
\sup_{\rx\in\bS_R} \|G_{R,N}(\rx,\Cdot)^m\|_{L_q(\bS_R)}
\\
=
\sup_{\rx\in\bS_R} \|G_{R,N}(\rx,\Cdot)\|_{L_{mq}(\bS_R)}^m
\leq 
C\sup_{\rx\in\bS_R} \|G_{R,N}(\rx,\Cdot)\|_{L^{(mq-2)/mq}_2(\bS_R)}^m
\\
=
C(4\pi R^2)^{-m/2}\,[\mathrm{Tr}(G_{R,N}(1-\Delta_R)^{(mq-2)/mq}G_{R,N})]^{m/2}.
\end{multline}
The first bound above is true because $R\geq L$. The second bound is a consequence of the Sobolev embedding stated in Lemma~\ref{lem:embedding_sphere}. It holds that
\begin{equation}
\mathrm{Tr}(G_{R,N}(1-\Delta_R)^{(mq-2)/mq}G_{R,N})
= \text{Tr}\,\big(  \frac{1}{(1-\Delta_R)^{1+2/mq}(1-(\Delta_R/N^2))^4} \big).
\end{equation}
Now, Item~(A) follows from Lemma~\ref{lem:trace} applied with $N'=1$ and $\kappa'=2/mq$.  To prove Item~(B), we use exactly the same strategy as above with $m=1$. There exists $C\in(0,\infty)$ depending on $p$ and $\kappa$ such that for all $R,N\in\bN_+$, one obtains
\begin{multline}
\begin{aligned}
 \bE\|\jmath_R^* (X_R-X_{R,N})\|_{L^{-\kappa}_p(\bR^2,v_L^{1/p})}^p
= \bE\|\jmath_R^*(1-K_{R,N})X_R\|_{L^{-\kappa}_p(\bR^2,v_L^{1/p})}^p 
 \\
\leq 
 \|v_L w_L^{-1/q}\|_{L_1(\bR^2)}\,\|\bE ((1-\Delta)^{-\kappa/2}\jmath_R^* (1-K_{R,N})X_R(\Cdot))^p\|_{L_\infty(\bR^2,w_L^{1/q})}
\\
\leq\, C\,\|\bE\left((1-\Delta)^{-\kappa/2}\jmath_R^*(1-K_{R,N})X_R(\Cdot)\right)^2\|_{L_\infty(\bR^2,w_L^{1/q})}^{p/2}\,.
 \end{aligned}
\end{multline}
where the last bound is a consequence of Lemma~\ref{lem:nelson}. One has
\begin{equation}
\bE\left(\jmath_R^*(1-K_{R,N})X_R\otimes \jmath_R^*(1-K_{R,N})X_R \right)=(\jmath_R^*\otimes\jmath_R^*)\,\tilde G_{R,N}(\Cdot,\Cdot)\,.
\end{equation}
Utilizing the result obtained in the previous part of the proof with $m=1$ and replacing $G_{R,N}\to\tilde G_{R,N}$ yields
\begin{multline}
\begin{aligned}
\mathrm{Tr}\big(\tilde G_{R,N}(1-\Delta_R)^{1-2/q}\tilde G_{R,N} \big)
=&\,\mathrm{Tr}\big( \frac{1}{(1-\Delta_R)^{1+2/q}}\,\big(\frac{-(\Delta_R/N^2)}{1-(\Delta_R/N^2)}\big)^4 \big)
\\
\leq & \,\,N^{-\kappa}\,\mathrm{Tr}\,\big(  \frac{1}{(1-\Delta_R)^{1+2/q-\kappa/2}}\big)\,.
\end{aligned}
\end{multline}
The above expression is converging provided that $2/q> \kappa/2$, which amount to the restriction $q<4/\kappa$.
This concludes the proof.
\end{proof}

\begin{lem}\label{lem:K_hat}
There exists $C\in(0,\infty)$ such that for all $R,N\in\bN_+$, it holds that
\begin{equation}
 |\hat K_{R,N}| \leq C\,\frac{1}{1-\Delta_R/N^2},
 \qquad
 |1-\hat K_{R,N}|\leq
 C\,\frac{(1-\Delta_R)/N^2}{1-\Delta_R/N^2}\,.
\end{equation} 
\end{lem}

\begin{Remark}
Observe that for all $R,N\in\bN_+$, one has $1-\Delta_R/N^2\geq 1$. Given the first bound in the statement of the lemma, considering separately the cases $-\Delta_R/N^2\leq 1$ and $-\Delta_R/N^2\geq 1$, and using the approximation property of $\hat K_{R,N}$, there exists $C\in(0,\infty)$ such that the second bound holds.
\end{Remark}

\begin{proof}
Using Remark~\ref{rem:spectral}, one has 
$$\hat K_{R,N}=\sum_{l=0}^\infty (2l+1)\mathrm{Tr}(\hat K_{R,N}\cP_{R,l})\,\cP_{R,l}.$$
Consequently, by the triangle inequality for commuting self-adjoint operators, it is enough to show that there exists $C\in(0,\infty)$ such that for all $R,N\in\bN_+$ and $l\in\bN_0$ 
\begin{equation}\label{eq:trace_L_operator_bounds_proof}
\begin{gathered}
 (1+l(l+1)/R^2N^2)\,|\mathrm{Tr}(\hat K_{R,N}\cP_{R,l})| \leq C,
 \\
 |\mathrm{Tr}((1-\hat K_{R,N})\mathcal{P}_{R,l})| \leq C\,(1+l(l+1))/R^2N^2\,,
\end{gathered} 
\end{equation}
where we used Remark~\ref{rem:expansion_covariance_harmo}. Using Remark~\ref{rem:funk-hecke}, one has
\begin{equation}
\begin{aligned}
 \mathrm{Tr}(\hat K_{R,N}\mathcal{P}_{R,l}) 
& =
\frac{1}{4\pi R^2}\,\int_{\bS_R}\int_{\bS_R}\,N^2\,h(RN\,\theta(\rx,\ry))\,P_l(\rx\cdot\ry/R^2)\,\rho_R(\rd \rx)\,\rho_R(\rd\ry)
\\
& =2\pi\int_0^{\pi}\,(NR)^2\,h(RN\theta)\,P_l(\cos\theta)\,\sin\theta\,\rd\theta
\\
& = 2\pi\int_0^1\,h(\theta)\,P_l(\cos(\theta/RN))\,RN\,\sin(\theta/RN)\,\rd\theta\,,
 \end{aligned}
\end{equation}
where we used a change of variable $\theta'=RN\,\theta$, the support property of the function $h$ and the rotation invariance property of the whole expression. Observe that $RN\sin(\theta/RN)\leq\theta$ and from Eq.~\eqref{eq:bound_legendre}, one has  $|P_l(\cos(\theta/RN))|\leq 1$. This verifies that $|\mathrm{Tr}(\hat K_{R,N}\mathcal{P}_{R,l})| \leq C$. Using integration by parts and the support property of the function $h$, one infers that there exists $C'\in (0,\infty)$ such that
\begin{equation}
\left|\int_0^{\pi}\,h(RN\theta)\left[
\frac{\rd}{\rd\theta}\left(\cos\theta\,P_{l}(\cos\theta)\right) 
-\,\frac{\rd^2}{\rd\theta^2}\,\left(\sin\theta\,P_{l}(\cos\theta)\right)\right]\rd\theta\right| <C'\,.
\end{equation}
The preceding bound combined with Eq.~\eqref{eq:legendre_1} implies that 
\begin{equation}
\big|2\pi\int_0^{\pi}\,(NR)^2\,h(RN\theta)\,P_l(\cos\theta)\,\sin\theta\,\rd\theta\big|
\leq\,C\,\frac{2\pi(NR)^2}{l(l+1)}\,.
\end{equation}
This concludes the proof of the first bound in Eq.~\eqref{eq:trace_L_operator_bounds_proof}. Next, using that
$$2\pi \int \theta\,h(\theta) \rd\theta=1,$$
we obtain that
\begin{equation}
\begin{aligned}
 \text{Tr}((1-\hat K_{R,N})\mathcal{P}_{R,l}) 
 &= 2\pi\int_0^1\,\big(\theta-\,RN\,\sin(\theta/RN)\,P_l(\cos(\theta/RN))\big)\,h(\theta)\,\rd\theta\,.
\end{aligned}
\end{equation}
One rewrites the RHS of the above expression
\begin{equation}
2\pi\int_0^1
\left[1-\frac{\sin(\theta/RN)}{\theta/RN}+\frac{\sin(\theta/RN)}{\theta/RN}
\bigl(1-P_l(\cos(\theta/RN))\bigr)
\right]\theta h(\theta)\,d\theta.
\end{equation}
Using Remark~\ref{rem:er_legendre}, one obtains
\begin{equation}
0\leq 1-P_l(\cos(\theta/RN)) 
\leq l(l+1)\,(1-\cos(\theta/RN))/2
\leq l(l+1)\,(\theta/RN)^2/4.
\end{equation}
Combining the above bound with the elementary estimate
$$0\leq 1 - \sin(\theta/RN)/(\theta/RN) \leq (\theta/RN)^2/6$$
and the support property of the function $h$, one verifies the second bound in~\eqref{eq:trace_L_operator_bounds_proof} and finishes the proof.
\end{proof}

\section{Proof of Lemma~\ref{lem:mild_regularity_time_space}}\label{sec:mild_regularity_time_space}
From Thm.~\ref{thm:mild_slotion_non-linear} a mild $L^1_2(\bS_R)$-valued solution of Eq.~\eqref{eq:varPsi_pde} on $[0,T]$ is a function $\varPsi^g_{R,N}\in C([0,T],L^1_2(\bS_R))$ such that it holds almost surely
\begin{equation}
\varPsi^g_{R,N}(t,\Cdot) = \re^{-t\,Q_{R,N}}\,\psi^g_{R,N} +\int_0^t\,\re^{-(t-s)\,Q_{R,N}}\, F(\varPsi^g_{R,N}(s,\Cdot))\,\rd s\quad\textrm{for all}\quad0\leq t\leq T\,,
\end{equation}
where $\psi^g_{R,N}\in L^1_2(\bS_R)$ and
$$F( \varPsi^g_{R,N}(s,\Cdot)) = P'((\varPsi^g_{R,N}+Z_{R,N})(s,\Cdot),c_{R,N}) +((\varPsi^g_{R,N}+Z_{R,N})(s,g))^{n-1}g(\Cdot)$$
for all $R,N\in \bN_+$. Note that 
$$F\,:\,L_2^1(\bS_R)\to L_2(\bS_R)\quad\text{and}\quad F\,:\,L_2^3(\bS_R)\to L_2(\bS_R)$$
are locally Lipschitz continuous. The latter fact follows from calculations similar to those in the discussion of Eqs.~\eqref{eq:non_linearity_contra_1}--\eqref{eq:R^2_local} with the help of Remark~\ref{rem:bounded_semigroup_L^3_2}. Moreover, Thm.~\ref{lem:global_solution_Psi} implies that 
$\varPsi^g_{R,N}\in C([0,\infty), L^1_2(\bS_R))$ is the global mild solution to Eq.~\eqref{eq:varPsi_pde}. To verify the statement, we shall first show that $\varPsi^g_{R,N}\in C((0,\infty),L^3_2(\bS_R))$, which implies $\varPsi^g_{R,N}\in C((0,\infty),L^{-3}_2(\bS_R))$. Observe that at $t=0$ we cannot have a better regularity than the one of the initial condition. Using the fact that for all $N\in \bN_+$, there are some $\tilde c,c\in (0,\infty)$ such that $$c\,Q^{1/3}_{R,N}\leq (1-\Delta_R)\leq \tilde c\,Q^{1/3}_{R,N},$$
one gets
\begin{equation}
\|\re^{-t{Q_{R,N}}}\,\psi^g_{R,N} \|_{L_2^3(\bS_R)}
\leq  \tilde c\, \|Q_{R,N}^{1/2}\,\re^{-t{Q_{R,N}}}\,\psi^g_{R,N}\|_{L^1_2(\bS_R)} 
\leq C\,\|\psi^g_{R,N}\|_{L^1_2(\bS_R)}\,,
\end{equation}
where $C\in (0,\infty)$. The last bound follows from the fact that for all $t>0$, the map $x\mapsto x^s\,\re^{-tx}$ is uniformly bounded for all $x>0$ by $C'/t^{s}$ for some $C'\in (0,\infty)$. Note that for all $N\in\bN_+$, there exists some $C\in (0,\infty)$ such that
\begin{equation}
\|\int_0^t\,\re^{-(t-s)Q_{R,N}}\, F(s,\Cdot)\,\rd s\|_{L^3_2(\bS_R)}
\leq C\,
\|Q_{R,N}^{1/2}\,\int_0^t\,\re^{-(t-s)Q_{R,N}}\, F(s,\Cdot)\,\rd s\|_{L_2(\bS_R)},\,
\end{equation}
where we wrote $F(s,\Cdot):=F(\varPsi^g_{R,N}(s,\Cdot))$ for brevity.
Testing with $\phi\in\text{Dom}(Q_{R,N}^{1/2})$ and using the Cauchy-Schwarz inequality one infers that the right hand side of the above expression can be bounded by
\begin{multline}
\|Q_{R,N}^{1/2}\int_0^t\,\re^{-(t-s)\,Q_{R,N}}\, F(s,\Cdot)\,\rd s\|_{L_2(\bS_R)}
\\
=
\sup_{\|\phi\|_{L_2(\bS_R)}\, \leq 1}\,|\, \langle \phi,Q_{R,N}^{1/2}\,\int_0^t\,\re^{-(t-s)\,Q_{R,N}}\, F(s,\Cdot)\,\rd s\rangle_{L_2(\bS_R)}\,|
\\
=
\sup_{\|\phi\|_{L_2(\bS_R)}\, \leq 1}\,|\,\int_0^t\, \langle  \,Q_{R,N}^{1/2}\,\re^{-(t-s)\,Q_{R,N}} \phi,\, F(s,\Cdot)\,\rangle_{L_2(\bS_R)}\,\rd s|
\\
=
\sup_{\|\phi\|_{L_2(\bS_R)}\, \leq 1}\,\Big[\int^t_0\|Q_{R,N}^{1/2}\,\re^{-(t-s)\,Q_{R,N}} \phi\,\|_{L_2(\bS_R)}\,\rd s \Big]\,\sup_{r\in[0,t]}\|F(r,\Cdot)\|_{L_2(\bS_R)}
\\
\leq  C'\,\sup_{\|\phi\|_{L_2(\bS_R)}\, \leq 1}\,\Big[ \langle Q_{R,N}^{1/2}\,\phi, \frac{\re^{-t\,Q_{R,N}} -1}{Q_{R,N}}\,Q_{R,N}^{1/2}\, \phi \rangle \Big]^{1/2}
\leq C''\,,
\end{multline}
where we used Fubini's theorem in the third line and $C''\in(0,\infty)$ is some constant. Observe that
$$\sup_{r\in[0,t]}\|F(r,\Cdot)\|_{L_2(\bS_R)}$$
is bounded. This verifies that for all $t>0$, the solution $\varPsi^g_{R,N}(t,\Cdot)$ takes values in $L^3_2(\bS_R)$. To proceed, we show that $\varPsi^g_{R,N}\in C((0,\infty), L^3_2(\bS_R))$. To this end, fix $t>0$ and observe that the limit
\begin{equation}
\lim_{\Delta t\to 0}\,\|\re^{-t\,Q_{R,N}} (\re^{-\Delta t\,Q_{R,N}} - 1)\, \psi_{R,N}^g\|_{L^3_2(\bS_R)} = 0\,.
\end{equation}
It follows from the semigroup property, the strong continuity of the semigroup $\re^{-t\,Q_{R,N}}$ and the fact that $\|Q_{R,N}^{1/2}\,\re^{-tQ_{R,N}}\|_{L_2(\bS_R)\to L_2(\bS_R)}$ is bounded. Next we aim to show that
\begin{multline}\label{eq:Psi_2_continuity}
\lim_{\Delta t\to 0}\|\,\int_0^{t}\,\re^{-(t-s)\,Q_{R,N}}\,(\re^{-\Delta t\, Q_{R,N}} - 1)\, F(s,\Cdot)\,\rd s
\\
 + \int_t^{t+\Delta t}\,\re^{-(t+\Delta t-s)\,Q_{R,N}}\, F(s,\Cdot)\,\rd s\|_{L^3_2(\bS_R)}
\end{multline}
is zero. To verify this claim observe that for $s\neq t$
\begin{equation}\label{eq:Psi_2_continuity_2}
\|Q_{R,N}^{1/2}\,\re^{-(t-s)\,Q_{R,N}}\,(\re^{-\Delta t\, Q_{R,N}} - 1)\, F(s,\Cdot)\|_{L_2(\bS_R)}
\leq \frac{C}{(t-s)^{1/2}}
\end{equation}
is integrable, where $C\in (0,\infty)$. Moreover, for any fixed $s\neq t $ the following limit
\begin{equation}\label{eq:Psi_2_continuity_3}
\lim_{\Delta t\to 0}\,\|Q_{R,N}^{1/2}\,\re^{-(t-s)\,Q_{R,N}}\,(\re^{-\Delta t\, Q_{R,N}} - 1)\,F(s,\Cdot)\|_{L_2(\bS_R)}
\end{equation}
is zero, which gives the point-wise convergence of the integral. Combining Eqs.~\eqref{eq:Psi_2_continuity_2} and \eqref{eq:Psi_2_continuity_3} from the dominated convergence theorem, one concludes that the first integral in Eq.~\eqref{eq:Psi_2_continuity} tends to zero. Now consider the second integral in Eq.~\eqref{eq:Psi_2_continuity}. One has
\begin{multline}
\|\int_t^{t+\Delta t}\re^{-(t+\Delta t-s)\,Q_{R,N}}\,F(s,\Cdot)\,\rd s\|_{L^3_2(\bS_R)}
\\
= \|\int_0^{\infty}\,\chi_{[0,\Delta t]}(s)\,\re^{-s\,Q_{R,N}}\,F(t+\Delta t -s,\Cdot)\, \rd s\|_{L^3_2(\bS_R)}\,.
\end{multline}
Note that
\begin{equation}\label{eq:Psi_2_continuity_21}
\|\chi_{[0,\Delta t]}(s)\,Q_{R,N}^{1/2}\,\re^{-s\,Q_{R,N}}\,F(t+\Delta t -s,\Cdot)\|_{L_2(\bS_R)}\leq C\,\chi_{[0,1]}(s)\,s^{-1/2}.
\end{equation}
Furthermore,
\begin{equation}\label{eq:Psi_2_continuity_22}
\lim_{\Delta t\to 0}\,\|\chi_{[0,\Delta t]}(s)\,Q_{R,N}^{1/2}\,\re^{-s\,Q_{R,N}}\,F(t+\Delta t -s,\Cdot)\|_{L_2(\bS_R)}= 0\,.
\end{equation}
Combining Eqs.~\eqref{eq:Psi_2_continuity_21} and \eqref{eq:Psi_2_continuity_22} from the dominated convergence theorem, one verifies that the second integral in Eq.~\eqref{eq:Psi_2_continuity} is zero as well. This finishes the proof that $\varPsi^g_{R,N}\in C((0,\infty), L^3_2(\bS_R)$. To proceed, we shall show that $\varPsi^g_{R,N}(t,\Cdot)\in C^1((0,\infty),L^{-3}_2(\bS_R))$. Using the fact that the desired derivative of $\re^{-t\,Q_{R,N}}\,\psi_{R,N}^g$ is $-Q_{R,N}\,\re^{-t\,Q_{R,N}}\,\psi^g_{R,N}$, we aim to verify that
\begin{equation}
\lim_{\Delta t\to 0}\, \|\Big(\frac{(\re^{-\Delta t\,Q_{R,N}}- 1)}{\Delta t} + Q_{R,N} \Big)\,\re^{-t\,Q_{R,N}}\,\psi^g_{R,N}\|_{L_2^{-3}(\bS_R)}=0.
\end{equation}
One rewrites the above expression
\begin{equation}\label{eq:Psi_3_continuity_1}
(\Delta t)\,\int_0^1\,\|Q_{R,N}^{-1/2}\, Q^2_{R,N}\,\re^{-t\,Q_{R,N}}\,\re^{-s\,\Delta t}(1-s)\,\psi^g_{R,N}\,\,\rd s\|_{L_2(\bS_R)}\,,
\end{equation}
where we used Remark~\ref{rem:taylor_semigroup} with $n=1$ to have
\begin{equation}
\re^{-\Delta t\,Q_{R,N}} = 1-\Delta t Q_{R,N} +Q_{R,N}^2(\Delta t)^2\,\int_0^1\re^{-s\,\Delta t\,Q_{R,N}}(1-s)\,\rd s\,.
\end{equation}
Observe that due to the positivity of the spectrum of the elliptic operator $Q_{R,N}$ there is $C\in(0,\infty)$ independent of $\Delta t$ such that 
\begin{equation}\label{eq:Psi_3_continuity_2}
\|\int_0^1\,\re^{-s\,\Delta t Q_{R,N}}(1-s)\,\rd s\|_{L_2(\bS_R)\to L_2(\bS_R)} \leq |\int_0^1(1-s)\,\rd s|\leq C\,.
\end{equation}
Note that $\|Q_{R,N}^{-1/2}\, Q^2_{R,N}\,\re^{-tQ_{R,N}} \|_{L_2(\bS_R)\to L_2(\bS_R)}$ is bounded. Thus, on account of Eq.~\eqref{eq:Psi_3_continuity_2}, and the dominated convergence, it follows that Eq.~\eqref{eq:Psi_3_continuity_1} tends to zero as $\Delta t\to 0$, i.e., $\psi^g_{R,N}\in C^1((0,\infty),L^{-3}_2(\bS_R))$. Now we need to verify that
\begin{equation}\label{eq:derivative_psi_3_1}
\begin{gathered}
\lim_{\Delta t\to 0}\,\|\frac{1}{\Delta t}\,\int_0^{t}\,\re^{-(t-s)\,Q_{R,N}}\left( \re^{-\Delta\, t\,Q_{R,N}}-1 \right) F(s,\Cdot)\,\rd s
\\+Q_{R,N}\,\int_0^t\,\re^{-(t-s)\,Q_{R,N}}F(s,\Cdot)\,\rd s\|_{L_2^{-3}(\bS_R)}=0
\\
\lim_{\Delta t\to 0}\,\|\frac{1}{\Delta t}\,\int_t^{t+\Delta t}\,\re^{-(t+\Delta t -s)\,Q_{R,N}}\,F(s,\Cdot)\,\rd s-F(t,\Cdot)\|_{L^{-3}_2(\bS_R)}=0\,.
\end{gathered}
\end{equation}
Consider the first integral in Eq.~\eqref{eq:derivative_psi_3_1} and note that
$$\|Q_{R,N}^{-1/2}\, Q_{R,N}\,\re^{-(t-s)\,Q_{R,N}} \|_{L_2(\bS_R)\to L_2(\bS_R)}\quad \text{and}\quad \|\int_0^1\,\re^{-s\,\Delta t Q_{R,N}}\,\rd s\|_{L_2(\bS_R)\to L_2(\bS_R)}$$
are bounded. One shows that pointwise 
\begin{equation}
\|\,Q_{R,N}^{-1/2} \,Q_{R,N}\int_0^t\,\re^{-(t-s)\,Q_{R,N}} \Big(1-\int_0^1\re^{-s\,\Delta t\,Q_{R,N}}\,\rd s\Big)\,F(s,\Cdot)\|_{L_2(\bS_R)}\,,
\end{equation}
converges to zero as $\Delta t\to 0$. It follows from Remark~\ref{rem:taylor_semigroup} with $n=0$
\begin{equation}
\re^{-\Delta t\,Q_{R,N}} -1= -Q_{R,N}\Delta t\int_0^1\,\re^{-s\Delta t\,Q_{R,N}}\, \rd s\,.
\end{equation}
Hence, the first integral in Eq.~\eqref{eq:derivative_psi_3_1} tends to zero as $\Delta t\to 0$ by the dominated convergence theorem. Consider the second integral in Eq.~\eqref{eq:derivative_psi_3_1}. One has
\begin{multline}
\|\frac{1}{\Delta t}\,\int_t^{t+\Delta t}\,\re^{-(t+\Delta t -s)\,Q_{R,N}}\,F(s,\Cdot)\rd s-F(t,\Cdot)\|_{L^{-3}_2(\bS_R)}
\\
=
\,\|\frac{1}{\Delta t}\,
\int_0^{\infty}\,\chi_{[0,\Delta t]}(s)\,\re^{-s\,Q_{R,N}}\,F(t+\Delta t-s,\Cdot)\,\rd s-F(t,\Cdot)\|_{L^{-3}_2(\bS_R)}.
\end{multline}
Applying the dominated convergence theorem as we did above, and using continuity of $t\mapsto F(t, \Cdot)$, one shows that the above limit is zero. This verifies that $\varPsi^g_{R,N}\in C^1((0,\infty),L^{-3}_2(\bS_R))$ and concludes the proof.

\section{Proof of Lemma~\ref{lem:wieght_plane_laplacian}}\label{sec:proof_wieght_plane_laplacian}
Notice that in what follows from line to line $C$ might be different. There exists $C\in(0,\infty)$ such that
$$|\vec\nabla(w_R^{-1})|\leq\,(C/L)\,w_R^{-1},\qquad
|\vec\nabla(v_L)|\leq \,(C/L)\,v_L,$$
where $v_L$ is given in Def.~\ref{def:weight_plane}. Using the product rule, one obtains
$$
|\vec\nabla(w_R^{-1}v_L)| \leq\,(C/L)\,|w_R^{-1}v_L|.$$
Similarly,
$$
|\Delta(w_R^{-1}v_L)| = |(\Delta w_R^{-1})\,v_L + w_R^{-1}\,(\Delta v_L)
+2 (\vec\nabla w_R^{-1})\cdot(\vec\nabla v_L)|\leq \,(C/L)^2\,|w_R^{-1}v_L|.$$
Consider Item~(A) of Lemma~\ref{lem:wieght_plane_laplacian}. Using the integration by parts one gets
\begin{equation}
\begin{aligned}
-\langle \Psi, v_L (w^{-1}_R \Delta)\Psi \rangle_{L_2(\mathbb{R}^2)} 
&= \langle \vec\nabla \Psi (v_Lw^{-1}),\nabla\Psi \rangle + \langle \Psi \vec\nabla(v_Lw^{-1}),\vec\nabla\Psi \rangle
\\ &\geq \langle \vec\nabla \Psi ,(v_Lw^{-1})\nabla\Psi \rangle -|\langle \Psi, \vec\nabla(v_Lw^{-1})\vec\nabla\Psi \rangle |
\\ &\geq \|\vec\nabla\Psi\|^2_{L_2(w_R^{-1/2}v_L^{1/2})} -\frac{C}{L}|\langle \Psi, (v_Lw^{-1}) \vec\nabla\Psi\rangle |\,.
\end{aligned}
\end{equation}
Applying H\"older's and Young's inequalities yields
\begin{multline}
-\langle \Psi, v_L (w^{-1}_R \Delta)\Psi \rangle_{L_2(\mathbb{R}^2)}
\geq (1-\frac{C}{2L})\,\|\vec\nabla\Psi\|^2_{L_2(w_R^{-1/2}v_L^{1/2})}
-\frac{C}{2L}\|\Psi\|^2_{L_2(w_R^{-1/2}v_L^{1/2})}\,.
\end{multline}
This finishes the proof of Item~(A). To prove Item~(B), note that
\begin{multline}
\langle \Psi, v_L (-w^{-1}_R\Delta)^2 \Psi \rangle = \langle\Psi, (v_L w^{-1}_R) \Delta  w^{-1}_R\Delta\Psi \rangle
\\= \langle \Delta\Psi (v_L w^{-1}_R), w^{-1}\Delta\Psi \rangle +  \langle \Psi, \Delta (v_L w^{-1}_R), w^{-1}\Delta\Psi \rangle +  2\langle \vec\nabla\Psi \cdot\vec\nabla(v_L w^{-1}_R), w^{-1}\Delta\Psi \rangle
\\ \geq  \langle \Delta\Psi (v_L w^{-1}_R), w^{-1}\Delta\Psi \rangle  -|  \langle \Psi, \Delta (v_L w^{-1}_R), w^{-1}\Delta\Psi \rangle |- 2|\,\langle \vec\nabla\Psi \cdot\vec\nabla(v_L w^{-1}_R), w^{-1}\Delta\Psi \rangle|
\\ \geq \|\Delta\Psi\|^2_{L_2(w_R^{-1}v_L^{1/2})} -\frac{C^2}{2L^2} \|\Psi\|^2_{L_2(w_R^{-1}v_L^{1/2})}
 -\frac{C^2}{2L^2}\|\Delta\Psi\|_{L_2(w_R^{-1}v_L^{1/2})}
 \\ -\frac{C}{L} \|\vec\nabla\Psi\|^2_{L_2(w_R^{-1}v_L^{1/2})} -\frac{C}{L}\|\Delta\Psi\|^2_{L_2(w_R^{-1}v_L^{1/2})}\,.
\end{multline}
Hence,
\begin{multline}\label{eq:by_part_2}
\langle \Psi, v_L (-w^{-1}_R\Delta)^2 \Psi \rangle \geq (1-\frac{C}{L} -\frac{C^2}{2L^2})\|\Delta\Psi\|^2_{L_2(w_R^{-1}v_L^{1/2})}
\\ -\frac{C}{L}\, \|\vec\nabla\Psi\|^2_{L_2(w_R^{-1}v_L^{1/2})} -\frac{C^2}{2L^2}\,\|\Psi\|^2_{L_2(w_R^{-1}v_L^{1/2})}\,.
\end{multline}
Observe that from the H\"older and the Young inequalities one has
\begin{multline}
\|\vec\nabla\Psi\|^2_{L_2(w_R^{-1}v_L^{1/2})} = \|w_R^{-1}v_L^{1/2}\vec\nabla\Psi\|^2_{L_2} =\sum_{j=1,2}\langle  \partial_j\Psi, w_R^{-2}v_L \partial_j\Psi\rangle
\\ \leq |\sum_{j=1,2}\langle \Psi , (-\partial_j)( w_R^{-2}v_L)\,\partial_j\Psi \rangle |+|\langle \Psi, ( w_R^{-2}v_L) (-\Delta)\Psi\rangle |
\\ \leq \frac{C}{L} \|\Psi\|_{L_2(w_R^{-1}v_L^{1/2})} \|\vec\nabla\Psi\|_{L_2(w_R^{-1}v_L^{1/2})}  +\|\Psi\|_{L_2(w_R^{-1}v_L^{1/2})} \|\Delta\Psi\|_{L_2(w_R^{-1}v_L^{1/2})} 
\\
\leq  \frac{C}{2L}\Big[ \|\Psi\|^2_{L_2(w_R^{-1}v_L^{1/2})} +\|\vec\nabla\Psi\|^2_{L_2(w_R^{-1}v_L^{1/2})}\Big]  +\frac{1}{2}\Big[ \|\Psi\|^2_{L_2(w_R^{-1}v_L^{1/2}) } + \|\Delta\Psi\|^2_{L_2(w_R^{-1}v_L^{1/2})} \Big]\,. 
\end{multline}
In particular, one has
\begin{equation}
(1-\frac{C}{2L})\|\vec\nabla\Psi\|^2_{L_2(w_R^{-1}v_L^{1/2})} \leq (\frac{1}{2} +\frac{C}{2L})\, \|\Psi\|^2_{L_2(w_R^{-1}v_L^{1/2})} +1/2\,\|\Delta\Psi\|^2_{L_2(w_R^{-1}v_L^{1/2})}\,.
\end{equation}
Consequently,
\begin{equation}
\|\vec\nabla\Psi\|^2_{L_2(w_R^{-1}v_L^{1/2})} \leq \frac{(\frac{1}{2} +\frac{C}{2L})}{(1-\frac{C}{2L})}\, \|\Psi\|^2_{L_2(w_R^{-1}v_L^{1/2})} +\frac{1/2}{(1-\frac{C}{2L})}\,\|\Delta\Psi\|^2_{L_2(w_R^{-1}v_L^{1/2})}
\end{equation}
provided that $(1-\frac{C}{2L}) > 0$. Therefore, one can rewrite Eq.~\eqref{eq:by_part_2}
\begin{equation}
\begin{aligned}
\langle \Psi, v_L (-w^{-1}_R\Delta)^2 \Psi \rangle 
&\geq (1-\frac{C}{L} -\frac{C^2}{2L^2} -\frac{C}{L}\,\,\frac{1/2}{(1-\frac{C}{2L})})\|\Delta\Psi\|^2_{L_2(w_R^{-1}v_L^{1/2})} 
\\&-(\frac{C}{L}\,\, \frac{(\frac{1}{2} +\frac{C}{2L})}{(1-\frac{C}{2L})} +\frac{C^2}{2L^2})\,\|\Psi\|^2_{L_2(w_R^{-1}v_L^{1/2})}\,.
\end{aligned}
\end{equation}
This finishes the proof of Item~(B). To prove Item~(C), consider
\begin{multline}\label{eq:three_lap}
\langle \Psi, v_L (-w_R^{-1}\Delta)^3 \Psi\rangle 
= -\langle \Delta \Psi (v_L w^{-1}_R), w_R^{-1}(\Delta w_R^{-1})\Delta \Psi\rangle
\\ -\langle \Psi \Delta(v_L w_R^{-1}),w_R^{-1}(\Delta w_R^{-1})\Delta \Psi\rangle
-2 \langle \vec\nabla \Psi\cdot\vec \nabla (v_L w^{-1}_R),w_R^{-1}(\Delta w_R^{-1})\Delta \Psi\rangle\,.
 \end{multline}
Take the first term on the RHS of Eq.~\eqref{eq:three_lap}. Using integration by parts, one writes
\begin{equation}
\begin{aligned}
-\langle \Delta \Psi (v_L w^{-1}_R) w_R^{-1},\Delta w_R^{-1}\,\Delta \Psi\rangle 
&= \langle \vec\nabla\Delta \Psi (v_L w^{-1}_R) w_R^{-1},\vec\nabla w_R^{-1}\Delta \Psi\rangle 
\\& +  \langle \Delta \Psi \vec\nabla(v_L w^{-1}_R) w_R^{-1},\vec\nabla w_R^{-1}\Delta \Psi\rangle
\\  &+\langle \Delta \Psi (v_L w^{-1}_R) \vec\nabla w_R^{-1},\vec\nabla w_R^{-1}\Delta \Psi\rangle.
\end{aligned}
\end{equation}
Hence,
\begin{equation}\label{eq:frist_term_with}
\begin{aligned}
-\langle \Delta \Psi (v_L w^{-1}_R) w_R^{-1},\Delta w_R^{-1}\,\Delta \Psi\rangle 
\geq  \langle \vec\nabla\Delta \Psi (v_L w^{-1}_R) w_R^{-1},\vec\nabla [w_R^{-1}\Delta \Psi]\rangle 
\\ -  |\langle \Delta \Psi \vec\nabla(v_L w^{-1}_R) w_R^{-1},\vec\nabla [w_R^{-1}\Delta \Psi]\rangle |
\\  - | \langle \Delta \Psi (v_L w^{-1}_R)\vec \nabla w_R^{-1},\vec\nabla [w_R^{-1}\Delta \Psi]\rangle |\,.
\end{aligned}
\end{equation}
Consider the first term on the RHS of Eq.~\eqref{eq:frist_term_with}. Utilizing the H\"older and the Young inequalities one writes
\begin{multline}
 \langle \vec\nabla\Delta \Psi (v_L w^{-1}_R) w_R^{-1},\vec\nabla [w_R^{-1}\Delta \Psi] \rangle 
\\ \geq \langle \vec\nabla\Delta \Psi (v_L w^{-3}_R), \vec\nabla\Delta \Psi \rangle - | \langle \nabla\Delta \Psi (v_L w^{-1}_R) w^{-1},(\nabla w^{-1}) \Delta \Psi \rangle | 
\\ \geq \|\vec\nabla\Delta\Psi\|^2_{L_2(w_R^{-3/2}v_L^{1/2})} -\frac{C}{2L}\|\nabla\Delta\Psi\|^2_{L_2(w_R^{-3/2}v_L^{1/2})}  -\frac{C}{2L}\|\Delta\Psi\|^2_{L_2(w_R^{-3/2}v_L^{1/2})}.
\end{multline}
Similarly, one treats the second term on the RHS of Eq.~\eqref{eq:frist_term_with}. It holds that 
\begin{multline}
|\langle \Delta \Psi \vec\nabla(v_L w^{-1}_R) w_R^{-1},\vec\nabla [w_R^{-1}\Delta \Psi]\rangle |
\\=
|\langle \Delta \Psi \vec\nabla(v_L w^{-1}_R) w_R^{-1},(\vec\nabla w_R^{-1}) \Delta \Psi]\rangle | 
+
|\langle \Delta \Psi \vec\nabla(v_L w^{-1}_R) w_R^{-1},w_R^{-1}\vec\nabla\Delta \Psi]\rangle | 
\\
\leq \frac{C^2}{L^2}\|\Delta\Psi\|^2_{L_2(w_R^{-3/2}v_L^{1/2})} +\frac{C}{2L} \|\vec\nabla\Delta\Psi\|^2_{L_2(w_R^{-3/2}v_L^{1/2})} 
+\frac{C}{2L} \|\Delta\Psi\|^2_{L_2(w_R^{-3/2}v_L^{1/2})}\,.
\end{multline}
Consider the third term on the RHS of Eq.~\eqref{eq:frist_term_with}. One obtains
\begin{multline}
|\langle \Delta \Psi (v_L w^{-1}_R)\vec \nabla w_R^{-1},\vec\nabla [w_R^{-1}\Delta \Psi]\rangle |
\\
\leq  
\frac{C^2}{L^2}\|\Delta\Psi\|^2_{L_2(w_R^{-3/2}v_L^{1/2})}
+\frac{C}{2L} \|\vec\nabla\Delta\Psi\|^2_{L_2(w_R^{-3/2}v_L^{1/2})} 
+\frac{C}{2L} \|\Delta\Psi\|^2_{L_2(w_R^{-3/2}v_L^{1/2})}\,.
\end{multline}
Combining the above results, one has
\begin{multline}\label{eq:first_term}
-\langle \Delta \Psi (v_L w^{-1}_R) w^{-1},\Delta w_R^{-1}\,\Delta \Psi\rangle 
\\ \geq (1-\frac{3\,C}{2L}) \|\vec\nabla\Delta\Psi\|^2_{L_2(w_R^{-3/2}v_L^{1/2})} -
(\frac{2C^2}{L^2}+\frac{3C}{2L})\|\Delta\Psi\|^2_{L_2(w_R^{-3/2}v_L^{1/2})}\,.
\end{multline} 
Consider the second on the RHS of Eq.~\eqref{eq:three_lap}. Using the fact that
$$\Delta(fgh)= (\Delta f)gh+f(\Delta g)h + fg(\Delta h)
\\+ 2 (\vec\nabla f)\cdot (\vec\nabla g)h + 2(\vec \nabla f)\cdot g(\vec\nabla h)+ 2 f(\vec\nabla g)\cdot(\vec\nabla h)$$ one has
\begin{multline}
\langle \Psi \Delta(v_L w^{-1}_R)w_R^{-1} ,(\Delta w_R^{-1})\Delta \Psi\rangle
\\=
\langle\Delta \Psi \Delta(v_L w^{-1}_R)w_R^{-1}, w_R^{-1}\Delta \Psi\rangle
+
\langle\Psi (\Delta)^2(v_L w^{-1}_R)w_R^{-1}, w_R^{-1}\Delta \Psi\rangle
\\
+\langle \Psi \Delta(v_L w^{-1}_R)(\Delta w_R^{-1}), w_R^{-1}\Delta \Psi\rangle
+2 \langle \vec\nabla \Psi\cdot\vec\nabla \Delta(v_L w^{-1}_R)\, w_R^{-1}, w_R^{-1}\Delta \Psi\rangle
\\
+2\langle \vec\nabla \Psi\cdot \Delta(v_L w^{-1}_R)\,\vec\nabla w_R^{-1}, w_R^{-1}\Delta \Psi\rangle
+2\langle \Psi\vec\nabla \Delta(v_L w^{-1}_R)\cdot\vec\nabla\,w_R^{-1}, w_R^{-1}\Delta \Psi\rangle.
\end{multline}
Using the H\"older and the Young inequalities, it holds that
\begin{multline}
|\langle \Psi \Delta(v_L w^{-1}_R)w_R^{-1} ,(\Delta w_R^{-1})\Delta \Psi\rangle|
\\
\leq \frac{C^2}{L^2}\|\Delta\Psi\|_{L_2(w_R^{-3/2}v_L^{1/2})} +\frac{2C^4}{L^4}\|\Psi\|^2_{L_2(w_R^{-3/2}v_L^{1/2})} +\frac{2C^4}{L^4}\|\Delta\Psi\|^2_{L_2(w_R^{-3/2}v_L^{1/2})}
\\
+\frac{2C^3}{L^3}\|\vec\nabla\Psi\|^2_{L_2(w_R^{-3/2}v_L^{1/2})}+\frac{2C^3}{L^3}\|\Delta\Psi\|^2_{L_2(w_R^{-3/2}v_L^{1/2})}\,.
\end{multline}
Combining the above results, we obtain a bound on the second term on the RHS of Eq.~\eqref{eq:three_lap}.
\begin{multline}\label{eq:second_term}
-|\langle \Psi \Delta(v_L w^{-1}_R)w^{-1} ,(\Delta w^{-1})\Delta \Psi\rangle|
\geq 
-\frac{2C^4}{L^4} \|\Psi\|^2_{L_2(w_R^{-3/2}v_L^{1/2})}-\frac{2C^3}{L^3}\|\vec\nabla\Psi\|^2_{L_2(w_R^{-3/2}v_L^{1/2})}
\\ - (\frac{2C^4}{L^4}+\frac{2C^3}{L^3} + \frac{C^2}{L^2}) \|\Delta\Psi\|^2_{L_2(w_R^{-3/2}v_L^{1/2})}\,.
\end{multline}
Consider the third term on the RHS of Eq.~\eqref{eq:three_lap}. It holds that
\begin{multline}
\langle\vec \nabla \Psi\ \cdot \vec\nabla (v_L w^{-1}_R)w^{-1},(\Delta [w_R^{-1}\Delta \Psi])\rangle 
\\
=
\langle \vec \nabla\Delta \Psi\ \cdot \vec\nabla (v_L w^{-1}_R)w^{-1},w^{-1}\Delta \Psi\rangle 
+\langle\vec \nabla \Psi\ \cdot \vec \nabla\Delta (v_L w^{-1}_R)w^{-1},w^{-1}\Delta \Psi\rangle 
\\
+\langle\vec \nabla \Psi\ \cdot \vec\nabla (v_L w^{-1}_R)(\Delta w^{-1}),w^{-1}\Delta \Psi\rangle 
+2 \langle\Delta\Psi\ \Delta (v_L w^{-1}_R)w^{-1},w^{-1}\Delta \Psi\rangle 
\\
+2 \langle\vec\nabla\Psi\  \Delta (v_L w^{-1}_R)\cdot\vec\nabla w_R^{-1}, w_R^{-1}\Delta \Psi\rangle 
+2 \langle\Delta\Psi\ \vec\nabla (v_L w^{-1}_R)\cdot\vec\nabla w_R^{-1}, w_R^{-1}\Delta \Psi\rangle .
\end{multline}
Hence,
\begin{multline}
|\langle\vec \nabla \Psi\ \cdot \vec\nabla (v_L w^{-1}_R)w^{-1},(\Delta w_R^{-1})\Delta \Psi\rangle |
\leq
\frac{C}{2L}\|\vec\nabla\Delta\Psi\|^2_{L_2(w_R^{-3/2}v_L^{1/2})} 
\\
+\frac{2C^3}{L^3}\|\vec\nabla\Psi\|^2_{L_2(w_R^{-3/2}v_L^{1/2})} 
+(\frac{2C^3}{L^3}+\frac{C}{2L}+\frac{4C^2}{L^2})\|\Delta\Psi\|^2_{L_2(w_R^{-3/2}v_L^{1/2})}\,.
\end{multline}
Therefore, one obtains a bound on the third term on the RHS of Eq.~\eqref{eq:three_lap}
\begin{multline}\label{eq:third_term}
-2\langle \vec\nabla \Psi \cdot\vec\nabla (v_L w^{-1}_R)w^{-1},(\Delta w^{-1})\Delta \Psi\rangle
\geq 
-\frac{C}{L}\|\vec\nabla\Delta\Psi\|^2_{L_2(w_R^{-3/2}v_L^{1/2})} 
\\ - \frac{4C^3}{L^3} \|\vec\nabla\Psi\|^2_{L_2(w_R^{-3/2}v_L^{1/2})}-(\frac{4C^3}{L^3}+\frac{C}{L}+\frac{8C^2}{L^2})\|\Delta\Psi\|^2_{L_2(w_R^{-3/2}v_L^{1/2})} \,.
\end{multline} 
Substituting Eqs.~\eqref{eq:first_term},~\eqref{eq:second_term}, and~\eqref{eq:third_term} into Eq.~\eqref{eq:three_lap} yields
 \begin{multline}\label{eq:three_delta}
\langle \Psi, v_L (-w^{-1}_R\Delta)^3 \Psi\rangle
\\ \geq (1-\frac{5C}{2L}) \|\vec\nabla\Delta\Psi\|^2_{L_2(w_R^{-3/2}v_L^{1/2})} 
- \frac{6C^3}{L^3} \|\vec\nabla\Psi\|^2_{L_2(w_R^{-3/2}v_L^{1/2})}
\\-\frac{2C^4}{L^4} \|\Psi\|^2_{L_2(w_R^{-3/2}v_L^{1/2})}
-(\frac{5C}{2L}+\frac{11C^2}{L^2}+\frac{6C^3}{L^3}+\frac{2C^4}{L^4} )\|\Delta\Psi\|^2_{L_2(w_R^{-3/2}v_L^{1/2})}\,.
\end{multline}
We aim to write $ \|\vec\nabla\Psi\|^2_{L_2(w_R^{-3/2}v_L^{1/2})}$ and $\|\Delta\Psi\|^2_{L_2(w_R^{-3/2}v_L^{1/2})}$ in terms of $ \|\vec\nabla\Delta\Psi\|^2_{L_2(w_R^{-3/2}v_L^{1/2})} $ and $ \|\Psi\|^2_{L_2(w_R^{-3/2}v_L^{1/2})}$. Consider
 \begin{multline}
\|\Delta\Psi\|^2_{L_2(w_R^{-3/2}v_L^{1/2})} = \|w_R^{-3/2}v_L^{1/2}\,\Delta\Psi\|^2_{L_2} =  \langle \Delta\Psi, (w_R^{-3}v_L) \Delta\Psi\rangle
\\ \leq |\langle\vec \nabla \Psi, [(-\vec\nabla)(w_R^{-3}v_L)] \Delta\Psi\rangle | + |\langle \vec\nabla\Psi, (w_R^{-3}v_L) (-\vec\nabla\,\Delta\Psi)\rangle  |
\\
\leq \frac{C}{L}\,\|\vec\nabla\Psi\|_{L_2(w_R^{-3/2}v_L^{1/2})}\,\|\Delta\Psi\|_{L_2(w_R^{-3/2}v_L^{1/2})} + \|\vec\nabla\Psi\|_{L_2(w_R^{-3/2}v_L^{1/2})} \,\|\vec\nabla\Delta\Psi\|_{L_2(w_R^{-3/2}v_L^{1/2})}
\\
\leq  \frac{C}{2L}\,\|\vec\nabla\Psi\|^2_{L_2(w_R^{-3/2}v_L^{1/2})} +\frac{C}{2L}\,\|\Delta\Psi\|^2_{L_2(w_R^{-3/2}v_L^{1/2})} +\frac{1}{2} \|\vec\nabla\Psi\|_{L_2(w_R^{-3/2}v_L^{1/2})}^2 +\frac{1}{2}\|\vec\nabla\Delta\Psi\|^2_{L_2(w_R^{-3/2}v_L^{1/2})}.
\end{multline}
Hence,
 \begin{multline}\label{eq:bound_for_delta}
\|\Delta\Psi\|^2_{L_2(w_R^{-3/2}v_L^{1/2})}
\\
\leq \frac{(1/2+\frac{C}{2L})}{(1-\frac{C}{2L})}\,\|\vec\nabla\Psi\|^2_{L_2(w_R^{-3/2}v_L^{1/2})}
+\frac{\frac{1}{2}}{(1-\frac{C}{2L})}\|\vec\nabla\Delta\Psi\|^2_{L_2(w_R^{-3/2}v_L^{1/2})}\,.
\end{multline}
Observe that
\begin{multline}
\|\vec\nabla\Psi\|^2_{L_2(w_R^{-3/2}v_L^{1/2})} = \|w_R^{-3/2}v_L^{1/2}\,\vec\nabla\Psi\|^2_{L_2} =\sum_{j=1,2}\langle  \partial_j\Psi, w_R^{-3}v_L \partial_j\Psi\rangle
\\ \leq\,\sum_{j=1,2}\, \langle |\Psi |, |(-\partial_j)( w_R^{-3}v_L)|\,|\partial_j\Psi|\rangle +|\langle \Psi, ( w_R^{-3}v_L) (-\Delta)\Psi\rangle |
\\ \leq \frac{C}{L} \|\Psi\|_{L_2(w_R^{-3/2}v_L^{1/2})} \|\vec\nabla\Psi\|_{L_2(w_R^{-3/2}v_L^{1/2})}  +\|\Psi\|_{L_2(w_R^{-3/2}v_L^{1/2})} \|\Delta\Psi\|_{L_2(w_R^{-3/2}v_L^{1/2})} 
\\
\leq  \frac{C}{2L}\Big[ \|\Psi\|^2_{L_2(w_R^{-3/2}v_L^{1/2})} +\|\vec\nabla\Psi\|^2_{L_2(w_R^{-3/2}v_L^{1/2})}\Big]  +\frac{1}{2}\Big[ \|\Psi\|^2_{L_2(w_R^{-3/2}v_L^{1/2}) } + \|\Delta\Psi\|^2_{L_2(w_R^{-3/2}v_L^{1/2})} \Big]\,. 
\end{multline}
One obtains
\begin{multline}\label{eq:bound_for_grad}
\|\vec\nabla\Psi\|^2_{L_2(w_R^{-3/2}v_L^{1/2})}
\\ \leq \frac{(1/2+\frac{C}{2L})}{(1-\frac{C}{2L})} \|\Psi\|^2_{L_2(w_R^{-3/2}v_L^{1/2})}+\frac{1/2}{(1-\frac{C}{2L})} \|\Delta\Psi\|^2_{L_2(w_R^{-3/2}v_L^{1/2})}\,.
\end{multline}
Observe that the preceding expression resembles Eq.~\eqref{eq:bound_for_delta} with the same coefficients.
Combining Eqs.~\eqref{eq:bound_for_delta} and \eqref{eq:bound_for_grad}, one infers
\begin{multline}
(1-\frac{(1/2+\frac{C}{2L})}{(1-\frac{C}{2L})} )\|\nabla\Psi\|^2_{L_2(w_R^{-3/2}v_L^{1/2})}  + (1-\frac{1/2}{(1-\frac{C}{2L})} )\, \|\Delta\Psi\|^2_{L_2(w_R^{-3/2}v_L^{1/2})}
\\
\leq \frac{(1/2+\frac{C}{2L})}{(1-\frac{C}{2L})} \|\Psi\|^2_{L_2(w_R^{-3/2}v_L^{1/2})} + \frac{1/2}{(1-\frac{C}{2L})} \,\|\nabla\Delta\Psi\|^2_{L_2(w_R^{-3/2}v_L^{1/2})}\,.
\end{multline}
Substitute Eqs.~\eqref{eq:bound_for_delta} and~\eqref{eq:bound_for_grad} into Eq.~\eqref{eq:three_delta} finishes the proof.
\section{Proof of Remark~\ref{rem:translation}} \label{sec:translation}
\newcommand{\fr}{\frac}
\newcommand{\pa}{\partial}
Recall the formula
\begin{equation}
 \cS_{R,\alpha}(x_1,x_2):= 
 \frac{2(R\sin (\alpha/R) (1-(x_1^2+x_2^2)/4R^2)+x_1 \cos(\alpha/R),x_2)}{1+\cos(\alpha/R)+ (x_1^2+x_2^2)/4R^2\,(1-\cos(\alpha/R)) -(x_1/R)\,\sin(\alpha/R)}\,.
\end{equation}
On account of Remark~\ref{rem:sphere_mobius_plan} rotations on the sphere correspond to M\"obius transformations on the plane via the stereographic mapping. Specifically, let us write $w=x_1+\ri x_2$  and consider the transformation
\beqa
w'=\frac{aw+b}{cw+d}= \fr{\cos(\alpha/(2R))w+2R\sin(\alpha/(2R))  }{ -\fr{\sin(\alpha/(2R))}{2R} w+\cos(\alpha/(2R))},
\eeqa
where $a=\cos(\alpha/2R)$, $b=2R\,\sin (\alpha/2R)$, $c=-\sin(\alpha/2R)/2R$, and $d=\cos(\alpha/2R)$ such that $ad-cb=1$. Set
\begin{equation}
N(w):=a(w)+b=(ax_1+b)+\ri ax_2,
\quad 
D(w):=cw+d=(cx_1+d)+\ri cx_2\,.
\end{equation}
This implies that
\begin{equation}
w'=\frac{N(w)}{D(w)}=\frac{N(w)\overline{D(w)}}{D(w)^2}= x'_1+\ri x'_2\,,
\end{equation}
where
\begin{equation}
\begin{gathered}
x'_1=\frac{(ax_1+b)(cx_1+d)+ac x_2^2}{c^2(x_1^2+x_2^2)+2cdx_1+d^2},
\quad
x'_2=\frac{x_2}{c^2(x_1^2+x_2^2)+2cdx_1+d^2}\,.
\end{gathered}
\end{equation}
Using the facts that $2\sin(\alpha/2R)\cos(\alpha/2R)= \sin(\alpha/R)$ and $\cos(\alpha/2R)^2-\sin(\alpha/2R)^2=\cos(\alpha/R)$, one gets
\begin{equation}
\begin{gathered}
(ax_1+b)(cx_1+d)+ac x_2^2= R\,\sin(\alpha/R)(1-\frac{1}{4R^2}\,(x_1^2+x_2^2))+ x_1\cos(\alpha/R),
\\
c^2(x_1^2+x_2^2)+2cdx_1+d^2=\frac{1}{4R^2}\sin(\alpha/2R)^2(x_1^2+x_2^2)+\cos(\alpha/2R)^2-\frac{x_1}{2R}\sin(\alpha/R).
\end{gathered}
\end{equation}
From the fact that $\cos(\alpha/2R)^2=\frac{1}{2}(1+\cos(\alpha/R))$ and $\sin(\alpha/2R)^2=\frac{1}{2}(1-\cos(\alpha/R))$, one rewrites the second relation 
\begin{equation}
\frac{1}{2}\big( (1-\cos(\alpha/R)\,\frac{1}{4R^2}(x_1^2+x_2^2)+(1+\cos(\alpha/R))-\frac{x_1}{R}\sin(\alpha/R)\big).
\end{equation}
This verifies that $(x_1', x_2')= \cS_{R,\alpha}(x_1,x_2)$. We write $N(w)=:c_1 w+ c_{\al}$, $D(w)=:c_0 w+c_1$, where the indices point to  the limits to which the coefficients tend as $R\to \infty$. Clearly, we have
\beqa
c_1=1+O(1/R^2), \quad c_{\al}=\al+O(1/R^2), \quad c_0=O(1/R^2).
\eeqa
Then, under the assumptions on the ranges of the parameter, stated in Remark~\ref{rem:translation},
\beqa
\fr{N(w)}{D(w)}=\fr{w+\alpha+O(1/R^2)}{1+O(1/R^2)}=w+\alpha+O(1/R^2),
\eeqa
which covers the case $a=0$ from this remark. Let us move on to the derivatives. Consider the expression
\beqa
F(w):=\fr{N(w)}{D(w)}-w-\alpha.
\eeqa
Let $\pa$ denote the derivative w.r.t. $x_1$ or $\ri x_2$.  We have
\begin{align}
\pa F(w)&=\fr{\pa N(w)}{D(w)}-\fr{(\pa D(w)) N(w)}{D(w)^2}-1\\
&= \fr{c_1}{c_0 w+c_1}
-\fr{c_0 (c_1 w+ c_{\al}) }{(c_0 w+c_1)^2}-1=\fr{c_0w}{c_0 w+c_1}-\fr{c_0 (c_1 w+ c_{\al}) }{(c_0 w+c_1)^2}
= O(1/R^2).
\end{align}
due to the factor $c_0$. We note that if we apply a second derivative to $F$, we will obtain two types of terms: First, those
involving $\pa D(w)$, hence $O(1/R^2)$. Second, those involving second derivatives of $N(w)$ or $D(w)$, hence zero. This reasoning extends to higher derivatives.

\afterpage{\blankpage}


\begin{thebibliography}{99}


\bibitem{Ai82}
M. Aizenman,  \emph{Geometric analysis of $\Phi^4$ fields and Ising models, I, II}, Commun. Math. Phys. {\bf(86)}, 1--48, (1982).


\bibitem{AD21}
M. Aizenman, H. Duminil-Copin, \emph{Marginal triviality of the scaling limits of critical 4D Ising and $\phi^4$ models}, Ann. of Math. (2) {\bf 194}, 163--235, (2021).


\bibitem{And13}
C.C. Anderson, \emph{Defining physics at imaginary time: Reflection positivity
for certain Riemannian manifolds}, (2013).




\bibitem{Aub98} T. Aubin, \emph{Some Nonlinear Problems in Riemannian Geometry}, Springer-Verlag Berlin Heidelberg, (1998).



\bibitem{AK17} S. Albeverio, S. Kusuoka,
\emph{The invariant measure and the flow associated to the $\Phi^4_3$ quantum field model},
Ann. Sc. Norm. Super. Pisa Cl. Sci., (5), 20, 1359--1427, 
\href{https://arxiv.org/pdf/1711.07108.pdf}{[arXiv:1711.07108]}.

\bibitem{AK21}  S. Albeverio, S. Kusuoka,
\emph{Construction of a non-Gaussian and rotation-invariant $\Phi^4$ measure and associated flow on $\bR^3$ through stochastic quantization},
\href{https://arxiv.org/pdf/2102.08040.pdf}{[arXiv:2102.08040]}.

\bibitem{ABRSS18} F. Arici, D. Becker, C. Ripken, F. Saueressig, W.D. van Suijlekom, \emph{Reflection positivity in higher derivative scalar theories}, J. Math. Phys. \textbf{59}, 082302, (2018).

\bibitem{Arfken}
G. B. Arfken, H. J. Weber, F. E. Harris, \emph{Mathematical Methods For Physicists}, A Comprehensive Guide, Seventh Edition, (2013).

\bibitem{AH12}
K. Atkinson, W. Han, \emph{Spherical harmonics and approximations on the unit sphere: an introduction}, Springer, (2012).





\bibitem{BCD11}
H. Bahouri, J-Y. Chemin, R. Danchin,
\emph{Fourier Analysis and Nonlinear Partial Differential Equations},
 Springer, (2011).


\bibitem{BarashkovSinh} N. Barashkov, F. C. De Vecchi,
\emph{Elliptic stochastic quantization of Sinh-Gordon QFT},
\href{https://arxiv.org/pdf/2108.12664.pdf}{[arXiv:2108.12664]}.


\bibitem{Bai23}
I. Bailleul, 
\emph{Uniqueness of the $\Phi_3^4$ measures on closed Riemannian 3-manifolds},
\href{https://arxiv.org/pdf/2306.07616}{[arXiv:2306.07616v2]}.


\bibitem{BDFLL}
I. Bailleul, N.V. Dang, L. Ferdinand, G. Leclerc, J. Lin,
\emph{Spectrally cut-off GFF, regularized $\Phi^4$ measure, and reflection positivity},
\href{https://arxiv.org/abs/2312.15511}{[arXiv:2312.15511]}.

\bibitem{BDFT1}
I. Bailleul, N.V. Dang, L. Ferdinand, T.D. T$\hat{o}$, 
\emph{Global harmonic analysis for $\Phi_3^4$ on closed Riemannian manifolds},
\href{https://arxiv.org/pdf/2306.07757}{[arXiv:2306.07757v3]}.



\bibitem{BDFT2}
I. Bailleul, N.V. Dang, L. Ferdinand, T.D. T$\hat{o}$, 
\emph{ $\Phi_3^4$ measures on compact Riemannian 3-manifolds},
\href{https://arxiv.org/pdf/2304.10185}{[arXiv:2304.1018]}.



\bibitem{BDH25}
R. Bauerschmidt, B. Dagallier, H. Weber, \emph{Holley–Stroock uniqueness method for the $\varphi^4_2$ dynamics}, \href{https://arxiv.org/abs/2504.08606}{[arXiv:2504.08606v1]}.

\bibitem{BZ81}
L. Baulieu, D. Zwanziger, \emph{Equivalence of stochastic quantization and the Faddeev-Popov Ansatz},
Nuclear Physics B,
{\bf193}(1), 163--172, (1981).

\bibitem{Berger1}
M. Berger, \emph{Geometry I}, translated from the French by M. Cole and S. Levy,
 Springer Berlin Heidelberg, (1987).
 
 \bibitem{Berger2}
M. Berger, \emph{Geometry II}, translated from the French by M. Cole and S. Levy,
 Springer Berlin Heidelberg, (1987).

\bibitem{Be01} W. Beckner, \emph{Sharp Sobolev inequalities on the sphere and the Moser-Trudinger
inequality}. Annals of Mathematics, \textbf{138}, 213--242, (1993).



\bibitem{BM01} 
H. Brezis, P. Mironescu, \emph{Gagliardo-Nirenberg, composition and products in fractional Sobolev spaces}, Journal of Evolution Equations, {\bf 1}(4),  387--404, (2001).




\bibitem{BFS83}
D. C. Brydges, J. Fr\"ohlich, A. D. Sokal,
\emph{A new proof of the existence and nontriviality
of the continuum and quantum field theories},
Commun. Math. Phys. {\bf 91}, 141--186, (1983).

\bibitem{Boga07}
V. I. Bogachev, \href{https://link.springer.com/book/10.1007/978-3-540-34514-5}{\emph{Measure Theory}}, Springer Berlin, Heidelberg, (2007).







\bibitem{Ca03}
S. M. Carroll,
\emph{Spacetime and Geometry: An Introduction to General Relativity}, Addison-Wesley, (2003).

\bibitem{CW17}
A. Chandra, H. Weber,
\href{https://www.numdam.org/article/AFST_2017_6_26_4_847_0.pdf}{
\emph{Stochastic PDEs, regularity structures and interacting particle systems}}, (2017).



\bibitem{DaPrato06}
G. Da Prato, \emph{An introduction to infinite dimensional-analysis}, Springer-Verlag Berlin Heidelberg, (2006).

\bibitem{daprato2003}
G. Da Prato, A. Debussche,
\emph{Strong solutions to the stochastic quantization equations},
Ann. Probab. {\bf 31}(4), 1900--1916, (2003).

\bibitem{DaZab96} 
G. Da Prato, J. Zabczyk, 
\emph{Ergodicity for infinite dimensional systems}, 
Cambridge University Press, (1996).

\bibitem{DaZab14} 
G. Da Prato, J. Zabczyk, 
\emph{Stochastic equations in infinite dimensions}
Cambridge University Press, Second Edition, (2014).

\bibitem{DX13}
F. Dai, Y. Xu, \emph{Approximation Theory and Harmonic Analysis on Spheres and Balls}, Springer New York, (2013).


\bibitem{DGR24}
J. Dereziński, C. Ga{\ss}, B. Ruba, \emph{Point potentials on Euclidean space,
hyperbolic space and sphere in any dimension},  Ann. Henri Poincar\'e {\bf 26}, 3477–3531 (2025). \href{https://doi.org/10.1007/s00023-024-01496-1}{https://doi.org/10.1007/s00023-024-01496-1}.

\bibitem{DiGenova} 
G. F. De Angelis, D. de Falco, G. Di Genova, \emph{Random fields on Riemannian manifolds: a constructive approach}, Commun. Math. Phys. {\bf 103}, 297-303, (1986).

\bibitem{dimock2004}
J. Dimock,
\emph{Markov quantum fields on a manifold},
Rev. Math. Phys. {\bf 16}(02), 243--255, \href{https://arxiv.org/pdf/math-ph/0305017.pdf}{[arXiv:math-ph/0305017]}.

\bibitem{DOR}
R. Dijkgraaf, D. Orlando, S. Reffert, \emph{Relating field theories via stochastic quantization},  \href{https://arxiv.org/abs/0903.0732}{[arXiv:0903.0732v2]}.



\bibitem{DH87}
P. H. Damgaard, H. H\"uffel, \emph{Stochastic quantization}, Physics Reports (Review Section of Physics Letters). {\bf 152}, 5 \& 6, 227--398, (1987).

\bibitem{Du21}
P. Duch, \emph{Flow equation approach to singular stochastic PDEs}, Probability and Mathematical Physics {\bf 6}, 327-437, (2025).


\bibitem{DDJ} 
P. Duch, W. Dybalski, A. Jahandideh, Stochastic quantization of two-dimensional $P(\Phi)$
 quantum field theory, Ann. Henri Poincar\'e \textbf{26}, 1055–1086, \href{https://doi.org/10.1007/s00023-024-01447-w}{https://doi.org/10.1007/s00023-024-01447-w}.


\bibitem{DGR23}
P. Duch, M. Gubinelli, P. Rinaldi, \emph{Parabolic stochastic quantisation of the fractional $\Phi^4_3$ model
in the full subcritical regime}, \href{https://arxiv.org/abs/2303.18112}{[arXiv:2303.18112]}.

\bibitem{DHYZ25}
P. Duch, M. Hairer, J. Yi, and W. Zhao,
\emph{Ergodicity of infinite volume $\Phi^4_3$ at high temperature}, \href{https://arxiv.org/abs/2508.07776v1}{[arXiv:2508.07776v1]}.




\bibitem{DS26}
R. C. Dalang, M. Sanz-Sole, \emph{ Stochastic Partial Differential Equations, Space-time White
Noise and Random Fields}, \href{https://doi.org/10.1007/978-3-032-01650-8}{https://doi.org/10.1007/978-3-032-01650-8}.
\bibitem{EE79} J.-P. Eckmann and H. Epstein, \emph{Time-Ordered Products and Schwinger Functions}, Commun. Math. Phys. \textbf{64}, 95--130, (1979). 


\bibitem{EdmundsTriebel} 
D. E. Edmunds, H. Triebel, 
\emph{Function spaces, entropy numbers, differential operators},
Cambridge University Press, (1996).

\bibitem{ISEM27}
M. Egert, R. Haller, S. Monniaux, P. Tolksdorf,
\href{https://www.mathematik.tu-darmstadt.de/media/analysis/lehrmaterial_anapde/ISem_complete_lecture_notes.pdf}{\emph{Harmonic Analysis Techniques for Elliptic Operators}},
27th International Internet Seminar, (2023--2024).



\bibitem{ethier2005}
S. N. Ethier, T.G. Kurtz, T,
\emph{Markov Processes: Characterization and Convergence},
Wiley, (2005).



\bibitem{Fr82}
J. Fr\"ohlich, \emph{On the triviality of $\phi^4_d$ models and the approach to the critical point in $d\ge 4$ dimensions}, Nucl. Phys. B {\bf 200}, 281--296, (1982).


\bibitem{Fe74} 
J. Feldman, \emph{The $\phi^4_3$ field theory in a finite volume}, Commun. Math. Phys. {\bf 37}, 93--120, (1974).

 
\bibitem{FO76} 
J. S. Feldman, K. Osterwalder, \emph{The Wightman axioms and the mass gap for weakly coupled $\phi^4_3$ quantum field theories}, Ann. Phys. {\bf 97}, 80--135, (1976).



\bibitem{FL83}
E. Floratos, J. Iliopoulos, \emph{Equivalence of stochastic and canonical quantisation in perturbation theory}, Nuclear Physics B, Volume 214, Issue 3, 392--404, (1983).



\bibitem{GJS74}
J. Glimm, A. Jaffe, T. Spencer, 
\emph{The Wightman axioms and particle structure in the $P(\Phi)_2$ quantum field model},
Ann. Math., {\bf 100}, 585--632, (1974).

\bibitem{GJ85}
J. Glimm, A. Jaffe, 
\emph{Quantum Field Theory and Statistical Mechanics: Expositions},
Birkh{\"a}user, (1985).


\bibitem{GlimmJaffe} J. Glimm, A. Jaffe,
\emph{Quantum Physics, A functional integral point of view},
Springer, (1987).

\bibitem{Glim68} J. Glimm, \emph{Boson fields with the $\Phi^4$ interaction in three dimensions. Communications in Mathematical Physics}, 10:1–47, (1968).


\bibitem{GJ73}J. Glimm, A. Jaffe,
\emph{Positivity of the $\Phi^4_3$ Hamiltonian}, Fortschritte der Physik. Progress of Physics, 21, 327–376, (1973).


\bibitem{Go}
N. Goldenfeld, \href{https://doi.org/10.1201/9780429493492}{\emph{Lectures on phase transitions and the renormalization group}}, CRC Press, Taylor \& Francis Group, (2018).


\bibitem{GH18} M. Gubinelli, M. Hofmanov$\acute{a}$,
\emph{Global Solutions to Elliptic and Parabolic $\Phi^4$ Models in Euclidean Space}, Commun. Math. Phys. {\bf368}, 1201–1266, \href{https://arxiv.org/abs/1804.11253}{[arXiv:1804.11253]}.

\bibitem{GH21}
M. Gubinelli, M. Hofmanov{\'a},
\emph{A PDE construction of the Euclidean $\Phi^4_3$ QFT},
Commun. Math. Phys., {\bf 382}, 1--75,
\href{https://arxiv.org/pdf/1810.01700.pdf}{[arXiv:1810.01700]}.

\bibitem{Gub21} M. Gubinelli, \emph{Lectures on Stochastic Quantization of $\phi^4_3$ – Part II: Stochastic quantization}.



\bibitem{Gri09}
A. Grigoryan, \emph{Heat Kernel and Analysis on Manifolds}, University of Bielefeld, Bielefeld, Germany,  (2009).





\bibitem{Hai14} 
M. Hairer, \emph{A theory of regularity structures}, Inventiones mathematicae, {\bf198}(2), 269--504, March, (2014).



\bibitem{Hai23}
M. Hairer, \emph{An Introduction to Stochastic PDEs},
\href{https://arxiv.org/abs/0907.4178}{[arXiv:0907.4178v2]}.


\bibitem{Hairer-Malliavin} M. Hairer,  \href{https://hairer.org/notes/Malliavin.pdf}{\emph{Introduction to Malliavin Calculus}}, Lecture notes, (2021).

\bibitem{HairerSteele} M. Hairer, R. Steele,
\emph{The $\Phi^4_3$ measure has sub-Gaussian tails}, J. Stat. Phys., {\bf 186}, 38, 
 \href{https://arxiv.org/pdf/2102.11685.pdf}{[arXiv:2102.11685]}.


\bibitem{Ha96}
R. Haag, \emph{Local Quantum Physics}, Springer Berlin, Heidelberg, \href{https://doi.org/10.1007/978-3-642-61458-3}{https://doi.org/10.1007/978-3-642-61458-3}.




\bibitem{JP23} A. Jagannath, N. Perkowski, \emph{A simple construction of the dynamical $\Phi^4_3$ model}, Trans. Amer. Math. Soc., {\bf376}(3), 1507--1522, (2023).




\bibitem{SJ97}
S. Janson,
\emph{Gaussian Hilbert Spaces},
Cambridge Topics in Mathematics,
Cambridge University Press, (1997).



\bibitem{Ku16}
A. Kupiainen, \emph{Renormalization group and stochastic PDE's}, Ann. Henri Poincaré {\bf 17}, 497--535, (2016).






\bibitem{Le11}
T. L\'evy,
\emph{Topological quantum field theories and Markovian random fields}, Bulletin des Sciences Math\'ematiques, {\bf 135}, Issues 6–7, 629--649, (2011).


\bibitem{Le97} I. M. Lee, \emph{Riemannian manifolds}, Graduate Texts in Mathematics, Vol. 176, Springer-Verlag, New York, (1997).



\bibitem{Li58}
A. Lichnerowicz, 
\emph{G\'eom\'etrie des groupes de transformations}, Travaux et Recherches Math\'ematiques, III, (1958).


\bibitem{Li09}
P. Li, \emph{Geometric Analysis}, Cambridge University Press, \href{https://doi.org/10.1017/CBO9781139105798}{
https://doi.org/10.1017/CBO9781139105798}.






\bibitem{MS76} 
J. Magnen, R. S\'en\'eor. \emph{The infinite volume limit of the $\phi^4_3$ model}, Ann. Inst. Henri Poincar\'e {\bf 24}, 95--159, (1976).



\bibitem{MW17a}
J. C. Mourrat, H. Weber, \emph{The dynamic $\phi^4_3$ model comes down from infinity}.
Comm. Math. Phys., {\bf 356}(3), 673--753, (2017). 


\bibitem{Muscalu}
C. Muscalu, W. Schlag, \emph{Classical and Multilinear Harmonic Analysis}, 
Cambridge University Press, (2013).




\bibitem{NO18}
K. H. Neeb, G. Olafsson, \emph{Reflection Positivity-A Representation Theoretic Perspective}, \href{https://arxiv.org/abs/1802.09037v1}{[arXiv:1802.09037v1]}.

\bibitem{Nelson}
E. Nelson,
\emph{A quartic interaction in two dimensions}, 
in Mathematical Theory of Elementary Particles
(Proc. Conf., Dedham, Mass., 1965), 69--73, M.I.T. Press, (1966).


\bibitem{Nualart}
D. Nualart, 
\emph{The Malliavin calculus and related topics},
Springer-Verlag, (2006).

\bibitem{OS75}
K. Osterwalder, R. Schrader,
\emph{Axioms for Euclidean Green's functions. II},
Commun. Math. Phys., {\bf 42}, 281--305, (1975).



\bibitem{parisi1981}
G. Parisi, Y. S. Wu, 
\emph{Perturbation theory without gauge fixing},
Sci. Sinica, {\bf 24}(4), 483--496, (1981).


\bibitem{PR87}
T. Parker, S. Rosenberg, \emph{Invariants of conformal Laplacians}, J. Differential Geom. {\bf 25}(2), 199--222 (1987). 


\bibitem{PT11}
G. Peccati, M. S. Taqqu,
\emph{Wiener Chaos: Moments, Cumulants and Diagrams,
A survey with computer implementation}, (2011).


\bibitem{PR07}
C. Pr\'ev$\hat{o}$t, M. R\"ockner,
\emph{A Concise Course
on Stochastic Partial Differential Equations}, Springer-Verlag Berlin Heidelberg, (2007).



\bibitem{PZ07}
S. Peszat, J. Zabczyk, \emph{Stochastic Partial Diﬀerential Equations with L\'evy
Noise}, Encyclopedia of Mathematics, Cambridge University Press, (2007).



\bibitem{RR73}
M. Reed, L. Rosen,
\emph{Support Properties of the Free Measure
for Boson Fields},
Commun. math. Phys. {\bf 36}, 123--132, (1974).


\bibitem{RSII}
 M.~Reed, B.~Simon,
\emph{Methods of Modern Mathematical Physics, Vol.~II: Fourier Analysis, Self-Adjointness},
Academic Press, New York, (1975).


\bibitem{RSI}
 M.~Reed, B.~Simon,
\emph{Methods of Modern Mathematical Physics, Vol.~I: Functional Analysis},
Academic Press, New York, (1975).

\bibitem{Res15} S. I. Resnick, \emph{A Probability Path},
Birkh\"auser Boston,~MA, \href{https://doi.org/10.1007/978-0-8176-8409-9}{https://doi.org/10.1007/978-0-8176-8409-9}.

\bibitem{Rudin91}
W. Rudin,
\emph{Functional Analysis}, (international series in pure and applied mathematics), Second Edition, (1991).

\bibitem{Rudin87}
W. Rudin, \emph{Real and Complex Analysis}, Third Edition, McGraw-Hill Book Company, (1987).

\bibitem{RZZ17}
M. R\"ockner, R. Zhu, X. Zhu, \emph{Restricted Markov uniqueness for the stochastic quantization of $P(\Phi)_2$ and its applications}, J. Funct. Anal., {\bf 272}(10), 4263--4303, (2017).



\bibitem{YS18} Y. Sawano, \emph{Theory of Besov Spaces, Springer Singapore},
\href{https://link.springer.com/book/10.1007/978-981-13-0836-9}{https://doi.org/10.1007/978-981-13-0836-9}.





\bibitem{STV18}
H. Schlichtkrull, P. Trapa, D. A. Vogan Jr. \emph{Laplacians on spheres}, \href{https://arxiv.org/abs/1803.01267}{[arXiv:1803.01267v2]}.
\bibitem{Sch08}
M. Schottenloher, \emph{A Mathematical Introduction to Conformal Field Theory}, 
Springer-Verlag Berlin Heidelberg, (2008).


\bibitem{Sh04}
I. Shigekawa, \emph{Stochastic analysis}, Translations of Mathematical Monographs, 
Iwanami Series in Modern Mathematics, (2004).



\bibitem{SRI}
B. Simon, M. C. Reed,
\emph{Methods of Modern Mathematical Physics, Functional Analysis}, Revised, Academic Press,
  (1980).
 
 \bibitem{SRII}
B. Simon, M. C. Reed,
\emph{Fourier analysis, Self-adjointness}, Academic Press,  (1976).
 
 \bibitem{SZZ23}
H. Shen, R. Zhu, X. Zhu,
\emph{An SPDE approach to perturbation theory of $\Phi^4_2$: asymptoticity and short distance behavior}, \href{https://arxiv.org/abs/2108.11312v2}{[arXiv:2108.11312v2]}.



\bibitem{SW80}
R. Streater, A. Wightman, \emph{PCT, Spin and Statistics, and All That}, Princeton University Press,  (1980).



\bibitem{St83}
R. S. Strichartz, \emph{Analysis of the Laplacian on the complete Riemannian manifold}, Journal of Functional 
Analysis {\bf 52}, 48--79, (1983).








\bibitem{TaylorII}
M. E. Taylor, \emph{Partial Differential Equations II}, Second Edition, Springer-Verlag, 
New York, (1996).

\bibitem{TaylorIII}
M. E. Taylor, \emph{Partial Differential Equations III}, Second Edition, Springer-Verlag, New York,  (1996).


\bibitem{Triebel1} H. Triebel, 
\emph{Theory of Function Spaces},
Birkh{\"a}user,  (1983).

\bibitem{Triebel2} H. Triebel, 
\emph{Theory of Function Spaces II},
Birkh{\"a}user,  (1992).

\bibitem{Triebel3} H. Triebel,
\emph{Theory of Function Spaces III},
Birkh{\"a}user, (2006).





\bibitem{VTC87}
N.~N.~Vakhania, V.~I.~Tarieladze, S.~A.~Chobanyan,
\textit{Probability Distributions on Banach Spaces},
D.~Reidel, Dordrecht, \href{https://doi.org/10.1007/978-94-009-3873-1}{https://doi.org/10.1007/978-94-009-3873-1}.


\bibitem{Varopoulos}
N. Varopoulos, L. Saloff-Coste, T. Coulhon,
\emph{Analysis and Geometry on Groups},
Cambridge University Press, (1992).



\bibitem{Wil91}
D. Williams, \emph{Probability with martingales}, Cambridge Mathematical Textbooks, Cambridge University Press, Cambridge, (1991).








\bibitem{Zab89}
J. Zabczyk, \emph{Symmetric solutions of semilinear stochastic equations}, In: Da Prato, G., Tubaro, L. (eds) Stochastic Partial Differential Equations and Applications II, Lecture Notes in Mathematics, vol 1390. Springer, Berlin, Heidelberg, \href{https://doi.org/10.1007/BFb0083952}{https://doi.org/10.1007/BFb0083952}.



\end{thebibliography}
\end{document}